\documentclass[11pt,a4paper]{article}
\usepackage{a4wide}
\usepackage[page,header]{appendix}
\usepackage{titletoc}
\usepackage{lipsum}
\usepackage{comment}

\usepackage{indentfirst}
\usepackage{setspace}
\usepackage[font=small,  margin=2cm]{caption}
\usepackage{amsthm,amsmath,amssymb}
\usepackage{natbib}
\usepackage{multirow}
\usepackage[pdftex]{graphicx}
\usepackage{subfigure}
\usepackage{makecell}
\usepackage{booktabs}
\usepackage{array}
\usepackage{pifont}
\usepackage{dsfont}
\usepackage{url}
\usepackage{ulem}
\usepackage{algorithm}
\usepackage{algpseudocode}
\usepackage{bm}
\usepackage{smile}
\usepackage{mathtools}
\usepackage{wrapfig}
\usepackage{lipsum}
\usepackage{mathrsfs}
\usepackage{dsfont}
\usepackage{rotating} 
\usepackage{color}
\usepackage{natbib}
\usepackage{multirow}
\usepackage{apalike}
\usepackage{appendix}
\usepackage{tabularx}
\usepackage{tikz}
\usetikzlibrary{arrows.meta,positioning,decorations.pathreplacing}
\usepackage{adjustbox}

\def\reals{{\mathbb R}}
\def\P{{\mathbb P}}
\def\E{{\mathbb E}}

\def\supp{\mathop{\text{supp}\kern.2ex}}

\def\argmax{\mathop{\text{\rm arg\,max}}}

\def\R{{\mathcal{R}}}
\def\sign{\mathrm{sign}}
\def\supp{\mathop{\text{supp}}}

\newcommand{\nb}[1]{{\bf\color{blue} [#1]}}

\usepackage[usenames,dvipsnames,svgnames,table]{xcolor}
\usepackage[colorlinks=true,
            linkcolor=blue,
            urlcolor=blue,
            citecolor=blue]{hyperref}

\numberwithin{equation}{section}
\numberwithin{theorem}{section}
\numberwithin{corollary}{section}
\numberwithin{asmp}{section}
\numberwithin{definition}{section}

\begin{document}
	


\title{A General U-Statistic Framework for High-Dimensional Multiple Change-Point Analysis}

\author{
	Bin Liu\\
	Department of Statistics and Data Science, School of Management, Fudan University, China\\
	\emph{email:} \texttt{liubin0145@gmail.com}
	\and
	Yufeng Liu\\
	Department of Statistics, University of Michigan, U.S.A\\
	\emph{email:} \texttt{yufliu@umich.edu}
}

\date{}

\maketitle
\vspace{-0.4in}
\begin{abstract}

High-dimensional change-point analysis is essential in modern statistical inference. However, existing methods are often designed either for specific parameters (e.g., mean or variance) or for particular tasks (e.g., testing or estimation), making them difficult to generalize. Moreover, they typically rely on restrictive distributional assumptions, limiting their robustness to heavy-tailed data. We propose a unified framework for testing, estimating, and inferring multiple change points in high-dimensional data. Our approach leverages a two-sample U-statistic within a moving window, allowing flexible kernel function selection to accommodate structural changes in general parameters such as variance changes or robust statistics. For testing, we develop an $\ell_{\infty}$-norm-based statistic with a high-dimensional multiplier bootstrap procedure, achieving minimax-optimal power under sparse alternatives. For estimation, we construct an initial estimator for the change-point number and locations and refine it using the U-statistic Projection Refinement Algorithm (U-PRA), attaining minimax-optimal localization rates. We further derive the asymptotic distribution of refined estimators, enabling valid confidence interval construction. Extensive numerical experiments demonstrate the better performance of our method across various settings, including heavy-tailed distributions. Applications to genomic copy number variation  data highlight its practical utility. An R package implementing the proposed method, \texttt{U-PRA}, is publicly available at \url{https://github.com/liubin0145/R-codes-UPRA/}.

\end{abstract}
\setcounter{page}{0}

\section{Introduction} 


Modern statistical applications frequently involve high-dimensional data sequences, where the dimensionality substantially exceeds the sample size. Due to the complex nature of data generation processes, heterogeneity is a prevalent characteristic in such datasets. These structural changes, occurring at unknown locations within the observed data sequence, are referred to as change points.
Detecting and identifying change points is crucial in many real-world problems. In biological and biomedical studies, for instance, change-point analysis has been widely used to uncover structural shifts in gene expression profiles (\cite{zhang2010detecting}), detect abrupt transitions in neural activity patterns (\cite{safikhani2022joint}), and identify critical stages in disease progression (\cite{ma2024disentangling}). Such insights are essential for understanding complex biological processes and improving early diagnosis strategies.

To formalize the problem, let $\bX=(X_1,\cdots,X_d)^\top$ be a $d$-dimensional random vector, and let $\bX_1,\cdots,\bX_n$ be an independent sequence of observations drawn from $\bX$, where $\bX_t=(X_{t,1},\cdots,X_{t,d})^\top$; $t=1,\ldots, n$. Suppose $\btheta \in \RR^d$ represents an underlying parameter of interest, and consider the following multiple change-point model:
\vspace{-0.2in}
\begin{equation}\label{equ: cpt hypothesis}
	\vspace{-0.2in}
	\begin{array}{l}
		\btheta_1=\cdots=\btheta_{\gamma_1}\neq \btheta_{\gamma_{1}+1}=\cdots=\btheta_{\gamma_{2}}=\cdots=\btheta_{\gamma_{M_0}}\neq \btheta_{\gamma_{M_0+1}}=\cdots=\btheta_n,
	\end{array}
\end{equation}
where $0<\gamma_1<\cdots<\gamma_{M_0}<n$ are the unknown change-point locations such that $\btheta_t=\btheta^{(m)}$ for $\gamma_{m-1}<t\leq \gamma_{m}$ with $1\leq m\leq M_0+1$. For notational convenience, we define $\gamma_{0}=0$ and $\gamma_{M_0+1}=n$. Given this model, high-dimensional change-point analysis typically involves three primary research questions: (1) detecting the presence of  change points, (2) estimating multiple change points, and (3) conducting statistical inference for the estimated change-point locations.

The first problem concerns single change-point detection with \(M_0\leq 1\), which has been extensively studied for high-dimensional mean changes with \(\btheta_t=\E(\bX_t)\). Existing methods typically aggregate componentwise signals through sparse, dense, or adaptive statistics; see \cite{Jirak2015Uniform, YuChen2021, wang2019inference, Liu2020Unified, zhang2021adaptive, wang2023computationally}. See the review in \cite{liu2022high}. However, these methods are mostly mean-based and often require sub-Gaussian, sub-exponential, or other restrictive moment conditions. They may therefore become unreliable under heavy-tailed distributions or outliers.

Beyond single change-point detection, multiple change-point estimation aims to recover both the number and locations of structural breaks. Existing approaches are mainly based on segmentation or penalized optimization, including binary segmentation and its variants and the fused LASSO-type estimators \cite{Liu2020Unified, zhang2021adaptive, safikhani2022joint}. While effective in many settings, these methods may require specific signal-strength conditions, assume the existence of change points, or lack a formal testing procedure. Moreover, they are mostly designed for mean shifts and are not readily extendable to variance changes or robust statistics.


The third challenge is statistical inference for estimated change points. Confidence interval construction remains difficult because it requires optimal localization rates and tractable limiting distributions. Existing results are mostly limited to univariate settings or high-dimensional mean-change models \cite{bai1998estimating, eichinger2018mosum, chen2022InferenceBreakpointsHighdimensionalb}. These methodologies are inherently difficult to generalize to broader parameter settings and typically assume that the data follow a light-tailed distribution. Consequently, when applied to heavy-tailed data, these methods often yield unreliable confidence intervals.

As discussed above, despite substantial progress in high-dimensional change-point analysis, several important challenges remain open. Existing methods are mainly mean-based, task-specific, and dependent on restrictive tail assumptions. These limitations make them difficult to apply to general structural changes, such as variance shifts or robust distributional contrasts, especially under heavy-tailed data or outliers. To address these challenges, we propose a unified U-statistic-based framework for testing, estimating, and inferring multiple change points in high-dimensional data. Our contributions can be summarized as follows.


1. \textbf{Multiple Change-Point Testing}: We propose a moving-window two-sample U-statistic based test with \(\ell_{\infty}\)-norm aggregation for sparse high-dimensional alternatives. By choosing different kernels, the framework accommodates mean shifts, variance changes, and robust distributional contrasts. A high-dimensional multiplier bootstrap procedure is developed to approximate the null distribution while preserving the dependence induced by overlapping windows. We prove that the test controls the significance level and achieves asymptotic power one, allowing \(d\) to grow exponentially with \(n\). With a proper bandwidth choice, it attains minimax-optimal detection rates under sparse alternatives and remains applicable to heavy-tailed data through robust kernel choices.


2. \textbf{Optimal Estimation of Change-Point Locations and Numbers}: 
We construct an initial estimator for the number and locations of multiple change points in high-dimensional settings, extending the moving-window idea of \cite{eichinger2018mosum, Zhou2025BootstrapMOSUM}. The estimator consistently recovers the true number of change points and achieves near-optimal localization rates. We further propose the U-statistic Projection Refinement Algorithm (U-PRA), which refines the initial estimators by aggregating information along the estimated signal direction and attains  optimal localization rates. This procedure applies to mean changes, variance changes, and robust statistics, and performs well under heavy-tailed distributions.

3. \textbf{Confidence Interval Construction for Multiple Change Points}: 
Building on the refined estimators, we derive their asymptotic distributions within the two-sample U-statistic framework. The limiting distribution is characterized as the \(\arg\max\) of a drifted weighted Brownian motion process, which provides the basis for constructing valid confidence intervals for multiple change points. Due to the U-statistic structure, the limit involves weighted sums of independent Brownian motions and differs from existing mean-based results. Numerical studies show that the proposed confidence intervals achieve empirical coverage close to the nominal level across a range of distributional settings even in heavy-tailed settings. Finally, we conduct extensive numerical studies validating both our theoretical results and methodological effectiveness. Our analysis on genomic copy number variation  further highlights the practical utility of our proposed framework.

The rest of this paper is organized as follows. 
Section~\ref{sec: method} provides the methodology includeing multiple change points testing, estimation and inference. 
Section~\ref{sec: Theoretical guarantees} establishes the theoretical guarantees, including size and power analysis for testing, consistency of the initial estimators, optimal localization of the refined estimators, and their limiting distributions for confidence interval construction. 
Section~\ref{section: empirical study} reports numerical studies.
Section~\ref{sec: real data analysis} applies the proposed method to genomic copy number variation data. 
Section~\ref{section: summary and discussion} concludes the paper.


We end this section with some notations. Let $\cX=\{\bX_1,\ldots,\bX_n\}$. We define the $\ell_p$ norm as $\|\bv\|_p=(\sum_{j=1}^d|v_j|^p)^{1/p}$ for $\bv=(v_1,v_2,\cdots,v_d)^\top \in \mathbb{R}^d$. For $p=\infty$, $\|\bv\|_\infty=\max_{1\leq j\leq d}|v_j|$. For two real numbered sequences $a_n$ and $b_n$, we set $a_n=O(b_n)$ if there exits a constant $C$ such that $|a_n|\leq C|b_n|$ for a sufficiently large $n$; $a_n=o(b_n)$ if $a_n/b_n\rightarrow0$ as $n\rightarrow\infty$; $a_n\asymp b_n$ if there exists constants $c$ and $C$ such that $c|b_n|\leq|a_n|\leq C|b_n|$ for a sufficiently large $n$. For a sequence of random variables (r.v.s) $\{\xi_1,\xi_2,\cdots\}$, we set $\xi_n\rightarrow \xi$ if $\xi_n$ converges to $\xi$ in probability as $n\rightarrow\infty$. We also denote $\xi_n=o_p(1)$ if $\xi_n\rightarrow 0$. 
\vspace{0cm}

\section{Methodology}\label{sec: method}
\vspace{0cm}

In this section, we introduce a kernel-based local two-sample contrast for detecting 
multiple change points. Let \(h(x,y)\in\RR\) be a user-specified two-sample kernel 
satisfying the antisymmetry condition $h(y,x)=-h(x,y)$.
For \(\bx=(x_1,\ldots,x_d)\) and \(\by=(y_1,\ldots,y_d)\), we apply the kernel 
coordinatewise and write $\bh(\bx,\by)
:=
\big(h(x_1,y_1),\ldots,h(x_d,y_d)\big)^\top
\in \RR^d$ .

The choice of \(h\) determines the structural feature targeted by the procedure. 
For example, a linear kernel $h(x,y)=y-x$ targets changes in mean levels, whereas a rank-based 
kernel $h(x,y)=\text{sign}(y-x)$ can be used to construct a more robust contrast. More examples of kernel choices for different structural changes are provided in Appendix~\ref{sec: examples}.

Our first goal is to test whether the sequence contains at least one structural break. 
To this end, for a generic time point \(t\), define the kernel-induced signal jump
\begin{equation}
	\btheta_t
	:=
	\mathbb{E}\bh(\bX_t,\bX_{t+1}),
	\qquad t=1,\ldots,n-1 .
\end{equation}
This quantity characterizes the population-level contrast between two adjacent 
distributions, as measured by the kernel \(h\). Its interpretation depends on the 
choice of the kernel. For example, if \(h(x,y)=y-x\), then $\btheta_t
=\mathbb{E}(\bX_{t+1}-\bX_t)$
so that \(\btheta_t\) represents the mean difference across adjacent segments. 
If \(h(x,y)=\operatorname{sign}(y-x)\), then
$\btheta_t
=
\mathbb{E}\{\operatorname{sign}(\bX_{t+1}-\bX_t)\}$,
where the sign function is applied coordinatewise. This gives a rank-based contrast 
between adjacent distributions and is less sensitive to heavy-tailed observations or 
outliers.

Based on \(\btheta_t\), we formulate the multiple change-point testing problem as
\begin{equation} 
	\label{equ: cpt hypothesis}
	\begin{aligned}
		\Hb_0 &: \btheta_t = \mathbf{0}, \quad \forall t = 1,\ldots,n-1, \\
		\Hb_1 &: \exists~ \gamma_1< \gamma_2<\cdots <\gamma_{M_0}, 
		\text{ such that } 
		\btheta_{\gamma_m} \neq \mathbf{0}, \quad m=1,\ldots,M_0.
	\end{aligned}
\end{equation}

Under \(\Hb_0\), the observations on the two sides of any time point have the same distribution this yields \(\btheta_t=\mathbb{E}\bh(\bX_t,\bX_{t+1})=\mathbf{0}\) for all 
\(t\). Under \(\Hb_1\), there are one or more locations \(\gamma_m\) where the 
kernel-induced contrast becomes nonzero with $\mathbb{E}\bh(\bX_{\gamma_m},\bX_{\gamma_m+1})\neq \mathbf{0}$, 
indicating structural changes in the parameter captured by the chosen kernel. For the \(m\)-th change point \(\gamma_m\), we write
$\btheta^{(m)}
:=
\btheta_{\gamma_m}
=
\mathbb{E}\bh(\bX_{\gamma_m},\bX_{\gamma_m+1})\in\RR^d$
and denote its \(j\)-th component by
$\theta_j^{(m)}
:=
\mathbb{E}h(X_{\gamma_m,j},X_{\gamma_m+1,j}),
j=1,\ldots,d .$
We further define the active set
$\bPi_m
=
\{j\in\{1,\ldots,d\}:\theta_j^{(m)}\neq 0\}$
with $s_m=|\Pi_m|$  being its cardinality.
Thus, \(\bPi_m\) collects the coordinates affected by the \(m\)-th change point.

\subsection{\textbf{Moving-window U-statistic for local two-sample comparison}}
\label{section: cpt detection}

Note that the above  testing problem allows the number of change points \(M_0\)
to scale with the sample size \(n\), and the single change-point testing problem is
included as a special case when \(M_0=1\). A natural idea for change-point testing is
to use a CUSUM-type statistic, which compares the observations before and after a
candidate split point \(k\).  However, in the presence of multiple change points,
the two samples formed by \(\{1,\ldots,k\}\) and \(\{k+1,\ldots,n\}\) may each contain
observations from several different regimes. As a result, the signals from different
change points can be mixed  or even cancelled out. 

{The above difficulty motivates}  a moving-window local two-sample comparison, in the
spirit of \cite{eichinger2018mosum}. Instead of comparing the whole left part of the
sequence with the whole right part, we  compare observations around each candidate location \(k\). Specifically, for a candidate location \(k=G,\ldots,n-G\) and for each coordinate \(j=1,\ldots,d\), we compare all cross-window pairs
\((X_{t_1,j},X_{t_2,j})\), with \(t_1\in \{k-G+1,\ldots,k\}\) and \(t_2\in \{k+1,\ldots,k+G\}\), through a
user-specified  kernel \(h\), where  \(G\) is the bandwidth parameter controlling the size of the local comparison. This leads
to the following moving-window two-sample U-statistic:
\begin{equation}\label{equation: wilcoxon each coordinate}
	T_j(k)
	:=
	\frac{1}{G^{3/2}}
	\sum_{t_1=k-G+1}^{k}
	\sum_{t_2=k+1}^{k+G}
	h(X_{t_1,j},X_{t_2,j}),
	\qquad j=1,\ldots,d .
\end{equation}


The intuition behind \eqref{equation: wilcoxon each coordinate} is transparent for several common kernels. 
For the linear kernel \(h(x,y)=y-x\), we have
\[
T_j(k)
=
\frac{1}{G^{3/2}}
\sum_{t_1=k-G+1}^{k}
\sum_{t_2=k+1}^{k+G}
\left(X_{t_2,j}-X_{t_1,j}\right).
\]
Thus, \(T_j(k)\) aggregates all pairwise differences between the right and left windows. 
Under \(\Hb_0\), these differences have no systematic direction, so their accumulated contribution is small. 
Under \(\Hb_1\), when \(k\) is close to a true change point, the two windows mainly come from two neighboring regimes, and the \(G^2\) pairwise differences tend to shift in the same direction. 
Their signals therefore reinforce each other, producing a peak of \(T_j(k)\) near the change point, as illustrated in Figure~\ref{fig:Tjk-H1}. 
For the sign kernel \(h(x,y)=\operatorname{sign}(y-x)\), the statistic instead aggregates the signs of the pairwise differences and provides a rank-based contrast that is less sensitive to heavy-tailed observations or outliers.

\subsection{\textbf{Theoretical intuition for detection and localization}}


	We next explain why \(T_j(k)\) can detect and localize change points. 
	Unlike the usual mean-based CUSUM statistic, which admits a direct signal-noise decomposition, \(T_j(k)\) is a two-sample U-statistic.   {This motivate the Hoeffding's decomposition which is
		used to obtain such a signal-noise separation.} To illustrate the idea, consider a generic two-sample kernel \(h(X,Y)\), where \(X\in \RR^1\) and
	\(Y\in \RR^1\) are independent observations from two possibly different distributions. Let
	$\theta=\mathbb E\{h(X,Y)\}$.
	The Hoeffding's decomposition gives 
	\[
	h(X,Y)
	=
	\theta
	+
	h_1(X)
	+
	h_2(Y)
	+
	g(X,Y),
	\]
	where
	$h_1(X)=\mathbb E\{h(X,Y)\mid X\}-\theta$,
	$h_2(Y)=\mathbb E\{h(X,Y)\mid Y\}-\theta$,
	and $g(X,Y)=h(X,Y)-\theta-h_1(X)-h_2(Y)$.
	The terms \(h_1(X)\) and \(h_2(Y)\) are
	the first-order random fluctuations and \(g(X,Y)\) 
	is the remaining higher-order fluctuation which is usually  negligible.

	As for the moving-window statistic in
	\eqref{equation: wilcoxon each coordinate}, suppose that the bandwidth \(G\) is chosen
	such that each scanning window contains at most one change point, say \(\gamma_m\). See Figure \ref{fig:moving-window-local}. 
	\vspace{0cm}
	\begin{figure}[H]
		\centering
		\begin{tikzpicture}[
			x=0.9cm,y=0.9cm,
			>=Latex,
			every node/.style={font=\small},
			obsM/.style={circle,draw=blue!70!black,fill=blue!15,inner sep=1.5pt},
			obsP/.style={circle,draw=red!70!black,fill=red!15,inner sep=1.5pt},
			braceL/.style={decorate,decoration={brace,amplitude=4pt},blue!70!black,thick},
			braceR/.style={decorate,decoration={brace,amplitude=4pt,mirror},red!70!black,thick}
			]
			
			\draw[->,thick] (0,0) -- (13.5,0) node[right] {time};
			
			\foreach \x in {0.6,1.1,1.6,2.1,2.6,3.1,3.6,4.1,4.6,5.1,5.6,6.1}{
				\node[obsM] at (\x,0) {};
			}
			
			\foreach \x in {6.9,7.4,7.9,8.4,8.9,9.4,9.9,10.4,10.9,11.4,11.9,12.4}{
				\node[obsP] at (\x,0) {};
			}
			
			\draw[dashed,thick] (2.4,-0.55) -- (2.4,0.55);
			\draw[dashed,thick,red!75!black] (6.6,-0.65) -- (6.6,0.75);
			\draw[dashed,thick] (10.8,-0.55) -- (10.8,0.55);
			
			\node[below] at (2.4,-0.55) {$\gamma_{m-1}$};
			\node[below,red!75!black] at (6.6,-0.65) {$\gamma_m$};
			\node[below] at (10.8,-0.55) {$\gamma_{m+1}$};
			
			\draw[thick] (6.15,-0.4) -- (6.15,0.9);
			\node[above] at (6.15,0.9) {$k$};
			
			\node[blue!70!black] at (4.3,0.45) {segment \(m\)};
			\node[red!70!black] at (8.6,0.45) {segment \(m+1\)};
			
			\draw[braceL] (3.15,0.9) -- (6.15,0.9)
			node[midway,above=5pt,blue!70!black] {};
			\node[blue!70!black] at (4.65,1.45) {length \(G\)};
			
			\draw[braceR] (6.25,-0.9) -- (9.25,-0.9)
			node[midway,below=5pt,red!70!black] {};
			\node[red!70!black] at (7.75,-1.45) {length \(G\)};
			
		\end{tikzpicture}
		\caption{Local moving-window configuration around a change point. The scanning location \(k\) is close to \(\gamma_m\), with \(L_k=\{k-G+1,\ldots,k\}\) and \(R_k=\{k+1,\ldots,k+G\}\). The two windows compare observations around \(\gamma_m\) from neighboring regimes.}
		\label{fig:moving-window-local}
	\end{figure}
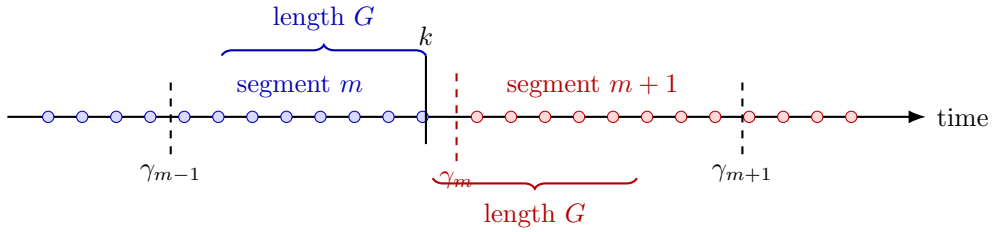
	\vspace{0cm}
	For coordinate
	\(j\), recall
	$\theta_j^{(m)}
	=
	\mathbb Eh(X_{\gamma_m,j},X_{\gamma_m+1,j})$
	is signal jump across the \(m\)-th change point. Applying Hoeffding's decomposition to the  kernel  yields a
	signal-noise representation of the moving-window statistic $T_{j}(k)$: 
	\begin{equation}\label{eq:Tj-signal-noise}
		T_j(k)
		=
		\delta_j^{(m)}(k)
		+
		\mathcal R_j^{(m)}(k), \quad  k=G+1,\ldots,n-G,
	\end{equation}
	where \(\delta_j^{(m)}(k)\) is the deterministic signal function and
	\(\mathcal R_j^{(m)}(k)\) collects the stochastic fluctuation terms generated by the
	first-order components and the higher-order residual in Hoeffding's decomposition.
	More explicitly, we have
	\[
	\delta_j^{(m)}(k)
	=
	\frac{G(k+G-\gamma_m)}{G^{3/2}}\theta_j^{(m)}
	\mathbf{1}\{\gamma_{m-1}< k\leq \gamma_m\}
	+
	\frac{G(\gamma_m-k+G)}{G^{3/2}}\theta_j^{(m)}
	\mathbf{1}\{\gamma_m< k\leq \gamma_{m+1}\}.
	\]

	Note that as \(k\) approaches \(\gamma_m\), more cross-window pairs contain observations from two different regimes, so \(|\delta_j^{(m)}(k)|\) increases and is maximized when \(k=\gamma_m\). 
	The remainder \(\mathcal R_j^{(m)}(k)\) collects stochastic fluctuations, whose explicit form is given in the Supplementary Material. 
	Under the regularity conditions, these fluctuations are uniformly controlled as $\max_{1\leq j\leq d}
	\max_{G\leq k\leq n-G}
	\left|
	\mathcal R_j^{(m)}(k)
	\right|
	=
	O_p\{\sqrt{\log(nd)}\}$.
	Therefore, under a suitable signal-to-noise condition, the deterministic signal dominates the stochastic fluctuation near the true change point. 
	This explains why the moving-window U-statistic can detect change points and provide accurate candidate locations, as also illustrated in Figure~\ref{fig:Tjk-H1}.

	The above decomposition is coordinatewise. To test whether any change point exists
	in a high-dimensional sequence, we scan over all candidate locations $k$ and aggregate
	over coordinates $j$. Since we focus on sparse alternatives, where only a subset of
	coordinates may change, we use the \(\ell_\infty\)-norm aggregation of
	\(\mathbf T(k)=(T_1(k),\ldots,T_d(k))^\top\) which has been used in \cite{Jirak2015Uniform,YuChen2021,yu2022robust}. Specifically, our  testing
	statistic is  defined as
	\begin{equation}\label{equ: final testing statistic}
		W
		=
		\max_{G\leq k\leq n-G}
		\max_{1\leq j\leq d}
		|T_j(k)|
		=
		\max_{G\leq k\leq n-G}
		\|\mathbf T(k)\|_\infty .
	\end{equation}
	Since $W$ is constructed by adopting the $\ell_\infty$ norm, it is more sensitive to sparse signals with strong perturbations on a small number of coordinates. {More importantly,  we  do  not employ the CUSUM statistic; instead, we utilize information of two-sample U statistics that constructs the testing statistic. This  ensures that the test may not use the data points directly  and may maintain high power while also being robust to heavy-tailed data and outliers.} 
	\vspace{0cm}
	\subsection{\textbf{Multiplier bootstrap for testing}}\label{sec: bootstrap}
	In the previous section, we introduced the \(\ell_{\infty}\)-norm based testing
	statistic \(W\) in \eqref{equ: final testing statistic}. To implement the test, we
	need to approximate the null distribution of \(W\). This is challenging because
	\(W\) takes the maximum over both scanning locations and coordinates. Moreover, due
	to the moving-window construction, the statistics
	\(\{\bT(k):G\leq k\leq n-G\}\) are strongly dependent across nearby scanning
	locations. One possible approach is to use an extreme-value approximation, such as a Gumbel
	limit as derived in \cite{Jirak2015Uniform}. However, such an approximation usually requires restrictive assumptions on
	the covariance structure and may converge slowly in finite samples. 
	
	{These difficulties motivate a multiplier
		bootstrap procedure}, which approximates the null distribution of \(W\) in a
	data-driven way. Bootstrap methods for \(U\)-statistics have been developed in low dimensions
	\citep{bucher2016dependent} and extended to high-dimensional settings
	\citep{chernozhukov2013gaussian,chernozhukov2017central,Jirak2015Uniform,Liu2020Unified}.
	We design a multiplier bootstrap procedure for moving-window two-sample
	\(U\)-statistics in multiple change-point testing. Let
	\(e_1^b,\ldots,e_n^b\) be i.i.d. \(N(0,1)\) random variables, independent of the
	data, for \(b=1,\ldots,B\). The \(b\)-th bootstrap version of \(T_j(k)\) is defined as
	\begin{equation}\label{equ: bootstrap based testing statistics for each coordinate}
		T_j^b(k)
		=
		\dfrac{1}{G^{3/2}}
		\sum_{t_1=k-G+1}^{k}
		\sum_{t_2=k+1}^{k+G}
		(e_{t_1}^b+e_{t_2}^b)
		h(X_{t_1,j},X_{t_2,j}),
		\qquad j=1,\ldots,d .
	\end{equation}
	
	Accordingly, the \(b\)-th bootstrap version of the final testing statistic is
	defined as
	\[
	W^b
	=
	\max_{G\leq k\leq n-G}
	\max_{1\leq j\leq d}
	|T_j^b(k)|
	=
	\max_{G\leq k\leq n-G}
	\|\bT^b(k)\|_{\infty}.
	\]
	Given the significance level \(\alpha\), let $c_{\alpha,W}
	:=
	\inf\{t\in\mathbb R:\mathbb P(W\leq t)\geq 1-\alpha\}$
	be the theoretical critical value. Based on the bootstrap statistics
	\(\{W^1,\ldots,W^B\}\), we estimate it by
	\begin{equation}\label{equ: bootstrap critical value}
		\widehat c_{\alpha,W}
		=
		\inf\left\{
		t:
		\frac{1}{B}
		\sum_{b=1}^{B}
		\ind\{W^b\leq t\}
		\geq 1-\alpha
		\right\}.
	\end{equation}
	The resulting test is defined as
	\begin{equation}\label{statistics: robust test}
		\Psi_{\alpha,W}
		=
		\ind\{W\geq \widehat c_{\alpha,W}\}.
	\end{equation}
	We reject \(\Hb_0\) if and only if \(\Psi_{\alpha,W}=1\).
	
	{The multiplier bootstrap offers several advantages}. \textbf{First}, as a result of the moving-window procedure, \(\{\bT(k), k=G,\ldots,n-G\}\) possesses
	a complex correlation structure across different scanning positions. We prove that
	the bootstrap-based test statistic can approximate this correlation structure in a
	fully data-driven manner, thereby circumventing the necessity of first estimating
	the covariance matrix and then generating the corresponding Gaussian random variables
	for approximation as used in \cite{chen2022InferenceBreakpointsHighdimensionalb}.
	\textbf{Second}, the bootstrap method is straightforward to implement, requiring no
	hyperparameter tuning or additional estimation of high-dimensional parameters.
	\textbf{Third}, when estimating multiple change points, we can utilize the
	bootstrap-based critical values as thresholds for identifying change-point intervals,
	thereby avoiding the need to manually specify a threshold as used in
	\cite{Wang2016High}. \textbf{Lastly}, since the proposed bootstrap procedure is built
	upon the general kernel \(h\), it is not limited to mean changes, but is also
	applicable to variance changes and robust location changes, highlighting its
	flexibility and broad applicability.
	\vspace{0cm}
	\subsection{\textbf{Initial estimation of multiple change points}}
	After introducing the testing statistic for Problem~\ref{equ: cpt hypothesis}, we next
	present our methodology for estimating the unknown number and locations of change
	points and for constructing their confidence intervals. The estimation procedure
	comes from the  moving-window statistic used in the testing step. As discussed
	above, \(\|\bT(k)\|_\infty\) tends to form local peaks around true change points,
	while it remains below the bootstrap critical value away from change points. Hence,
	the bootstrap threshold naturally separates candidate change-point regions from
	noise-dominated regions. See Figure~\ref{fig:Tjk-initial-estimation} (b) for an direct illustration. 
	
	Recall $\bT(k)=(T_1(k),\ldots,T_{d}(k))^\top$ with $T_{j}(k)$ defined in (\ref{equation: wilcoxon each coordinate}).  Consider all pairs of indices $(v_j,w_j)\subset\{G,\ldots,n-G\}$ such that 
	\begin{equation}\label{equ: multiple cpt principle}
		\begin{array}{ll}
			(1)\|\bT(k)\|_{\infty}\geq \hat{c}_{\alpha,W}, ~\text{for}~k=v_j,\ldots,w_j;
			(2) \|\bT(k)\|_{\infty}<\hat{c}_{\alpha,W}, ~\text{for}~k=v_j-1,w_j+1,\\
			(3)w_j-v_j\geq \eta G,\text{for some fixed but arbitary  }~0<\eta<1/2,
		\end{array}
	\end{equation}
	where $\hat{c}_{\alpha,W}$ is the $1-\alpha$ quantile level obtained  from the  bootstrap procedure in Section \ref{sec: bootstrap}. Condition \ref{equ: multiple cpt principle} is necessary to avoid overestimation by spurious local maxima exceeding the critical value on the boundary between significant and insignificant areas. Based on  (\ref{equ: multiple cpt principle}), our initial estimation for change point number and locations are defined as:
	\begin{equation}\label{equation: intial number}
		\hat{M}_0=\text{number of pairs } (v_j,w_j),~~\text{and}~\hat{\gamma}_{m}=\argmax_{v_m\leq k\leq w_m}\|\bT(k)\|_{\infty}~\text{for}~m=1,\ldots,\hat{M}_0.
	\end{equation}
	Note that the principle in (\ref{equ: multiple cpt principle}) was proposed in \cite{eichinger2018mosum} for one-dimensional sequences and recently modified by \cite{chen2022InferenceBreakpointsHighdimensionalb}. Differently, we consider the high dimensional setting using the $\ell_\infty$-norm aggregations for the two-sample U-statistic-based moving window process. 
	\vspace{0cm}
	
	\begin{figure}[H]
		\centering
		
		\begin{minipage}{\textwidth}
			\centering
			\includegraphics[width=\textwidth]{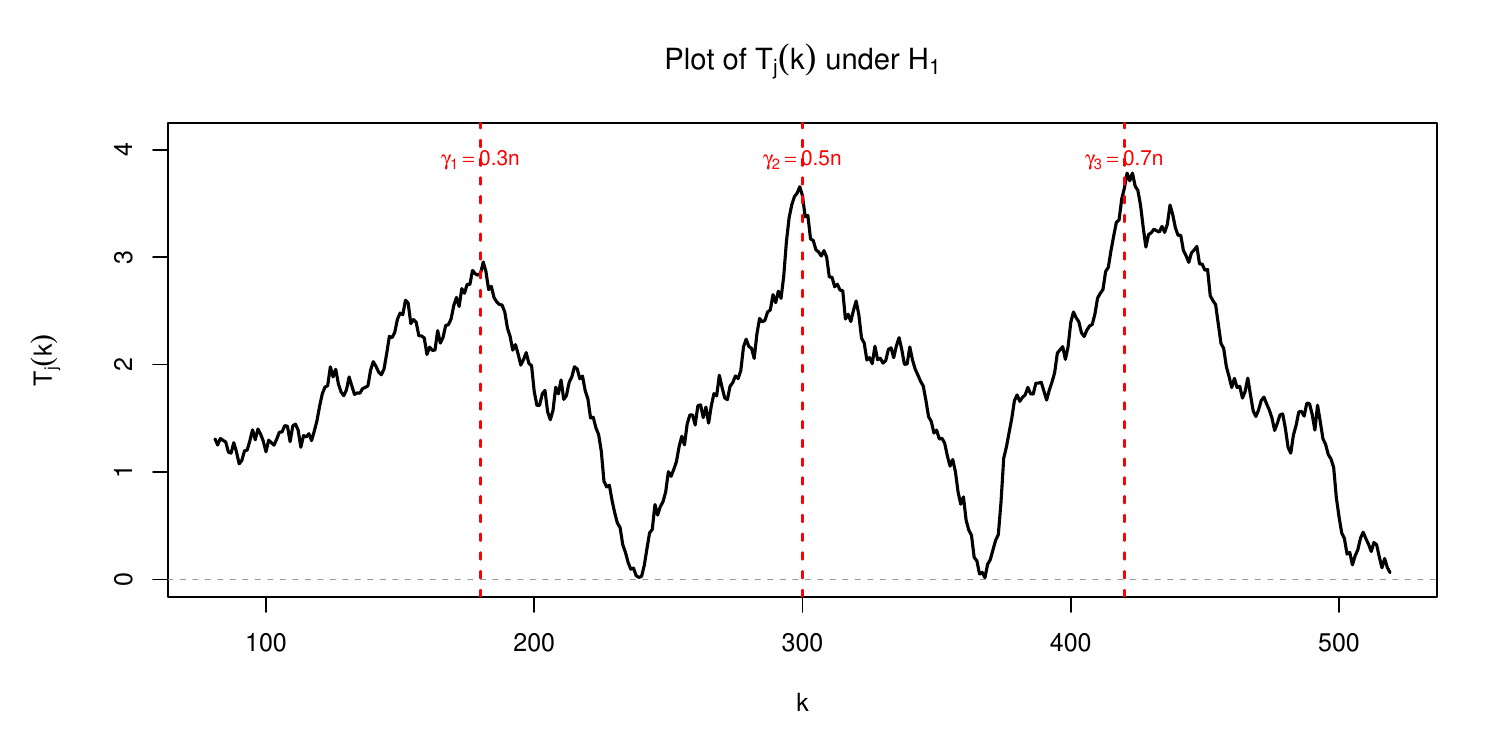}
			\caption*{(a) Coordinatewise statistic \(T_j(k)\)}
		\end{minipage}
		
		\vspace{0.5em}
		
		\begin{minipage}{\textwidth}
			\centering
			\includegraphics[width=\textwidth]{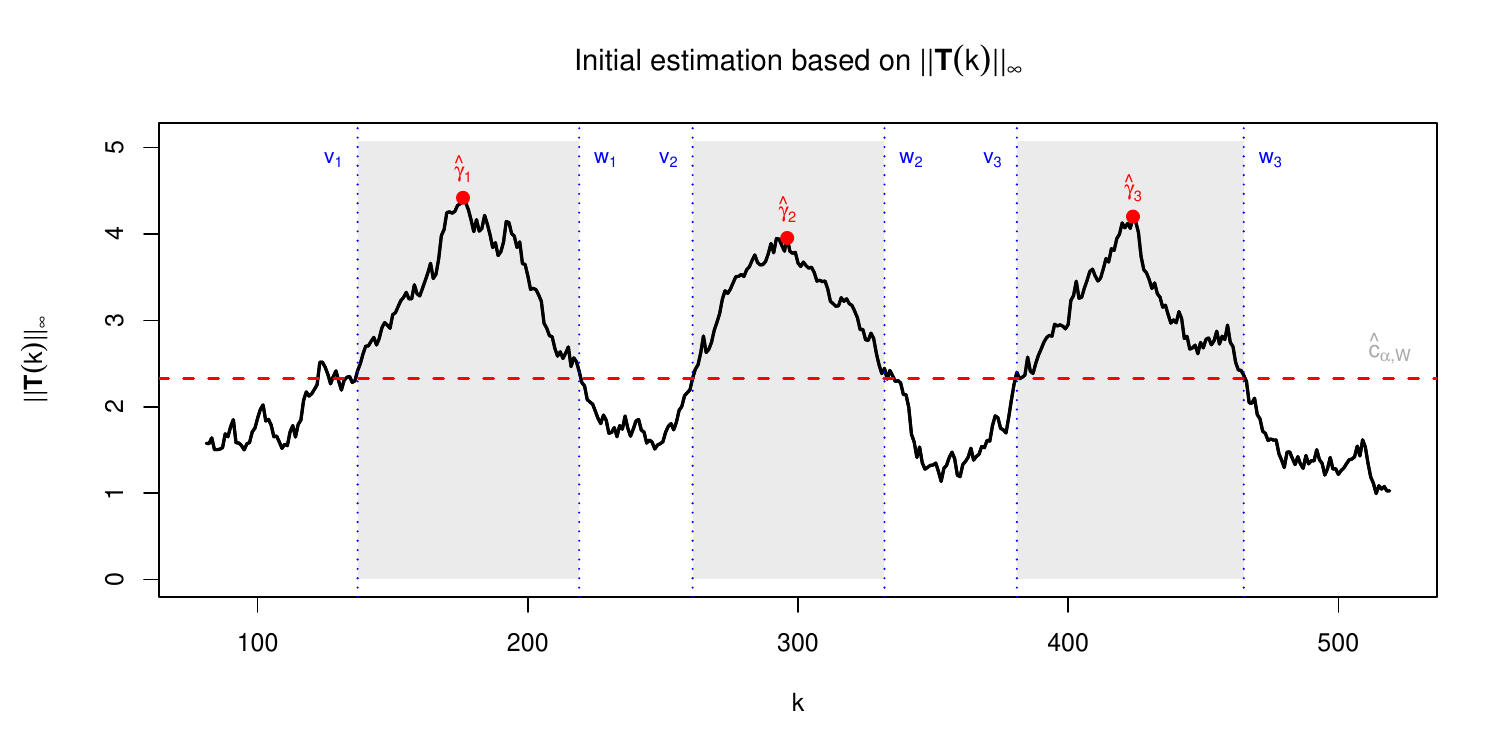}
			\caption*{(b) Initial estimation using \(\|\mathbf T(k)\|_\infty\)}
			\label{fig:initial-estimation}
		\end{minipage}
		
		\caption{
			Illustration of the moving-window statistic and the initial estimation procedure.
			Panel (a) plots \(T_j(k)\) for one coordinate under the simulation setting with three true change points and the linear kernel \(h(x,y)=y-x\). 
			The statistic forms clear peaks around the true change-point locations. 
			Panel (b) plots the aggregated statistic \(\|\mathbf T(k)\|_\infty\) over scanning locations. 
			The horizontal dashed line represents the bootstrap critical value \(\widehat c_{\alpha,W}\), and the above-threshold intervals \([v_j,w_j]\) are used as candidate change-point regions. 
			The maximizer within each interval gives the initial estimator \(\widehat\gamma_j\).
		}	\label{fig:Tjk-H1}
	\end{figure}

	\vspace{0cm}
	\subsection{\textbf{Projection-based refinement}}\label{method: refinement}
	The initial estimators obtained in the previous subsection provide candidate
	locations of the change points. However, they are mainly designed for detection and
	coarse localization.  Specifically,  the statistic \(\|\bT(k)\|_\infty\) is driven by the
	largest coordinatewise signal, it may not fully use the information contained in
	all coordinates with change points. {This motivates the refined procedure,}
	where we refine the initial estimators by aggregating the  coordinates along the signal direction.
	
	We first explain the idea at an oracle level. Suppose  that, for the
	\(m\)-th change point, the active set $\bPi_m=\{j:\theta_j^{(m)}\neq 0\}$ and the signal direction
	\(\btheta^{(m)}=(\theta_1^{(m)},\ldots,\theta_d^{(m)})^\top\) are known, where $\theta_j^{(m)}
	=
	\mathbb Eh(X_{\gamma_m,j},X_{\gamma_m+1,j})$. Then a
	natural way to combine the coordinatewise moving-window statistics is to project
	\(\bT(k)\) onto \(\btheta^{(m)}\), namely $\sum_{j\in\bPi_m}\theta_j^{(m)} T_j(k)$.
	
	As a result,
	signals from different active coordinates reinforce each other, while inactive
	coordinates are excluded. The intuition is particularly clear for the linear kernel \(h(x,y)=y-x\). In this
	case, \(\theta_j^{(m)}\) is the mean shift of the \(j\)-th coordinate across
	\(\gamma_m\). If \(k\) is close to \(\gamma_m\), then the local contrast
	\(T_j(k)\) has expectation proportional to \(\theta_j^{(m)}\). Therefore, after
	multiplying by \(\theta_j^{(m)}\), the expected contribution of the \(j\)-th active
	coordinate becomes proportional to \((\theta_j^{(m)})^2\). Summing over \(j\in\bPi_m\) accumulates the total signal jump with
	\begin{equation*}
		\sum_{j\in\bPi_m}\theta_j^{(m)}T_j(k)
		\;\propto\;
		\sum_{j\in\bPi_m}\bigl(\theta_j^{(m)}\bigr)^2
		=
		\|\btheta^{(m)}\|_2^2.
	\end{equation*}

	In contrast, the initial \(\ell_\infty\)-based statistic mainly uses the strongest
	single coordinate and is driven by \(\|\btheta^{(m)}\|_\infty\). Thus, projection
	can substantially improve the effective signal-to-noise ratio when several
	coordinates carry change-point information. This idea is related to projection-based change-point methods
	\citep{Wang2016High,Aston2018Change,chen2022InferenceBreakpointsHighdimensionalb}.
	However, existing approaches mainly rely on sample means to determine projection
	directions. In contrast, our projection direction is constructed from the general
	two-sample kernel \(h(x,y)\), so the refinement step inherits the flexibility of the
	proposed U-statistic framework and can be applied beyond mean changes.

	Note that the oracle projection
	requires the active set \(\bPi_m\) and the signal direction \(\btheta^{(m)}\), both
	of which are unknown in practice. We therefore estimate them locally around the
	initial change-point estimator \(\widehat\gamma_m\). Let \(\widehat\gamma_0=1\) and
	\(\widehat\gamma_{\widehat M_0+1}=n\). For the \(m\)-th detected change point
	\(\widehat\gamma_m\), we estimate the segment lengths on its two sides by
	$	D_m^-=\widehat\gamma_m-\widehat\gamma_{m-1}$ and $D_m^+=\widehat\gamma_{m+1}-\widehat\gamma_m $. 
	Given a trimming parameter \(\rho=0.1\), define
	$ s_{1,m}
	=
	\left\lfloor
	\widehat\gamma_m-(1-\rho)D_m^-
	\right\rfloor$, 
	$e_{1,m}
	=
	\left\lfloor
	\widehat\gamma_m-\rho D_m^-
	\right\rfloor$, 
	$s_{2,m}
	=
	\left\lceil
	\widehat\gamma_m+\rho D_m^+
	\right\rceil$ and 
	$e_{2,m}
	=
	\left\lceil
	\widehat\gamma_m+(1-\rho)D_m^+
	\right\rceil $. 
	The corresponding local neighborhoods are defined as 
	$\mathcal L_m=\{s_{1,m},\ldots,e_{1,m}\}$ and  $\mathcal R_m=\{s_{2,m},\ldots,e_{2,m}\}$. See Figure	\ref{fig:local-neighborhood-refine} for an illustration.  Since the initial
	estimator may have a nonzero localization error, observations too close to
	\(\widehat\gamma_m\) may come from either side of the true change point. Removing
	this transition region helps ensure that \(\mathcal L_m\) and \(\mathcal R_m\)
	mainly contain observations from the two adjacent regimes.
	
	For each coordinate \(j=1,\ldots,d\), we estimate the kernel based signal jump by
	\begin{equation}\label{equation: signal jump estimation}
		\widehat\theta_j^{(m)}
		=
		\frac{1}{|\mathcal L_m||\mathcal R_m|}
		\sum_{t_1\in\mathcal L_m}
		\sum_{t_2\in\mathcal R_m}
		h(X_{t_1,j},X_{t_2,j}),
		\qquad j=1,\ldots,d .
	\end{equation}
	This quantity estimates the population contrast \(\theta_j^{(m)}\) across the
	\(m\)-th change point. By Theorem~\ref{theory: initial estimators}, we have
	\(|\widehat\gamma_m-\gamma_m|=o(G)\) with probability tending to one. Therefore,
	the two trimmed local samples used in \eqref{equation: signal jump estimation} are
	not mixed across regimes with high probability. We then estimate the active set by
	thresholding:
	\begin{equation}\label{equ: estimation for the cpt set}
		\widehat\bPi_m
		=
		\{1\leq j\leq d: |\widehat\theta_j^{(m)}|\geq w^+\}.
	\end{equation}
	Here \(w^+\) is a threshold parameter used to distinguish truly active coordinates
	from noise coordinates. Our Theorem \ref{theorem: support recovery} gives a sufficient condition for consistent active-set recovery. {The threshold \(w^+\) plays an important role in determining the estimated active set
		\(\widehat{\bPi}_m\). Although our theorem gives a sufficient condition for exact support
		recovery, this condition involves unknown population quantities and is not directly available
		in practice. We therefore adopt a data-driven rule for selecting \(w^+\), as described
		in Appendix \ref{sec: pratical guidence}. Extensive simulations indicate that the resulting
		procedure is stable and yields reliable active-set recovery.}
	\vspace{0cm}
	\begin{figure}[H]
		\centering
		\begin{tikzpicture}[
			x=0.95cm,y=0.95cm,
			>=Stealth,
			every node/.style={font=\small}
			]
			
			\def\gmone{3}
			\def\gm{7}
			\def\gpone{11}
			
			\def\sone{4.1}
			\def\eone{6.1}
			\def\stwo{7.9}
			\def\etwo{9.9}
			
			\fill[gray!20] (\eone,-0.35) rectangle (\stwo,0.35);
			\node[gray!70!black] at (\gm,0.62) {trimmed buffer};
			
			\draw[thick,->] (0,0) -- (14,0) node[right] {time};
			
			\foreach \x in {1.2,1.7,2.2,2.7,3.2,3.7,4.2,4.7,5.2,5.7,6.2,6.7}{
				\draw[blue] (\x,0) circle (0.08);
			}
			\foreach \x in {7.4,7.9,8.4,8.9,9.4,9.9,10.4,10.9,11.4,11.9,12.4,12.9}{
				\draw[red] (\x,0) circle (0.08);
			}
			
			\draw[dashed] (\gmone,-0.85) -- (\gmone,1.05);
			\draw[dashed,thick] (\gm,-1.15) -- (\gm,1.15);
			\draw[dashed] (\gpone,-0.85) -- (\gpone,1.05);
			
			\node[below] at (\gmone,-0.85) {$\widehat\gamma_{m-1}$};
			\node[below] at (\gm,-1.15) {$\widehat\gamma_m$};
			\node[below] at (\gpone,-0.85) {$\widehat\gamma_{m+1}$};
			
			\draw[decorate,decoration={brace,amplitude=5pt},blue,thick]
			(\sone,1.05) -- (\eone,1.05);
			\node[blue] at ({(\sone+\eone)/2},1.55)
			{$\mathcal L_m=\{s_{1,m},\ldots,e_{1,m}\}$};
			
			\draw[decorate,decoration={brace,amplitude=5pt},red,thick]
			(\stwo,1.05) -- (\etwo,1.05);
			\node[red] at ({(\stwo+\etwo)/2},1.55)
			{$\mathcal R_m=\{s_{2,m},\ldots,e_{2,m}\}$};
			
			\draw[blue,thick] (\sone,-0.18) -- (\sone,0.18);
			\draw[blue,thick] (\eone,-0.18) -- (\eone,0.18);
			\draw[red,thick] (\stwo,-0.18) -- (\stwo,0.18);
			\draw[red,thick] (\etwo,-0.18) -- (\etwo,0.18);
			
			\node[below=5pt,blue] at (\sone,0) {$s_{1,m}$};
			\node[below=5pt,blue] at (\eone,0) {$e_{1,m}$};
			\node[below=5pt,red] at (\stwo,0) {$s_{2,m}$};
			\node[below=5pt,red] at (\etwo,0) {$e_{2,m}$};
			
			\draw[decorate,decoration={brace,amplitude=5pt,mirror},black,thick]
			(\gmone,-1.75) -- (\gm,-1.75);
			\node[below=6pt] at ({(\gmone+\gm)/2},-1.75) {$D_m^-=\widehat\gamma_m-\widehat\gamma_{m-1}$};
			
			\draw[decorate,decoration={brace,amplitude=5pt,mirror},black,thick]
			(\gm,-1.75) -- (\gpone,-1.75);
			\node[below=6pt] at ({(\gm+\gpone)/2},-1.75) {$D_m^+=\widehat\gamma_{m+1}-\widehat\gamma_m$};
			
		\end{tikzpicture}
		\caption{
			Local neighborhoods used in the projection-based refinement step. Around the
			initial estimator \(\widehat\gamma_m\), we define the estimated segment lengths
			\(D_m^-=\widehat\gamma_m-\widehat\gamma_{m-1}\) and
			\(D_m^+=\widehat\gamma_{m+1}-\widehat\gamma_m\), and then construct a left
			neighborhood \(\mathcal L_m=\{s_{1,m},\ldots,e_{1,m}\}\) and a right neighborhood
			\(\mathcal R_m=\{s_{2,m},\ldots,e_{2,m}\}\). The trimmed buffer around
			\(\widehat\gamma_m\) excludes observations close to the estimated change point,
			so that the two local samples are less affected by the transition region.
		}
		\label{fig:local-neighborhood-refine}
	\end{figure}
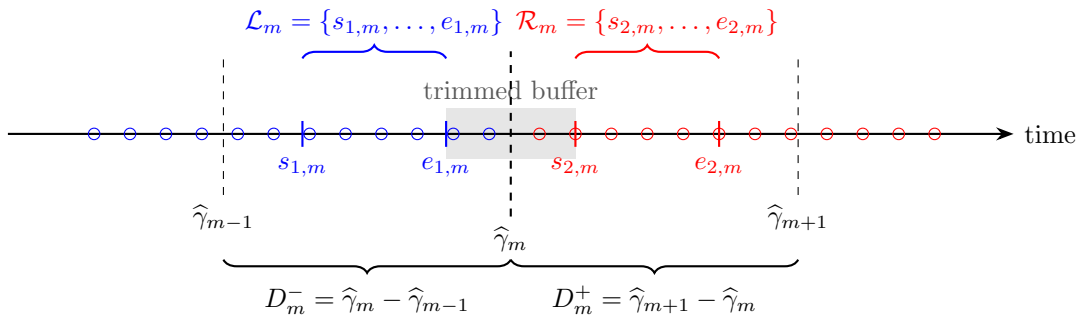
	\vspace{0cm}
	
	Using the estimated active set and signal direction, we construct the projected
	moving-window statistic. For each estimated change point \(\widehat\gamma_m\) and
	each search location \(k\), define
	\begin{equation}\label{equation: wilcoxon each coordinate 2}
		\widetilde T_j(k)
		:=
		\frac{1}{G}
		\sum_{t_1=k-G+1}^{k}
		\sum_{t_2=k+1}^{k+G}
		\widehat\theta_j^{(m)}
		h(X_{t_1,j},X_{t_2,j}),
		\qquad j=1,\ldots,d .
	\end{equation}
	Based on the aggregated projected statistic, we refine the change-point location by
	performing a local search over the region around each initial change point estimator $\hat{\gamma}_m$. Specifically, let 
	$\mathcal K_m
	=
	\left\{
	\left\lfloor \widehat\gamma_m-(1-\rho)D_m^- \right\rfloor,
	\ldots,
	\left\lceil \widehat\gamma_m+(1-\rho)D_m^+ \right\rceil
	\right\}.$
	The refined change-point estimator is defined as
	\begin{equation}\label{equ: refined estimator1}
		\widetilde\gamma_m
		=
		\arg\max_{k\in\mathcal K_m}
		\sum_{j\in\widehat\bPi_m}
		\widetilde T_j(k).
	\end{equation}
	The projection-based refinement has two main benefits. \textbf{First}, by aggregating the
	selected coordinates along the estimated signal direction, it uses more
	change-point information than the initial \(\ell_\infty\)-based estimator and
	achieves the theoretically optimal localization rate. See Theorems \ref{theory: initial estimators} and \ref{theorem: refined estimation}. \textbf{Second}, the resulting
	projected one-dimensional statistic admits a tractable limiting distribution. Under
	suitable regularity conditions, we show in Theorem \ref{theorem: refined estimation} that $\|\btheta^{(m)}\|_2^2
	\bigl(\widetilde\gamma_m-\gamma_m\bigr)
	\Rightarrow
	\arg\max_{s\in\mathbb R} Z_m(s)$,
	where \(Z_m(s)\) is defined in (\ref{equation: refined estimation}). This distributional
	result provides the basis for constructing confidence intervals for the change-point
	locations, with implementation details given in Appendix~\ref{sec: implementation for estimation confidence}.
	\vspace{0cm}
	\subsection{\textbf{Practical choice of bandwith G}}\label{sec: tuning parameters}
	The proposed moving-window procedure depends on the bandwidth \(G\). In general,
	\(G\) should be neither too small nor too large: a small \(G\) preserves locality but
	may increase stochastic fluctuation, whereas a large \(G\) stabilizes the local
	statistic but may blur nearby change points. Our theoretical requirements on \(G\)
	are given in Section~\ref{sec: Theoretical guarantees}. To reduce the sensitivity to a single
	bandwidth choice in practice, we further develop a multiscale implementation that
	aggregates information from multiple bandwidths; see Appendix~\ref{sec: pratical guidence}.
	Our numerical studies show that the multiscale procedure is stable for both localization and confidence interval construction.
	\vspace{0cm}
	\section{Theoretical guarantees}\label{sec: Theoretical guarantees}
	\subsection{\textbf{Theory for change-point testing}}\label{sec: Theory for change-point testing}
	\subsubsection{\textbf{Size control under the null}}
	First, we provide theoretical results for the bootstrap based procedure under $\Hb_0$. 
	Let $\bX=(X_{1},\ldots,X_{d})^\top \sim F(\bx)$ under $\Hb_0$ and $\bX'=(X'_{1},\ldots,X'_{d})^\top$ be the independent copy of $\bX$. Under $\Hb_0$, for each coordinate $j$ with the kernel $h(x,y)$, consider the following Hoeffding's decomposition 
	\begin{equation}\label{equ: hoeff decomposition}
		h(X_j,X_j')=0+h_{1,j}(X_j)+h_{2,j}(X_j')+g_j(X_j,X_j'),~j=1,\ldots,d,
	\end{equation}
	where $h_{1,j}(X_j):=\E[h(X_j,X'_j)|X_j]$, $h_{2,j}(X'_j):=\E[h(X_j,X'_j)|X'_j]$, and $g_j(X_j,X_j')=h(X_j,X_j')-h_{1,j}(X_j)-h_{2,j}(X'_j)$. Note that due to the antisymmetry property of $h(x,y)$, we have $h_{2,j}(x)=-h_{1,j}(x)$. The following {Assumptions ~(A.1)} -- (A.3) impose some moment conditions on the kernel. Specifically,
	
	\noindent$\mathbf{Assumption ~(A.1)}$: There is some constant $b$ such that $\E (h_{1,j}(X_{j}))^2\geq b$ for all $j=1,\ldots,d$.

	\noindent$\mathbf{Assumption ~(A.2)}$: There exists some constant $D$ such that $\|h_{1,j}(X_{j})\|_{\psi_1}\leq D$ for all $j=1,\ldots,d$, where  for any random variable $X$, $\|X\|_{\psi_{1}}$ is defined as
	$\|X\|_{\psi_{1}}:=\inf\big\{C>0: \E\psi_{1}(|X|/C)|)\leq 1\big\}$ with $\psi_{1}(x):=\exp(x)-1$.

	\noindent$\mathbf{Assumption ~(A.3)}$: There exists some constant $D$ such that $\E|h_{1,j}(X_{j})|^{2+\ell}\leq D^\ell$ with $\ell=1,2$ for all $j=1,\ldots,d$.
	
	\noindent$\mathbf{Assumption ~(A.4)}$ For the bandwidth parameter $G$, we require 
	\begin{equation*}
		\dfrac{\log^7((n-2G+1)d)}{G}\rightarrow 0,~\text{as}~(n,d)\rightarrow\infty.
	\end{equation*}
	Assumptions {(A.1)} -- {(A.4)} are crucial for proving the results of Gaussian approximations for the {$\ell_\infty$-norm-based} testing statistic. Assumption {(A.1)} requires that the two sample kernel based $U$-statistics are non-degenerate.  Assumption {(A.2)} requires that $h_{1,j}(X_{j})$ follows sub-exponential distributions. Many bounded kernels such as  Wilcoxon can satisfy this condition. Assumption {(A.3)} requires that $h_{1,j}(X_{j})$ has bounded third and forth moments. Lastly, Assumption ~{(A.4)} characterizes the scaling relationship between $(G,n,d)$ which allows the data dimension $d$ can grow exponentially with $n$ as long as  $G\gg \log^7((n-2G+1)d)$ holds. To make these kernel-level assumptions more interpretable, we provide in Appendix \ref{sec: Verification of the  assumptions} a detailed verification of these conditions for concrete mean and variance change-point models. 
	
	Bases on the above assumptions, the following Theorem \ref{theorem: gaussian approximation} justifies the validity of our proposed bootstrap procedure.
	\vspace{0cm}
	\begin{theorem}\label{theorem: bootstrap validity}\label{theorem: gaussian approximation}
		Suppose Assumptions $\mathbf{(A.1)}$ -- $\mathbf{(A.4)}$ hold.  Under $\Hb_{0}$,  we have
		\begin{equation}\label{equ: gaussian approximation result}
			\sup_{z\in (0,\infty)}\big|\P(W\leq z)-\P(W^{b}\leq z|\cX)\big|=o_p(1), ~\text{as}~  n,d\rightarrow\infty.
		\end{equation}
	\end{theorem}
	Theorem \ref{theorem: gaussian approximation} shows that we can uniformly approximate the distribution of $W$ using that of $W^b$. As a corollary, the following Corollary \ref{corollary: size} shows that our proposed new test can control the Type I error asymptotically  for any given pre-specified significant level $\alpha$.
	\vspace{0cm}
	\begin{corollary}\label{corollary: size}
		Assume Assumptions $\mathbf{(A.1)}$--$\mathbf{(A.4)}$ hold. 	Under $\Hb_{0}$,  we have
			$				\P(\Psi_{\alpha,W}=1)\rightarrow\alpha, ~ \text{as} ~n,d,B\rightarrow\infty.$
	\end{corollary}

	The proof of Theorem \ref{theorem: bootstrap validity} relies on the essential modifications of high dimensional Gaussian approximation theorey of \cite{chernozhukov2017central} to the two-sample U-statistic-based multiple change point detection with moving window. The proof mainly  proceeds in three steps. We only sketch the proof here. For any $1\leq t_1\neq t_2\leq n$, define the $\RR^d$ dimensional kernel and the corresponding leading and residual terms as:
	\begin{equation*}
		\begin{array}{ll}
			\bh(\bX_{t_1},\bX_{t_2}):=\big(h(X_{t_1,1},X_{t_2,1}),\ldots,h(X_{t_1,d},X_{t_2,d})\big)^\top,
			\bh_1(\bX_{t}):=\big(h_{1,1}(X_{t,1}),\ldots,h_{1,d}(X_{t,d})\big)^\top,\\
			~~~	\text{and}~~~	\bg(\bX_{t_1},\bX_{t_2}):=\big(g_1(X_{t_1,1},X_{t_2,1}),\ldots,g_{d}(X_{t_1,d},X_{t_2,d})\big)^\top,\\
		\end{array}
	\end{equation*}
	where $h_{1,j}(x)$ and $g_j(x,y)$ are defiend in (\ref{equ: hoeff decomposition}). Let $\bGamma=(\gamma_{i,j})\in \RR^{d\times d}=\text{Cov}(\bh_1(\bX_1))$. 
	
	\textbf{Step (i).} In this step, we  approximate the distribution of \( W \) under the null hypothesis. {Note that under the null hypothesis, we have $\E h(X_j,X'_j)=0$.} For each fixed $G\leq k\leq n-G$ and $t=1,\ldots,n$, let
	$a_{t}(k)=\dfrac{\sqrt{n}}{\sqrt{G}}\big(\mathbf{1}\{k-G+1\leq t\leq k\}-\mathbf{1}\{k+1\leq t\leq k+G\}\big).$
	Then, we can rewrite $\bT(k)$ as the following form:
	\begin{equation}\label{equ: leading term}
		\bT(k)=\dfrac{1}{\sqrt{n}}\sum_{t=1}^n a_t(k)\bh_1(\bX_{t})+\dfrac{1}{G^{3/2}}\sum_{t_1=k-G+1}^k\sum_{t_2=k+1}^{k+G}\bg(\bX_{t_1},\bX_{t_2}).
	\end{equation}
	Our theorem shows that the residual term of the Hoeffding's decomposition can be uniformly negligible over $k$ and $1\leq j\leq d$ in the sense that 
	\begin{equation*}
		\max_{G\leq k\leq n-G}  \big\|\bT(k)-\dfrac{1}{\sqrt{n}}\sum_{t=1}^n a_t(k)\bh_1(\bX_{t})\big\|_{\infty}=O_p({\log(nd)}/{\sqrt{G}}).
	\end{equation*}
	Hence, by the Hoeffding's decomposition, we can approximate $W$ by $\max_{G\leq k\leq n-G}\|\bT^{(1)}(k)\|_{\infty}$, where $\bT^{(1)}(k):=\dfrac{1}{\sqrt{n}}\sum\limits_{t=1}^n a_t(k)\bh_1(\bX_{t})$.
	
	
	\textbf{Step (ii).} 
	Note that we don’t know the exact distribution of $\max_{G\leq k\leq n-G}\|\bT^{(1)}(k)\|_{\infty}$. As $\bT^{(1)}(k)$ is a sum of independent random vectors with zero mean, this motivates us to  use a Gaussian random vector to further approximate $\max_{G\leq k\leq n-G}\|\bT^{(1)}(k)\|_{\infty}$.
	Specifically,  let $\bG_{1},\ldots,\bG_n$ be $i.i.d$ Gaussian random vectors with $\bG_{t}=(G_{t,1},\ldots,G_{t,d})\in \RR^d$ and $\bG_{t}\sim N(0,\bGamma)$.  We define the Gaussian random vector-based testing statistic $\bT^{\bG}(k)=(\bT^{\bG}_1(k),\ldots,T^{\bG}_d(k))^\top$ with
	\begin{equation}\label{equ: TG}
		\begin{array}{ll}
			\bT^{\bG}(k)&=\dfrac{1}{\sqrt{n}}\sum\limits_{t=1}^n a_t(k)\bG_{t},~~\text{with}~~G\leq k\leq n-G.
		\end{array}
	\end{equation}
	Our theorem shows that we can approximate $\bT^{(1)}(k)$ using $\bT^{\bG}(k)$ in the sense that 
	\begin{equation*}
		\small
		\sup_{z\in(0,\infty)} \big|\P(\max_{G\leq k\leq n-G}\|\bT^{(1)}(k)\|_{\infty}\leq z\big)-\P(\max_{G\leq k\leq n-G}\|\bT^{G}(k)\|_{\infty}\leq z\big)\big|\leq C_1\Big(\dfrac{\log^7((n-2G+1)d)}{G}\Big)^{1/6}.
	\end{equation*}
	Hence, we can use the distribution of $\max\limits_{G\leq k\leq n-G}\|\bT^{G}(k)\|_{\infty}$ to approximate that of $\max\limits_{G\leq k\leq n-G}\|\bT^{(1)}(k)\|_{\infty}$.
	
	\textbf{Step (iii).} Since the covariance structure $\bGamma$ is typically unknown, one way is to obtain some estimator $\hat{\bGamma}$ and use the corresponding Gaussian random vectors to approximate the testing statistic. However, the estimation of the covariance matrix imposes certain requirements, which can be particularly challenging in ultra-high dimensional problems. Instead, we can use the multiplier bootstrap-based testing statistic in (\ref{equ: bootstrap based testing statistics for each coordinate}) to approximate $\bT^{\bG}(k)$ in a data-driven way. Specifically, for the Gaussian random vector-based testing statistic $\bT^{\bG}(k)$, it has very complicated covariance structure in terms of the search location:
	\begin{equation}\label{equ: covariance}
		\small
		\text{Cov}(\bT^{\bG}(k_1),\bT^{\bG}(k_2)) = 
		\left\{
		\begin{array}{ll}
			0,& \text{if}~|k_2-k_1|>2G-1; \\
			-\dfrac{2G-|k_2-k_1|}{G}\bGamma,& \text{if}~G-1<|k_2-k_1|\leq2G-1;  \\
			\dfrac{2G-3|k_2-k_1|}{G}\bGamma,& \text{if}~0<|k_2-k_1|\leq G-1;  \\
			2\bGamma,& \text{if}~|k_2-k_1|=0. \\
		\end{array}
		\right.
	\end{equation}
	Our theorem shows that the multiplier bootstrap-based testing statistic can capture the above covariance structure accurately in the sense that 
	\begin{equation*}
		\max_{G\leq k_1,k_2\leq n-G}\|\text{Cov}(\bT^{\bG}(k_1),\bT^{\bG}(k_2))-\text{Cov}(\bT^{b}(k_1),\bT^{b}(k_2))\|_{\infty}=o_p(1).
	\end{equation*}

	The above results allow us to use $\|\bT^{b}(k)\|_{\infty}$ to approximate $\max_{G\leq k\leq n-G}\|\bT^{G}(k)\|_{\infty}$. Lastly,  we analyze the estimation error 
	and combine results from \textbf{Steps (i)} - \textbf{(iii)} to 
	finish the proof of Theorem \ref{theorem: bootstrap validity}. The detailed proof is in  the supplementary materials.
	
	\vspace{0cm}
	\subsubsection{\textbf{Power under alternatives}}\label{sec: Power under alternatives}
	After analyzing the size, we discuss the power properties. To this end, we need some additional notations. Recall $\bPi_m:=\{j\in\{1,\ldots,d\}: \theta_{j}^{(m)}\neq 0\}$ as the set of coordinates having a change point at $\gamma_m$.  For each $m=1,\ldots,M_0$, define the signal jump vector 
	$\btheta^{(m)}:=(\theta^{(m)}_1,\ldots,\theta^{(m)}_d)^\top$. Define $\btheta^{\diamondsuit}=\max_{1\leq m\leq M_0}\|\btheta^{(m)}\|_{\infty}$ be the maximum signal jump among all change points and let $\Delta:=\min_{1\leq m\leq M_0}(\gamma_{m+1}-\gamma_m)$ be the minimum segmentation length between change points with $\gamma_{M_0+1}=n$. Next, we present some conditions for $\Hb_1$. 
	
	\noindent$\mathbf{Assumption ~(B.1)}$: For the bandwidth parameter $G$, we require 
	\begin{equation*}
		\Delta\geq 2G,~~\text{and}~~\dfrac{\log^2(nd)\log^2(Gd)}{G}\rightarrow 0 ~\text{as}~(n,d)\rightarrow\infty.
	\end{equation*}
	Additionally, suppose there exist some constants such that $\max_{1\leq m\leq M_0}\|\btheta^{(m)}\|_\infty\leq C^*$.
	
	\noindent$\mathbf{Assumption ~(B.2)}$: There exists some constants $D^\ell$ such that $\E|h(X_{\gamma_m,j},X_{\gamma_m+1,j})-\theta_j^{(m)}|^{2+\ell}\leq D^\ell$ with $\ell=1,2$ for all $j=1,\ldots,d$ and $m=1,\ldots,M_0$.
	
	\noindent$\mathbf{Assumption ~(B.3)}$: There exists some constant $D$ such that $\|h(X_{\gamma_m,j},X_{\gamma_m+1,j})-\theta_j^{(m)}\|_{\psi_1}\leq D$ for all $j=1,\ldots,d$ and $m=1,\ldots,M_0$.
	
	{Assumption ~(B.1)} requires that the bandwidth should not be too large in the sense that $G\leq \Delta/2$ such that there is at most one change point in any candidate moving window $[k-G+1,k+G]$. Moreover, the bandwidth is required to have an order $G\gg \log^2(nd)\log^2(Gd)$ so that the large sample property holds. See also \cite{eichinger2018mosum}. {Assumptions ~(B.2)} and (B.3) make some similar  moment conditions on the kernel  to {Assumptions ~(A.2)} and (A.3). Notably, our moment condition assumptions are primarily imposed on the kernel rather than the data itself. Therefore, for many heavy-tailed datasets, the theoretical framework of the model remains applicable as long as an appropriate kernel function is chosen.
	
	With the above notations and some regularity conditions,  the following Theorem \ref{theorem: power results} shows that we can reject the null hypothesis of no change point with overwhelming probability.
	\vspace{0cm}
	\begin{theorem}\label{theorem: power results}
		Suppose Assumptions $\mathbf{(B.1)}$ -- $\mathbf{(B.3)}$ hold. If $\btheta^{\diamondsuit}$ satisfies
		\begin{equation}\label{inequality: theoretical signal strengh}
			\sqrt{G}\times\btheta^{\diamondsuit}\geq C_0\Big(\sqrt{2\log(dn)}+\sqrt{2\log(\alpha^{-1})}\Big),
		\end{equation}
		under $\Hb_{1}$, we have
			$		\P(\Phi_{\alpha,W}=1)\rightarrow 1, \text{as}~n,d,B\rightarrow\infty,$
		where $C_0$ is a large enough universal positive constant not depending on $n$ or $d$.
	\end{theorem}
	\vspace{0cm}
	\indent Theorem \ref{theorem: power results} demonstrates that with probability tending to one, our proposed new test can detect the existence of multiple change points as long as the corresponding signal to noise ratio satisfies (\ref{inequality: theoretical signal strengh}). Specifically, considering Assumption (B.1) with $\Delta\geq 2G$, {if the bandwith parameter is chosen properly with $\Delta \approx G$} it mainly requires that  $\max_{1\leq m\leq M_0}\|\btheta^{(m)}\|_{\infty}\geq \sqrt{{\log(nd)}/{\Delta}}$. Hence,  with a larger signal jump and larger segmentation length, it is more likely to trigger a rejection of  the null hypothesis. From the above conditions, we can also observe that to detect multiple change points, it is sufficient for the signal-to-noise ratio condition of just one of the change points to meet the requirements. This is fundamentally different from the conditions for change point estimation discussed later.

	\vspace{0cm}
	\subsection{\textbf{Theory for initial estimation}}
	In the following, we  provide consistency results of the initial estimated change point number and locations. To that end, we need the following conditions on the signal strengh $\btheta_{\diamondsuit}=\min_{1\leq m\leq M_0}\|\btheta^{(m)}\|_{\infty}$ {with $\btheta^{(m)}:=\E\bh(\bX_{\gamma_m},\bX_{\gamma_m+1})$.}
	
	\noindent$\mathbf{Assumption ~(C.1)}$  Assume $G\times\btheta_{\diamondsuit}^2\gg \log(nd)$. Based on the assumptions, we have the following result.
	\vspace{0cm}
	\begin{theorem} \label{theory: initial estimators}
		(\textbf{Consistency of initial estimators}) Suppose Assumptions \textbf{(B.1)-(B.3)} and \textbf{(C.1)} hold. Then we have 
		
		(1) $\mathbb{P}\{\hat{M}_0={M}_0\}\to 1$, as $n,d \to \infty$;
		
		(2)   $|\hat{\gamma}_m-{\gamma}_m|=O_P\Big(\dfrac{\log(nd)}{\|\btheta^{(m)}\|_{\infty}^2}\Big)$, for $m=1,\ldots,M_0$.
	\end{theorem}
	
	Theorem \ref{theory: initial estimators} (1) shows that the initial procedure can correctly identify the change point number $M_0$.  This is important for the refining procedure in the following step. More importantly,  Theorem \ref{theory: initial estimators} (2) demonstrates that the initial estimator $\hat{\gamma}_m$ has an estimation error bound of $O_p(\log(nd)/\|\btheta^{(m)}\|_{\infty}^2)$, which is consistent. 

	\vspace{-0.3cm}
	\subsection{\textbf{Theory for projection-based refinement}}\label{sec: theory for refinement}
	Next, we provide theorem for the refined procedure.  Recall $\bPi_{m}=\{j\in\{1,\cdots,d\}: \theta_j^{(m)}\neq 0 \}$ as the coordinates having a change point at $\gamma_m$ and  $s_m=|\Pi_{m}|$ is the total number. Without loss of generality, we assume $\bPi_{m}=\{1,\ldots,s_m \}$.
	Recall $\hat{\bPi}_m$ in (\ref{equ: estimation for the cpt set}) as the estimation for $\bPi_{m}$. The next theorem shows that  we can recover ${\bPi}_m$ with probability tending to one.
	\vspace{0cm}
	\begin{theorem}\label{theorem: support recovery}
		{	Suppose Assumptions $\mathbf{(B.1)}$--$\mathbf{(B.3)}$ and $\mathbf{(C.1)}$ hold. } Furthermore, for the threshold parameter $w^{+}$, we require
		\begin{equation*}
			\begin{array}{ll}
				w^+\leq \min\limits_{1\leq m\leq M_0}\min\limits_{1\leq j\leq d}|\theta_{j}^{(m)}|-C_2D\Big(\sqrt{\dfrac{\log(nd)}{G}}+\dfrac{\log(Gd)\log(nd)}{G}\Big)\\
				\text{and}~~	w^+\geq C_2D\Big(\sqrt{\dfrac{\log(nd)}{G}}+\dfrac{\log(Gd)\log(nd)}{G}\Big),
			\end{array}
		\end{equation*} 
		for some $C_2>0$.  Then we have 
			$	\P(\hat{\bPi}_m={\bPi}_{m})\rightarrow 1.$
	\end{theorem}
	
	The next theorem shows that after refinemet,  our proposed U-PRA can have an optimal convergence rate. Thanks to this result, we are able to construct valid confidence intervals for the positions of  change points. To that end, we need more notations. For each $m=1,\ldots,M_0$ and $j=1,\ldots,d$, define the centralized kernel $\tilde{h}(x,y)=h(x,y)-\theta_{j}^{(m)}$ and let 
	\begin{equation*}
		\begin{array}{ll}
			h_{1,j}(X_{\gamma_m,j}):=\E [\tilde{h}(X_{\gamma_m,j},X_{\gamma_m+1,j}|X_{\gamma_m,j})],~~\underline{h}_{1,j}(X_{\gamma_m,j}):=\E[ h(X_{\gamma_m-1,j},X_{\gamma_m,j}|X_{\gamma_m,j})],\\ h_{2,j}(X_{\gamma_m+1,j}):=\E [\tilde{h}(X_{\gamma_m,j},X_{\gamma_m+1,j}|X_{\gamma_m+1,j})],~~ \overline{h}_{1,j}(X_{\gamma_m+1,j}):=\E [h(X_{\gamma_m+1,j},X_{\gamma_m+2,j}|X_{\gamma_m+1,j})],\\
		\end{array}
	\end{equation*}
	with $\theta_j^{(m)}:=\E[h(X_{\gamma_m,j},X_{\gamma_m+1,j}]$ being the signal jump of the $j$-th coordinate at $\gamma_m$. Note that by above definitions, $h_{1,j}(X_{\gamma_m,j})$ and $h_{2,j}(X_{\gamma_m+1,j})$ are the Hoeffding's projections for the two-sample kernel $h(x,y)$ with heterogeneous data and $\underline{h}_{1,j}(X_{\gamma_m,j})$, $\overline{h}_{1,j}(X_{\gamma_m+1,j})$ are the Hoeffding's projections for the two-sample kernel before and after the change point with homogeneous data. Moreover, define the four $\RR^{s_m}$-dimensional  vectorized version of Hoeffding's projections as follows
	\begin{equation*}
		\begin{array}{ll}
			\bh_1(\bX_{\gamma_m})=(h_{1,1}(X_{\gamma_m,1}),\ldots,h_{1,s_m}(X_{\gamma_m,s_m}))^\top,
			\bh_2(\bX_{\gamma_m+1})=(h_{2,1}(X_{\gamma_m+1,1}),\ldots,h_{2,s_m}(X_{\gamma_m+1,s_m}))^\top,\\
			\overline{\bh}_1(\bX_{\gamma_m+1})=(\overline{h}_{1,1}(X_{\gamma_m+1,1}),\ldots,\overline{h}_{1,s_m}(X_{\gamma_m+1,s_m}))^\top,
			\underline{\bh}_1(\bX_{\gamma_m})=(\underline{h}_{1,1}(X_{\gamma_m,1}),\ldots,\underline{h}_{1,s_m}(X_{\gamma_m,s_m}))^\top,\\
		\end{array}
	\end{equation*}
	and define the following four $\RR^{s_m\times s_m}$-dimensional covariance matrices as 
	\begin{equation}\label{equ: four covs}
		\begin{array}{ll}
			\bSigma_1^{(m)}=\text{Cov}(\bh_1(\bX_{\gamma_m})),\quad \bSigma_2^{(m)}=\text{Cov}(\bh_2(\bX_{\gamma_m+1})),\\ \bSigma_3^{(m)}=\text{Cov}(\overline{\bh}_1(\bX_{\gamma_m+1})-\bh_2(\bX_{\gamma_m+1})), \quad
			\bSigma_4^{(m)}=\text{Cov}(\underline{\bh}_1(\bX_{\gamma_m})+\bh_1(\bX_{\gamma_m})).
		\end{array}
	\end{equation}
	To present the change point inference results, we need the following additional assumption.
	
	\noindent{\bf{Assumption (C.2):}} We require there exists constants $\kappa_1>0$ and $\kappa_2>0$ such that
	\begin{equation}
		\begin{array}{ll}
			\kappa_1\leq \min(\lambda_{\min}(\bSigma_1^{(m)}),\lambda_{\min}(\bSigma_2^{(m)}),\lambda_{\min}(\bSigma_3^{(m)}),\lambda_{\min}(\bSigma_4^{(m)}),\\ \max(\lambda_{\max}(\bSigma_1^{(m)}),\lambda_{\max}(\bSigma_2^{(m)}),\lambda_{\max}(\bSigma_3^{(m)}),\lambda_{\max}(\bSigma_4^{(m)})\leq \kappa_2, ~\text{for}~m=1,\ldots,M_0.
		\end{array}
	\end{equation}
	With the above notations, we are ready to give the change point inference results.
	\vspace{0cm}
	\begin{theorem}\label{theorem: refined estimation}
		Suppose Assumptions $\mathbf{(B.1)}$--$\mathbf{(B.3)}$ and $\mathbf{(C.1)}-\mathbf{(C.2)}$ hold.  Then we have 
		\begin{equation}\label{equation: refined estimation}
			\begin{array}{ll}
				(1) |\tilde{\gamma}_m-{\gamma}_m|=O_P\big(\dfrac{1}{\|\btheta^{(m)}\|^2}\big), \text{for}~ m=1,\ldots,M_0; \\
				(2) \|\btheta^{(m)}\|^2(\tilde{\gamma}_m-\gamma_m)\Rightarrow \argmax_{s\in \RR} Z^{(m)}(s), \text{for}~m=1,\ldots,M_0,\\
			\end{array}
		\end{equation}
		with
		\[
		Z^{m}(s)=
		\begin{cases} 
			-s+\sigma_{1,*}^{(m)}W^{(m)}_{1}(s)+ \sigma_{3,*}^{(m)}W^{(2)}_{2}(s)+\sigma_{2,*}^{(m)}W^{(m)}_{3}(s), & s >0, \\
			0, & s = 0, \\
			s+\sigma_{1,*}^{(m)}W^{(m)}_{1}(s)+ \sigma_{4,*}^{(m)}W^{(m)}_{2}(s)+\sigma_{2,*}^{(m)}W^{(m)}_{3}(s), & s <0,
		\end{cases}
		\]
		where $\{W^{(m)}_{1}(s)\}$, $\{W^{(m)}_{2}(s)\}$, $\{W^{(m)}_{3}(s)\}$  are standard Brownian motion processes on $(-\infty,\infty)$ which are independent with each other and $\sigma_{j,*}^{(m)}; j=1,2,3,4$  are defined as 
		\begin{equation}\label{equ: sigma1-star-4-star}
			\begin{array}{ll}
				\sigma_{1,*}^{(m)}=\dfrac{((\btheta^{(m)})^\top	\bSigma_1^{(m)}\btheta^{(m)})^{1/2}}{\|\btheta^{(m)}\|}, \sigma_{2,*}^{(m)}=\dfrac{((\btheta^{(m)})^\top	\bSigma_2^{(m)}\btheta^{(m)})^{1/2}}{\|\btheta^{(m)}\|},\\ \sigma_{3,*}^{(m)}=\dfrac{((\btheta^{(m)})^\top	\bSigma_3^{(m)}\btheta^{(m)})^{1/2}}{\|\btheta^{(m)}\|}, \sigma_{4,*}^{(m)}=\dfrac{((\btheta^{(m)})^\top	\bSigma_4^{(m)}\btheta^{(m)})^{1/2}}{\|\btheta^{(m)}\|}.
			\end{array}
		\end{equation}
		Moreover, we can prove that  $\|\btheta^{(m)}\|^2(\tilde{\gamma}_m-\gamma_m)$'s \text{are asymptotically independent}.
	\end{theorem}
	\vspace{0cm}
	
	Theorem \ref{theorem: refined estimation} (1) demonstrates that the U-statistic projection refinement algorithm U-PRA has an estimation error bound of $O_P(\dfrac{1}{\|\btheta^{(m)}\|^2})$, which has been improved compared to the initial estimators obtained in Theorem \ref{theory: initial estimators}. To our best knowledge, the convergence rate is optimal.  Moreover, this refinement is crucial from an inferential standpoint,  which ensures the presence of a limiting distribution.   Theorem \ref{theorem: refined estimation} (2) shows that the centered and scaled change point estimator, after refining, has an limiting distribution which  follows the $\arg\max$ of a drifted weighted Brownian motion process. Hence, we can make valid confidence intervals for the detected change points. In Appdenx \ref{sec: implementation for estimation confidence}, we provide more discussions and implementation details for making confidence interval.

	\vspace{0cm}
	\section{Numerical performance}\label{section: empirical study}

	In this section, we examine the numerical performance of our proposed method in terms of change point detection, estimation as well as inference for the mean shift and variance change point models. We also compare our method with several existing state-of-art techniques.

	\subsection{\textbf{Changes in the mean}}\label{sec: changes in mean}

	We consider the location change point problem. In particular, let $\bX_1,\cdots,\bX_n$ be independent random vectors with $\bX_t\in \RR^d$ for $1 \leq t \leq n$.
	We generate data from the following mean shift model with possible three change points at $\gamma_{1},\ldots,\gamma_{3}$:
	\begin{equation*}
		\bX_t=\bmu_t+\bepsilon_t,
	\end{equation*}
	where $\{\bepsilon_t\}_{t=1}^{n}$ are i.i.d innovations with mean zero, $\bmu_{t}$ are the mean vectors at each time point. To show the robustness of our method, we generate the innovations $\{\bepsilon_t\}_{t=1}^{n}$ from four types of data distributions including:\\
	(1) \textbf{Multivariate normal distributions} $N(\mathbf{0},\bSigma)$; \\
	(2) \textbf{Multivariate student $t$ distributions $t(v,\bSigma)$} with a degree of freedom $v\in\{2,3\}$
	according to $\bZ/\sqrt{W/v}$, where $\bZ\sim N(\mathbf{0},\bSigma)$ and $W\sim\chi^2(v)$, $\bZ$ and $W$ are independent with each other; \\
	(3) \textbf{Contaminated Gaussian distribution}: $\bepsilon_i=N(\mathbf{0},\bSigma)+\mathbf{1}_d\times \epsilon_iZ_i$, where $\mathbf{1}_d$ is a
	$d$-dimensional vector 
	with each element equaling to 1, $\epsilon_i\sim B(1,p)$ with $p=0.1$,  $Z_i$ is a random variable taking 5 or -5 with equal probability and $\epsilon_i$ and $Z_i$ are independent with each other.

	Moreoover,  {we generate $\bSigma$} from the  blocked diagonal model: 
	by setting $\bSigma=\bSigma^{*}$, where $\bSigma^\star=(\sigma^\star_{ij})\in\reals^{d\times d}$ with $\sigma^\star_{ii}=1$,  $\sigma^\star_{ij}=0.5$ for
	$5(k-1)+1\le i\ne j\le 5k$ ($k=1,\ldots,\lfloor d/5\rfloor$), and $\sigma^\star_{ij} = 0$ otherwise. 
	
	Throughout numerical experiments,  for all models, we set $n=600$ with the dimension $d\in\{100,200,300\}$. {We consider the bandwith parameter $G\in\{60,80,100\}$ and the number of bootstrap replications $B=200$.} We set the significance level $\alpha=0.05$.  All numerical results are based on 500 replications under $\Hb_0$ and 200 replications under $\Hb_1$.



	\subsection{{\textbf{Empirical sizes}}}\label{sec: empirical sizes}

	\noindent In this section, we investigate the empirical size performance. It is worth noting that our method applies to a broad class of  kernel functions \( \bh(\bx,\by) \). To demonstrate the wide applicability of our approach, we have selected two specific cases: one is \(\bh_1(\bx,\by)=\by-\bx\), and the other is \(\bh_2(\bx,\by)=\text{sign}(\by-\bx)\). The former is equivalent to the high-dimensional CUSUM method, while the latter corresponds to the well-known Wilcoxon-Mann-Whitney test, which has not yet been explored in the context of ultra-high dimensions for multiple change point detection.

	As shown in Table \ref{tab:empirical-size-model1},	For Gaussian distribution, both methods  are able to effectively control the theoretical significance level (\( 0.05 \)) across both Models 1 and  2, as well as for various dimensions (\( d = 100, 200, 300 \)). This demonstrates that both methods perform well in maintaining the desired type I error when the data follows a light tailed  distribution.
	
	However, when the data follows a heavy-tailed distribution, such as \( t_3 \), the method based on \( \bh_1(\bx, \by) = \by - \bx \) becomes notably conservative. This conservatism becomes more pronounced as the data exhibits heavier tails, such as  \( t_2 \) and contaminated normal distributions. One possible explanation for this behavior is  that the increased presence of extreme outliers in heavy-tailed data leads to a larger variance in the bootstrap statistics. As a result, the test becomes overly conservative, reducing the rejection rate of the null hypothesis. On the other hand, the method based on \(\bh_2(\bx,\by)=\text{sign}(\by-\bx)\) is able to maintain size in both light-tailed and heavy-tailed distributions. This suggests that the Wilcoxon-based method is more robust in the presence of outliers, as it effectively mitigates the impact of extreme values on the test’s performance by using the ranks instead of the raw data.

	\begin{table}[H]
		\centering
		\caption{Empirical size performance for mean shifts under different bandwidth choices.}
		\label{tab:empirical-size-model1}
		\setlength{\tabcolsep}{10pt}
		\renewcommand{\arraystretch}{1.15}
		\begin{tabular}{llcccccc}
			\toprule
			&  & \multicolumn{3}{c}{\(h(x,y)=y-x\)} 
			& \multicolumn{3}{c}{\(h(x,y)=\operatorname{sign}(y-x)\)} \\
			\cmidrule(lr){3-5} \cmidrule(lr){6-8}
			\(G\) & Distribution 
			& \(d=100\) & \(d=200\) & \(d=300\)
			& \(d=100\) & \(d=200\) & \(d=300\) \\
			\midrule
			
			60 & Normal       & 0.030 & 0.030 & 0.020 & 0.040 & 0.025 & 0.025 \\
			60 & \(t_3\)      & 0.000 & 0.005 & 0.005 & 0.030 & 0.025 & 0.020 \\
			60 & \(t_2\)      & 0.000 & 0.005 & 0.005 & 0.030 & 0.025 & 0.020 \\
			60 & Contam-Norm  & 0.025 & 0.020 & 0.020 & 0.030 & 0.030 & 0.050 \\
			
			\midrule
			80 & Normal       & 0.045 & 0.035 & 0.025 & 0.035 & 0.050 & 0.045 \\
			80 & \(t_3\)      & 0.010 & 0.000 & 0.005 & 0.060 & 0.040 & 0.020 \\
			80 & \(t_2\)      & 0.010 & 0.000 & 0.005 & 0.060 & 0.040 & 0.020 \\
			80 & Contam-Norm  & 0.010 & 0.055 & 0.040 & 0.040 & 0.030 & 0.050 \\
			
			\midrule
			100 & Normal      & 0.030 & 0.055 & 0.020 & 0.025 & 0.020 & 0.040 \\
			100 & \(t_3\)     & 0.010 & 0.005 & 0.010 & 0.050 & 0.030 & 0.045 \\
			100 & \(t_2\)     & 0.010 & 0.005 & 0.010 & 0.050 & 0.030 & 0.045 \\
			100 & Contam-Norm & 0.035 & 0.035 & 0.030 & 0.045 & 0.050 & 0.040 \\
			
			\bottomrule
		\end{tabular}
	\end{table}

	\subsection{\textbf{Empirical powers}}\label{sec: empirical powers}
	
	Next, we consider the power performance. The data generating processes are considered
	below:
	\[
	(\mathbf H_{1,3}) \quad
	\text{(three change-point alternative)}:\quad
	\bmu_t
	=
	\bdelta \mathbf 1\{\gamma_1<t\leq \gamma_2\}
	+
	\bdelta \mathbf 1\{\gamma_3<t\leq n\},
	\]
	where \(\bdelta=(\delta_1,\ldots,\delta_d)^\top\in\mathbb R^d\) is the signal
	jump vector. We consider a sparse alternative in which only the first five coordinates
	have nonzero signals. Specifically, let
	\[
	\delta_j
	=
	C\sqrt{\frac{\log(d)}{n}}\,s_j,
	\qquad j=1,\ldots,5,
	\]
	where $(s_1,s_2,s_3,s_4,s_5)=(1,-1,1,-1,1)$
	and $\delta_j=0,\text{for}~ j=6,\ldots,d$. Here, \(C>0\) controls the signal strength. In this case, we set \(n=600\) and
	\(d=200\). Under \(\mathbf H_{1,3}\), the three change points are set at
	$(\gamma_1,\gamma_2,\gamma_3)=(0.3n,0.5n,0.7n).$

	Similar to Section \ref{sec: empirical sizes}, we consider the kernel \(\bh_1(\bx,\by)=\by-\bx\) and \(\bh_2(\bx,\by)=\text{sign}(\by-\bx)\). In addition, we compare our approach with methods specifically designed for sparse alternative including the $(s_0,p)$-norm based CUSUM method in \cite{Liu2020Unified}, where we set $s_0=1$ which reduces to the $\ell_\infty$-norm based method (denoted by LZZL);  Notably, based on existing research, the method in \cite{Liu2020Unified} based on $\ell_\infty$-norm exhibits similar performance to those of \cite{Jirak2015Uniform} and \cite{yu2021finite}. Therefore, we did not include comparisons with the latter two methods. We also compare with the double cusum based method in \cite{cho2016change} with parameter $\psi=0$ which is powerful for sparse alternatives (DC) and the sparse projection based method in \cite{Wang2016High} (Inspect). Note that \cite{Wang2016High} aim to estimate change points and  we adopt their single change point version function in corresponding R packages and convert them to tests using their default threshold computing functions. We use the R package "AdaptiveCpt", "hdbinseg", and ``inspect" to implement the above three methods with defalut parameters. Note that we do not compare with \cite{wang2019inference,jiang2023robust} since they are designed for the dense alternative patterns which are not discussed in this paper.
	
	Figures \ref{fig: power} shows the empirical power performance under three change-point alternative  with different signal strength and data distributions. It can be observed that when the errors  follow a normal distribution, the performance of the existing methods is comparable. This is primarily because the normal distribution is a light-tailed distribution, which satisfies the theoretical assumptions of these tests. However, as the data transition to a heavy-tailed distribution, the effectiveness of these methods gradually diminishes.
	
	Specifically, the Inspect method fails to control the  size (\( C=0 \)), indicating a tendency to overestimate the number of change points. This overestimation is further corroborated in the subsequent numerical simulations in Section \ref{sec: emprical change point estimation} for multiple change-point estimation. The DC method is severely affected in the presence of heavy-tailed distributions, demonstrating an overly conservative behavior that leads to almost no change points being detected. The primary reason for this lies in its reliance on the bootstrap procedure for determining critical values, where the presence of outliers inflates the bootstrap critical value excessively. These findings are consistent with the results reported in \cite{yu2022robust}. Although the LZZL method is relatively less affected, its power in detecting multiple change points deteriorates significantly under extreme heavy-tailed distributions and the presence of outliers. In contrast, the method based on \( \bh_2(\bx,\by) \) outperforms all existing approaches, suggesting that constructing multiple change-point testing statistics based on ranks and moving windows  achieves a desirable balance between testing powers and robustness.

	\begin{figure}[H]
		\centering
		\includegraphics[width=1\textwidth]{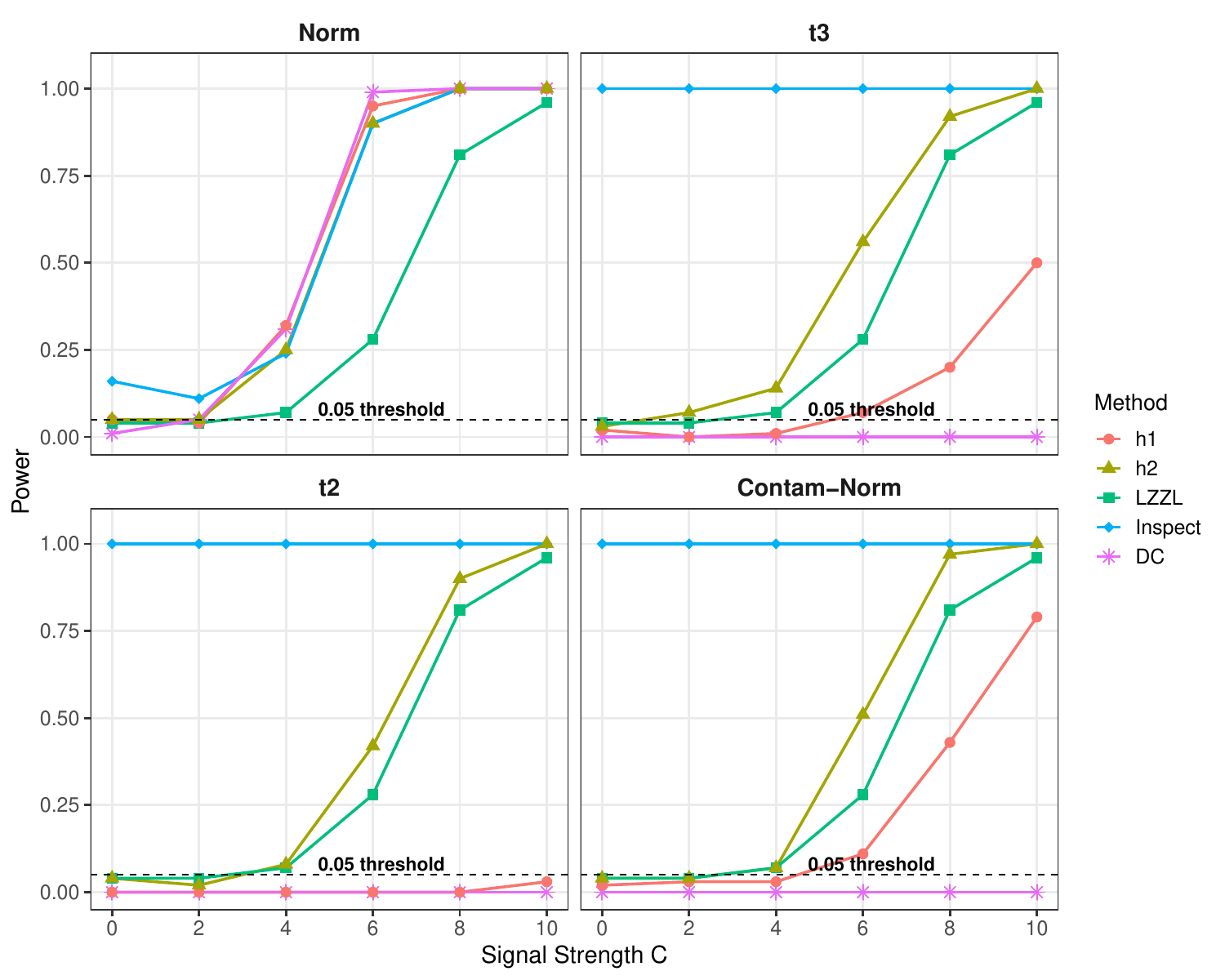}
		\caption{Empirical powers under with different distributions. }
		\label{fig: power}
	\end{figure}

	\subsection{{\textbf{Multiple change point estimation}}}\label{sec: emprical change point estimation}
	

	In this section, we consider the performance of multiple change point detection and compare our method with the existing techniques. In this numerical study, we consider the three change point model $\Hb_{1,3}$ with $C=12$ for $N(0,\bSigma)$ and $C=15$ for the heavy tailed  distributions,  which has been introduced in the previous section.  To evaluate the performance in identifying the change point, we use  the Hausdorff distance (denoted as Haus)  to evaluate the performance in change point estimation, which is defined as:
	\begin{equation*}
		d(\cS_1,\cS_2)={\max(\max_{s_1\in\cS_1}\min_{s_2\in\cS_2} |s_1-s_2|,\max_{s_2\in\cS_2}\min_{s_1\in\cS_1} |s_1-s_2|  )},
	\end{equation*}
	where $\cS_1$ are the true change points and $\cS_2$ are the estimated change points. Additionally, we have recorded the proportion of overestimated and underestimated change points across 200 simulations, denoted as $\hat{M}_0<M_0$ and $\hat{M}_0>M_0$, respectively. Similar to the previous sections, for the two types of kernels, we recorded the results of the initial estimation in (\ref{equation: intial number}) and the refined estimation in (\ref{equ: refined estimator1}), denoted as $h_1-{initial}$, $h_1-{refined}$, $h_2-{initial}$, $h_2-{refined}$, respectively. We compare our method with \cite{Liu2020Unified}, which employs the CUSUM statistic for single change-point mean testing and combine it with the binary segmentation algorithm for multiple change-point detection. Additionally, we compare  with the sparse projection-based method in \cite{Wang2016High}, which primarily assumes a Gaussian distribution and utilizes the WBS algorithm for multiple change-point detection.

	Table 	\ref{tab:haus-model1}	 shows that when the data follows a normal distribution, the method based on \( \bh_1(\bx, \by) \) generally outperforms the method based on \( \bh_2(\bx, \by) \). This is reflected in both smaller Hausdorff distances and higher proportions of \( \hat{M}_0 = M_0 \). The results indicate that the CUSUM based  method is better at accurately detecting the change points and estimating their locations under the light tailed  data. As the data transitions to a heavy-tailed distribution, such as \( t_3 \), the method based on \( \bh_1(\bx, \by) \) struggles to accurately identify both the number and positions of the change points. This difficulty becomes more pronounced as the data becomes even more heavy-tailed  and  its performance  becomes almost ineffective in such cases. In contrast, the method based on \(\bh_2(\bx,\by)=\text{sign}(\by-\bx)\), which is the Wilcoxon based statistic, demonstrates better performance in detecting change points even in the presence of heavy tails and contaminated data. This illustrates the robustness of rank-based methods for multiple change point estimation.

	For both \( \bh_1(\bx, \by) \) and \( \bh_2(\bx, \by) \), it is evident that the refined change point estimates, obtained after projection, are more accurate than the initial estimates (see Figure \ref{fig:multiscale-haus}). This observation validates the theoretical results in Theorem \ref{theorem: refined estimation} and demonstrates that projection allows for the extraction of more informative components of the change points.  This reinforces the conclusion that, compared to using the \( \ell_\infty \)-norm alone, projection enhances the method by integrating additional change point components. Interestingly, even in the presence of heavy-tailed and contaminated data, the projection-based change point estimates constructed using the Wilcoxon statistic still remain effective.
	
	{For other methods, we observe that under the normal distribution, the Inspect method outperforms the LZZL approach}. However, our two methods, which are based on the moving window and projection techniques, achieve better performance compared to both of these approaches. The primary reason for this improvement is that our bootstrap-based approach allows for the determination of critical values in multiple change-point detection, thereby enabling accurate identification of the number of change points. In contrast, LZZL is specifically designed for single change-point detection, and  Inspect  is highly sensitive to critical values, making them less effective in accurately estimating the number of change points. Furthermore, when the data exhibit heavy-tailed distributions, both  LZZL and Inspect  are significantly affected. Specifically, their performance deteriorates due to sensitivity to outliers, leading to either overestimation or underestimation of the number of change points, which in turn reduces the accuracy of multiple change-point estimation. This finding suggests that existing methods for multiple change-point estimation lack robustness to heavy-tailed data and outliers. Therefore, we recommend employing our proposed \(\bh_2(\bx,\by)\)-based  method when the data exhibit heavy-tailed properties.

	\begin{table}[H]
		\centering
		\caption{Experimental results for multiple change-point estimation  with different bandwidth choices under  mean shifts .}
		\label{tab:haus-model1}
		\setlength{\tabcolsep}{3.2pt}
		\renewcommand{\arraystretch}{1.2}
		\begin{adjustbox}{max width=\textwidth}
			\begin{tabular}{llcccccccccccc}
				\toprule
				&  & \multicolumn{3}{c}{\(G=60\)}
				& \multicolumn{3}{c}{\(G=80\)}
				& \multicolumn{3}{c}{\(G=100\)}
				& \multicolumn{3}{c}{Multiscale} \\
				\cmidrule(lr){3-5}
				\cmidrule(lr){6-8}
				\cmidrule(lr){9-11}
				\cmidrule(lr){12-14}
				Error  & Method
				& Haus & \(\widehat M_0=M_0\) & \(\widehat M_0>M_0\)
				& Haus & \(\widehat M_0=M_0\) & \(\widehat M_0>M_0\)
				& Haus & \(\widehat M_0=M_0\) & \(\widehat M_0>M_0\)
				& Haus & \(\widehat M_0=M_0\) & \(\widehat M_0>M_0\) \\
				\midrule
				
				Normal
				& \(h_1\)-initial & 6.205 & 0.945 & 0.055 & 5.210 & 0.995 & 0.005 & 5.730 & 0.990 & 0.010 & 6.460 & 0.940 & 0.060 \\
				& \(h_1\)-refined & 0.570 & 0.945 & 0.055 & 0.595 & 0.995 & 0.005 & 0.945 & 0.990 & 0.010 & 0.955 & 0.940 & 0.060 \\
				& \(h_2\)-initial & 7.205 & 0.915 & 0.075 & 5.285 & 0.995 & 0.005 & 5.375 & 0.990 & 0.010 & 6.515 & 0.930 & 0.070 \\
				& \(h_2\)-refined & 1.965 & 0.915 & 0.075 & 0.725 & 0.995 & 0.005 & 0.755 & 0.990 & 0.010 & 0.765 & 0.930 & 0.070 \\
				& LZZL    & 23.8 & 0.68 & 0.32 & - & - & - & - & -& -& - & -& - \\
				& Inspect & 20.13 & 0.63 & 0.36 & - & - & - & - & -& -& - & -& -  \\
				
				\midrule
				\(t_3\)
				& \(h_1\)-initial & 116.149 & 0.193 & 0.000 & 47.888 & 0.653 & 0.041 & 17.096 & 0.920 & 0.011 & 17.319 & 0.883 & 0.053 \\
				& \(h_1\)-refined & 114.833 & 0.193 & 0.000 & 43.088 & 0.653 & 0.041 & 11.362 & 0.920 & 0.011 & 10.729 & 0.883 & 0.053 \\
				& \(h_2\)-initial & 8.470 & 0.925 & 0.050 & 5.690 & 0.985 & 0.015 & 6.090 & 0.990 & 0.010 & 7.415 & 0.935 & 0.065 \\
				& \(h_2\)-refined & 4.100 & 0.925 & 0.050 & 1.110 & 0.985 & 0.015 & 1.115 & 0.990 & 0.010 & 1.405 & 0.935 & 0.065 \\
				& LZZL    & 26.86& 0.77 & 0.16& - & - & - & - & -& -& - & -& - \\
				& Inspect & 83.57 & 0.00 & 1.00 & - & - & - & - & -& -& - & -& -  \\
				
				\midrule
				\(t_2\)
				& \(h_1\)-initial & 155.714 & 0.000 & 0.000 & 155.714 & 0.000 & 0.000 & 127.440 & 0.200 & 0.000 & 142.517 & 0.103 & 0.000 \\
				& \(h_1\)-refined & 154.143 & 0.000 & 0.000 & 154.143 & 0.000 & 0.000 & 125.720 & 0.200 & 0.000 & 141.345 & 0.103 & 0.000 \\
				& \(h_2\)-initial & 42.335 & 0.690 & 0.050 & 8.480 & 0.955 & 0.035 & 6.545 & 0.985 & 0.015 & 7.825 & 0.910 & 0.090 \\
				& \(h_2\)-refined & 39.250 & 0.690 & 0.050 & 2.890 & 0.955 & 0.035 & 1.590 & 0.985 & 0.015 & 1.715 & 0.910 & 0.090 \\
				& LZZL    & 124.51 & 0.33 & 0.05& - & - & - & - & -& -& - & -& -  \\
				& Inspect & 87.38 & 0.00 & 1.00& - & - & - & - & -& -& - & -& -  \\
				
				\midrule
				Contam-Norm
				& \(h_1\)-initial & 158.674 & 0.000 & 0.000 & 137.023 & 0.078 & 0.039 & 80.978 & 0.439 & 0.044 & 84.373 & 0.432& 0.032 \\
				& \(h_1\)-refined & 157.826 & 0.000 & 0.000 & 135.876 & 0.078 & 0.039 & 77.050 & 0.439 & 0.044 & 80.449 & 0.432 & 0.032 \\
				& \(h_2\)-initial & 92.824 & 0.332 & 0.026 & 19.935 & 0.825 & 0.080 & 8.160 & 0.990 & 0.010 & 9.19 & 0.920 & 0.080\\
				& \(h_2\)-refined & 91.544 & 0.332 & 0.026 & 13.885 & 0.825 & 0.080 & 2.990 & 0.990 & 0.010 & 2.665 & 0.920 & 0.080 \\
				& LZZL    & 97.17 & 0.39 & 0.18 & - & - & - & - & -& -& - & -& -  \\
				& Inspect & 86.24 & 0.00 & 1.00 & - & - & - & - & -& -& - & -& -  \\
				
				\bottomrule
			\end{tabular}
		\end{adjustbox}
	\end{table}

	\begin{figure}[H]
		\centering
		\includegraphics[width=0.95\textwidth]{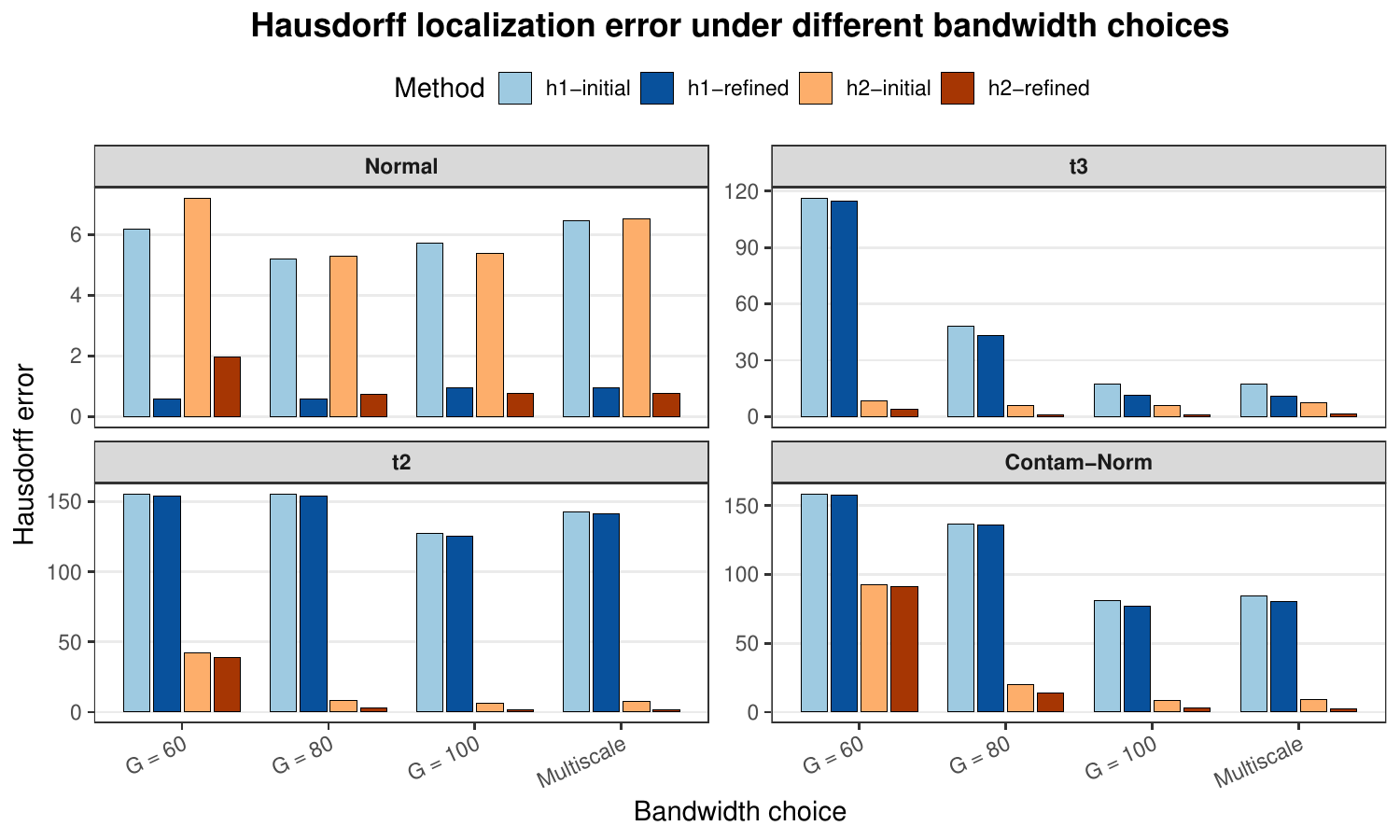}
		
		\caption{Hausdorff localization errors under different error distributions and bandwidth choices. The four panels correspond to Normal, Student \(t_3\), Student \(t_2\), and contaminated normal errors. For each error setting, we compare the initial and refined estimators under \(h_1(x,y)=y-x\) and \(h_2(x,y)=\operatorname{sign}(y-x)\), using \(G=60\), \(G=80\), \(G=100\), and the multiscale scheme.}
		\label{fig:multiscale-haus}
	\end{figure}

	\subsection{\textbf{Multiple change point inference}}\label{sec: emprical change point inference}


	Lastly, we analyze the performance of change point inference. Let $\cS_i=\{\hat{\gamma_{1}}^{(i)},\ldots,\hat{\gamma}^{(i)}_{\hat{M}^{(i)}_0}\}$ with $1\leq i\leq 200$ be the estimated change points among the 200 replications with $\hat{M}^{(i)}_0$ being the estimated change point number. Let $\hat{\Delta}_i:=\min_{1\leq m\leq \hat{M}^{(i)}_0 }=\hat{\gamma}_{m+1}^{(i)}-\hat{\gamma}_{m}^{(i)}$ be the estimated segmentation length for the $i$-th simulation. For each true change point $\gamma_m$ with $m=1,\ldots,M_0$, under each replication $i$, we first identify the estimated change point that is within a distance of less than $\hat{\Delta}_i/4$ and closest to it from $\cS_i$, say $\hat{\gamma}_{m,i}^{*}$.  Using the estimated change point $\{\hat{\gamma}_{m,i}^{*}\}_{i=1}^{200}$, we construct the corresponding confidence interval for the true change point $\gamma_m$ using (\ref{equ: estimated confidence}). Finally, we calculate the proportion of all constructed confidence intervals among the 200 replications that successfully cover the true change points.

	Table 	\ref{tab:ci-coverage-model1} presents the confidence interval coverage results for three change points using two methods: CUSUM ($\bh_1(\bx,\by)$) and Wilcoxon ($\bh_2(\bx,\by)$), under different errors  and model settings. First, when the data follows a normal distribution, both methods exhibit a good ability to approximate the theoretical coverage level (0.95) for all three change points. 
	Second, as the data transitions to a heavy-tailed distribution, such as \( t_3 \), the CUSUM-based method struggles to accurately cover the true change points due to the influence of extreme values. As the tail of the distribution becomes even heavier, the coverage rate further decreases. In contrast, the confidence intervals constructed using the Wilcoxon method continue to accurately cover the true change points, even when the data is heavy-tailed, such as in \( t_2 \) and Contam-Norm distributions. This highlights the robustness of the Wilcoxon-based method in change point inference. Note that for the CUSUM method, we do not report results for \( t_2 \), primarily because the method almost fails to detect the change points under the current model setups.

	\begin{table}[H]
		\centering
		\caption{Empirical confidence inference results  for mean shifts  with different bandwidth choices.}
		\label{tab:ci-coverage-model1}
		\small
		\setlength{\tabcolsep}{9pt}
		\renewcommand{\arraystretch}{1.18}
		\begin{tabular}{llcccccc}
			\toprule
			&  & \multicolumn{3}{c}{\(h_1(x,y)=y-x\)}
			& \multicolumn{3}{c}{\(h_2(x,y)=\operatorname{sign}(y-x)\)} \\
			\cmidrule(lr){3-5} \cmidrule(lr){6-8}
			Bandwidth & Distribution
			& \(\gamma_1\) & \(\gamma_2\) & \(\gamma_3\)
			& \(\gamma_1\) & \(\gamma_2\) & \(\gamma_3\) \\
			\midrule
			
			\(G=60\) & Norm        & 0.930 & 0.965 & 0.955 & 0.925 & 0.950 & 0.930 \\
			& \(t_3\)     & 0.871 & 0.871 & 0.953 & 0.934 & 0.904 & 0.950 \\
			& \(t_2\)     & -&  -&  -& 0.899 & 0.932 & 0.900 \\
			& Contam-Norm & 0.833 & 0.750 & 0.824 & 0.934 & 0.946 & 0.919 \\
			
			\midrule
			\(G=80\) & Norm        & 0.960 & 0.935 & 0.935 & 0.970 & 0.945 & 0.955 \\
			& \(t_3\)     & 0.856 & 0.866 & 0.898 & 0.950 & 0.905 & 0.940 \\
			& \(t_2\)     & -&  -&  - & 0.940 & 0.919 & 0.915 \\
			& Contam-Norm & 0.770 & 0.679 & 0.768 & 0.944 & 0.927 & 0.917 \\
			
			\midrule
			\(G=100\) & Norm        & 0.955 & 0.970 & 0.960 & 0.955 & 0.965 & 0.945 \\
			& \(t_3\)     & 0.862 & 0.863 & 0.891 & 0.945 & 0.925 & 0.935 \\
			& \(t_2\)     & -&  -&  -& 0.940 & 0.930 & 0.890 \\
			& Contam-Norm & 0.686 & 0.699 & 0.785 & 0.910 & 0.930 & 0.890 \\
			
			\midrule
			Multiscale & Norm        & 0.950 & 0.955 & 0.965 & 0.960 & 0.975 & 0.950 \\
			& \(t_3\)     & 0.870 & 0.899 & 0.902 & 0.925 & 0.915 & 0.930 \\
			& \(t_2\)     & -&  -&  -& 0.900 & 0.935 & 0.900 \\
			& Contam-Norm & 0.715 & 0.710 & 0.752 & 0.945 & 0.93 & 0.915\\
			
			\bottomrule
		\end{tabular}
	\end{table}

	\subsection{\textbf{Changes in the variances}}\label{sec: changes in variances}

	In this section, we consider high dimensional multiple change point detection for the variances. 
	In particular, let $\bX_1,\cdots,\bX_n$ be independent random vectors with $\bX_i\in \RR^d$ for $1 \leq i \leq n$.
	We generate data from the following  model with possible three change points at $\gamma_{1},\ldots,\gamma_{3}$: $$\bX_t=\bsigma_t\odot\bepsilon_t, ~~t=1,\ldots,n$$, where $\odot$ denotes the element-wise product (Hadamard product) of two vectors and $\{\bepsilon_t\}_{t=1}^{n}$ are i.i.d innovations with mean zero and variance ones. Under $\Hb_0$ with no change points, we set $\bsigma_t=\bsigma^{(1)}:=\sqrt{0.3}\mathbf{1}_d$ for $1\leq t\leq n$. Under $\Hb_{1,3}$ with three change points, we set
	\begin{equation*}
		\bsigma_t =\bsigma^{(1)}\mathbf{1}\{t\leq \gamma_{1}\}+\bsigma^{(2)}\{\gamma_1< t\leq \gamma_{2}\}+\bsigma^{(1)}\{\gamma_2< t\leq \gamma_{3}\}+\bsigma^{(2)}\mathbf{1}\{\gamma_3< t\},
	\end{equation*}
	where $\bsigma^{(2)}=\blambda\odot\bsigma^{(1)}$ and $\blambda=(\lambda_1,\ldots,\lambda_d)$ is the relative signal jump. We set $\lambda_j=\sqrt{2.5}$ for $1\leq j\leq 5$ and $\lambda_j=1$ for $6\leq j\leq d$, which means that the first five entries of $\bX_t$ have a change point of variances. We generate the innovations $\{\bepsilon_i\}_{i=1}^{n}$ from two types of data distributions including: (1)  Multivariate normal distributions $N(\mathbf{0},\mathbf{I}_d)$; (2)  Scaled multivariate student $t_5$ distributions $t(5,\mathbf{I}_d)/\sqrt{5/3}$. Under the current model, we consider two methods with $\bh_3(\bx,\by)=\by^2-\bx^2$ and $\bh_4(\bx,\by)=\text{sign}(\by^2-\bx^2)$ and the parameters of the methods are the same as those in Section \ref{sec: changes in mean}. 
	
	Table \ref{table: empirical size of variance} shows the size results under various data dimensions and bandwith choices. We observe that the method based on \( \bh_4(\bx,\by) \) effectively controls the significance level  across both types of distributions and different data dimensions. In contrast, while the method based on \( \bh_3(\bx,\by) \) also maintains control, it tends to be overly conservative. The primary reason for this conservativeness lies in the nature of variance change-point problems: fluctuations in variance lead to an increase in higher-order moments of the data, which subsequently enlarges the variance of the bootstrap test statistic. As a result, the test becomes excessively conservative. In comparison, rank-based methods are less affected by this issue.
	
	{Tables 	\ref{tab:variance_multicp_hausdorff} and \ref{tab:var-ci} demonstrate the multiple change point estimation and inference results. } It can be observed that, similar to the previous results, the estimation accuracy of change points improves for both methods after variance projection compared to the initial estimates. The method based on \( \bh_3(\bx,\by) \)  struggles to accurately estimate the number and locations of change points, even under the normal distribution setting, and its confidence intervals  fail to  cover the true change points. The primary reason for this issue is that variance shifts increase the likelihood of outliers in the data, and the method based on \( \bh_3(\bx,\by) \) is highly sensitive to such outliers. In contrast, the method based on \( \bh_4(\bx,\by) \)  consistently outperforms the method based on \( \bh_3(\bx,\by) \) across all model settings. Even under heavy-tailed distributions, the empirical confidence interval coverage for three change points remains close to the nominal level (95\%).

	\begin{table}[H]
		\centering
		\caption{Empirical sizes for variance change points under different kernels, dimensions, error distributions, and bandwidths.}
		\label{table: empirical size of variance}
		\begin{tabular}{llcccccc}
			\toprule
			\multirow{2}{*}{Error} 
			& \multirow{2}{*}{Bandwidth}
			& \multicolumn{3}{c}{Linear kernel}
			& \multicolumn{3}{c}{Sign kernel} \\
			\cmidrule(lr){3-5} \cmidrule(lr){6-8}
			& & $d=100$ & $d=200$ & $d=300$
			& $d=100$ & $d=200$ & $d=300$ \\
			\midrule
			\multirow{3}{*}{Normal}
			& $G=60$  & 0.000 & 0.005 & 0.005 & 0.055 & 0.040 & 0.035 \\
			& $G=80$  & 0.000 & 0.005 & 0.005 & 0.050 & 0.045 & 0.045 \\
			& $G=100$ & 0.000 & 0.005 & 0.010 & 0.055 & 0.025 & 0.050 \\
			\midrule
			\multirow{3}{*}{$t_5$}
			& $G=60$  & 0.000 & 0.000 & 0.000 & 0.035 & 0.040 & 0.030 \\
			& $G=80$  & 0.000 & 0.000 & 0.000 & 0.030 & 0.025 & 0.035 \\
			& $G=100$ & 0.000 & 0.000 & 0.000 & 0.080 & 0.030 & 0.040 \\
			\bottomrule
		\end{tabular}
	\end{table}

	
	\begin{table}[H]
		\centering
		\caption{Hausdorff distance and change-point number estimation results
			under different bandwidth settings for variance change point models. }
		\label{tab:variance_multicp_hausdorff}
		\resizebox{\textwidth}{!}{
			\begin{tabular}{llcccccccccccc}
				\toprule
				\multirow{2}{*}{Error}
				& \multirow{2}{*}{Method}
				& \multicolumn{3}{c}{$G=60$}
				& \multicolumn{3}{c}{$G=80$}
				& \multicolumn{3}{c}{$G=100$}
				& \multicolumn{3}{c}{Multiscale} \\
				\cmidrule(lr){3-5}
				\cmidrule(lr){6-8}
				\cmidrule(lr){9-11}
				\cmidrule(lr){12-14}
				&
				& Haus.
				& $\widehat{M}_0=M_0$
				& $\widehat{M}_0>M_0$
				& Haus.
				& $\widehat{M}_0=M_0$
				& $\widehat{M}_0>M_0$
				& Haus.
				& $\widehat{M}_0=M_0$
				& $\widehat{M}_0>M_0$
				& Haus.
				& $\widehat{M}_0=M_0$
				& $\widehat{M}_0>M_0$ \\
				\midrule
				
				\multirow{4}{*}{Normal}
				& $h_3$-initial
				& 112.137 & 0.200 & 0.017
				& 28.392  & 0.784 & 0.085
				& 13.920  & 0.940 & 0.050
				& 15.870  & 0.835 & 0.160 \\
				
				& $h_3$-refined
				& 110.474 & 0.200 & 0.017
				& 20.266  & 0.784 & 0.085
				& 4.505   & 0.940 & 0.050
				& 5.195   & 0.835 & 0.160 \\
				
				& $h_4$-initial
				& 26.360 & 0.795 & 0.045
				& 6.420  & 0.990 & 0.010
				& 6.365  & 1.000 & 0.000
				& 7.805  & 0.945 & 0.055 \\
				
				& $h_4$-refined
				& 22.400 & 0.795 & 0.045
				& 2.025  & 0.990 & 0.010
				& 1.815  & 1.000 & 0.000
				& 1.665  & 0.945 & 0.055 \\
				
				\midrule
				
				\multirow{4}{*}{$t_5$}
				& $h_3$-initial
				& -- & -- & --
				& -- & -- & --
				& -- & -- & --
				& -- & -- & -- \\
				
				& $h_3$-refined
				& -- & -- & --
				& -- & -- & --
				& -- & -- & --
				& -- & -- & -- \\
				
				& $h_4$-initial
				& 80.813 & 0.378 & 0.047
				& 17.120 & 0.860 & 0.085
				& 8.875  & 0.980 & 0.015
				& 10.500 & 0.880 & 0.120 \\
				
				& $h_4$-refined
				& 78.487 & 0.378 & 0.047
				& 10.580 & 0.860 & 0.085
				& 3.110  & 0.980 & 0.015
				& 2.615  & 0.880 & 0.120 \\
				
				\bottomrule
			\end{tabular}
		}
	\end{table}
	
	\begin{table}[H]
		\centering
		\caption{Empirical coverage probabilities of confidence intervals for variance change-point models.}
		\label{tab:var-ci}
		\begin{tabular}{llcccccc}
			\toprule
			\multirow{2}{*}{Bandwidth}
			& \multirow{2}{*}{Distribution}
			& \multicolumn{3}{c}{Linear kernel}
			& \multicolumn{3}{c}{Sign kernel} \\
			\cmidrule(lr){3-5} \cmidrule(lr){6-8}
			& & $\gamma_1$ & $\gamma_2$ & $\gamma_3$
			& $\gamma_1$ & $\gamma_2$ & $\gamma_3$ \\
			\midrule
			\multirow{2}{*}{$G=60$}
			& Normal & 0.625 & 0.646 & 0.780 & 0.907 & 0.910 & 0.928 \\
			& $t_5$ & -- & -- & -- & 0.921 & 0.943 & 0.909 \\
			\midrule
			\multirow{2}{*}{$G=80$}
			& Normal & 0.594 & 0.670 & 0.738 & 0.895 & 0.895 & 0.900 \\
			& $t_5$ & -- & -- & -- & 0.896 & 0.902 & 0.930 \\
			\midrule
			\multirow{2}{*}{$G=100$}
			& Normal & 0.625 & 0.693 & 0.688 & 0.910 & 0.895 & 0.930 \\
			& $t_5$ & -- & -- & -- & 0.935 & 0.920 & 0.910 \\
			\midrule
			\multirow{2}{*}{Multiscale}
			& Normal & 0.660 & 0.665 & 0.663 & 0.905 & 0.910 & 0.955 \\
			& $t_5$ & -- & -- & -- & 0.930 & 0.925 & 0.895 \\
			\bottomrule
		\end{tabular}
	\end{table}
	
	\section{Application to bladder tumor aCGH profiles}\label{sec: real data analysis}
	
	We apply the proposed method to a bladder tumor array comparative genomic hybridization 
	(aCGH) dataset. Array CGH is a genome-wide 
	technology for measuring DNA copy-number variation by comparing the hybridization 
	intensity of DNA fragments from a tumor sample with that from a reference sample. 
	The resulting observations are usually represented as log-intensity ratios along ordered 
	genomic loci. The used bladder tumor aCGH dataset was originally reported by \citet{stransky2006regional} and  contain aCGH profiles (individuals) for 57 bladder tumor samples 
	measured at 2215 ordered genomic loci. Following the preprocessing strategy used in \cite{matteson2014nonparametric},  individuals with more than \(7\%\) missing 
	measurements had been removed, and the remaining missing values had been imputed by 
	the average of their neighboring values. The resulting data matrix contains
	$n = 2215$ \text{ordered genomic loci and} and 
	$d = 43$ \text{individuals}.
	This processed dataset is publicly available as the \texttt{ACGH} dataset in the R package 
	\texttt{ecp}.

	For the implementation of our method, we considered two kernels: $ h_1(x,y) = y - x$ and   $h_2(x,y) = \operatorname{sign}(y-x)$. The linear kernel  is used as the main specification because  the classical copy-number segmentation problem is primarily concerned with changes in the mean levels. The sign kernel is used as a robustness check. Agreement between the two kernels  provides additional evidence that the detected changepoints are valid. To reduce the dependence on a single bandwidth choice, we used the multiscale version of the procedure in Appendix \ref{sec: pratical guidence} with bandwidths $ G \in \{30,40,50,60,70,80\}$.

	Table~\ref{tab:multiscale-cpts} reports that the linear kernel detects 29 change points while the the sign kernel detects more changepoints than the linear kernel, which is expected since the sign kernel can be more sensitive to local distributional changes and rank-based shifts. Despite these differences, the two kernels agree on a set of stable transition regions. Table~\ref{tab:common-cpts-linear-sign} reports the jointly supported changepoints and their localization confidence intervals. Note that two changepoints detected by the two kernels are regarded as jointly supported if their localization confidence intervals overlap.  In total, 19 representative changepoints are supported by both kernels.  Most of the corresponding confidence intervals are narrow, suggesting stable localization of the main transition regions. 
	
	In addition, \cite{YuChen2021} and \cite{yu2022robust} have also been applied to the same aCGH dataset, and their reported changepoints are listed in Table~\ref{tab:multiscale-cpts}. We regard two changepoints as overlapping if their distance is within five loci. Under this criterion, the linear-kernel based results overlap with 16 change points of \cite{YuChen2021} and 13 of \cite{yu2022robust}. The sign-kernel   overlaps with 21 of \cite{YuChen2021}and 19 of \cite{yu2022robust}. These overlaps suggest that our multiscale procedure recovers many stable changepoint locations previously identified by existing methods.
	
	Since our theoretical analysis is developed under an independence assumption, we further examine the timeseries dependence of the aCGH profiles along the ordered genomic loci.  As shown in Figure~\ref{fig:residual-acf}, the autocorrelation is positive at very small lags, especially at lag one, indicating that adjacent genomic loci are not exactly independent after centering. However, the autocorrelation decays rapidly as the lag increases. 
	Such local dependence is reasonable for aCGH profiles. Neighboring probes are physically close on the genome and may share similar experimental noise or local hybridization effects.  Nevertheless, the observed dependence is weak and short-range. In addition, our additional numerical simulations in the appendix show that the proposed new procedure remains stable under moderate serial dependence. Considering the agreement between our method and the existing techniques,
	this suggests that the proposed method remains practically applicable to this aCGH dataset, even though the real profiles may deviate from the idealized independence assumption.
	\vspace{0cm}

	\begin{figure}[H]
		\centering
		
		\begin{minipage}{\textwidth}
			\centering
			\includegraphics[width=\textwidth]{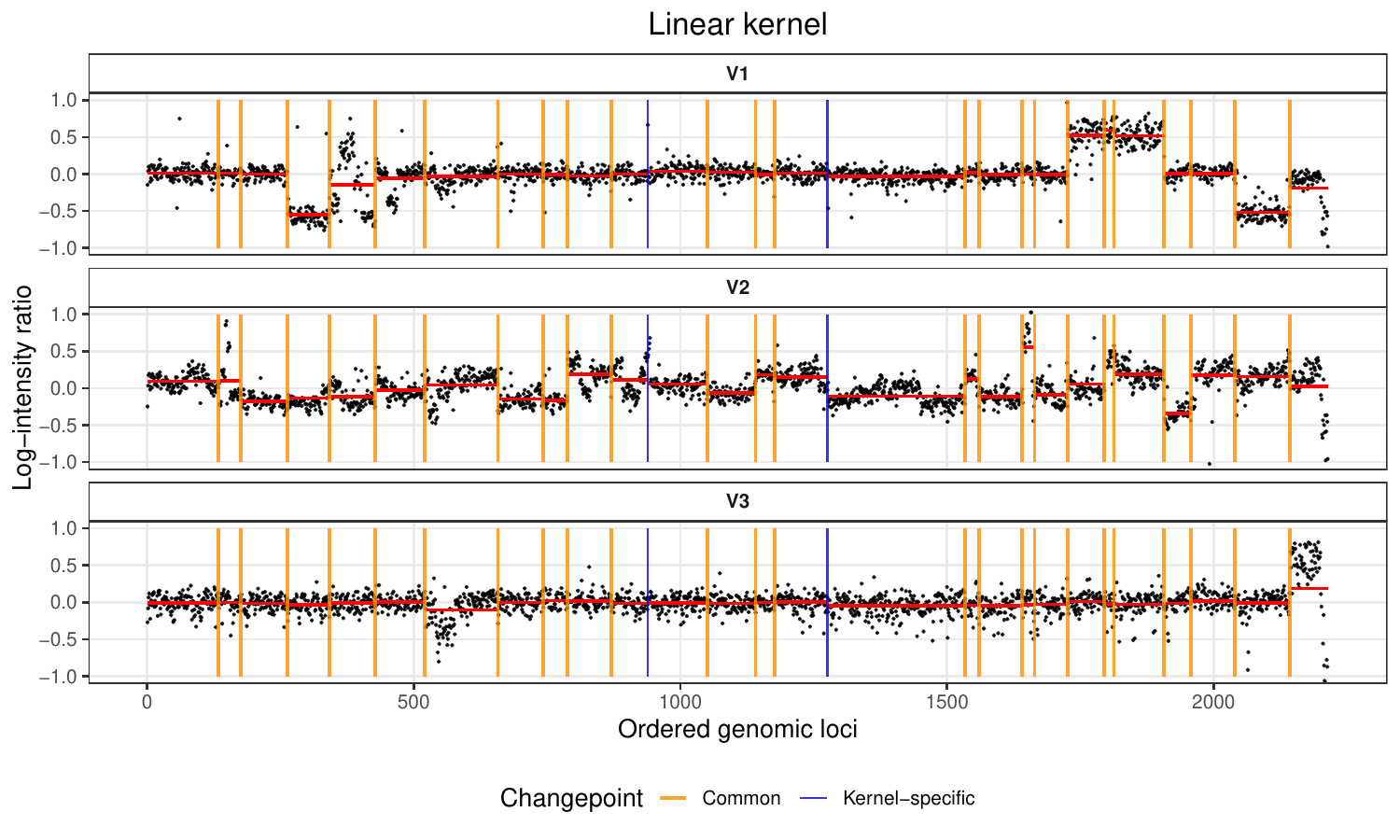}
			\caption*{(a) Linear kernel}
		\end{minipage}
		
		\vspace{0.5em}
		
		\begin{minipage}{\textwidth}
			\centering
			\includegraphics[width=\textwidth]{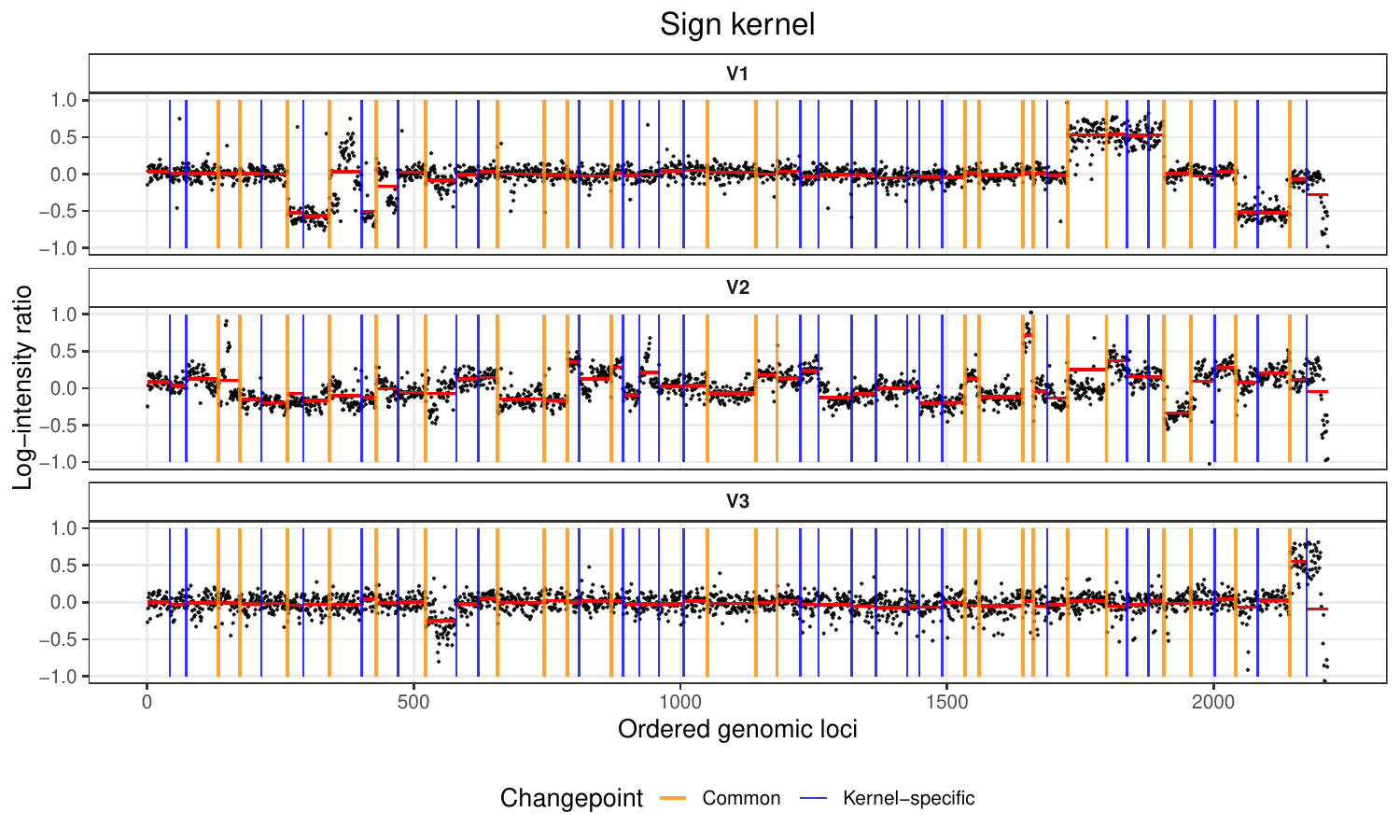}
			\caption*{(b) Sign kernel}
		\end{minipage}
		
		\caption{Segmented arrayCGH profiles using the linear and sign kernels. 
			The black points represent the observed log-intensity ratios for the first three individuals, 
			and the red horizontal segments represent the fitted segment-wise means. 
			Vertical lines indicate detected changepoints: orange lines correspond to changepoints 
			jointly supported by both the linear and sign kernels, whereas blue lines represent 
			kernel-specific changepoints detected only by the corresponding kernel.}
		\label{fig:linear-sign-segmentation}
	\end{figure}

	\vspace{0cm}
	\begin{figure}[H]
		\centering
		\includegraphics[width=0.75\textwidth]{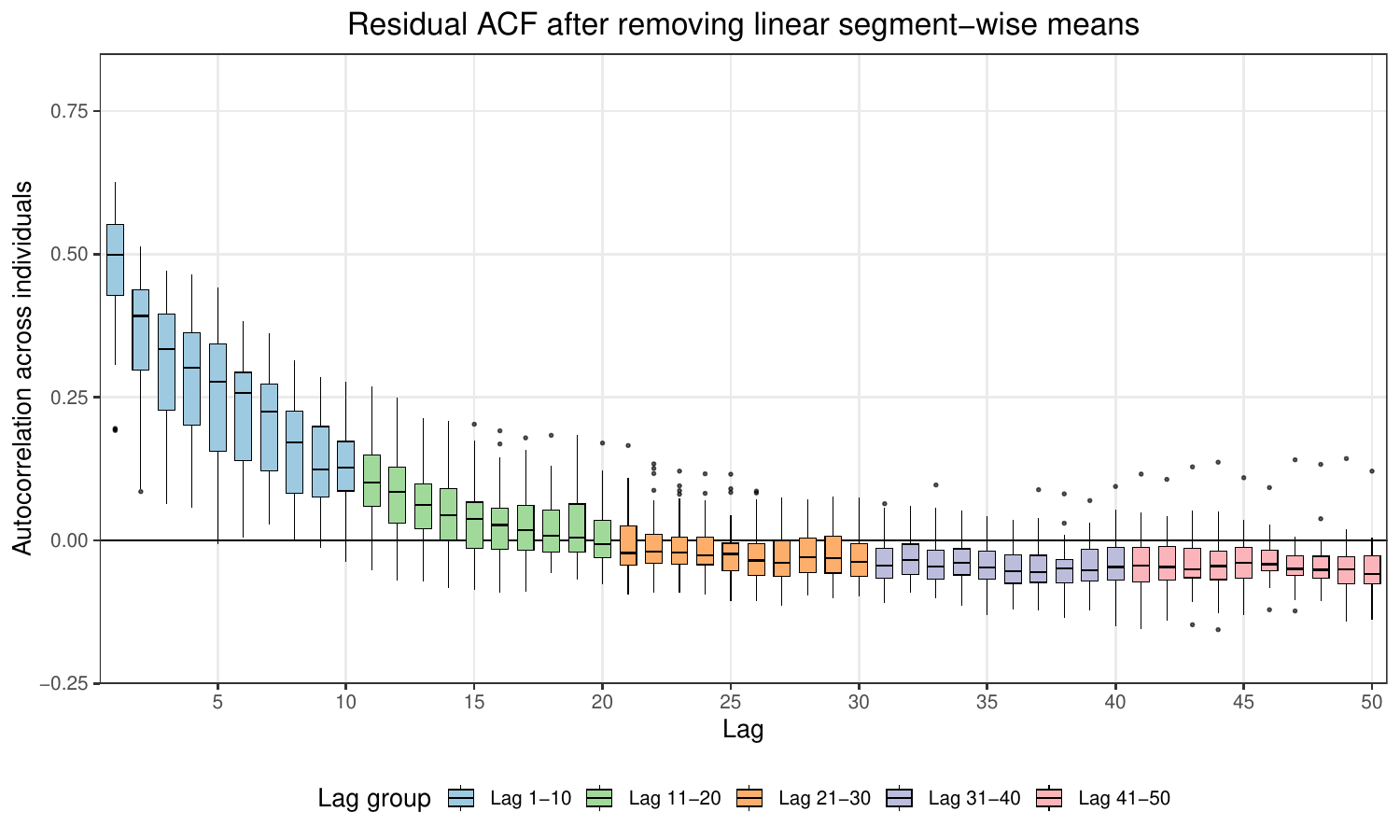}
		\caption{Residual autocorrelation after removing segment-wise means. 
			The boxplots summarize the autocorrelation coefficients across the 43 individuals at different lags.}
		\label{fig:residual-acf}
	\end{figure}
	\vspace{0cm}
	\begin{table}[htbp]
		\centering
		\caption{Detected changepoints from the linear and sign kernel based multiscale procedures}
		\label{tab:multiscale-cpts}
		\small
		\renewcommand{\arraystretch}{1.25}
		\begin{tabularx}{\textwidth}{lX}
			\toprule
			Method & Detected changepoints \\
			\midrule
			$h_1$-initial 
			& 134, 185, 263, 344, 428, 521, 657, 670, 741, 809, 871, 961, 1050, 1141, 1177, 1259, 1534, 1562, 1641, 1664, 1724, 1771, 1792, 1827, 1876, 1963, 2031, 2120, 2142 \\
			\hline
			$h_1$-refined 
			& 134, 176, 263, 342, 428, 521, 657, 658, 743, 788, 871, 939, 1051, 1141, 1177, 1276, 1534, 1560, 1641, 1664, 1726, 1795, 1795, 1813, 1907, 1957, 2040, 2143, 2143 \\
			\hline
			$h_2$-initial 
			& 43, 73, 134, 173, 214, 263, 297, 342, 401, 450, 479, 522, 561, 581, 621, 657, 744, 788, 811, 871, 891, 922, 958, 1010, 1051, 1141, 1177, 1225, 1259, 1319, 1367, 1425, 1456, 1489, 1534, 1564, 1641, 1664, 1690, 1724, 1795, 1839, 1878, 1906, 1956, 1992, 2041, 2079, 2102, 2143, 2174 \\
			\hline
			$h_2$-refined 
			& 43, 73, 134, 174, 214, 263, 293, 342, 402, 430, 471, 522, 580, 581, 621, 657, 744, 788, 810, 871, 892, 923, 960, 1006, 1051, 1142, 1181, 1225, 1259, 1321, 1367, 1425, 1448, 1491, 1534, 1560, 1642, 1661, 1688, 1726, 1799, 1837, 1878, 1906, 1957, 2002, 2041, 2079, 2084, 2143, 2174 \\
			\hline 
			\cite{YuChen2021}
			& 73, 185, 263, 342, 428, 521, 581, 657, 741, 801, 871, 960, 1051, 1141, 1216, 1276, 1367, 1427, 1503, 1563, 1664, 1724, 1836, 1905, 1965, 2044, 2143 \\
			\hline
			\cite{yu2022robust}
			& 74, 136, 174, 248, 280, 344, 448, 528, 544, 624, 658, 744, 810, 876, 932, 1022, 1050, 1140, 1220, 1282, 1366, 1418, 1500, 1560, 1642, 1726, 1850, 1908, 1964, 2022, 2084, 2142 \\
			\bottomrule
		\end{tabularx}
	\end{table}
	\vspace{0cm}
	\begin{table}[H]
		\centering
		\caption{Common changepoints detected by the linear and sign multiscale procedures}
		\label{tab:common-cpts-linear-sign}
		\small
		\renewcommand{\arraystretch}{1.15}
		\begin{tabular}{c|ccc}
			\hline
			No. & Linear refined cp & Linear CI & Sign refined cp \quad Sign CI \\
			\hline
			1  & 134  & [132.21, 134.48] & 134  \quad [133.28, 134.21] \\
			2  & 263  & [262.11, 263.77] & 263  \quad [262.30, 263.88] \\
			3  & 342  & [341.13, 342.73] & 342  \quad [340.18, 342.82] \\
			4  & 428  & [425.98, 428.38] & 430  \quad [427.15, 431.10] \\
			5  & 521  & [519.59, 524.29] & 522  \quad [521.73, 522.44] \\
			6  & 658  & [657.08, 659.16] & 657  \quad [656.36, 657.13] \\
			7  & 743  & [742.19, 743.41] & 744  \quad [742.72, 744.40] \\
			8  & 788  & [787.26, 789.12] & 788  \quad [787.50, 788.75] \\
			9  & 871  & [868.44, 874.92] & 871  \quad [870.01, 871.41] \\
			10 & 1051 & [1049.26, 1052.32] & 1051 \quad [1049.13, 1052.61] \\
			11 & 1141 & [1139.83, 1141.44] & 1142 \quad [1141.24, 1142.37] \\
			12 & 1534 & [1531.88, 1535.76] & 1534 \quad [1532.42, 1535.29] \\
			13 & 1560 & [1556.99, 1561.06] & 1560 \quad [1557.80, 1561.59] \\
			14 & 1641 & [1639.96, 1641.46] & 1642 \quad [1641.38, 1643.54] \\
			15 & 1726 & [1725.24, 1726.93] & 1726 \quad [1725.24, 1726.73] \\
			16 & 1907 & [1905.33, 1907.93] & 1906 \quad [1905.40, 1906.13] \\
			17 & 1957 & [1956.43, 1958.04] & 1957 \quad [1956.59, 1958.96] \\
			18 & 2040 & [2035.31, 2040.61] & 2041 \quad [2040.22, 2041.32] \\
			19 & 2143 & [2142.37, 2143.31] & 2143 \quad [2142.73, 2143.73] \\
			\hline
		\end{tabular}
	\end{table}
	
	
	\vspace{0cm}
	\section{ Summary}\label{section: summary and discussion}
	This paper introduces a general framework for high-dimensional change-point detection, estimation, and inference based on U-statistics within a moving window. By enabling flexible kernel function selection, our method extends beyond traditional mean-based change-point analysis to accommodate a broader range of structural changes, including variance  and robust mean shifts. We develop a minimax-optimal test statistic tailored for sparse alternatives and propose a high-dimensional multiplier bootstrap procedure that ensures valid significance control while maintaining high power. For estimation, we construct an initial estimator for the number and locations of change points and refine it using the U-statistic Projection Refinement Algorithm (U-PRA), achieving minimax-optimal localization rates. Furthermore, we derive the asymptotic distribution of the refined estimators, facilitating the construction of valid confidence intervals. The effectiveness of our method is demonstrated through both theoretical analysis and empirical validation on synthetic and real-world datasets. Interesting future directions include extending the change-point analysis framework to other statistical settings, such as graphical models.

\bibliographystyle{ims}
\bibliography{ref_robust}

\clearpage

\begin{center}
	{\Large \bfseries Supplementary Material to ``A General U-Statistic Framework for}\\[0.6em]
	{\Large \bfseries High-Dimensional Multiple Change-Point Analysis''}
\end{center}

\vspace{1em}

\appendix

\section{Several specific examples}\label{sec: examples}

\noindent\textbf{Example 1.}	When the kernel function is defined as $h(x, y) = y - x$, the resulting U-statistic focuses on detecting changes in the mean of the distribution. 
Specifically, in this case, we have $\theta_j^{(m)}=\E (X_{\gamma_m+1}-X_{\gamma_m})$  as  the mean shift and $T_j(k)$ reduces to
\begin{equation*}
	\small
	T_{j}(k)=	\dfrac{1}{\sqrt{G}}\Big(\sum\limits_{t_1=k-G+1}^kX_{t_1,j}-\sum\limits_{t_2=k+1}^{k+G}X_{t_2,j}\Big),~~\text{for}~~j=1,\ldots,d.
\end{equation*}
which has been recently extensively studied in \cite{eichinger2018mosum,chen2022InferenceBreakpointsHighdimensionalb}.

\noindent\textbf{Example 2.}	To handle cases where the data may contain outliers or heavy-tailed distributions, we consider a robust kernel function:
$h(x, y) = \operatorname{sign}(y - x)$, where $\text{sign}(x)=\mathbf{1}\{x\geq 0\}-\mathbf{1}\{x< 0\}$ is the indicator function.  This kernel function leads to a \textbf{rank-based U-statistic}, making the test resistant to the influence of outliers when detecting data with mean shifts. To see this,
$X_{\gamma_m,j}=\mu_j^{(m)}+\epsilon_j$ and $X_{\gamma_m+1,j}=\mu_j^{(m+1)}+\epsilon'_j$ and $\epsilon_j$, $\epsilon_j'$ are the random erros with mean zeros. Let $\Delta_j^{(m)}:=\mu_j^{(m+1)}-\mu_j^{(m)}$ be the signal jump at the $m$-th change point for the $j$-th coordidate and $\xi_j=\epsilon'_j-\epsilon_{j}$. By the defintion of $\theta_j^{(m)}$, we have
\begin{equation}
	\begin{array}{ll}
		\theta^{(m)}_j&=\E h(X_{\gamma_m,j},X_{\gamma_m+1,j})\\
		&=_{(1)}\E[\mathbf{1}\{\mu_j^{(m+1)}-\mu_j^{(m)}+\epsilon'_j-\epsilon_j\geq 0\}]-\E[\mathbf{1}\{\mu_j^{(m+1)}-\mu_j^{(m)}+\epsilon'_j-\epsilon_j< 0\}]\\
		&=_{(2)}\E\mathbf{1}\{\xi_j\geq -\Delta_j^{(m)}\}-\E\mathbf{1}\{\xi_j< -\Delta_j^{(m)}\}
		=_{(3)}-2(F_j(-\Delta_j^{(m)})-F(0))\approx_{(4)} 2f_j(0)\Delta_j^{(m)},
	\end{array}
\end{equation}
where $F_j(x)$ and $f_j(x)$ are the CDF and pdf of $\xi_j$ respectively, Equation (3) is derived from the symmetry of $\xi_j$, and  (4) follows from a Taylor expansion.  Thus, $\theta_j^{(m)}$ with vanishing signal jumps encapsulates information about the signal jump  and  shares equivalent magnitudes with $\Delta_j^{(m)}$, indicating that any distortion of the signal through the sign kernel is limited to a multiplicative constant. 


\noindent\textbf{Example 3.}	Consider a sequence of observations $\{X_{t,j}\}_{t=1}^n$ with $j=1,\ldots,d$, where the variance changes at $\gamma_m$:
\[
X_{t,j} =
\begin{cases}
	\sigma_{j}^{(m)}\times\epsilon_{t,j}	\sim (0, (\sigma_{j}^{(m)})^2), & \gamma_{m-1}<t \leq \gamma_m, \\
	\sigma_{j}^{(m+1)}\times\epsilon'_{t,j}\sim	(0, (\sigma_{j}^{(m+1)})^2), &  \gamma_m<t\leq \gamma_{m+1},
\end{cases}
\]
where $(\sigma_{j}^{(m+1)})^2=\lambda_j^{(m)}(\sigma_{j}^{(m)})^2$ with $\lambda_j^{(m)}$ represents the relative change in variance and $\epsilon_{t,j}$ or $\epsilon'_{t,j}$ are i.i.d random variables with mean zero and variance one.  We can use the kernel $h(x, y) = y^2 - x^2$. In this case, $\theta_{j}^{(m)}=(\lambda^{(m)}_j-1)\sigma_{1,j}^2$ denotes the value of the variance shift. If we use $h(x,y)=\text{sign}(y^2-x^2)$, we can construct a robust variance change-point detection method that is more effective for heavy-tailed data. Specifically, similar to Example 2, we can prove that $\theta_{j}^{(m)}$ is proportional to the relative changes of variances with
\begin{equation*}
	\theta_{j}^{(m)}=\E h(X_{\gamma_m,j},X_{\gamma_m+1,j})\approx 2f_{\epsilon^2_{t,j}/\epsilon'^2_{t,j}}(1)(\lambda^{(m)}_j-1),
\end{equation*}
where $f_{\epsilon^2_{t,j}/\epsilon'^2_{t,j}}(x)$ is the pdf of $\epsilon^2_{t,j}/\epsilon'^2_{t,j}$.

\section{Possible extensions to covariance matrices and  distribution function changes}\label{sec: Possible extensions to covariance}

We briefly discuss two possible  extensions of our proposed  framework. 

\textbf{Example 1: high dimensional  covariance matrix $\bSigma $}.  Consider covariance matrix changes. Suppose, for simplicity, that the observations are mean zero within each segment. Define the transformed variable
\[
\bZ_i=\operatorname{vech}(\bX_i\bX_i^\top)\in\mathbb R^q,
\qquad q=\frac{d(d+1)}{2}.
\]
Then we have 
\[
\E(\bZ_i)=\operatorname{vech}\{\E(\bX_i\bX_i^\top)\}
=\operatorname{vech}\{\operatorname{Cov}(\bX_i)\}.
\]
Therefore, detecting a change in the covariance matrix of \(\bX_i\) is equivalent to detecting a mean change in the transformed sequence \(\bZ_i\).

For the transformed data \(\bZ_i\), we can use the same antisymmetric two-sample kernel as in the mean-change point  case,
\[
h(\bz_i,\bz_j)=\bz_j-\bz_i.
\]
For a given window size \(G\), define the left and right local averages
\[
\bar \bZ_k^-=\frac{1}{G}\sum_{i=k-G+1}^{k} \bZ_i,
\qquad
\bar \bZ_k^+=\frac{1}{G}\sum_{i=k+1}^{k+G} \bZ_i,
\qquad G\le k\le n-G.
\]
Then have 
\[
\bar \bZ_k^+-\bar \bZ_k^-
=
\frac{1}{G^2}
\sum_{i=k-G+1}^{k}\sum_{j=k+1}^{k+G}
h(\bZ_i,\bZ_j).
\]
Thus, the testing statistic for covariance matrix changes is
\[
T^{\bSigma}
=
\max_{G\le k\le n-G}
\sqrt{G}\,
\left\|
\bar \bZ_k^+-\bar \bZ_k^-
\right\|_\infty.
\]
This has the same form as the statistic in the main paper, except that the original data dimension is replaced by the transformed dimension \(q=d(d+1)/2\). A rigorous theory would require verifying the corresponding tail, non-degeneracy, and signal-to-noise conditions for the transformed variables \(\bZ_i\).

\textbf{Example 2: change point in distribution functions}.  Consider changes in marginal distribution functions at finitely  fixed points.  Specifically, let \(t_1,\ldots,t_m\in \RR\) be prespecified values. For each observation \(\bX_i=(X_{i1},\ldots,X_{id})^\top\), define the transformed vector
\[
\bZ_i
=
\big(
\mathbf 1\{X_{i1}\le t_1\},\ldots,\mathbf 1\{X_{i1}\le t_m\},
\ldots,
\mathbf 1\{X_{id}\le t_1\},\ldots,\mathbf 1\{X_{id}\le t_m\}
\big)^\top
\in\mathbb R^{dm}.
\]
The mean of this transformed vector is
\[
\E(\bZ_i)
=
\big(
F_1(t_1),\ldots,F_1(t_m),
\ldots,
F_d(t_1),\ldots,F_d(t_m)
\big)^\top,
\]
where \(F_j\) denotes the marginal distribution function of the \(j\)-th coordinate of $\bX$. Hence, detecting changes in the marginal distribution functions at the fixed grid points \(t_1,\ldots,t_m\) can be reduced to detecting a mean change in \(\bZ_i\).

Using the same antisymmetric kernel
$h(\bz_i,\bz_j)=\bz_j-\bz_i$, on the transformed data \(\bZ_i\), the corresponding moving-window  based statistic can be written as
\[
T^{F}
=
\max_{G\le k\le n-G}
\sqrt{G}
\left\|
\frac{1}{G^2}
\sum_{i=k-G+1}^{k}
\sum_{j=k+1}^{k+G}
h(\bZ_i,\bZ_j)
\right\|_\infty .
\]

Note that, in this  setting, the above statistic can also be written directly in terms of the original observations. Specifically, define the componentwise antisymmetric kernel
\[
h_{r,\ell}(\bX_i,\bX_j)
=
\mathbf 1\{X_{j,r}\le t_\ell\}
-
\mathbf 1\{X_{i,r}\le t_\ell\},
\qquad
1\le r\le d,\quad 1\le \ell\le m .
\]
Then, based on the new kernel,  the corresponding  statistic can be written equivallently  as
\[
T^F
=
\max_{G\le k\le n-G}
\sqrt{G}
\max_{1\le r\le d,\ 1\le \ell\le m}
\left|
\frac{1}{G^2}
\sum_{i=k-G+1}^{k}
\sum_{j=k+1}^{k+G}
h_{r,\ell}(\bX_i,\bX_j)
\right|.
\]

Thus, under finite  fixed points, marginal distribution changes can be embedded into our current framework either by transforming the data into indicators or by choosing the above new  kernel.

Note that if we aim to detect changes in the full marginal distribution functions, rather than their values at fixed points $t_1,\ldots,t_m$,  the problem becomes substantially different. For example, to test changes in the full marginal distribution functions, we may consider the following kernel
\[
h_{r,t}(\bX_i,\bX_j)
=
\mathbf 1\{X_{j,r}\le t\}
-
\mathbf 1\{X_{i,r}\le t\},
\qquad
1\le r\le d,\quad t\in\mathbb R .
\]
Then the corresponding moving-window based statistic would take the form
\[
T^{F,\infty}
=
\max_{G\le k\le n-G}
\sqrt{G}
\max_{1\le r\le d}
\sup_{t\in\mathbb R}
\left|
\frac{1}{G^2}
\sum_{i=k-G+1}^{k}
\sum_{j=k+1}^{k+G}
h_{r,t}(\bX_i,\bX_j)
\right|.
\]
Compared with the finite-grid statistic, the maximum is  taken not only over the marginal coordinate \(r\), but also over a continuum of  points \(t\in \RR\). Therefore, the parameter is no longer finite-dimensional. Instead, the statistic is an indexed-kernel process over the class
\[
\mathcal H
=
\{h_{r,t}:1\le r\le d,\ t\in\mathbb R\}.
\]

This setting is beyond the scope of the present paper. Our current theory is developed for finite-dimensional vector-valued kernels and the resulting moving-window maximum statistic. In contrast, the above statistic would require controlling the supremum of an infinite-dimensional indexed process jointly over the scanning location \(k\), the marginal coordinate \(r\), and the continuum of the points \(t\).

We note that there is relevant theoretical work on \(U\)-process suprema recently. For example, \cite{ChenKato2020JMB} studied Gaussian and jackknife multiplier bootstrap approximations for statistics of the form
\[
Z_n=\sup_{h\in\mathcal H}\frac{\mathbb U_n(h)}{r},
\]
where, for each fixed \(h\), \(U_n(h)\) is a scalar \(U\)-statistic and \(\mathcal H\) is a possibly infinite function class. Their results provide finite-sample Gaussian and bootstrap approximations to the supremum of a \(U\)-process. However, their setting does not consider change-point detection, moving-window scanning, multiple change-point localization, or projection-based refinement and inference. Extending our framework to the full distribution-function case would therefore require new arguments. For example, we may need to combine moving-window change-point analysis with empirical-process and \(U\)-process tools, including entropy conditions for the indexed class, Gaussian coupling, anti-concentration, multiplier bootstrap approximation, and uniform control of Hoeffding projection and degenerate remainder terms over both \(k\) and \(t\). We view this as an interesting but nontrivial direction for future research.

\newpage

\section{{Comparisons with \cite{Zhou2025BootstrapMOSUM}}}\label{sec: Comparisons with zhou}

\subsection{\textbf{Connection with \cite{Zhou2025BootstrapMOSUM}}}

We briefly discuss the connection between our method and the bootstrap MOSUM approach of \cite{Zhou2025BootstrapMOSUM}. 

\textbf{First}, both methods are developed for high-dimensional change-point analysis and share a moving-window/MOSUM-type construction. \cite{Zhou2025BootstrapMOSUM} focus on high-dimensional mean changes and compare the sample means in two adjacent local windows.  In particular, when \(h(x,y)=y-x\), our statistic reduces to \cite{Zhou2025BootstrapMOSUM}
which can be viewed as a special case of our framework.

\textbf{Second}, the two methods use a similar high-dimensional aggregation strategy. \cite{Zhou2025BootstrapMOSUM} aggregate the coordinate-wise MOSUM statistics by the \(\ell_\infty\)-norm, while our testing statistic also uses an \(\ell_\infty\)-norm aggregation of coordinate-wise moving-window U-statistics. Thus, both methods are designed to be sensitive to sparse high-dimensional changes and both require theoretical control of maximum-type statistics over coordinates and candidate locations.

\textbf{Third}, regarding  the assumptions and bootstrap theory, \cite{Zhou2025BootstrapMOSUM} impose non-degeneracy, finite moment, and sub-exponential or maximum moment conditions on the noise variables. In our framework, analogous conditions are imposed on the kernel and its Hoeffding projections. In terms of bootstrap approximation, \cite{Zhou2025BootstrapMOSUM} obtain an error term of the order
$\left\{\frac{\log^7(n d)}{G}\right\}^{1/6}$. Our multiplier bootstrap approximation has a similar high-dimensional structure and requires 
\[
\frac{\log^7\{(n-2G+1)d\}}{G}\to 0.
\]
Hence, both theories require the local window size to be sufficiently large relative to the logarithmic complexity induced by the dimension and the number of scanning locations.

\textbf{Fourth}, the power conditions in \cite{Zhou2025BootstrapMOSUM} and in our paper share the same high-dimensional signal-to-noise structure. In both cases, the local jump size, measured in an \(\ell_\infty\)-norm, needs to dominate the stochastic fluctuation of a high-dimensional maximum statistic. Up to logarithmic factors and significance-level constants, the detectable signal strength has the common form
\[
\text{signal size}
\gtrsim
\sqrt{
	\frac{\log(\text{dimension} \times \text{number of scanning locations})}
	{\text{effective local sample size}}
}.
\]
Thus, both power results reflect the same sparse high-dimensional detection principle.

\subsection{\textbf{Differences from \cite{Zhou2025BootstrapMOSUM}}}

Although our method and \cite{Zhou2025BootstrapMOSUM} share a similar moving-window/MOSUM-type procedure, there are several important differences.

\textbf{First}, the model settings are different. \cite{Zhou2025BootstrapMOSUM} mainly focus on high-dimensional mean-change detection under an at-most-one-change-point model. In contrast, our paper considers a more general high-dimensional parameter change model based on two-sample U-statistics. By choosing different kernels, our framework can cover mean changes, variance changes, robust mean changes, and robust variance changes. Moreover, our model allows multiple change points and aims to estimate both the number and the locations of these change points.

\textbf{Second}, the assumptions are imposed on different objects. Since \cite{Zhou2025BootstrapMOSUM} study mean changes, their assumptions for bootstrap approximation, size control, and power analysis are mainly imposed on the original observations , such as non-degeneracy, moment, and tail conditions. In our framework, the statistic is a two-sample U-statistic, and therefore many regularity conditions are naturally imposed on the kernel function and its Hoeffding projections.

\textbf{Third}, the proof strategy is different because our method handles both multiple change points and U-statistics. In the multiple-change-point setting, the moving-window bandwidth \(G\) needs to satisfy a two-sided requirement. It should be large enough to control high-dimensional stochastic fluctuations and bootstrap approximation errors, but it should also be smaller than the minimum spacing between adjacent change points so that one local window does not cover more than one structural break. Therefore, the theoretical requirements for power and localization involve a more delicate scaling relationship among the bandwidth \(G\), the minimum spacing between adjacent change points, and the minimum signal jump size. In addition, the U-statistic structure requires a Hoeffding decomposition. The leading stochastic fluctuation is driven by the first-order Hoeffding projections, while the degenerate remainder terms need to be controlled uniformly over both coordinates and scanning locations. These issues do not arise in the same form for the mean-based MOSUM statistic in \cite{Zhou2025BootstrapMOSUM}.

\textbf{Finally}, our paper further develops estimation refinement and statistical inference, which are not the focus of \cite{Zhou2025BootstrapMOSUM}. After obtaining initial estimators, we introduce a U-statistic projection refinement algorithm. Furthermore, we derive the limiting distribution of the refined estimator and use it to construct confidence intervals for the change-point locations. This inference problem is not considered in \cite{Zhou2025BootstrapMOSUM}.

Overall, \cite{Zhou2025BootstrapMOSUM} provide an important bootstrap MOSUM approach for high-dimensional mean-change detection. Our work extends this line of ideas to a general two-sample U-statistic framework, allows multiple change points, develops projection-based refinement, and provides limiting distribution theory and confidence interval construction for the refined estimators.

\newpage

\section{{Verification of the  assumptions}}\label{sec: Verification of the  assumptions}

In this subsection, we verify the kernel-level assumptions for several representative
kernels considered in Appendix \ref{sec: examples}. The objective is to provide sufficient conditions,
stated directly in terms of the data-generating distribution and the choice of kernel
function, under which Assumptions \textbf{(A.1)--(A.3)} hold.

Let \(X_j\) and \(Y_j\) be two independent copies from the same marginal distribution
\(F_j\) within a stationary segment. For an antisymmetric kernel \(h(x,y)\), define
the first-order Hoeffding projection by
\[
h_{1,j}(X_j)
=
\E\{h(X_j,Y_j)\mid X_j\}.
\]
Assumptions (A.1)--(A.3) require that \(h_{1,j}(X_j)\) is non-degenerate, has a
uniformly bounded sub-exponential norm, and has uniformly bounded third and fourth
moments. Specifically, for some constants \(b>0\) and \(D>0\), these conditions can
be written as
\[
\E\{h_{1,j}^2(X_j)\} \geq b,
\qquad
\|h_{1,j}(X_j)\|_{\psi_1}\leq D,
\]
and
\[
\E|h_{1,j}(X_j)|^{2+\ell}\leq D^\ell,
\qquad \ell=1,2,
\]
uniformly over \(j=1,\ldots,d\). 
Next, we give sufficient conditions under which the three assumptions hold for the four typed kernels. 

\noindent\textbf{Example 1. Mean changes with \(h(x,y)=y-x\).}
Consider the linear kernel $h(x,y)=y-x.$
Let \(\mu_j=\E X_j\). Then we have 
\[
h_{1,j}(x)
=
\E(Y_j-x)
=
\mu_j-x,
\]
and hence
\[
h_{1,j}(X_j)
=
-(X_j-\mu_j).
\]
Therefore, the sub-exponential condition in Assumption (A.2) is satisfied if
\[
\sup_{1\leq j\leq d}
\|X_j-\mu_j\|_{\psi_1}<\infty.
\]
This condition is implied, in particular, by uniform sub-exponential distributions  of the coordinates. Under the same condition,
Assumption (A.3) follows from standard moment bounds for sub-exponential random
variables. The non-degeneracy condition in Assumption (A.1) reduces to
\[
\inf_{1\leq j\leq d}
\operatorname{Var}(X_j)>0.
\]
Thus, for the unbounded mean-change kernel \(h(x,y)=y-x\), the kernel-level assumptions
are implied by conventional light-tail and  non-degenerate variance conditions on the
original observations $\bX=(X_1,\ldots,X_p)$.

\noindent\textbf{Example 2. Robust mean changes with \(h(x,y)=\operatorname{sign}(y-x)\).} Next, consider the rank-based kernel $h(x,y)=\operatorname{sign}(y-x)$. Its first-order Hoeffding's projection is
\[
h_{1,j}(x)
=
\E\{\operatorname{sign}(Y_j-x)\}
=
P(Y_j\geq x)-P(Y_j<x).
\]
If \(F_j\) is continuous, then we have 
\[
h_{1,j}(x)=1-2F_j(x).
\]
In all cases, we have 
\[
|h_{1,j}(x)|\leq 1.
\]
Consequently, for some universal constant \(C>0\),
\[
\|h_{1,j}(X_j)\|_{\psi_1}\leq C,
\qquad
E|h_{1,j}(X_j)|^q\leq 1,
\quad q\geq 1.
\]
Hence Assumptions (A.2) and (A.3) hold automatically, without imposing sub-Gaussian
or sub-exponential tail conditions on \(X_j\). This is the key reason why the rank-based
kernel is suitable for heavy-tailed or contaminated distributions.

It remains to verify non-degeneracy. A sufficient condition is
\[
\inf_{1\leq j\leq d}
\operatorname{Var}\{h_{1,j}(X_j)\}>0.
\]
For example, when \(F_j\) is continuous, \(F_j(X_j)\sim U(0,1)\), and therefore
\[
\operatorname{Var}\{h_{1,j}(X_j)\}
=
\operatorname{Var}\{1-2F_j(X_j)\}
=
\frac{1}{3}.
\]
Thus, the robust mean-change kernel satisfies Assumptions (A.1)--(A.3) under mild
non-degeneracy conditions on the marginal distribution.

\noindent\textbf{Example 3. Variance changes with \(h(x,y)=y^2-x^2\).} For the variance-change kernel $h(x,y)=y^2-x^2$
we have
\[
h_{1,j}(x)
=
\E(Y_j^2-x^2)
=
\E X_j^2-x^2.
\]
Therefore,
\[
h_{1,j}(X_j)
=
-(X_j^2-E X_j^2).
\]
The sub-exponential condition in Assumption (A.2) is satisfied if
\[
\sup_{1\leq j\leq d}
\|X_j\|_{\psi_2}<\infty.
\]
Indeed, under the above condition, we have
\[
\|X_j^2-E X_j^2\|_{\psi_1}
\leq
C\|X_j\|_{\psi_2}^2
\]
for a universal constant \(C>0\), uniformly over \(j=1,\ldots,d\). Assumption (A.3)
then follows from the same sub-exponential bound. The non-degeneracy condition in
Assumption (A.1) is implied by
\[
\inf_{1\leq j\leq d}
\operatorname{Var}(X_j^2)>0.
\]
Therefore, for the unbounded variance-change point kernel \(h(x,y)=y^2-x^2\), a convenient
sufficient data-level condition is uniform sub-Gaussianity of the original data, together
with uniform non-degeneracy of \(X_j^2\) for $j=1,\ldots,d$.

\noindent\textbf{Example 4. Robust variance changes with \(h(x,y)=\operatorname{sign}(y^2-x^2)\)}.  Finally, consider the bounded kernel $h(x,y)=\operatorname{sign}(y^2-x^2)$.
Its first-order Hoeffding's projection is
\[
h_{1,j}(x)
=
\E\{\operatorname{sign}(Y_j^2-x^2)\}
=
\P(Y_j^2\geq x^2)-\P(Y_j^2<x^2).
\]
Since
\[
\operatorname{sign}(Y_j^2-x^2)\in\{-1,1\},
\]
we have $|h_{1,j}(x)|\leq 1$. Consequently, for some universal constant \(C>0\),
\[
\|h_{1,j}(X_j)\|_{\psi_1}\leq C,
\qquad
E|h_{1,j}(X_j)|^q\leq 1,
\quad q\geq 1.
\]
Hence Assumptions (A.2) and (A.3) hold automatically, without any light-tail requirement
on \(X_j\).

For Assumption (A.1), it is sufficient to require
\[
\inf_{1\leq j\leq d}
\operatorname{Var}
\left[
\P(Y_j^2\geq X_j^2)-\P(Y_j^2<X_j^2)
\right]>0.
\]
If the distribution of \(X_j^2\) is continuous, denoted by \(F_{j,2}\), then
\[
h_{1,j}(X_j)
=
1-2F_{j,2}(X_j^2),
\]
and \(F_{j,2}(X_j^2)\sim U(0,1)\). Therefore, we have 
\[
\operatorname{Var}\{h_{1,j}(X_j)\}
=
\operatorname{Var}\{1-2F_{j,2}(X_j^2)\}
=
\frac{1}{3}.
\]
Thus, the robust variance-change point kernel satisfies Assumptions (A.1)--(A.3) automatically without conditions on the distribution of \(X_j^2\).



\newpage

\section{{Implementation details for making confidence interval}}
\label{sec: implementation for estimation confidence}

In this section, we describe the construction of confidence intervals for the detected change points. 
Theorem~\ref{theorem: refined estimation} shows that, after refinement, the centered and scaled estimator satisfies an argmax limiting distribution of a drifted weighted Brownian motion process. 
This result provides a basis for constructing both marginal and simultaneous confidence intervals.

Let \(q^{(m)}_{\alpha/2}\) and \(q^{(m)}_{1-\alpha/2}\) denote the \(\alpha/2\) and \(1-\alpha/2\) quantiles of 
\[
\underset{s\in\mathbb R}{\arg\max}\, Z^{(m)}(s).
\]
Then, for each change point \(\gamma_m\), the theoretical marginal confidence interval is given by
\[
\Big[\tilde l_{\alpha}(m),\tilde u_{\alpha}(m)\Big]
:=
\left[
\tilde\gamma_m-
\dfrac{q^{(m)}_{1-\alpha/2}}{\|\btheta^{(m)}\|_2^2},
\,
\tilde\gamma_m-
\dfrac{q^{(m)}_{\alpha/2}}{\|\btheta^{(m)}\|_2^2}
\right],
\qquad m=1,\ldots,M_0.
\]
By Theorem~\ref{theorem: refined estimation}, we have
\[
\P\left\{
\gamma_m\in
\Big[\tilde l_{\alpha}(m),\tilde u_{\alpha}(m)\Big]
\right\}
\to 1-\alpha,
\qquad m=1,\ldots,M_0.
\]

Moreover, Theorem~\ref{theorem: refined estimation} also establishes the asymptotic independence of the refined estimators for different change points. 
This allows us to construct simultaneous confidence intervals. 
Specifically, for any \(1\le J\le M_0\) and any selected change points
\(\gamma_{m_1},\ldots,\gamma_{m_J}\subset\{\gamma_1,\ldots,\gamma_{M_0}\}\), define
\[
\alpha_J = 1-(1-\alpha)^{1/J}.
\]
For each \(j=1,\ldots,J\), construct the marginal interval
\[
\Big[\tilde l_{\alpha_J}(m_j),\tilde u_{\alpha_J}(m_j)\Big]
:=
\left[
\tilde\gamma_{m_j}-
\dfrac{q^{(m_j)}_{1-\alpha_J/2}}{\|\btheta^{(m_j)}\|_2^2},
\,
\tilde\gamma_{m_j}-
\dfrac{q^{(m_j)}_{\alpha_J/2}}{\|\btheta^{(m_j)}\|_2^2}
\right].
\]
Then, as \((n,d)\to\infty\),
\[
\P\left(
\gamma_{m_1}\in
\Big[\tilde l_{\alpha_J}(m_1),\tilde u_{\alpha_J}(m_1)\Big],
\ldots,
\gamma_{m_J}\in
\Big[\tilde l_{\alpha_J}(m_J),\tilde u_{\alpha_J}(m_J)\Big]
\right)
\to 1-\alpha.
\]

Note that this procedure involves the estimation for
\(\sigma_{1,*}^{(m)}-\sigma_{4,*}^{(m)}\). Specifically, for each
\(m=1,\ldots,M_0\), we put \(\tilde{\gamma}_m\) into
(\ref{equation: signal jump estimation}) and obtain estimation for
\(\btheta^{(m)}\). After obtaining the refined estimator \(\tilde{\gamma}_m\), we estimate the  left and right segment length according to the  refined
estimators. 
Let 
\[
\tilde{\gamma}_0=1,
\qquad
\tilde{\gamma}_{M_0+1}=n .
\]
For the \(m\)th refined estimator \(\tilde{\gamma}_m\), we estimate the  segment lengths on its two sides by
\[
D_m^-=\tilde{\gamma}_m-\tilde{\gamma}_{m-1},
\qquad
D_m^+=\tilde{\gamma}_{m+1}-\tilde{\gamma}_m .
\]
Let $\rho=0.1$, we define
\[
s_{1,m}
=
\left\lfloor
\tilde{\gamma}_m-(1-\rho)D_m^-
\right\rfloor,
\qquad
e_{1,m}
=
\left\lfloor
\tilde{\gamma}_m-\rho D_m^-
\right\rfloor,
\]
and
\[
s_{2,m}
=
\left\lceil
\tilde{\gamma}_m+\rho D_m^+
\right\rceil,
\qquad
e_{2,m}
=
\left\lceil
\tilde{\gamma}_m+(1-\rho)D_m^+
\right\rceil.
\]
The local left and right neighborhoods are then given by
\[
\mathcal L_m=\{s_{1,m},\ldots,e_{1,m}\},
\qquad
\mathcal R_m=\{s_{2,m},\ldots,e_{2,m}\}.
\]
Thus, \(\mathcal L_m\) and \(\mathcal R_m\) are constructed from the estimated
stationary segments adjacent to \(\tilde{\gamma}_m\), while observations within a
\(\rho\)-fraction neighborhood of the refined estimator are excluded to reduce
local contamination near the change point.

Define the estimation for the Hoeffding's projections
\(\bh_1(\bX_t)\) and \(\bh_2(\bX_t)\) by
\[
\begin{array}{ll}
	\hat{\bh}_1(\bX_{t_1})
	=
	\dfrac{1}{|\mathcal R_m|}
	\sum\limits_{t_2\in\mathcal R_m}
	\left[
	\bh(\bX_{t_1},\bX_{t_2})-\hat{\btheta}^{(m)}
	\right],
	&
	t_1\in\mathcal L_m,
	\\[0.18in]
	\hat{\bh}_2(\bX_{t_2})
	=
	\dfrac{1}{|\mathcal L_m|}
	\sum\limits_{t_1\in\mathcal L_m}
	\left[
	\bh(\bX_{t_1},\bX_{t_2})-\hat{\btheta}^{(m)}
	\right],
	&
	t_2\in\mathcal R_m .
\end{array}
\]
Then, we can estimate \(\bSigma^{(m)}_1\) and \(\bSigma^{(m)}_2\) by
\[
\hat{\bSigma}^{(m)}_1
=
\left(
\dfrac{1}{|\mathcal L_m|}
\sum_{t\in\mathcal L_m}
\hat{\bh}_1(\bX_t)\hat{\bh}_1(\bX_t)^\top
\right)_{\hat{\Pi}_m\times \hat{\Pi}_m},
\]
and
\[
\hat{\bSigma}^{(m)}_2
=
\left(
\dfrac{1}{|\mathcal R_m|}
\sum_{t\in\mathcal R_m}
\hat{\bh}_2(\bX_t)\hat{\bh}_2(\bX_t)^\top
\right)_{\hat{\Pi}_m\times \hat{\Pi}_m},
\]
where \((A)_{\hat{\Pi}_m\times\hat{\Pi}_m}\) is the sub-matrix of \(A\)
according to \(\hat{\Pi}_m\) rows and \(\hat{\Pi}_m\) columns.

For estimating \(\bSigma^{(m)}_3\) and \(\bSigma^{(m)}_4\), we further split the
adaptive left and right local neighborhoods into two parts. Let
\[
n_{1,m}=|\mathcal L_m|=e_{1,m}-s_{1,m}+1,
\qquad
n_{2,m}=|\mathcal R_m|=e_{2,m}-s_{2,m}+1.
\]
Define
\[
u_{1,m}=s_{1,m},\qquad
u_{2,m}=s_{1,m}+\lfloor n_{1,m}/2\rfloor-1,\qquad
u_{3,m}=e_{1,m},
\]
and
\[
v_{1,m}=s_{2,m},\qquad
v_{2,m}=s_{2,m}+\lfloor n_{2,m}/2\rfloor-1,\qquad
v_{3,m}=e_{2,m}.
\]
Thus, \(\{u_{1,m},\ldots,u_{2,m}\}\) and
\(\{u_{2,m}+1,\ldots,u_{3,m}\}\) are the first and second halves of the
left local neighborhood, while \(\{v_{1,m},\ldots,v_{2,m}\}\) and
\(\{v_{2,m}+1,\ldots,v_{3,m}\}\) are the first and second halves of the
right local neighborhood of $\cL_m$ and $\cR_m$, respectively.  

Let
\[
\widehat{\bv}^{(m)}
=
\frac{1}{(u_{2,m}-u_{1,m}+1)(u_{3,m}-u_{2,m})}
\sum_{t_1=u_{1,m}}^{u_{2,m}}
\sum_{t_2=u_{2,m}+1}^{u_{3,m}}
\bh(\bX_{t_1},\bX_{t_2}),
\]
and
\[
\widehat{\bw}^{(m)}
=
\frac{1}{(v_{2,m}-v_{1,m}+1)(v_{3,m}-v_{2,m})}
\sum_{t_1=v_{1,m}}^{v_{2,m}}
\sum_{t_2=v_{2,m}+1}^{v_{3,m}}
\bh(\bX_{t_1},\bX_{t_2}),
\]
where \(\widehat{\bv}^{(m)}\) and \(\widehat{\bw}^{(m)}\) estimate the within-segment
parameters before and after the change point \(\gamma_m\), respectively.

Let the estimation for \(\underline{\bh}_{1}(\bX_t)\) and
\(\overline{\bh}_{1}(\bX_t)\) be
\[
\begin{array}{ll}
	\widehat{\underline{\bh}}_1(\bX_{t_1})
	=
	\dfrac{1}{u_{3,m}-u_{2,m}}
	\sum\limits_{t_2=u_{2,m}+1}^{u_{3,m}}
	\left[
	\bh(\bX_{t_1},\bX_{t_2})-\widehat{\bv}^{(m)}
	\right],
	&
	t_1=u_{1,m},\ldots,u_{2,m},
	\\[0.18in]
	\widehat{\overline{\bh}}_1(\bX_{t_1})
	=
	\dfrac{1}{v_{3,m}-v_{2,m}}
	\sum\limits_{t_2=v_{2,m}+1}^{v_{3,m}}
	\left[
	\bh(\bX_{t_1},\bX_{t_2})-\widehat{\bw}^{(m)}
	\right],
	&
	t_1=v_{1,m},\ldots,v_{2,m}.
\end{array}
\]

Then,  we estimate the remaining covariance
matrices by
\[
\widehat{\bSigma}^{(m)}_3
=
\left(
\dfrac{1}{v_{2,m}-v_{1,m}+1}
\sum\limits_{t=v_{1,m}}^{v_{2,m}}
\big(
\widehat{\overline{\bh}}_1(\bX_t)
-
\widehat{\bh}_2(\bX_t)
\big)
\big(
\widehat{\overline{\bh}}_1(\bX_t)
-
\widehat{\bh}_2(\bX_t)
\big)^\top
\right)_{\widehat{\Pi}_m\times \widehat{\Pi}_m},
\]
and
\[
\widehat{\bSigma}^{(m)}_4
=
\left(
\dfrac{1}{u_{2,m}-u_{1,m}+1}
\sum\limits_{t=u_{1,m}}^{u_{2,m}}
\big(
\widehat{\underline{\bh}}_1(\bX_t)
+
\widehat{\bh}_1(\bX_t)
\big)
\big(
\widehat{\underline{\bh}}_1(\bX_t)
+
\widehat{\bh}_1(\bX_t)
\big)^\top
\right)_{\widehat{\Pi}_m\times \widehat{\Pi}_m}.
\]

Using the above estimtors, we put $\hat{\btheta}^{(m)},	\hat{\bSigma}^{(m)}_1,\ldots,	\hat{\bSigma}^{(m)}_4$ into the definition of $\sigma_{1,*}^{(m)} -	\sigma_{4,*}^{(m)}$ in (\ref{equ: sigma1-star-4-star}) of the main paper and obtain the estimators $\hat{\sigma}_{1,*}^{(m)} -	\hat{\sigma}_{4,*}^{(m)}$. 

Lastly, based on the estimated parameters, we use the Monte Carlo method to simulate the process $\{ Z^{(m)}(s),s\in(-\infty,\infty)\}$ 2000 times. In each simulation, we compute \( \arg\max_{s} Z^{(m)}(s) \), resulting in 2000 values of ``\( \arg\max \)''. We obtain estimates of \( q^{(m)}_{\alpha/2} \) and \( q^{(m)}_{1-\alpha/2} \), say \( \hat{q}^{(m)}_{\alpha/2} \) and \( \hat{q}^{(m)}_{1-\alpha/2} \), by taking the \( \alpha/2 \) and \( 1-\alpha/2 \) quantiles of these 2000 values. Finally, the empirical confidence interval for each detected change point is constructed as:
\begin{equation}\label{equ: estimated confidence}
	\Big[\hat{l}_{\alpha}(m),\hat{u}_\alpha(m)\Big]:=\Big[\tilde{\gamma}_m- \dfrac{\hat{q}^{(m)}_{1-\alpha/2}}{\|\hat{\btheta}^{(m)}\|^2},	\tilde{\gamma}_m- \dfrac{\hat{q}^{(m)}_{\alpha/2}}{\|\hat{\btheta}^{(m)}\|^2}\Big],~\text{for}~m=1,\ldots,M_0.
\end{equation}

Extensive numerical studies show that the empirical coverage probability of our proposed method can closely match the nominal confidence level under various model settings. See Section \ref{sec: emprical change point inference} for more details. 

\section{{Practical Guidance on Tuning Parameter Selection}}\label{sec: pratical guidence}
\subsection{\textbf{Practical choice of the minimum-length parameter \(\eta\)}}
The parameter \(\eta\) is used to impose a minimum-length requirement on the
significant regions in the initial detection step as {defined in (\ref{equ: multiple cpt principle}}) of the main paper. Its main purpose is to filter out isolated threshold exceedances caused by local
fluctuations, rather than to directly determine the final change-point locations. In our implementation, we set \(\eta=0.15\) by default.
Extensive numerical experiments show that the proposed method is not sensitive to
this choice. In particular, when \(\eta\) varies between \(0.1\) and \(0.3\), the
method exhibits similar performance in terms of estimating the number and locations
of change points, as well as the empirical coverage of the resulting confidence
intervals. Therefore, we use \(\eta=0.15\) as a stable default choice in all numerical
experiments.
\subsection{\textbf{Data-driven selection of the active-set threshold}}

The projection refinement step requires the specification of the active-set threshold
\(w^+\) as in {Theorem \ref{theorem: support recovery}}. 
In principle, this threshold should be sufficiently large to exclude inactive
coordinates, but not so large as to remove coordinates that carry  change-point
signals. To reduce the dependence on a manually fixed constant, we adopt a data-driven
calibration rule for the multiplicative constant in \(w^+\).

For each initially estimated change point \(\widehat\gamma_m\), we first construct local
left and right neighborhoods as 
\[
\mathcal L_m=\{s_{1,m},\ldots,e_{1,m}\},
\qquad
\mathcal R_m=\{s_{2,m},\ldots,e_{2,m}\},
\]
and define the effective local sample size as
\[
b_m=\min\{|\mathcal L_m|,|\mathcal R_m|\}.
\]
For each coordinate \(j=1,\ldots,d\), the local signal strength is estimated by
\[
\widehat\theta_j^{(m)}
=
\frac{1}{|\mathcal L_m||\mathcal R_m|}
\sum_{t\in \mathcal L_m}
\sum_{s\in \mathcal R_m}
h(X_{t,j},X_{s,j}),
\]
where the same kernel \(h\) as in the main testing and estimation procedure is used.

We consider thresholds of the form
\[
w_m^+(c)
=
c\sqrt{\frac{\log(dn)}{2b_m}},
\]
where \(c>0\) is a multiplicative constant. Given \(c\), the corresponding selected active
set is
\[
\widehat\Pi_m(c)
=
\left\{
j: |\widehat\theta_j^{(m)}|>w_m^+(c)
\right\},
\qquad
\widehat k_m(c)=|\widehat\Pi_m(c)|.
\]

To determine \(c\), let
\[
a_{m,(1)}\geq a_{m,(2)}\geq \cdots \geq a_{m,(d)}
\]
denote the decreasing rearrangement of $|\widehat\theta_1^{(m)}|,\ldots,|\widehat\theta_d^{(m)}|$. We restrict the selected model size to a moderate range. Let
\[
K_{\min}=2,
\qquad
K_{\max}=\lfloor 2\log(d)\rfloor,
\]
so that the selected set contains at least two coordinates but does not become excessively
large. Accordingly, the candidate interval for \(c\) is chosen as
\[
\mathcal C_m
=
\left[
\frac{a_{m,(K_{\max}+1)}}{\sqrt{\log(dn)/(2b_m)}},
\,
\frac{a_{m,(K_{\min})}}{\sqrt{\log(dn)/(2b_m)}}
\right),
\]
which is implemented through a finite grid over this interval.

For each \(c\in\mathcal C_m\), we evaluate the separation ratio
\[
R_m(c)
=
\frac{
	\left(
	\sum_{\ell=1}^{\widehat k_m(c)}
	a_{m,(\ell)}^2
	\right)^{1/2}
}{
	\left(
	\sum_{\ell=\widehat k_m(c)+1}^{2\widehat k_m(c)}
	a_{m,(\ell)}^2
	\right)^{1/2}
},
\]
provided that \(2\widehat k_m(c)\leq d\). Note that the numerator measures the aggregate magnitude of the coordinates selected under the threshold \(w_m^+(c)\), whereas the denominator
measures the aggregate magnitude of the next group of coordinates of the same size. Hence, \(R_m(c)\) quantifies the empirical separation between the selected coordinates and the remaining leading unselected coordinates.

The constant is then selected by
\[
\widehat c_m
=
\arg\max_{c\in\mathcal C_m} R_m(c),
\]
and the final threshold and active set are defined as
\[
\widehat w_m^+
=
\widehat c_m\sqrt{\frac{\log(dn)}{2b_m}},
\qquad
\widehat\Pi_m
=
\left\{
j: |\widehat\theta_j^{(m)}|>\widehat w_m^+
\right\}.
\]

Overall, this calibration strategy provides a fully data-driven way to determine the
multiplicative constant in the active-set threshold, rather than requiring it to be fixed
subjectively in advance. In our
extensive numerical studies, this approach has shown stable finite-sample performance, where
it selects all  active signal coordinates with high probability, while keeping the
number of falsely selected inactive coordinates at a moderate level.

\subsection{\textbf{Multiscale bandwidth aggregation}}

The moving-window detector depends on the bandwidth parameter \(G\). In finite samples,
a single bandwidth may not be uniformly suitable for all change points. A smaller bandwidth
is more sensitive to closely spaced changes, whereas a larger bandwidth may
produce more stable detection statistics for weaker changes over longer stationary segments.
To reduce the dependence on a single bandwidth, inspired by the multiscale idea of
\cite{cho2024highdimensional}, we have modified their strategy to fit our moving-window
\(U\)-statistic framework and proposed a multiscale bandwidth aggregation procedure.

\textbf{Step 1: Multiscale scanning.} Let
\[
\mathcal G=\{G_1,\ldots,G_H\}, \qquad G_1<\cdots<G_H,
\]
be a collection of candidate bandwidths. For example, in our numerical experiments, we use
\[
\mathcal G=\{60,80,100\}.
\]
For each \(G\in\mathcal G\), we compute the coordinate-wise moving-window U-statistic
\[
T_{j,G}(k)
=
\frac{1}{(2G)^{3/2}}
\sum_{t_1=k-G}^{k}
\sum_{t_2=k+1}^{k+G}
h(X_{t_1,j},X_{t_2,j}),
\qquad j=1,\ldots,d,
\]
and its \(\ell_\infty\)-based testing statistic
\[
S_G(k)=\max_{1\leq j\leq d}|T_{j,G}(k)|.
\]

For each bandwidth \(G\), we compute a bandwidth-specific bootstrap critical value
\(c_{\alpha,G}\) using the proposed multiplier bootstrap.  We then identify the significant regions at this scale.  Specifically,  let
\(L_G\) denote the number of significant regions detected under bandwidth \(G\). For
\(\ell=1,\ldots,L_G\), the \(\ell\)-th significant interval is denoted by
\[
I_{\ell,G}=[v_{\ell,G},w_{\ell,G}],
\]
where
\[
S_G(k)\geq c_{\alpha,G},
\qquad k=v_{\ell,G},\ldots,w_{\ell,G},
\]
with boundary conditions
\[
S_G(v_{\ell,G}-1)<c_{\alpha,G},
\qquad
S_G(w_{\ell,G}+1)<c_{\alpha,G},
\]
and the minimum-length requirement
\[
w_{\ell,G}-v_{\ell,G}\geq \eta G .
\]
For each $G$ and \(\ell=1,\ldots,L_G\), the associated bandwidth-specific pre-estimator is defined as
\begin{equation}\label{equ: pre estimation}
	\widehat\gamma_{\ell,G}
	=
	\arg\max_{v_{\ell,G}\leq k\leq w_{\ell,G}} S_G(k).
\end{equation}

\textbf{Step 2: Candidate pooling.} Pooling over all bandwidths gives the multiscale candidate set
\[
\widehat{\mathcal P}
=
\left\{
p_{\ell,G}:G\in\mathcal G,\ \ell=1,\ldots,L_G
\right\},
\]
where each candidate is represented by
\[
p_{\ell,G}
=
\left(
\widehat\gamma_{\ell,G},
G,
S_G(\widehat\gamma_{\ell,G}),
I_{\ell,G}
\right),
\]
where, \(\widehat\gamma_{\ell,G}\) is the candidate location for the $\ell$-th change point, \(G\) is the bandwidth
under which it is detected, \(S_G(\widehat\gamma_{\ell,G})\) is the corresponding
statistic, and \(I_{\ell,G}\) is its significant interval. We further define the
normalized statistic
\begin{equation}\label{equ: normalized sta}
	Z_{\ell,G}
	=
	\frac{S_G(\widehat\gamma_{\ell,G})}{c_{\alpha,G}},   \ell=1,\ldots,L_G,
\end{equation}

which makes candidates obtained under different bandwidths comparable.

\textbf{Step 3: Anchor identification.} Since the same true change point may be detected under multiple bandwidths, we merge
the pooled candidates through an anchor-based procedure. Candidates are first ordered
from finer to coarser bandwidths. 
Starting from
\(\widehat{\mathcal A}_0=\varnothing\), suppose that before examining candidate
\(p_{\ell,G}\), the current anchor set is
\[
\widehat{\mathcal A}_R
=
\{p_{\ell_r,G_r}:r=1,\ldots,R\},
\]
where, \(R\) is the current number of retained anchors, \(p_{\ell_r,G_r}\) is the
\(r\)-th retained anchor, \(\ell_r\) denotes the index of the significant region from
which this anchor is obtained, and \(G_r\) is its associated bandwidth. The candidate
\(p_{\ell,G}\) is retained as a new anchor if its significant interval does not overlap
with any previously retained anchor interval, namely,
\[
I_{\ell,G}\cap I_{\ell_r,G_r}=\varnothing,
\qquad r=1,\ldots,R.
\]
If this condition holds, we update
\[
\widehat{\mathcal A}_{R+1}
=
\widehat{\mathcal A}_{R}\cup\{p_{\ell,G}\};
\]
otherwise, \(p_{\ell,G}\) is regarded as a repeated detection of an already retained
anchor. After all candidates have been processed, the final anchor set is denoted by
\[
\widehat{\mathcal A}
=
\{a_r=p_{\ell_r,G_r}:r=1,\ldots,\widehat M\},
\]
where \(\widehat M\) is the number of retained anchors which denote the number of detected change points and
\[
a_r:
=
p_{\ell_r,G_r}
=
\left(
\widehat\gamma_{\ell_r,G_r},
G_r,
S_{G_r}(\widehat\gamma_{\ell_r,G_r}),
I_{\ell_r,G_r}
\right).
\]

\textbf{Step 4: Anchor-based clustering.} After the anchors are identified, the pooled candidates are assigned to anchor-based
clusters. For the \(r\)-th anchor \(a_r=p_{\ell_r,G_r}\), define
\[
\mathcal C_r
=
\left\{
p_{\ell,G}\in\widehat{\mathcal P}:
I_{\ell,G}\cap I_{\ell_r,G_r}\neq\varnothing
\right\},
\qquad r=1,\ldots,\widehat M.
\]
Thus, \(\mathcal C_r\) collects candidates whose significant intervals overlap with the
significant interval of the \(r\)-th anchor, and these candidates are treated as multiscale
detections of the same underlying change point. Note that if a candidate overlaps with more than
one anchor interval, it is assigned to the closest anchor according to
\[
r^\star(\ell,G)
=
\arg\min_{1\leq r\leq \widehat M}
\left|
\widehat\gamma_{\ell,G}
-
\widehat\gamma_{\ell_r,G_r}
\right|.
\]

\textbf{Step 5: Representative candidate selection.} For each cluster \(\mathcal C_r\), we retain a single representative initial estimator. The
representative candidate is chosen as the one with the largest normalized  statistic:
\[
(\ell_r^\star,G_r^\star)
=
\arg\max_{p_{\ell,G}\in\mathcal C_r}
Z_{\ell,G}.
\]
where $Z_{\ell,G}$  is in (\ref{equ: normalized sta}). 
The merged initial location is then defined as
\[
\widehat\gamma_r^{\,0}
=
\widehat\gamma_{\ell_r^\star,G_r^\star},
\qquad r=1,\ldots,\widehat M.
\]
where $\widehat\gamma_r^{\,0}$ is defined in (\ref{equ: pre estimation}).

The final merged initial set is
\[
\widehat{\mathcal G}^{\,0}
=
\left\{
\widehat\gamma_1^{\,0},\ldots,
\widehat\gamma_{\widehat M}^{\,0}
\right\}
=
\left\{
\widehat\gamma_{\ell_r^\star,G_r^\star}:
r=1,\ldots,\widehat M
\right\}.
\]

\textbf{Step 6: Local projection refinement and inference.}
The subsequent steps are identical to those in the single-bandwidth based  procedure, except that
they are initialized by the merged multiscale estimators. For each
\(\widehat\gamma_m^{\,0}\in\widehat{\mathcal G}^{\,0}\), we construct local left and right
neighborhoods,
\[
\mathcal L_m=\{s_{1,m},\ldots,e_{1,m}\},
\qquad
\mathcal R_m=\{s_{2,m},\ldots,e_{2,m}\},
\]
using the neighboring merged initial estimators. Based on these
local neighborhoods, we estimate the coordinate-wise signal vector by
\[
\widehat\theta_j^{(m)}
=
\frac{1}{|\mathcal L_m||\mathcal R_m|}
\sum_{t\in\mathcal L_m}
\sum_{s\in\mathcal R_m}
h(X_{t,j},X_{s,j}),
\qquad j=1,\ldots,d.
\]
The active set \(\widehat\Pi_m\) is then selected using the data-driven thresholding rule
described in the previous section.

Finally, using the active set \(\widehat\Pi_m\), we perform a local projection refinement around \(\widehat\gamma_m^{\,0}\) and obtain the final refined change point estimator $\widetilde\gamma_m$. Lastly, confidence intervals for the $m$-th detected change point are then constructed using the same 
procedure as in the single-bandwidth implementation.

In summary, the multiscale procedure is designed to improve the stability of the initial
detection stage by aggregating information across several bandwidths. Extensive simulations show that the multiscale implementation is stable, it can accurately estimate the number and locations of change points and provides satisfactory confidence-interval construction.


\newpage

\section{{Numerical stability analysis}}\label{sec: Numerical stability analysis}
\subsection{{\textbf{Numerical stability of signal direction estimation}}}

We first examine the finite-sample stability of signal direction estimation. The simulation setting is the same as that used for the three-change-point mean-shift model in {Section~\ref{sec: changes in mean}.}  In particular, we consider structural changes in the $d$-dimensional parameter vector
\[
\btheta=\E\{h(\bx,\by)\},
\]
where $\bx$ and $\by$ denote observations before and after a change point, respectively. Under our design, at each true change point only the first five coordinates of $\btheta$ are allowed to change, while the remaining $d-5$ coordinates remain equal to zero. Hence, the active set is fixed as
$S^\star=\{1,2,3,4,5\}.
$ For the $m$th change point , we estimate the signal jump, which yields an estimator $\hat{\btheta}_m$ for the $m$-th change point. We then compare $\hat{\btheta}_m$ with the corresponding true signal vector $\btheta_m$. 
To assess estimation accuracy, we consider two types of metrics. The first type evaluates the entire $d$-dimensional vector:
\[
\mathrm{RelErr}^{\mathrm{full}}_m
=
\frac{\|\hat{\btheta}_m-\btheta_m\|_2}{\|\btheta_m\|_2},
\qquad
\mathrm{Cos}^{\mathrm{full}}_m
=
\frac{\hat{\btheta}_m^\top \btheta_m}{\|\hat{\btheta}_m\|_2\|\btheta_m\|_2}.
\]
These two quantities measure, respectively, the overall relative estimation error and the cosine similarity between the estimated and true signal directions. The second type focuses only on the active coordinates in $S^\star$. Let $\btheta_{m,S^\star}$ and $\hat{\btheta}_{m,S^\star}$ denote the subvectors restricted to the active set. We compute
\[
\mathrm{RelErr}^{\mathrm{act}}_m
=
\frac{\|\hat{\btheta}_{m,S^\star}-\btheta_{m,S^\star}\|_2}{\|\btheta_{m,S^\star}\|_2},
\qquad
\mathrm{Cos}^{\mathrm{act}}_m
=
\frac{\hat{\btheta}_{m,S^\star}^\top \btheta_{m,S^\star}}
{\|\hat{\btheta}_{m,S^\star}\|_2\|\btheta_{m,S^\star}\|_2}.
\]
We report these four measures under both Gaussian and $t_3$ distributions, and for the two kernel  $h(\bx,\by)=\by-\bx$
and
$h(\bx,\by)=\sign(\by-\bx)$.


The results are given in Table 	\ref{tab:signal-direction-stability}. \textbf{First}, the proposed method estimates the signal jump direction accurately on the active coordinates. Across the three change points, the relative errors restricted to the active set are small, and the corresponding cosine similarities are close to one. This shows that the estimated direction is well aligned with the true signal direction. \textbf{Second}, the results are stable with respect to the bandwidth \(G\). When \(G\) changes from 60 to 100, the four reported measures vary only mildly in most settings. This indicates that the signal direction estimation step is not sensitive to the specific bandwidth choice. \textbf{Third}, the full-vector errors are larger than the active-coordinate errors. This is expected, since the full \(d\)-dimensional vector includes many inactive coordinates. Estimation noise on these inactive coordinates affects the full relative error and full cosine similarity.

\textbf{Lastly}, the sign kernel is more robust under heavy-tailed data. Under the \(t_3\) distribution, the linear kernel has larger full-vector errors, whereas the sign kernel maintains performance close to the Gaussian case. This confirms the robustness of the  sign based kernel.


\begin{table}[H]
	\centering
	\caption{Finite-sample stability of signal  jump estimation under different bandwidths \(G\). Entries are mean (sd).}
	\label{tab:signal-direction-stability}
	\footnotesize
	\setlength{\tabcolsep}{5pt}
	\renewcommand{\arraystretch}{1.5}
	\begin{tabular}{llccccc}
		\toprule
		Dist. & Kernel & \(G\) & RelErr (full) & RelErr (active) & Cos (full) & Cos (active) \\
		\midrule
		
		\multicolumn{7}{c}{\textbf{Cpt1} at 0.3n} \\
		\midrule
		Gaussian& \(y-x\) & 60  & 0.721 (0.060) & 0.101 (0.049) & 0.810 (0.025) & 0.995 (0.007) \\
		&         & 80  & 0.721 (0.061) & 0.102 (0.050) & 0.810 (0.026) & 0.994 (0.007) \\
		&         & 100 & 0.718 (0.061) & 0.102 (0.050) & 0.812 (0.025) & 0.994 (0.006) \\
		\cmidrule(lr){2-7}
		& \(\operatorname{sign}(y-x)\) & 60  & 0.815 (0.066) & 0.093 (0.043) & 0.774 (0.028) & 0.996 (0.005) \\
		&                              & 80  & 0.814 (0.066) & 0.093 (0.043) & 0.774 (0.028) & 0.996 (0.006) \\
		&                              & 100 & 0.814 (0.066) & 0.094 (0.043) & 0.775 (0.028) & 0.996 (0.006) \\
		\midrule
		\(t_3\) & \(y-x\) & 60  & 0.903 (0.092) & 0.184 (0.084) & 0.723 (0.050) & 0.985 (0.017) \\
		&         & 80  & 0.947 (0.104) & 0.149 (0.075) & 0.719 (0.047) & 0.989 (0.015) \\
		&         & 100 & 0.958 (0.108) & 0.142 (0.069) & 0.722 (0.042) & 0.989 (0.013) \\
		\cmidrule(lr){2-7}
		& \(\operatorname{sign}(y-x)\) & 60  & 0.817 (0.066) & 0.100 (0.049) & 0.766 (0.032) & 0.995 (0.006) \\
		&                              & 80  & 0.816 (0.066) & 0.101 (0.051) & 0.767 (0.031) & 0.995 (0.007) \\
		&                              & 100 & 0.818 (0.065) & 0.102 (0.048) & 0.766 (0.031) & 0.995 (0.006) \\
		
		\midrule
		\multicolumn{7}{c}{\textbf{Cpt2} at 0.5n} \\
		\midrule
		Gaussian & \(y-x\) & 60  & 0.757 (0.062) & 0.113 (0.052) & 0.795 (0.028) & 0.993 (0.007) \\
		&         & 80  & 0.755 (0.062) & 0.112 (0.049) & 0.796 (0.029) & 0.994 (0.006) \\
		&         & 100 & 0.754 (0.061) & 0.110 (0.048) & 0.797 (0.026) & 0.994 (0.006) \\
		\cmidrule(lr){2-7}
		& \(\operatorname{sign}(y-x)\) & 60  & 0.854 (0.067) & 0.100 (0.042) & 0.759 (0.030) & 0.995 (0.004) \\
		&                              & 80  & 0.851 (0.065) & 0.100 (0.040) & 0.759 (0.030) & 0.995 (0.004) \\
		&                              & 100 & 0.853 (0.065) & 0.099 (0.041) & 0.758 (0.029) & 0.995 (0.004) \\
		\midrule
		\(t_3\) & \(y-x\) & 60  & 0.945 (0.098) & 0.245 (0.110) & 0.670 (0.073) & 0.980 (0.021) \\
		&         & 80  & 0.977 (0.107) & 0.176 (0.074) & 0.701 (0.054) & 0.985 (0.013) \\
		&         & 100 & 1.000 (0.121) & 0.159 (0.068) & 0.703 (0.047) & 0.987 (0.012) \\
		\cmidrule(lr){2-7}
		& \(\operatorname{sign}(y-x)\) & 60  & 0.858 (0.069) & 0.110 (0.045) & 0.750 (0.035) & 0.995 (0.004) \\
		&                              & 80  & 0.856 (0.067) & 0.109 (0.042) & 0.751 (0.033) & 0.995 (0.005) \\
		&                              & 100 & 0.855 (0.065) & 0.110 (0.042) & 0.751 (0.032) & 0.995 (0.004) \\
		
		\midrule
		\multicolumn{7}{c}{\textbf{Cpt3} at 0.7n} \\
		\midrule
		Gaussian & \(y-x\) & 60  & 0.722 (0.058) & 0.111 (0.047) & 0.809 (0.026) & 0.994 (0.006) \\
		&         & 80  & 0.720 (0.057) & 0.109 (0.046) & 0.810 (0.024) & 0.994 (0.005) \\
		&         & 100 & 0.721 (0.056) & 0.110 (0.046) & 0.810 (0.024) & 0.994 (0.005) \\
		\cmidrule(lr){2-7}
		& \(\operatorname{sign}(y-x)\) & 60  & 0.816 (0.064) & 0.098 (0.041) & 0.774 (0.029) & 0.995 (0.004) \\
		&                              & 80  & 0.814 (0.066) & 0.097 (0.038) & 0.775 (0.028) & 0.995 (0.004) \\
		&                              & 100 & 0.814 (0.062) & 0.096 (0.038) & 0.776 (0.026) & 0.995 (0.004) \\
		\midrule
		\(t_3\) & \(y-x\) & 60  & 0.906 (0.080) & 0.177 (0.062) & 0.720 (0.045) & 0.987 (0.010) \\
		&         & 80  & 0.922 (0.091) & 0.163 (0.073) & 0.724 (0.046) & 0.987 (0.012) \\
		&         & 100 & 0.947 (0.106) & 0.151 (0.069) & 0.721 (0.047) & 0.988 (0.012) \\
		\cmidrule(lr){2-7}
		& \(\operatorname{sign}(y-x)\) & 60  & 0.818 (0.064) & 0.105 (0.042) & 0.767 (0.030) & 0.995 (0.004) \\
		&                              & 80  & 0.818 (0.063) & 0.104 (0.039) & 0.768 (0.029) & 0.995 (0.004) \\
		&                              & 100 & 0.817 (0.064) & 0.104 (0.038) & 0.768 (0.028) & 0.995 (0.004) \\
		
		\bottomrule
	\end{tabular}
\end{table}


\subsection{{\textbf{Numerical stability of covariance matrix estimation}}}

We next examine the finite-sample stability of covariance estimation in the proposed inference procedure. Recall that, for each change point, the limiting distribution depends on four covariance matrices, denoted by
$\bSigma_{m}^{(1)}, \bSigma_{m}^{(2)}, \bSigma_{m}^{(3)}, \bSigma_{m}^{(4)}$
for the $m$th change point, {which are defined in (\ref{equ: four covs})} of the main paper. 

For each refined change point, we construct local left and right neighborhoods centered at the refined estimator, and then estimate the four covariance matrices using the Hoeffding projection-based procedure described in {Section~\ref{sec: implementation for estimation confidence}}. More specifically, let $\tilde{\tau}_m$ denote the refined estimator associated with the $m$th true change point, we form the local samples on the two sides of $\tilde{\tau}_m$ , and apply the covariance estimation method in {Section~\ref{sec: implementation for estimation confidence}} to obtain $\hat\bSigma^{(1)}_{m}, \hat\bSigma^{(2)}_{m},\hat\bSigma^{(3)}_{m}$ and $\hat\bSigma^{(4)}_{m}$.

It is worth pointing out that, in each numerical  replication, the selected active set used to construct the Hoeffding projections may vary with the data. To isolate the numerical stability of covariance estimation itself from the additional randomness induced by active-set selection, we evaluate covariance estimation only on the true active coordinates. That is, we restrict attention to the fixed active set $\Pi_{m}=\{1,2,3,4,5\}$
and construct the projection-based covariance estimators on this set only. In this way,  the resulting comparison reflects the finite-sample error of covariance estimation rather than the variability caused by support recovery. Moreover, 
note that the true covariance matrices $\bSigma_{m}^{(1)}, \bSigma_{m}^{(2)}, \bSigma_{m}^{(3)}, \bSigma_{m}^{(4)}$ are analytically unavailable even when the underlying distribution and kernel function are known, we approximate them numerically using a large-scale Monte Carlo experiment under the same data-generating mechanism as in the numerical studies. These large-sample Monte Carlo approximations are then treated as oracle benchmarks in the finite-sample comparison.

To quantify estimation accuracy, we use the relative Frobenius norm error. For $\ell\in\{1,2,3,4\}$ and each change point $m$, we compute
\[
\mathrm{RelErr}(\hat\bSigma_{\ell,m})
=
\frac{\|\hat\bSigma_{\ell,m}-\bSigma_{\ell,m}\|_{F}}
{\|\bSigma_{\ell,m}\|_{F}},
\]
where $\|\cdot\|_F$ denotes the Frobenius norm. 
We report the relative Frobenius errors for all four covariance matrices at each change point under both Gaussian and $t_3$ distributions, and for the two kernel choices
$h(\bx,\by)=\by-\bx$ and  $h(\bx,\by)=\sign(\by-\bx)$.

The results are reported in  Table 	\ref{tab:covariance-estimation-stability}.
\textbf{First}, the covariance estimators are reasonably stable with respect to the bandwidth \(G\). Across the three change points, the relative Frobenius errors do not change substantially when \(G\) varies from 60 to 100. 

\textbf{Second},  the covariance estimation errors are moderate rather than negligible, which is reasonable in this finite-sample setting. Although the full sample size is \(n=600\), each covariance matrix is estimated only from local neighborhoods around a refined change point rather than the full sample. For \(\bSigma_3\) and \(\bSigma_4\), the local neighborhoods are further split into two halves, so the effective number of observations involved in the covariance estimation is very small.  As a result, the reported relative Frobenius errors should be interpreted as finite-sample errors under limited samples.

\textbf{Third}, the errors for \(\bSigma_3\) and \(\bSigma_4\) are generally larger than those for \(\bSigma_1\) and \(\bSigma_2\). This is also expected from the construction of the estimators. The estimation of \(\bSigma_1\) and \(\bSigma_2\) mainly uses the local samples on the left and right sides of the change point, respectively. In contrast, the estimation of \(\bSigma_3\) and \(\bSigma_4\) requires further splitting the local neighborhoods. This introduces extra variability and reduces the effective sample size, leading to larger relative errors.

\textbf{Lastly}, under the Gaussian distribution, both kernels yield stable covariance estimation results.  However, under the \(t_3\) distribution, the linear kernel becomes less stable, especially for \(\bSigma_3\) and \(\bSigma_4\). In comparison, the sign kernel remains much more stable under heavy-tailed data, which reflects the robustness of using a bounded kernel.


\begin{table}[H]
	\centering
	\caption{Finite-sample stability of covariance matrix estimation under different bandwidths \(G\).
		Entries are mean (sd) of the relative Frobenius errors
		\(\|\widehat\bSigma_{\ell,m}-\bSigma_{\ell,m}\|_F/\|\bSigma_{\ell,m}\|_F\), for
		\(\ell=1,2,3,4\) at each change point.}
	\label{tab:covariance-estimation-stability}
	\footnotesize
	\setlength{\tabcolsep}{6pt}
	\renewcommand{\arraystretch}{1.5}
	\begin{tabular}{llccccc}
		\toprule
		Dist. & Kernel & \(G\) & \(\bSigma_1\) & \(\bSigma_2\) & \(\bSigma_3\) & \(\bSigma_4\) \\
		\midrule
		
		\multicolumn{7}{c}{\textbf{Cpt1} at 0.3n} \\
		\midrule
		Gaussian & \(y-x\) & 60  & 0.133 (0.053) & 0.159 (0.068) & 0.195 (0.078) & 0.208 (0.092) \\
		&         & 80  & 0.134 (0.054) & 0.159 (0.069) & 0.195 (0.079) & 0.208 (0.093) \\
		&         & 100 & 0.134 (0.054) & 0.158 (0.068) & 0.195 (0.079) & 0.208 (0.092) \\
		\cmidrule(lr){2-7}
		& \(\operatorname{sign}(y-x)\) & 60  & 0.184 (0.058) & 0.204 (0.064) & 0.190 (0.067) & 0.203 (0.078) \\
		&                              & 80  & 0.184 (0.059) & 0.203 (0.061) & 0.189 (0.068) & 0.202 (0.075) \\
		&                              & 100 & 0.183 (0.058) & 0.203 (0.060) & 0.189 (0.068) & 0.202 (0.076) \\
		\midrule
		\(t_3\) & \(y-x\) & 60  & 0.359 (0.421) & 0.361 (0.241) & 0.586 (0.850) & 0.468 (0.194) \\
		&         & 80  & 0.363 (0.365) & 0.385 (0.332) & 0.515 (0.606) & 0.509 (0.415) \\
		&         & 100 & 0.370 (0.354) & 0.387 (0.351) & 0.514 (0.604) & 0.522 (0.460) \\
		\cmidrule(lr){2-7}
		& \(\operatorname{sign}(y-x)\) & 60  & 0.205 (0.065) & 0.237 (0.079) & 0.211 (0.071) & 0.231 (0.082) \\
		&                              & 80  & 0.205 (0.064) & 0.236 (0.079) & 0.211 (0.072) & 0.230 (0.079) \\
		&                              & 100 & 0.205 (0.064) & 0.236 (0.078) & 0.212 (0.074) & 0.228 (0.078) \\
		
		\midrule
		\multicolumn{7}{c}{\textbf{Cpt2} at 0.5n} \\
		\midrule
		Gaussian & \(y-x\) & 60  & 0.163 (0.075) & 0.148 (0.059) & 0.218 (0.102) & 0.221 (0.097) \\
		&         & 80  & 0.163 (0.074) & 0.148 (0.059) & 0.219 (0.102) & 0.222 (0.100) \\
		&         & 100 & 0.163 (0.074) & 0.148 (0.060) & 0.218 (0.102) & 0.222 (0.099) \\
		\cmidrule(lr){2-7}
		& \(\operatorname{sign}(y-x)\) & 60  & 0.211 (0.080) & 0.196 (0.061) & 0.211 (0.085) & 0.210 (0.082) \\
		&                              & 80  & 0.211 (0.077) & 0.197 (0.061) & 0.211 (0.083) & 0.209 (0.081) \\
		&                              & 100 & 0.209 (0.077) & 0.196 (0.060) & 0.210 (0.084) & 0.209 (0.080) \\
		\midrule
		\(t_3\) & \(y-x\) & 60  & 0.311 (0.136) & 0.338 (0.147) & 0.405 (0.241) & 0.463 (0.176) \\
		&         & 80  & 0.355 (0.309) & 0.359 (0.223) & 0.478 (0.502) & 0.491 (0.350) \\
		&         & 100 & 0.385 (0.394) & 0.385 (0.234) & 0.527 (0.612) & 0.531 (0.368) \\
		\cmidrule(lr){2-7}
		& \(\operatorname{sign}(y-x)\) & 60  & 0.238 (0.080) & 0.222 (0.070) & 0.232 (0.085) & 0.237 (0.085) \\
		&                              & 80  & 0.235 (0.075) & 0.223 (0.071) & 0.232 (0.083) & 0.237 (0.085) \\
		&                              & 100 & 0.235 (0.076) & 0.223 (0.071) & 0.231 (0.083) & 0.236 (0.084) \\
		
		\midrule
		\multicolumn{7}{c}{\textbf{Cpt3 } at 0.7n} \\
		\midrule
		Gaussian & \(y-x\) & 60  & 0.149 (0.062) & 0.142 (0.053) & 0.222 (0.099) & 0.196 (0.082) \\
		&         & 80  & 0.149 (0.063) & 0.142 (0.052) & 0.222 (0.102) & 0.196 (0.082) \\
		&         & 100 & 0.149 (0.063) & 0.143 (0.056) & 0.223 (0.100) & 0.199 (0.095) \\
		\cmidrule(lr){2-7}
		& \(\operatorname{sign}(y-x)\) & 60  & 0.196 (0.067) & 0.192 (0.066) & 0.209 (0.080) & 0.183 (0.064) \\
		&                              & 80  & 0.195 (0.062) & 0.192 (0.065) & 0.208 (0.080) & 0.184 (0.063) \\
		&                              & 100 & 0.195 (0.061) & 0.192 (0.066) & 0.207 (0.079) & 0.183 (0.064) \\
		\midrule
		\(t_3\) & \(y-x\) & 60  & 0.380 (0.229) & 0.284 (0.101) & 0.581 (0.582) & 0.390 (0.236) \\
		&         & 80  & 0.357 (0.239) & 0.337 (0.290) & 0.544 (0.533) & 0.508 (0.696) \\
		&         & 100 & 0.390 (0.267) & 0.380 (0.382) & 0.576 (0.566) & 0.573 (0.865) \\
		\cmidrule(lr){2-7}
		& \(\operatorname{sign}(y-x)\) & 60  & 0.230 (0.084) & 0.218 (0.067) & 0.241 (0.092) & 0.219 (0.079) \\
		&                              & 80  & 0.225 (0.074) & 0.219 (0.067) & 0.242 (0.092) & 0.219 (0.080) \\
		&                              & 100 & 0.225 (0.074) & 0.219 (0.066) & 0.241 (0.089) & 0.219 (0.081) \\
		
		\bottomrule
	\end{tabular}
\end{table}

\subsection{\textbf{Effect of estimation error on the limiting distribution}}
We next study how estimation error propagates into the limiting distribution used for inference. As established in {Theorem \ref{theorem: refined estimation}}, the asymptotic distribution of the refined estimator at the $m$th change point satisfies 
\begin{equation} 
	\|\btheta^{(m)}\|_2^2\bigl(\tilde\gamma_m-\gamma_m\bigr) \;\Rightarrow\; \argmax_{s\in\mathbb{R}} Z^{(m)}(s), \qquad m=1,\ldots,M_0, \label{eq:limit_refined}
\end{equation} 
where
\begin{equation} 
	Z^{(m)}(s)= \begin{cases} -s+\sigma_{1,*}^{(m)}W_1^{(m)}(s)+\sigma_{3,*}^{(m)}W_2^{(m)}(s)+\sigma_{2,*}^{(m)}W_3^{(m)}(s), & s>0,\\[0.3em] 0, & s=0,\\[0.3em] \phantom{-}s+\sigma_{1,*}^{(m)}W_1^{(m)}(s)+\sigma_{4,*}^{(m)}W_2^{(m)}(s)+\sigma_{2,*}^{(m)}W_3^{(m)}(s), & s<0, \end{cases} \label{eq:limit_process}
\end{equation} and $\{W_1^{(m)}(s)\}$, $\{W_2^{(m)}(s)\}$, and $\{W_3^{(m)}(s)\}$ are independent standard Brownian motions on $(-\infty,\infty)$. The scale parameters in~\eqref{eq:limit_process} are determined by the signal direction $\btheta^{(m)}$ and the four covariance matrices $\bSigma_1^{(m)},\ldots,\bSigma_4^{(m)}$ through 
\begin{equation} \sigma_{1,*}^{(m)} = \frac{\bigl((\btheta^{(m)})^\top \Sigma_1^{(m)} \btheta^{(m)}\bigr)^{1/2}} {\|\btheta^{(m)}\|_2}, \qquad \sigma_{2,*}^{(m)} = \frac{\bigl((\btheta^{(m)})^\top \Sigma_2^{(m)} \btheta^{(m)}\bigr)^{1/2}} {\|\btheta^{(m)}\|_2}, \label{eq:sigma12star} 
\end{equation} and 
\begin{equation} \sigma_{3,*}^{(m)} = \frac{\bigl((\btheta^{(m)})^\top \Sigma_3^{(m)} \btheta^{(m)}\bigr)^{1/2}} {\|\btheta^{(m)}\|_2}, \qquad \sigma_{4,*}^{(m)} = \frac{\bigl((\btheta^{(m)})^\top \Sigma_4^{(m)} \btheta^{(m)}\bigr)^{1/2}} {\|\btheta^{(m)}\|_2}. \label{eq:sigma34star} 
\end{equation} 
The purpose of this experiment is to study the  effect of estimation error of 
$\btheta^{(m)}, \bSigma_{1}^{(m)},\ldots,\bSigma_{4}^{(m)}$ on the limiting distribution in $\argmax_{s\in\mathbb{R}} Z^{(m)}(s)$. 

We compare two versions of the asymptotic distribution in~\eqref{eq:limit_refined}--\eqref{eq:limit_process}. 
\begin{itemize}
	\item \textbf{Oracle version}: For each $m$th change point, we first construct an \emph{oracle} version by replacing the signal jump  and covariance matrices with their true counterparts, namely $\btheta^{(m)}$ and $\bSigma_1^{(m)},\ldots,\bSigma_4^{(m)}$. Using the Monte Carlo procedure described in {Section~\ref{sec: implementation for estimation confidence}}, we then simulate the process in~\eqref{eq:limit_process} and obtain the corresponding oracle quantiles of the argmax distribution. 
	\item \textbf{Plug-in version}: We then construct a \emph{plug-in} version by replacing the unknown quantities with their estimators, namely the estimated signal jump $\hat{\btheta}^{(m)}$ and the four estimated covariance matrices $\hat\bSigma_1^{(m)},\ldots,\hat\bSigma_4^{(m)}$. Plugging these estimators into~\eqref{eq:sigma12star}--\eqref{eq:sigma34star} yields estimated scale parameters, which are again used in the Monte Carlo simulation of the limit process to obtain the corresponding plug-in quantiles of the argmax distribution. 
\end{itemize}

The results are summarized in Table	\ref{tab:argmax-quantile-gaussian-compact} and
\ref{tab:argmax-quantile-t3-compact}. \textbf{First}, the plug-in quantiles are generally close to the corresponding oracle quantiles across the three change points and different bandwidths \(G\). This indicates that the estimation errors in the signal jump and covariance matrices do not substantially distort the limiting distribution used for inference.
\textbf{Second}, the results are stable with respect to the bandwidth \(G\). For both kernels, the absolute errors change  mildly when \(G\) varies from 60 to 100. This suggests that the plug-in approximation is not sensitive to the particular bandwidth choice, which is consistent with the results observed in the signal-jump and covariance matrix estimation.
\textbf{Third},  under the \(t_3\) distribution, it  is more challenging to estimate the quantiles than the Gaussian setting, especially for the linear kernel. 
\textbf{Lastly},  the discrepancies are more pronounced at the tail quantiles than at the central quantiles. This is reasonable, since extreme quantiles of the argmax distribution are more sensitive to errors in the estimated scale parameters.  Overall, Tables	\ref{tab:argmax-quantile-gaussian-compact} and
\ref{tab:argmax-quantile-t3-compact} show that the plug-in limiting distribution provides a stable approximation to the oracle limiting distribution.

\begin{sidewaystable}
	\centering
	\caption{Gaussian error: comparison between plug-in and oracle quantiles of the limiting argmax distribution under different bandwidths \(G\).}
	\label{tab:argmax-quantile-gaussian-compact}
	\small
	\setlength{\tabcolsep}{3pt}
	\renewcommand{\arraystretch}{1.12}
	
	\begin{adjustbox}{max width=\textwidth}
		\begin{tabular}{lcc*{18}{c}}
			\toprule
			Kernel & \(G\) & cpt
			& \multicolumn{3}{c}{0.025}
			& \multicolumn{3}{c}{0.05}
			& \multicolumn{3}{c}{0.10}
			& \multicolumn{3}{c}{0.90}
			& \multicolumn{3}{c}{0.95}
			& \multicolumn{3}{c}{0.975} \\
			\cmidrule(lr){4-6}
			\cmidrule(lr){7-9}
			\cmidrule(lr){10-12}
			\cmidrule(lr){13-15}
			\cmidrule(lr){16-18}
			\cmidrule(lr){19-21}
			& &
			& Oracle & Plug-in & Abs.Err.
			& Oracle & Plug-in & Abs.Err.
			& Oracle & Plug-in & Abs.Err.
			& Oracle & Plug-in & Abs.Err.
			& Oracle & Plug-in & Abs.Err.
			& Oracle & Plug-in & Abs.Err. \\
			\midrule
			
			\(h(x,y)=y-x\)
			& 60 & cp1
			& -8.844 & -8.838 & 1.598 (1.166)
			& -6.063 & -6.242 & 1.178 (1.002)
			& -3.760 & -3.869 & 0.764 (0.673)
			& 3.722 & 3.806 & 0.740 (0.715)
			& 6.101 & 6.170 & 1.163 (1.037)
			& 8.862 & 8.744 & 1.539 (1.274) \\
			& 60 & cp2
			& -8.182 & -8.870 & 1.692 (1.460)
			& -5.481 & -6.275 & 1.305 (1.303)
			& -3.260 & -3.868 & 0.896 (0.888)
			& 3.622 & 3.953 & 0.822 (0.678)
			& 6.580 & 6.379 & 1.305 (0.837)
			& 9.781 & 8.959 & 1.825 (1.166) \\
			& 60 & cp3
			& -9.260 & -9.133 & 1.732 (1.213)
			& -6.741 & -6.455 & 1.345 (0.917)
			& -4.320 & -3.965 & 0.902 (0.594)
			& 3.662 & 3.920 & 0.807 (0.725)
			& 6.021 & 6.307 & 1.241 (0.998)
			& 8.124 & 8.916 & 1.746 (1.396) \\
			\cmidrule(lr){2-21}
			
			& 80 & cp1
			& -8.844 & -8.925 & 1.603 (1.287)
			& -6.063 & -6.329 & 1.201 (1.084)
			& -3.760 & -3.913 & 0.788 (0.712)
			& 3.722 & 3.856 & 0.774 (0.740)
			& 6.101 & 6.257 & 1.216 (1.040)
			& 8.862 & 8.867 & 1.606 (1.292) \\
			& 80 & cp2
			& -8.182 & -8.966 & 1.721 (1.515)
			& -5.481 & -6.342 & 1.325 (1.332)
			& -3.260 & -3.917 & 0.903 (0.915)
			& 3.622 & 4.004 & 0.861 (0.733)
			& 6.580 & 6.457 & 1.343 (0.885)
			& 9.781 & 9.093 & 1.894 (1.179) \\
			& 80 & cp3
			& -9.260 & -9.114 & 1.721 (1.159)
			& -6.741 & -6.446 & 1.329 (0.854)
			& -4.320 & -3.959 & 0.886 (0.528)
			& 3.662 & 3.915 & 0.808 (0.735)
			& 6.021 & 6.301 & 1.236 (1.024)
			& 8.124 & 8.929 & 1.717 (1.448) \\
			\cmidrule(lr){2-21}
			
			& 100 & cp1
			& -8.844 & -8.831 & 1.566 (1.167)
			& -6.063 & -6.238 & 1.167 (0.970)
			& -3.760 & -3.858 & 0.762 (0.654)
			& 3.722 & 3.826 & 0.757 (0.703)
			& 6.101 & 6.192 & 1.194 (0.998)
			& 8.862 & 8.784 & 1.558 (1.238) \\
			& 100 & cp2
			& -8.182 & -8.926 & 1.680 (1.473)
			& -5.481 & -6.323 & 1.319 (1.291)
			& -3.260 & -3.900 & 0.907 (0.897)
			& 3.622 & 4.006 & 0.820 (0.671)
			& 6.580 & 6.449 & 1.278 (0.807)
			& 9.781 & 9.065 & 1.812 (1.093) \\
			& 100 & cp3
			& -9.260 & -9.079 & 1.726 (1.130)
			& -6.741 & -6.407 & 1.315 (0.849)
			& -4.320 & -3.940 & 0.880 (0.532)
			& 3.662 & 3.882 & 0.799 (0.725)
			& 6.021 & 6.252 & 1.231 (1.006)
			& 8.124 & 8.841 & 1.657 (1.427) \\
			
			\midrule
			
			\(h(x,y)=\operatorname{sign}(y-x)\)
			& 60 & cp1
			& -3.001 & -2.636 & 0.564 (0.437)
			& -1.881 & -1.843 & 0.300 (0.310)
			& -1.162 & -1.128 & 0.196 (0.205)
			& 1.022 & 1.134 & 0.191 (0.189)
			& 1.781 & 1.846 & 0.278 (0.279)
			& 2.540 & 2.633 & 0.400 (0.396) \\
			& 60 & cp2
			& -2.200 & -2.654 & 0.530 (0.472)
			& -1.520 & -1.846 & 0.376 (0.320)
			& -0.980 & -1.121 & 0.209 (0.183)
			& 1.000 & 1.140 & 0.214 (0.214)
			& 1.600 & 1.857 & 0.350 (0.338)
			& 2.401 & 2.654 & 0.438 (0.417) \\
			& 60 & cp3
			& -2.520 & -2.726 & 0.440 (0.432)
			& -1.861 & -1.897 & 0.295 (0.272)
			& -1.100 & -1.162 & 0.189 (0.182)
			& 1.080 & 1.137 & 0.188 (0.182)
			& 1.701 & 1.867 & 0.312 (0.316)
			& 2.500 & 2.677 & 0.437 (0.436) \\
			\cmidrule(lr){2-21}
			
			& 80 & cp1
			& -3.001 & -2.611 & 0.547 (0.360)
			& -1.881 & -1.830 & 0.287 (0.251)
			& -1.162 & -1.115 & 0.186 (0.146)
			& 1.022 & 1.124 & 0.189 (0.187)
			& 1.781 & 1.830 & 0.281 (0.271)
			& 2.540 & 2.614 & 0.405 (0.398) \\
			& 80 & cp2
			& -2.200 & -2.660 & 0.540 (0.441)
			& -1.520 & -1.859 & 0.390 (0.312)
			& -0.980 & -1.127 & 0.211 (0.173)
			& 1.000 & 1.145 & 0.216 (0.221)
			& 1.600 & 1.872 & 0.364 (0.347)
			& 2.401 & 2.678 & 0.451 (0.434) \\
			& 80 & cp3
			& -2.520 & -2.700 & 0.416 (0.396)
			& -1.861 & -1.885 & 0.285 (0.265)
			& -1.100 & -1.152 & 0.179 (0.172)
			& 1.080 & 1.129 & 0.180 (0.183)
			& 1.701 & 1.853 & 0.299 (0.318)
			& 2.500 & 2.658 & 0.413 (0.424) \\
			\cmidrule(lr){2-21}
			
			& 100 & cp1
			& -3.001 & -2.599 & 0.564 (0.346)
			& -1.881 & -1.822 & 0.301 (0.243)
			& -1.162 & -1.111 & 0.193 (0.142)
			& 1.022 & 1.124 & 0.190 (0.186)
			& 1.781 & 1.829 & 0.284 (0.266)
			& 2.540 & 2.610 & 0.406 (0.373) \\
			& 100 & cp2
			& -2.200 & -2.655 & 0.533 (0.437)
			& -1.520 & -1.850 & 0.379 (0.311)
			& -0.980 & -1.125 & 0.209 (0.174)
			& 1.000 & 1.149 & 0.221 (0.237)
			& 1.600 & 1.874 & 0.365 (0.369)
			& 2.401 & 2.691 & 0.462 (0.491) \\
			& 100 & cp3
			& -2.520 & -2.705 & 0.427 (0.432)
			& -1.861 & -1.887 & 0.295 (0.293)
			& -1.100 & -1.156 & 0.186 (0.186)
			& 1.080 & 1.134 & 0.182 (0.186)
			& 1.701 & 1.863 & 0.306 (0.324)
			& 2.500 & 2.669 & 0.420 (0.440) \\
			
			\bottomrule
		\end{tabular}
	\end{adjustbox}
	
	\vspace{0.5em}
	\begin{minipage}{0.98\textwidth}
		\footnotesize
		Note: Oracle denotes the average oracle quantile, Plug-in denotes the average plug-in quantile,
		and Abs.Err. denotes the mean absolute error with its standard deviation in parentheses. We compare the quantiles of 0.025,0.05,0.1,0.9,0.95,0.975.
	\end{minipage}
\end{sidewaystable}

\begin{sidewaystable}
	\centering
	\caption{ Student \(t_3\) error: comparison between plug-in and oracle quantiles of the limiting argmax distribution under different bandwidths \(G\).}
	\label{tab:argmax-quantile-t3-compact}
	\small
	\setlength{\tabcolsep}{3pt}
	\renewcommand{\arraystretch}{1.12}
	
	\begin{adjustbox}{max width=\textwidth}
		\begin{tabular}{lcc*{18}{c}}
			\toprule
			Kernel & \(G\) & cpt
			& \multicolumn{3}{c}{0.025}
			& \multicolumn{3}{c}{0.05}
			& \multicolumn{3}{c}{0.10}
			& \multicolumn{3}{c}{0.90}
			& \multicolumn{3}{c}{0.95}
			& \multicolumn{3}{c}{0.975} \\
			\cmidrule(lr){4-6}
			\cmidrule(lr){7-9}
			\cmidrule(lr){10-12}
			\cmidrule(lr){13-15}
			\cmidrule(lr){16-18}
			\cmidrule(lr){19-21}
			& &
			& Oracle & Plug-in & Abs.Err.
			& Oracle & Plug-in & Abs.Err.
			& Oracle & Plug-in & Abs.Err.
			& Oracle & Plug-in & Abs.Err.
			& Oracle & Plug-in & Abs.Err.
			& Oracle & Plug-in & Abs.Err. \\
			\midrule
			
			\(h(x,y)=y-x\)
			& 60 & cp1
			& -17.942 & -17.276 & 1.962 (1.810)
			& -15.701 & -14.885 & 2.730 (2.114)
			& -10.762 & -11.167 & 3.102 (1.619)
			& 8.782 & 9.215 & 2.662 (1.981)
			& 12.701 & 13.291 & 2.669 (1.874)
			& 16.949 & 16.328 & 2.057 (1.841) \\
			& 60 & cp2
			& -16.901 & -17.629 & 2.254 (1.252)
			& -13.845 & -15.385 & 3.330 (1.563)
			& -9.242 & -11.681 & 3.877 (2.069)
			& 8.262 & 10.253 & 2.789 (2.275)
			& 12.868 & 14.437 & 2.613 (1.697)
			& 16.500 & 17.245 & 1.769 (1.103) \\
			& 60 & cp3
			& -17.162 & -16.232 & 1.851 (1.447)
			& -14.141 & -13.048 & 2.486 (1.651)
			& -9.746 & -8.880 & 2.290 (1.526)
			& 8.804 & 9.913 & 2.852 (2.124)
			& 13.324 & 13.841 & 2.849 (1.713)
			& 16.620 & 16.602 & 2.133 (1.526) \\
			\cmidrule(lr){2-21}
			
			& 80 & cp1
			& -17.942 & -16.858 & 1.920 (1.716)
			& -15.701 & -14.017 & 2.870 (2.031)
			& -10.762 & -10.113 & 2.843 (1.829)
			& 8.782 & 9.555 & 2.576 (1.946)
			& 12.701 & 13.590 & 2.666 (1.757)
			& 16.949 & 16.480 & 1.927 (1.589) \\
			& 80 & cp2
			& -16.901 & -16.516 & 2.200 (1.640)
			& -13.845 & -13.674 & 2.874 (1.782)
			& -9.242 & -9.752 & 2.887 (1.974)
			& 8.262 & 9.639 & 2.533 (2.226)
			& 12.868 & 13.659 & 2.472 (1.699)
			& 16.500 & 16.654 & 1.804 (1.279) \\
			& 80 & cp3
			& -17.162 & -16.406 & 1.826 (1.291)
			& -14.141 & -13.343 & 2.443 (1.513)
			& -9.746 & -9.241 & 2.376 (1.554)
			& 8.804 & 9.354 & 2.738 (1.932)
			& 13.324 & 13.257 & 2.751 (1.807)
			& 16.620 & 16.167 & 2.172 (1.656) \\
			\cmidrule(lr){2-21}
			
			& 100 & cp1
			& -17.942 & -16.415 & 2.128 (1.952)
			& -15.701 & -13.448 & 3.052 (2.219)
			& -10.762 & -9.430 & 2.856 (1.826)
			& 8.782 & 9.618 & 2.527 (1.886)
			& 12.701 & 13.627 & 2.635 (1.702)
			& 16.949 & 16.506 & 1.931 (1.512) \\
			& 100 & cp2
			& -16.901 & -16.239 & 2.152 (1.617)
			& -13.845 & -13.311 & 2.745 (1.785)
			& -9.242 & -9.354 & 2.696 (1.963)
			& 8.262 & 9.415 & 2.449 (2.074)
			& 12.868 & 13.447 & 2.432 (1.616)
			& 16.500 & 16.495 & 1.817 (1.287) \\
			& 100 & cp3
			& -17.162 & -16.472 & 1.763 (1.415)
			& -14.141 & -13.450 & 2.350 (1.641)
			& -9.746 & -9.334 & 2.347 (1.631)
			& 8.804 & 9.482 & 2.648 (1.891)
			& 13.324 & 13.443 & 2.651 (1.670)
			& 16.620 & 16.370 & 2.040 (1.489) \\
			
			\midrule
			
			\(h(x,y)=\operatorname{sign}(y-x)\)
			& 60 & cp1
			& -3.262 & -3.080 & 0.523 (0.582)
			& -2.302 & -2.154 & 0.395 (0.404)
			& -1.420 & -1.309 & 0.258 (0.271)
			& 1.360 & 1.357 & 0.227 (0.185)
			& 2.080 & 2.205 & 0.350 (0.326)
			& 3.000 & 3.153 & 0.504 (0.435) \\
			& 60 & cp2
			& -3.260 & -3.164 & 0.558 (0.491)
			& -2.322 & -2.206 & 0.409 (0.334)
			& -1.382 & -1.346 & 0.247 (0.219)
			& 1.382 & 1.326 & 0.245 (0.176)
			& 2.362 & 2.165 & 0.416 (0.272)
			& 3.400 & 3.110 & 0.569 (0.395) \\
			& 60 & cp3
			& -2.942 & -3.268 & 0.571 (0.517)
			& -2.061 & -2.287 & 0.408 (0.373)
			& -1.340 & -1.396 & 0.244 (0.214)
			& 1.360 & 1.360 & 0.235 (0.206)
			& 2.041 & 2.230 & 0.377 (0.361)
			& 3.060 & 3.194 & 0.520 (0.459) \\
			\cmidrule(lr){2-21}
			
			& 80 & cp1
			& -3.262 & -3.070 & 0.453 (0.325)
			& -2.302 & -2.143 & 0.340 (0.241)
			& -1.420 & -1.302 & 0.221 (0.152)
			& 1.360 & 1.366 & 0.223 (0.192)
			& 2.080 & 2.222 & 0.350 (0.333)
			& 3.000 & 3.162 & 0.489 (0.462) \\
			& 80 & cp2
			& -3.260 & -3.176 & 0.540 (0.456)
			& -2.322 & -2.208 & 0.393 (0.312)
			& -1.382 & -1.348 & 0.242 (0.202)
			& 1.382 & 1.340 & 0.244 (0.208)
			& 2.362 & 2.185 & 0.414 (0.310)
			& 3.400 & 3.134 & 0.570 (0.446) \\
			& 80 & cp3
			& -2.942 & -3.272 & 0.568 (0.473)
			& -2.061 & -2.289 & 0.397 (0.346)
			& -1.340 & -1.394 & 0.238 (0.191)
			& 1.360 & 1.357 & 0.242 (0.206)
			& 2.041 & 2.221 & 0.377 (0.369)
			& 3.060 & 3.193 & 0.523 (0.474) \\
			\cmidrule(lr){2-21}
			
			& 100 & cp1
			& -3.262 & -3.062 & 0.490 (0.360)
			& -2.302 & -2.140 & 0.362 (0.263)
			& -1.420 & -1.300 & 0.234 (0.163)
			& 1.360 & 1.369 & 0.213 (0.179)
			& 2.080 & 2.226 & 0.345 (0.312)
			& 3.000 & 3.176 & 0.497 (0.430) \\
			& 100 & cp2
			& -3.260 & -3.141 & 0.532 (0.436)
			& -2.322 & -2.183 & 0.391 (0.288)
			& -1.382 & -1.337 & 0.239 (0.199)
			& 1.382 & 1.325 & 0.244 (0.194)
			& 2.362 & 2.160 & 0.421 (0.290)
			& 3.400 & 3.098 & 0.582 (0.416) \\
			& 100 & cp3
			& -2.942 & -3.263 & 0.562 (0.477)
			& -2.061 & -2.289 & 0.398 (0.355)
			& -1.340 & -1.391 & 0.236 (0.197)
			& 1.360 & 1.356 & 0.240 (0.197)
			& 2.041 & 2.221 & 0.372 (0.349)
			& 3.060 & 3.195 & 0.515 (0.446) \\
			
			\bottomrule
		\end{tabular}
	\end{adjustbox}
	
	\vspace{0.5em}
	\begin{minipage}{0.98\textwidth}
		\footnotesize
		Note: Oracle denotes the average oracle quantile, Plug-in denotes the average plug-in quantile,
		and Abs.Err. denotes the mean absolute error with its standard deviation in parentheses. We compare the quantiles of 0.025,0.05,0.1,0.9,0.95,0.975.
	\end{minipage}
\end{sidewaystable}

\subsection{{\textbf{Practical accuracy of the asymptotic approximation}}}

We  assess the practical accuracy of the asymptotic approximation, where we aim to study oracle limiting distribution can be approximated by  the finite-sample based refined change-point estimators.  Note that different from the previous subsection, where we compared the plug-in and oracle versions of the {limiting distribution itself},  the object  is the empirical distribution of the refined estimator obtained from repeated simulations.

Specifically, for the $m$th change point, we collect the refined estimators from $200$ Monte Carlo replications and use them to form the empirical quantiles at the levels
$0.025,0.05, 0.10$ $0.90,0.95, 0.975$.
These empirical quantiles are then compared with the corresponding quantiles of the oracle asymptotic distribution.

Note that  we construct the empirical quantiles using the refined estimators on the \emph{original location scale}, namely $\tilde{\tau}_m$, without centering by the true change point or scaling by $\|\btheta^{(m)}\|_2^2$. The reason is that this comparison is more directly interpretable in finite samples. 

The results are given in Tables	\ref{tab:raw-quantile-accuracy-gaussian} and 	\ref{tab:raw-quantile-accuracy-t3}.


\begin{table}[H]
	\centering
	\caption{Gaussian error:  comparison between empirical quantiles of the refined estimators and oracle  quantiles under different bandwidths \(G\). }
	\label{tab:raw-quantile-accuracy-gaussian}
	\small
	\setlength{\tabcolsep}{4pt}
	\renewcommand{\arraystretch}{1}
	
	\begin{adjustbox}{max width=\textwidth}
		\begin{tabular}{lcc*{12}{c}}
			\toprule
			Kernel & \(G\) & cpt
			& \multicolumn{2}{c}{0.025}
			& \multicolumn{2}{c}{0.05}
			& \multicolumn{2}{c}{0.10}
			& \multicolumn{2}{c}{0.90}
			& \multicolumn{2}{c}{0.95}
			& \multicolumn{2}{c}{0.975} \\
			\cmidrule(lr){4-5}
			\cmidrule(lr){6-7}
			\cmidrule(lr){8-9}
			\cmidrule(lr){10-11}
			\cmidrule(lr){12-13}
			\cmidrule(lr){14-15}
			& &
			& Emp. & Oracle
			& Emp. & Oracle
			& Emp. & Oracle
			& Emp. & Oracle
			& Emp. & Oracle
			& Emp. & Oracle \\
			\midrule
			
			\(h(x,y)=y-x\)
			& 60 & cp1 & 178.000 & 178.609 & 179.000 & 179.046 & 179.000 & 179.409 & 180.000 & 180.585 & 181.000 & 180.960 & 181.000 & 181.394 \\
			& 60 & cp2 & 299.000 & 298.713 & 299.000 & 299.138 & 300.000 & 299.487 & 300.000 & 300.570 & 301.000 & 301.035 & 301.000 & 301.538 \\
			& 60 & cp3 & 418.975 & 418.544 & 419.000 & 418.940 & 420.000 & 419.321 & 420.000 & 420.576 & 421.000 & 420.947 & 421.000 & 421.278 \\
			
			\cmidrule(lr){2-15}
			& 80 & cp1 & 179.000 & 178.609 & 179.000 & 179.046 & 180.000 & 179.409 & 180.000 & 180.585 & 181.000 & 180.960 & 181.000 & 181.394 \\
			& 80 & cp2 & 299.000 & 298.713 & 299.000 & 299.138 & 300.000 & 299.487 & 300.000 & 300.570 & 301.000 & 301.035 & 302.000 & 301.538 \\
			& 80 & cp3 & 418.000 & 418.544 & 419.000 & 418.940 & 420.000 & 419.321 & 421.000 & 420.576 & 421.000 & 420.947 & 421.000 & 421.278 \\
			
			\cmidrule(lr){2-15}
			& 100 & cp1 & 179.000 & 178.609 & 179.000 & 179.046 & 179.000 & 179.409 & 180.000 & 180.585 & 181.000 & 180.960 & 181.000 & 181.394 \\
			& 100 & cp2 & 299.000 & 298.713 & 299.000 & 299.138 & 300.000 & 299.487 & 301.000 & 300.570 & 301.000 & 301.035 & 301.000 & 301.538 \\
			& 100 & cp3 & 418.000 & 418.544 & 419.000 & 418.940 & 419.000 & 419.321 & 420.000 & 420.576 & 421.000 & 420.947 & 421.000 & 421.278 \\
			
			\midrule
			
			\(h(x,y)=\operatorname{sign}(y-x)\)
			& 60 & cp1 & 178.000 & 178.180 & 179.000 & 178.859 & 179.000 & 179.295 & 180.000 & 180.620 & 181.000 & 181.080 & 181.050 & 181.541 \\
			& 60 & cp2 & 299.000 & 298.665 & 299.000 & 299.078 & 300.000 & 299.406 & 301.000 & 300.607 & 301.000 & 300.971 & 302.000 & 301.456 \\
			& 60 & cp3 & 418.975 & 418.471 & 419.000 & 418.871 & 419.000 & 419.333 & 421.000 & 420.655 & 421.000 & 421.032 & 422.000 & 421.517 \\
			
			\cmidrule(lr){2-15}
			& 80 & cp1 & 178.975 & 178.180 & 179.000 & 178.859 & 179.000 & 179.295 & 180.000 & 180.620 & 181.000 & 181.080 & 181.000 & 181.541 \\
			& 80 & cp2 & 299.000 & 298.665 & 299.000 & 299.078 & 300.000 & 299.406 & 300.100 & 300.607 & 301.000 & 300.971 & 302.000 & 301.456 \\
			& 80 & cp3 & 418.000 & 418.471 & 419.000 & 418.871 & 420.000 & 419.333 & 421.000 & 420.655 & 421.000 & 421.032 & 421.025 & 421.517 \\
			
			\cmidrule(lr){2-15}
			& 100 & cp1 & 179.000 & 178.180 & 179.000 & 178.859 & 179.000 & 179.295 & 180.000 & 180.620 & 181.000 & 181.080 & 182.000 & 181.541 \\
			& 100 & cp2 & 299.000 & 298.665 & 299.000 & 299.078 & 299.000 & 299.406 & 301.000 & 300.607 & 301.000 & 300.971 & 301.000 & 301.456 \\
			& 100 & cp3 & 418.000 & 418.471 & 418.000 & 418.871 & 419.000 & 419.333 & 420.000 & 420.655 & 421.000 & 421.032 & 421.000 & 421.517 \\
			
			\bottomrule
		\end{tabular}
	\end{adjustbox}
	
	\vspace{0.5em}
	\begin{minipage}{0.98\textwidth}
		\footnotesize
		Note: Emp. denotes the empirical quantile of the refined change-point estimator over 200 replications, and Oracle denotes the average oracle raw quantile on the original location scale.
	\end{minipage}
\end{table}


\begin{table}[H]
	\centering
	\caption{Student \(t_3\) error: comparison between empirical quantiles of the refined estimators and oracle  quantiles under different bandwidths \(G\). }
	\label{tab:raw-quantile-accuracy-t3}
	\setlength{\tabcolsep}{4pt}
	\renewcommand{\arraystretch}{1.12}
	
	\begin{adjustbox}{max width=\textwidth}
		\begin{tabular}{lcc*{12}{c}}
			\toprule
			Kernel & \(G\) & cpt
			& \multicolumn{2}{c}{0.025}
			& \multicolumn{2}{c}{0.05}
			& \multicolumn{2}{c}{0.10}
			& \multicolumn{2}{c}{0.90}
			& \multicolumn{2}{c}{0.95}
			& \multicolumn{2}{c}{0.975} \\
			\cmidrule(lr){4-5}
			\cmidrule(lr){6-7}
			\cmidrule(lr){8-9}
			\cmidrule(lr){10-11}
			\cmidrule(lr){12-13}
			\cmidrule(lr){14-15}
			& &
			& Emp. & Oracle
			& Emp. & Oracle
			& Emp. & Oracle
			& Emp. & Oracle
			& Emp. & Oracle
			& Emp. & Oracle \\
			\midrule
			
			\(h(x,y)=y-x\)
			& 60 & cp1 & 178.000 & 178.194 & 178.050 & 178.420 & 179.000 & 178.917 & 181.000 & 180.884 & 182.000 & 181.278 & 189.325 & 181.706 \\
			& 60 & cp2 & 294.175 & 298.299 & 295.450 & 298.606 & 298.000 & 299.070 & 301.100 & 300.832 & 302.000 & 301.295 & 303.550 & 301.661 \\
			& 60 & cp3 & 419.000 & 418.273 & 419.150 & 418.577 & 420.000 & 419.019 & 420.000 & 420.886 & 421.000 & 421.341 & 422.000 & 421.673 \\
			
			\cmidrule(lr){2-15}
			& 80 & cp1 & 174.000 & 178.194 & 177.000 & 178.420 & 179.000 & 178.917 & 181.000 & 180.884 & 182.000 & 181.278 & 182.200 & 181.706 \\
			& 80 & cp2 & 293.525 & 298.299 & 297.000 & 298.606 & 299.000 & 299.070 & 301.000 & 300.832 & 302.000 & 301.295 & 303.000 & 301.661 \\
			& 80 & cp3 & 416.650 & 418.273 & 419.000 & 418.577 & 419.000 & 419.019 & 421.000 & 420.886 & 421.700 & 421.341 & 423.000 & 421.673 \\
			
			\cmidrule(lr){2-15}
			& 100 & cp1 & 177.000 & 178.194 & 178.000 & 178.420 & 179.000 & 178.917 & 181.000 & 180.884 & 182.000 & 181.278 & 182.500 & 181.706 \\
			& 100 & cp2 & 296.525 & 298.299 & 297.050 & 298.606 & 299.000 & 299.070 & 301.000 & 300.832 & 301.950 & 301.295 & 302.475 & 301.661 \\
			& 100 & cp3 & 418.000 & 418.273 & 419.000 & 418.577 & 419.000 & 419.019 & 421.000 & 420.886 & 423.000 & 421.341 & 424.275 & 421.673 \\
			
			\midrule
			
			\(h(x,y)=\operatorname{sign}(y-x)\)
			& 60 & cp1 & 178.000 & 178.019 & 178.000 & 178.602 & 179.000 & 179.137 & 181.000 & 180.826 & 182.000 & 181.263 & 182.000 & 181.822 \\
			& 60 & cp2 & 298.000 & 298.020 & 298.800 & 298.590 & 299.000 & 299.161 & 301.000 & 300.839 & 302.000 & 301.435 & 303.000 & 302.066 \\
			& 60 & cp3 & 418.975 & 418.213 & 419.000 & 418.748 & 419.000 & 419.186 & 421.000 & 420.826 & 421.000 & 421.240 & 422.000 & 421.859 \\
			
			\cmidrule(lr){2-15}
			& 80 & cp1 & 178.000 & 178.019 & 178.000 & 178.602 & 179.000 & 179.137 & 181.000 & 180.826 & 181.000 & 181.263 & 181.025 & 181.822 \\
			& 80 & cp2 & 297.000 & 298.020 & 298.000 & 298.590 & 299.000 & 299.161 & 301.000 & 300.839 & 302.000 & 301.435 & 302.000 & 302.066 \\
			& 80 & cp3 & 418.000 & 418.213 & 419.000 & 418.748 & 420.000 & 419.186 & 421.000 & 420.826 & 422.000 & 421.240 & 422.000 & 421.859 \\
			
			\cmidrule(lr){2-15}
			& 100 & cp1 & 178.000 & 178.019 & 178.000 & 178.602 & 179.000 & 179.137 & 181.000 & 180.826 & 181.000 & 181.263 & 182.000 & 181.822 \\
			& 100 & cp2 & 298.000 & 298.020 & 298.950 & 298.590 & 299.000 & 299.161 & 301.000 & 300.839 & 301.000 & 301.435 & 302.000 & 302.066 \\
			& 100 & cp3 & 418.000 & 418.213 & 418.000 & 418.748 & 419.000 & 419.186 & 421.000 & 420.826 & 421.000 & 421.240 & 422.000 & 421.859 \\
			
			\bottomrule
		\end{tabular}
	\end{adjustbox}
	
	\vspace{0.5em}
	\begin{minipage}{0.98\textwidth}
		\footnotesize
		Note: Emp. denotes the empirical quantile of the refined change-point estimator over 200 replications, and Oracle denotes the average oracle raw quantile on the original location scale.
	\end{minipage}
\end{table}

\textbf{First}, the empirical quantiles of the refined estimators are generally close to the oracle  quantiles. Under the Gaussian setting, the empirical quantiles for both kernels approximate the oracle quantiles well across the three change points.  \textbf{Second}, the results are stable with respect to the bandwidth \(G\).  This suggests that the practical accuracy of the asymptotic approximation is not very sensitive to the specific bandwidth choice. \textbf{Third}, the sign kernel shows more stable performance than the linear kernel, especially under the \(t_3\) distribution. Specifically, for \(h(x,y)=\operatorname{sign}(y-x)\), the empirical and oracle quantiles remain close across different bandwidths and change points. In contrast, for \(h(x,y)=y-x\), the discrepancies become larger under the heavy-tailed \(t_3\) setting, particularly at the extreme quantile levels.

\textbf{Lastly}, it is worth mentioning that the empirical quantiles appear somewhat discrete because the refined change-point estimators are integer-valued locations on the original time scale.  Specifically, the empirical quantiles are computed from 200 Monte Carlo replications, many of them take values around the true change points such as 179, 180, or 181.  Overall, the above results provide finite-sample evidence that the oracle asymptotic distribution gives a relative stable and practically accurate approximation to the distribution of the refined change-point estimator.

\subsection{\textbf{Sensitivity to active-set misspecification}}

\noindent We next examine the finite-sample sensitivity of the proposed refinement and inference
procedure to the misspecification of the active set. Recall that, for the \(m\)-th change point,
the active set is defined as
$\bPi_m=\{j: \theta_j^{(m)}\neq 0\}$,
and the U-PRA refinement step first estimates \(\Pi_m\) and then projects the local
U-statistic process onto the selected coordinates. 

Although {Theorem~\ref{theorem: support recovery}} establishes the
consistency of active-set recovery, in finite samples
the selected active set may still contain false positive or false negative coordinates. Since
the limiting distribution and the resulting confidence intervals are constructed after the
projection step, it is important to evaluate how such active-set misspecification affects the
final inference results.

To isolate this effect, we conduct a controlled sensitivity analysis by manually perturbing
the active set used in the refinement step. In the simulation design considered here, the true
active set is fixed as
$\Pi_m=\{1,2,3,4,5\},  m=1,2,3.$
For each  initial change-point estimate, we keep all other steps of the proposed
procedure unchanged, the only modification is the coordinate set used in the projection step.

Specifically, we consider the following active-set choices:
\[
\begin{aligned}
	&\text{Oracle:} && \mathcal A_m=\Pi_m,\\
	&\text{Remove 1:} && \mathcal A_m \subset \Pi_m,\quad |\Pi_m\setminus\mathcal A_m|=1,\\
	&\text{Remove 3:} && \mathcal A_m \subset \Pi_m,\quad |\Pi_m\setminus\mathcal A_m|=3,\\
	&\text{Add 1:} && \mathcal A_m=\Pi_m\cup \mathcal I_m,\quad |\mathcal I_m|=1,\\
	&\text{Add 3:} && \mathcal A_m=\Pi_m\cup \mathcal I_m,\quad |\mathcal I_m|=3,\\
	&\text{Add 5:} && \mathcal A_m=\Pi_m\cup \mathcal I_m,\quad |\mathcal I_m|=5,\\
	&\text{Add 8:} && \mathcal A_m=\Pi_m\cup \mathcal I_m,\quad |\mathcal I_m|=8,
\end{aligned}
\]
where \(\mathcal I_m\subset \Pi_m^c\) denotes a set of inactive coordinates.  This construction allows us to separately examine
whether the inference procedure is more sensitive to missing true signal coordinates or to
including additional noise coordinates.

For each perturbed active set \(\mathcal A_m\), we recompute the refined estimator by
\[
\widetilde\gamma_m(\mathcal A_m)
=
\arg\max_{|k-\widehat\gamma_m|\leq G/4}
\sum_{j\in \mathcal A_m}
\widetilde T_j(k),
\]
where \(\widetilde T_j(k)\) denotes the projected local U-statistic. We then construct the corresponding confidence interval using the same
inference procedure as in {Section~\ref{sec: implementation for estimation confidence}} based on the manually selected active set. For each setting, we report three quantities:
the refined localization error, the empirical coverage probability of the confidence interval,
and the average confidence interval length.

The results in  Table 	\ref{tab:active-set-sensitivity}. \textbf{First}, the oracle active set leads to accurate and stable inference. For both kernels, the refined localization error is small, and the empirical coverage probabilities are close to the nominal level across the three change points. This confirms the good finite-sample performance of the refinement and inference procedure when the active set is correctly specified.

\textbf{Second}, the method is relatively insensitive to adding inactive coordinates. When one, three, five, or even eight noise coordinates are added to the true active set, the refined localization error increases only mildly. The coverage probabilities also remain close to the nominal level. \textbf{Third}, missing true signal coordinates has a more serious effect. Deleting one active coordinate only slightly changes the localization error and coverage probabilities. However, deleting three active coordinates leads to a clear increase in the refined localization error. This is expected because removing true signal coordinates directly weakens the projected signal used in the refinement step.



Overall,  the results show that the proposed refinement and inference procedure is reasonably stable to moderate active-set misspecification. The method is particularly robust to false positive coordinates, while false negative coordinates have a stronger impact.


\begin{table}[H]
	\centering
	\caption{Sensitivity of refined change-point inference to active-set misspecification.}
	\label{tab:active-set-sensitivity}
	\small
	\setlength{\tabcolsep}{8pt}
	\renewcommand{\arraystretch}{1.12}
	\begin{tabular}{llcccc}
		\toprule
		Kernel & Active set & Refind Haus & CPT1 coverage & CPT2 coverage & CPT3 coverage \\
		\midrule
		
		\multicolumn{6}{c}{{Normal} Distribution} \\
		\midrule
		\(h(x,y)=y-x\) & Oracle  & 0.88 & 0.955 & 0.965 & 0.965 \\
		& Delete 1 & 0.89 & 0.950 & 0.950 & 0.970 \\
		& Delete 3 & 1.65 & 0.955 & 0.940 & 0.955 \\
		& Add 1   & 0.89 & 0.955 & 0.955 & 0.965 \\
		& Add 3   & 0.90 & 0.965 & 0.960 & 0.970 \\
		& Add 5   & 0.94 & 0.960 & 0.960 & 0.965 \\
		& Add 8   & 0.98 & 0.945 & 0.965 & 0.965 \\
		\midrule
		\(h(x,y)=\operatorname{sign}(y-x)\) & Oracle  & 0.65 & 0.960 & 0.975 & 0.965 \\
		& Delete 1 & 0.67 & 0.960 & 0.940 & 0.970 \\
		& Delete 3 & 1.51 & 0.945 & 0.935 & 0.940 \\
		& Add 1   & 0.65 & 0.965 & 0.980 & 0.975 \\
		& Add 3   & 0.72 & 0.965 & 0.975 & 0.965 \\
		& Add 5   & 0.79 & 0.950 & 0.970 & 0.960 \\
		& Add 8   & 0.80 & 0.950 & 0.960 & 0.965 \\
		
		\midrule
		\multicolumn{6}{c}{{Student \(t_3\)} Distribution} \\
		\midrule
		\(h(x,y)=y-x\) & Oracle  & 0.88 & 0.955 & 0.965 & 0.965 \\
		& Delete 1 & 0.89 & 0.950 & 0.950 & 0.970 \\
		& Delete 3 & 1.65 & 0.955 & 0.940 & 0.955 \\
		& Add 1   & 0.89 & 0.955 & 0.955 & 0.965 \\
		& Add 3   & 0.90 & 0.965 & 0.960 & 0.970 \\
		& Add 5   & 0.94 & 0.960 & 0.960 & 0.965 \\
		& Add 8   & 0.98 & 0.945 & 0.965 & 0.965 \\
		\midrule
		\(h(x,y)=\operatorname{sign}(y-x)\) & Oracle  & 0.65 & 0.960 & 0.975 & 0.965 \\
		& Delete 1 & 0.67 & 0.960 & 0.940 & 0.970 \\
		& Delete 3 & 1.51 & 0.945 & 0.935 & 0.940 \\
		& Add 1   & 0.65 & 0.965 & 0.980 & 0.975 \\
		& Add 3   & 0.72 & 0.965 & 0.975 & 0.965 \\
		& Add 5   & 0.79 & 0.950 & 0.970 & 0.960 \\
		& Add 8   & 0.80 & 0.950 & 0.960 & 0.965 \\
		
		\bottomrule
	\end{tabular}
\end{table}

\section{{Additional Robustness and Sensitivity Analyses}}\label{sec: Additional Robustness and Sensitivity Analyses}

\subsection{{\textbf{Sensitivity to time-series dependence}}}

We further examine the sensitivity of the proposed method to time-series dependence within each segment. The data-generating mechanism is the same as that in {Section~\ref{section: additional empirical study}}, except that the independent Gaussian errors are replaced by serially dependent Gaussian errors. Specifically, for \(t=1,\ldots,n\), we generate
\[
\bX_t=\bmu_t+\bvarepsilon_t,
\]
where the piecewise constant mean function is given by
\[
\bmu_t
=
\bmu_1 \mathbf{1}\{1\le t\le \gamma_1\}
+
\bmu_2 \mathbf{1}\{\gamma_1<t\le \gamma_2\}
+
\bmu_3 \mathbf{1}\{\gamma_2<t\le \gamma_3\}
+
\bmu_4 \mathbf{1}\{\gamma_3<t\le n\}.
\]
where the true change points are located at
$\gamma_1=0.3n,\gamma_2=0.5n, \gamma_3=0.7n$.
The segment mean vectors \(\bmu_1,\ldots,\bmu_4\) and the signal strength are kept the same as those in Section~\ref{sec: changes in mean}. 

To introduce temporal dependence, we generate the error process from a Gaussian AR(1) model with cross-sectional covariance structure:
\[
\bvarepsilon_t=\rho\bvarepsilon_{t-1}+c_\rho \mathbf{\eta}_t,
\qquad \eta_t\sim N(0,\Sigma),
\]
where \(\bSigma\) is the same cross-sectional covariance matrix as in the independent setting. We consider three levels of serial dependence with $\rho\in\{0.2,0.4,0.6\}$. The scaling constant \(c_\rho\) is chosen as
\[
c_\rho=(1-\rho)
\left(\frac{1+\rho}{1-\rho}\right)^{\kappa/2},
\]
with \(\kappa=0.5\) fixed throughout the experiment. This construction provides an intermediate normalization between fixing the marginal covariance and fixing the long-run covariance, and allows the long-run noise level to increase moderately as \(\rho\) becomes larger.

For each simulated dataset, we apply the proposed multiscale procedure. Specifically, we combine the candidate change points obtained from multiple bandwidths with $G\in\{60,80,100\}$ and then perform the same local refinement and inference steps as in the main simulation study. 

The change-point estimation results are reported in Table \ref{tab:serial-dependence}. When \(\rho=0.2\), the method still provides accurate localization after refinement, and the probability of correctly estimating the number of change points remains high. As \(\rho\) increases to \(0.4\) and \(0.6\), the Hausdorff localization error increases and the probability of over-estimating the number of change points becomes larger. This indicates that stronger serial dependence mainly affects the screening step and may generate additional spurious local peaks.

The confidence interval results are reported in Table \ref{tab:serial-dependence-coverage}. The empirical coverage probabilities remain reasonably close to the nominal level under mild serial dependence, while they become less stable when the temporal dependence becomes stronger. Overall, these results suggest that the proposed method is reasonably robust to mild time-series dependence, but strong serial dependence can deteriorate both change-point estimation and finite-sample inference accuracy.

\begin{table}[htbp]
	\centering
	\caption{Simulation results under serial dependence. The table reports the Hausdorff localization error and the proportions of correctly and over-estimated numbers of change points under the multiscale scheme.}
	\label{tab:serial-dependence}
	\begin{tabular}{llccc}
		\toprule
		Serial dependence & Method & Haus & $\widehat{M}_0=M_0$ & $\widehat{M}_0>M_0$ \\
		\midrule
		$\rho=0.2$ 
		& $h_1$-initial & 9.240 & 0.900 & 0.100 \\
		& $h_1$-refined & 1.500 & 0.900 & 0.100 \\
		& $h_2$-initial & 9.925 & 0.890 & 0.110 \\
		& $h_2$-refined & 3.135 & 0.890 & 0.110 \\
		\midrule
		$\rho=0.4$ 
		& $h_1$-initial & 24.660 & 0.660 & 0.340 \\
		& $h_1$-refined & 12.850 & 0.660 & 0.340 \\
		& $h_2$-initial & 48.115 & 0.275 & 0.725 \\
		& $h_2$-refined & 35.000 & 0.275 & 0.725 \\
		\midrule
		$\rho=0.6$ 
		& $h_1$-initial & 65.300 & 0.120 & 0.880 \\
		& $h_1$-refined & 56.580 & 0.120 & 0.880 \\
		& $h_2$-initial & 82.920 & 0.005 & 0.995 \\
		& $h_2$-refined & 81.645 & 0.005 & 0.995 \\
		\bottomrule
	\end{tabular}
\end{table}

\begin{table}[htbp]
	\centering
	\caption{Empirical coverage probabilities under serial dependence with the multiscale scheme.}
	\label{tab:serial-dependence-coverage}
	\begin{tabular}{lcccccc}
		\toprule
		\multirow{2}{*}{Serial dependence} 
		& \multicolumn{3}{c}{$h_1(x,y)=y-x$} 
		& \multicolumn{3}{c}{$h_2(x,y)=\operatorname{sign}(y-x)$} \\
		\cmidrule(lr){2-4} \cmidrule(lr){5-7}
		& $\gamma_1$ & $\gamma_2$ & $\gamma_3$ 
		& $\gamma_1$ & $\gamma_2$ & $\gamma_3$ \\
		\midrule
		$\rho=0.2$ & 0.915 & 0.930 & 0.925 & 0.920 & 0.900 & 0.905 \\
		$\rho=0.4$ & 0.885 & 0.915 & 0.885 & 0.890 & 0.875 & 0.855 \\
		$\rho=0.6$ & 0.855 & 0.925 & 0.900 & 0.843 & 0.795 & 0.844 \\
		\bottomrule
	\end{tabular}
\end{table}

\subsection{{\textbf{Sensitivity to within-segment heteroskedasticity}}}

We next examine the sensitivity of the proposed method to within-segment heteroskedasticity. The baseline data-generating mechanism, the true change-point locations, and the signal strength are kept the same as those in {Section~\ref{sec: changes in mean}}. The only modification is that the Gaussian errors are allowed to have time-varying variances within each stationary segment.

Specifically, we generate
\[
\bvarepsilon_t=\bsigma_t \mathbf{\eta}_t,
\qquad \eta_t\sim N(0,\Sigma),
\]
where \(\bSigma\) is the same cross-sectional covariance matrix as in the independent Gaussian setting. Let \(r_t\) denote the relative location of \(t\) within its own segment. If \(\gamma_{j-1}<t\le \gamma_j\), with \(\gamma_0=0\) and \(\gamma_4=n\), we define
\[
r_t=\frac{t-\gamma_{j-1}}{\gamma_j-\gamma_{j-1}}.
\]
The variance scale is set as
\[
\sigma_t^2=1+a\sin(2\pi r_t),
\]
where \(a\in\{0.3,0.5,0.7\}\) controls the strength of within-segment heteroskedasticity. This design introduces smooth variance changes within each segment. 

For each simulated dataset, we apply the same multiscale procedure as in {Section~\ref{sec: pratical guidence}}, combining the results from \(G=60,80,100\). The change-point estimation results are reported in Table~	\ref{tab:within-segment-heteroskedasticity}, and the empirical coverage probabilities are reported in Table~\ref{tab:heteroskedasticity-coverage}.

\textbf{First}, the change-point estimation results show that within-segment heteroskedasticity has  some effect on the initial detection step. Across different values of \(a\), the initial Hausdorff errors increase slightly as the heteroskedasticity becomes stronger, indicating that stronger variance variation can make the initial screening step  more difficult. \textbf{Second}, the refined estimators remain stable for both kernels, and the effect of within-segment heteroskedasticity on the final localization accuracy is mild. \textbf{Third}, the results for estimating the number of change points are also relatively stable across different heteroskedasticity levels.  The proportions of correctly estimating and over-estimating  vary  mildly as a increases from 0.3 to 0.7. \textbf{Fourthly}, for confidence interval construction, the empirical coverage probabilities are close to the nominal level when \(a=0.3\). As \(a\) increases, the coverage probabilities decrease moderately, especially when \(a=0.7\). This is expected because stronger within-segment heteroskedasticity makes local covariance estimation more difficult.

Overall, the proposed multiscale procedure remains relative stable under mild and moderate within-segment heteroskedasticity especially for refined localization, while stronger heteroskedasticity mainly affects finite-sample coverage probability of the confidence interval.

\begin{table}[H]
	\centering
	\caption{Simulation results under within-segment heteroskedasticity with the multiscale scheme.}
	\label{tab:within-segment-heteroskedasticity}
	\begin{tabular}{llccc}
		\toprule
		Heteroskedasticity level & Method & Haus & $\widehat{M}_0=M_0$ & $\widehat{M}_0>M_0$ \\
		\midrule
		$a=0.3$
		& $h_1$-initial & 6.385 & 0.955 & 0.045 \\
		& $h_1$-refined & 1.250 & 0.955 & 0.045 \\
		& $h_2$-initial & 6.630 & 0.925 & 0.075 \\
		& $h_2$-refined & 0.720 & 0.925 & 0.075 \\
		\midrule
		$a=0.5$
		& $h_1$-initial & 6.945 & 0.935 & 0.065 \\
		& $h_1$-refined & 1.015 & 0.935 & 0.065 \\
		& $h_2$-initial & 7.405 & 0.930 & 0.070 \\
		& $h_2$-refined & 1.025 & 0.930 & 0.070 \\
		\midrule
		$a=0.7$
		& $h_1$-initial & 6.975 & 0.935 & 0.065 \\
		& $h_1$-refined & 0.950 & 0.935 & 0.065 \\
		& $h_2$-initial & 7.510 & 0.950 & 0.050 \\
		& $h_2$-refined & 0.810 & 0.950 & 0.050 \\
		\bottomrule
	\end{tabular}
\end{table}

\begin{table}[htbp]
	\centering
	\caption{Empirical coverage probabilities under within-segment heteroskedasticity with the multiscale scheme.}
	\label{tab:heteroskedasticity-coverage}
	\begin{tabular}{llcccccc}
		\toprule
		Bandwidth & Heteroskedasticity level
		& \multicolumn{3}{c}{$h_1(x,y)=y-x$}
		& \multicolumn{3}{c}{$h_2(x,y)=\operatorname{sign}(y-x)$} \\
		\cmidrule(lr){3-5} \cmidrule(lr){6-8}
		& & $\gamma_1$ & $\gamma_2$ & $\gamma_3$
		& $\gamma_1$ & $\gamma_2$ & $\gamma_3$ \\
		\midrule
		Multiscale & $a=0.3$ & 0.940 & 0.955 & 0.935 & 0.950 & 0.950 & 0.950 \\
		& $a=0.5$ & 0.895 & 0.945 & 0.910 & 0.930 & 0.925 & 0.935 \\
		& $a=0.7$ & 0.890 & 0.900 & 0.875 & 0.895 & 0.855 & 0.900 \\
		\bottomrule
	\end{tabular}
\end{table}

\newpage
\section{Useful Lemmas}\label{section: useful lemmas}
For $\beta>0$, we define the function $\psi_{\beta}: [0,\infty)\rightarrow [0,\infty)$ as $\psi_{\beta}(x):=\exp(x^\beta)-1$. Then, for any random variable $X$, we define 
\begin{equation*}
	\|X\|_{\psi_{\beta}}:=\inf\big\{C>0: \E\psi_{\beta}(|X|/C)|)\leq 1\big\}.
\end{equation*}

\begin{lemma}\label{lemma: concentration for maxsimum sub-exponential1}
	Let $X_1,\ldots,X_n\in \RR^1$ be independent random variables following sub-exponential distribution. Then, there is some $C_1>3$ and some large enough constant $C_2>0$ it holds that 
	\begin{equation*}
		\P\Big(|\sum_{t=\tau_1(k)}^{\tau_2(k)} X_{t}|\geq C_2\max\Big(\sqrt{(\tau_2(k)-\tau_1(k))\log(nd)},\log(nd)\Big),k=\underline{k},\ldots,\overline{k}\Big)\leq 2(nd)^{-C_1}.
	\end{equation*}
	
\end{lemma}

\begin{lemma}\label{lemma: concentration for maxsimum sub-exponential}
	Let $X_1,\ldots,X_n\in \RR^1$ be independent random variables following sub-exponential distribution with $\|X_t\|_{\psi_1}<\infty$. Let $1\leq \tau_1(k)\leq  \tau_2(k)\leq n$. Then, with probability at least $1-2(np)^{-C_1}$ with $C_1>3$ for sufficient large $n$ and $d$, we have
	\begin{equation*}
		\max_{\underline{k}\leq k\leq \overline{k}}	|\sum_{t=\tau_1(k)}^{\tau_2(k)} X_{t}|\leq C_2\max_{k}(\tau_2(k)-\tau_1(k))\max\Big(\sqrt{\dfrac{\log(nd)}{\min_k(\tau_2(k)-\tau_1(k))}},{\dfrac{\log(nd)}{\min_k(\tau_2(k)-\tau_1(k))}}\Big).
	\end{equation*}
\end{lemma}

\begin{lemma}\label{lemma: Hoeffding's residual}
	Let $h(x,y): R^2\rightarrow \R$ be a measurable two sample kernel. Suppose $X\sim F(x)$ and $Y\sim G(y)$ and let $\theta=\E h(X,Y)$.  Consider the following Hoeffding's decomposition:
	\begin{equation*}
		h(x,y)-\theta=h_1(x)+h_2(y)+g(x,y)
	\end{equation*}
	where $h_1(x):=\E h(x,Y)-\theta$, $h_2(y)=\E h(X,y)-\theta$, and $g(x,y)=h(x,y)-\theta-h_1(x)-h_2(y)$. Let $\sigma^2=\text{Var}[h(X,Y)]$ and 
	let 
	\begin{equation*}
		\begin{array}{ll}
			w_1(k)=\overline{\tau}_1(k)-\underline{\tau}_1(k),~~ \underline{w}_1=\min_{\underline{k}\leq k\leq \overline{k}}w_1(k),~~\overline{w}_1=\max_{\underline{k}\leq k\leq \overline{k}}w_1(k)\\
			w_2(k)=\overline{\tau}_2(k)-\underline{\tau}_2(k),~~ \underline{w}_2=\min_{\underline{k}\leq k\leq \overline{k}}w_2(k),~~\overline{w}_2=\max_{\underline{k}\leq k\leq \overline{k}}w_2(k).\\
		\end{array}
	\end{equation*}
	There exist some big enough constant $C_1>0$ such that with probability at least $1-(nd)^{-C_2}$ with $C_2>2$, the following holds:
	\begin{equation*}
		\max\limits_{\underline{k}\leq k \leq \overline{k}}\left| \sum\limits_{t_1=\underline{\tau}_1(k)}^{\overline{\tau}_1(k)} \sum\limits_{t_2=\underline{\tau}_2(k)}^{\overline{\tau}_2(k)} g(X_{t_1}, Y_{t_2}) \right|\leq C_1\log(nd)  [\overline{w}_1]^{1/2} [\overline{w}_2]^{1/2} \Big(\sqrt{\sigma^2}+\sqrt{\dfrac{\log^2(nd)}{\underline{w}_1}}+\sqrt{\dfrac{\log^2(nd)}{\underline{w}_2}}\Big).
	\end{equation*}
	Moreover, there exist some big enough constant $C_1>0$ such that with probability at least $1-(nd)^{-C_2}$ with $C_2>3$, the following holds uniformly over $k=\underline{k},\ldots,\overline{k}$:
	\begin{equation*}
		\left| \sum\limits_{t_1=\underline{\tau}_1(k)}^{\overline{\tau}_1(k)} \sum\limits_{t_2=\underline{\tau}_2(k)}^{\overline{\tau}_2(k)} g(X_{t_1}, Y_{t_2}) \right|\leq C_1\log(nd)  [w_1(k)]^{1/2}[w_2(k)]^{1/2} \Big(\sqrt{\sigma^2}+\sqrt{\dfrac{\log^2(nd)}{w_1(k)}}+\sqrt{\dfrac{\log^2(nd)}{w_2(k)}}\Big).
	\end{equation*}
\end{lemma}

\begin{lemma}(\cite{Zhou2017An})\label{lemma:maximum inequality}
	Let $\bW=(W_1,\ldots,W_p)^\top$ be a random vector with a marginal distribution  $N(0,\sigma_i^2)$ ($1\leq i\leq p$).
	Suppose $\exists A_0>0$ such that $\max_i\sigma_i^2\leq A^2_0$. Then, for any $t>0$, we have
	\begin{equation*}
		\E \big(\max_{1\leq i\leq p}|W_i|\big)\leq \dfrac{\log(2p)}{t}+\dfrac{tA_0^2}{2}.
	\end{equation*}
\end{lemma}
\begin{lemma}[Nazarovs inequality in \cite{nazarov2003maximal}] \label{lemma:anti consentration inequality}
	Let $\bW=(W_1,W_2,\cdots,W_d)^\top \in \mathbb{R}^{d}$ be centered Gaussian random vector with  $\inf_{1\leq k\leq d}\E(W_k)^2\geq b>0$. Then for
	any $\bx\in \mathbb{R}^d$ and $a>0$, we have
	\begin{equation*}
		\P(\bW\leq \bx+a)-\P(\bW\leq \bx)\leq Ca\sqrt{\log d},
	\end{equation*}
	where $C$ is a constant only depending on $b$.
\end{lemma}
\newpage
\section{Proofs of main results}\label{section: proof of main results}
\subsection{Proof of Theorem \ref{theorem: gaussian approximation}}
In this section, we prove the Gaussian approximation results. Firstly, under $\Hb_0$, due to the antisymmetric assumption of the kernel  and by the  Hoeffding's decomposition, for any $1\leq t_1\neq t_2\leq n$ and $j=1,\ldots,d$, we have 
\begin{equation*}
	h(X_{t_1,j},X_{t_2,j})=0+h_{1,j}(X_{t_1,j})-h_1(X_{t_2,j})+g_j(X_{t_1,j},X_{t_2,j}),
\end{equation*}
where $h_{1,j}(X_{t_1,j}):=\E h(X_{t_1,j},X_{t_2,j}|X_{t_1,j})$, $h_{2,j}(X_{t_2,j}):=\E h(X_{t_1,j},X_{t_2,j}|X_{t_2,j})=-h_1(X_{t_1,j})$ and  $g_j(X_{t_1,j},X_{t_2,j})=h(X_{t_1,j},X_{t_2,j})-h_{1,j}(X_{t_1,j})-h_{2,j}(X_{t_2,j})$ is the residual term of the Hoeffding's decomposition. Define the $\RR^d$ dimensional kernel and the corresponding leading and residual term as:
\begin{equation*}
	\begin{array}{ll}
		\bh(\bX_{t_1},\bX_{t_2}):=\big(h(X_{t_1,1},X_{t_2,1}),\ldots,h(X_{t_1,d},X_{t_2,d})\big)^\top,\\
		\bh_1(\bX_{t}):=\big(h_{1,1}(X_{t,1}),\ldots,h_{1,d}(X_{t,d})\big)^\top,\\
		\bg(\bX_{t_1},\bX_{t_2}):=\big(g_1(X_{t_1,1},X_{t_2,1}),\ldots,g_{d}(X_{t_1,d},X_{t_2,d})\big)^\top.\\
	\end{array}
\end{equation*}
Let $\bGamma=(\gamma_{i,j})\in \RR^{d\times d}=\text{Cov}(\bh_1(\bX_1))=\E \bh(\bX_1,\bX_2)\bh(\bX_1,\bX_3)^\top$. Note that by definition, we have $\gamma_{i,j}=\E [h_{1,i}(X_{t,i})h_{1,j}(X_{t,j})]$. Moreover, let $\bSigma=(\sigma_{i,j})\in \RR^{d\times d}=\E \bh(\bX_1,\bX_2)\bh(\bX_1,\bX_2)^\top=\text{Cov}(\bh(\bX_1,\bX_2))$. Recall the testing statistics $\bT(k)$ defined as:
\begin{equation*}
	\bT(k)=\dfrac{1}{G^{3/2}}\sum_{t_1=k-G+1}^k\sum_{t_2=k+1}^{k+G}\bh(\bX_{t_1},\bX_{t_2}).
\end{equation*}
With the above notations, we are ready to prove the main results, which have four steps. 

\textbf{Step~1} (Decomposition of $\bT(k)$). Based on the Hoeffding's decomposition, for each search location $G\leq k\leq n-G$, we have
\begin{equation*}
	\bT(k)=\underbrace{\dfrac{1}{\sqrt{G}}\sum_{t_1=k-G+1}^{k}\bh_{1}(\bX_{t_1})-\dfrac{1}{\sqrt{G}}\sum_{t_2=k+1}^{k+G}\bh_1(\bX_{t_2})}_{\bT^{(1)}(k)}+\underbrace{\dfrac{1}{G^{3/2}}\sum_{t_1=k-G+1}^k\sum_{t_2=k+1}^{k+G}\bg(\bX_{t_1},\bX_{t_2})}_{\bT^{(2)}(k)}.
\end{equation*}
The following Lemma \ref{lemma: bootstraped negligible} shows that the residual term $\bT^{(2)}(k)$ can be uniformly negligible over $k$ and $1\leq j\leq d$. The proof of Lemma \ref{lemma: bootstraped negligible} is given in Section \ref{section: proof of bootstraped negligible}.
\begin{lemma}\label{lemma: bootstraped negligible}
	Assume {Assumptions A.1-A.4} hold. Under $\Hb_0$, we have
	\begin{equation}\label{equation: bootstraped negligible}
		\P\Big(\max_{G\leq k\leq n-G} \big\|\bT(k)-\bT^{(1)}(k)\big\|_{\infty}\geq \epsilon\Big)=o(1),
	\end{equation}
	where $\epsilon=C_1 \dfrac{\log(nd)}{\sqrt{G}} $, and $C_1$ is a universal constant not depending on $n$ or $d$.
\end{lemma}	
\textbf{Step~2} (Gaussian approximation I). In Step~1, we have defined the  leading term $\bT^{(1)}(k)$. For each fixed $G\leq k\leq n-G$ and $t=1,\ldots,n$, let
\begin{equation*}
	a_{t}(k)=\dfrac{\sqrt{n}}{\sqrt{G}}\big(\mathbf{1}\{k-G+1\leq t\leq k\}-\mathbf{1}\{k+1\leq t\leq k+G\}\big).
\end{equation*}
Then, we can rewrite $\bT^{(1)}(k)$ as the following form:
\begin{equation}\label{equ: leading term}
	\bT^{(1)}(k)=\dfrac{1}{\sqrt{n}}\sum_{t=1}^n a_t(k)\bh_1(\bX_{t}).
\end{equation}
In this step, we aim to use the Gaussian random vectors based process to approximate $\max_{k}\|\bT^{(1)}(k)\|_{\infty}$. Specifically,  let $\bG_{1},\ldots,\bG_n$ be $i.i.d$ Gaussian random vectors with $\bG_{t}=(G_{t,1},\ldots,G_{t,d})\in \RR^d$ and $\bG_{t}\sim N(0,\bGamma)$.  We define the Gaussian random vectors based testing statistics $\bT^{\bG}(k)=(T^{\bG}_1(k),\ldots,T^{\bG}_d(k))^\top$ as:
\begin{equation}\label{equ: TG}
	\begin{array}{ll}	\bT^{\bG}(k)&=\dfrac{1}{\sqrt{G}}\sum\limits_{t_1=k-G+1}^{k}\bG_{t_1}-\dfrac{1}{\sqrt{G}}\sum\limits_{t_2=k+1}^{k+G}\bG_{t_2}=\dfrac{1}{\sqrt{n}}\sum\limits_{t=1}^n a_t(k)\bG_{t}.
	\end{array}
\end{equation}
The following lemma shows that we can approximate $\bT^{(1)}(k)$ using $\bT^{\bG}(k)$. The proof of Lemma \ref{lemma: gaussian approxiation 1} is given in Section \ref{section: proof of gaussian approximation 1}.
\begin{lemma}\label{lemma: gaussian approxiation 1}
	Assume {Assumptions A.1-A.4} hold. Under $\Hb_0$, there are $C_1>0$ such that we have
	\begin{equation*}
		\sup_{z\in(0,\infty)} \big|\P(\max_{G\leq k\leq n-G}\|\bT^{(1)}(k)\|_{\infty}\leq z\big)-\P(\max_{G\leq k\leq n-G}\|\bT^{G}(k)\|_{\infty}\leq z\big)\big|\leq C_1\Big(\dfrac{\log^7((n-2G+1)d)}{G}\Big)^{1/6}.
	\end{equation*}
\end{lemma}	
\textbf{Step~3} (Gaussian approximation II) In Step 2, we introduce the $d$-dimensional Gaussian random vector based moving window based process $\bT^{G}(k)$ in (\ref{equ: TG}) to approximate the leading process $\bT^{(1)}(k)$ in (\ref{equ: leading term}) with $G\leq k\leq n-G$.
In this step, we aim to show that we can use $\bT^b(k)$ in
(\ref{equ: bootstrap based testing statistics for each coordinate}) to approximate $\bT^{G}(k)$. To this end, 
let  $\bT^{\bG}=(\bT^{\bG}(G),\ldots,\bT^{\bG}(n-G))$ be the  $d\times (n-2G+1)$-dimensional random matrix, which is defined as
\begin{equation}\label{statistics: cusum matrix of gaussian random vectors}
	\bT^{\bG}=\Big(T^{\bG}_j(k)\Big)_{1\leq j\leq d,G\leq k\leq n-G}=\left(
	\begin{array}{ccc}
		T_{1}^{\bG}\big(G\big)&,\cdots,&T_{1}^{\bG}\big(n-G\big)\\
		T_{2}^{\bG}\big(G\big )&,\cdots,&T_{2}^{\bG}(n-G\big)\\
		\vdots&\cdots&\vdots\\
		T_{d}^{\bG}\big(G\big)&,\cdots,&T_{d}^{\bG}\big(n-G\big)\\
	\end{array}
	\right).
\end{equation}
Denote $\text{Vec}(\bT^{\bG})$ as the vectorized form of $\bT^{\bG}$, which is defined as
\begin{equation}\label{statistics: vectorized process for Gaussian CUSUM matrix}
	\text{Vec}(\bT^{\bG})=\Big(T_{1}^{\bG}\big(G\big),\ldots,T_{1}^{\bG}\big(n-G\big),\ldots,T_{d}^{\bG}\big(n-G\big),\cdots,T_{d}^{\bG}\big(n-G\big)\Big)^\top.
\end{equation}
Let $T_{i}^{\bG}(k_1)$ and $T_{j}^{\bG}(k_2)$ be any two components in $\text{Vec}(\bT^{\bG})$ with $1\leq i,j\leq d$ and $G\leq k_1,k_2\leq n-G$. By the definition of $T^{\bG}_j(k)$ in (\ref{equ: TG}), some calculations show that  for $G\leq k_1,k_2\leq n-G$ and $1\leq i,j\leq d$, we have
\begin{equation*}
	\text{Cov}(T_{i}^{\bG}(k_1),T_{j}^{\bG}(k_2)) = 
	\left\{
	\begin{array}{ll}
		0,& \text{if}~|k_2-k_1|>2G-1 \\
		-\dfrac{2G-|k_2-k_1|}{G}\gamma_{i,j},& \text{if}~G-1<|k_2-k_1|\leq2G-1  \\
		\dfrac{2G-3|k_2-k_1|}{G}\gamma_{i,j},& \text{if}~0<|k_2-k_1|\leq G-1  \\
		2\gamma_{i,j},& \text{if}~|k_2-k_1|=0 \\
	\end{array}
	\right.
\end{equation*}
where $\gamma_{i,j}$ is the (i,j)-th element in $\bGamma:=\text{Cov}(\bh_1(\bX))$.
Hence,  by definition, $\text{Vec}(\bT^{\bG})$ in (\ref{equ: TG}) follows a $d(n-2G+1)$-dimensional multivariate Gaussian distribution 
\begin{equation*}
	\text{Vec}(\bT^{\bG})\sim N(0,\bTheta^{\bG})
\end{equation*}
with $\bTheta^{\bG}=(\Theta_{k_1,k_2,i,j})=\text{Cov}(T_{i}^{\bG}(k_1),T_{j}^{\bG}(k_2)))\in \RR^{d(n-2G+1)\times d(n-2G+1)}$. In this step, we aim to approximate $\bT^{\bG}(k)$ by $\bT^b(k)$ with $G\leq k\leq n-G$, where $\bT^b(k)=(T_1^b(k),\ldots,T_d^b(k))^\top$ is the bootstrap based testing statistic, which  is defined as 
\begin{equation}\label{equ: bootstrap based testing statistics }
	\bT^b(k)=\dfrac{1}{G^{3/2}}\sum_{t_1=k-G+1}^{k}\sum_{t_2=k+1}^{k+G}(e_{t_1}+e_{t_2})\bh(\bX_{t_1},\bX_{t_2}).
\end{equation}
Similar to $\bT^{\bG}(k)$, let  $\bT^{b}=(\bT^{b}(G),\ldots,\bT^{b}(n-G))$ be the  $d\times (n-2G+1)$-dimensional random matrix, which is defined as
\begin{equation}\label{statistics: cusum matrix of bootstraprandom vectors}
	\bT^{b}=\Big(T^{b}_j(k)\Big)_{1\leq j\leq d,G\leq k\leq n-G}=\left(
	\begin{array}{ccc}
		T_{1}^{b}\big(G\big)&,\cdots,&T_{1}^{b}\big(n-G\big)\\
		T_{2}^{b}\big(G\big )&,\cdots,&T_{2}^{b}(n-G\big)\\
		\vdots&\cdots&\vdots\\
		T_{d}^{b}\big(G\big)&,\cdots,&T_{d}^{b}\big(n-G\big)\\
	\end{array}
	\right).
\end{equation}
Denote $\text{Vec}(\bT^{b})$ as the vectorized form of $\bT^{b}$, which is defined as
\begin{equation}\label{statistics: vectorized process for bootstrap  CUSUM matrix}
	\text{Vec}(\bT^{b})=\Big(T_{1}^{b}\big(G\big),\ldots,T_{1}^{b}\big(n-G\big),\ldots,T_{d}^{b}\big(n-G\big),\cdots,T_{d}^{b}\big(n-G\big)\Big)^\top.
\end{equation}
Let $T_{i}^{\bG}(k_1)$ and $T_{j}^{\bG}(k_2)$ be any two components in $\text{Vec}(\bT^{\bG})$ with $1\leq i,j\leq d$ and $G\leq k_1,k_2\leq n-G$. By the definition of $T^{b}_j(k)$ in (\ref{equ: bootstrap based testing statistics }), conditional on $\cX=\{\bX_1,\ldots,\bX_n\}$, it follows the multivariate Gaussian distribution, which is defined as:
\begin{equation*}
	\text{Vec}(\bT^{b})|\cX\sim N(0,\bTheta^{b})
\end{equation*}
with $\bTheta^{b}=(\Theta^b_{k_1,k_2,i,j})=\text{Cov}(T_{i}^{b}(k_1),T_{j}^{b}(k_2)))\in \RR^{d(n-2G+1)\times d(n-2G+1)}$. To achieve the above  goal, we need the following lemma. 
\begin{lemma}\label{lemma: large deviation}
	Assume {Assumptions A.1-A.4 hold}. Under $\Hb_0$, with probability at least $1-(nd)^{-C_1}$ for some $C_1>0$, it holds that
	\begin{equation*}
		\|\bTheta^{\bG}-\bTheta^{b}\|_{\infty}=\max_{G\leq k_1,k_2\leq n-G,\atop 1\leq i,j\leq d}|\text{Cov}(T_{i}^{\bG}(k_1),T_{j}^{\bG}(k_2))-\text{Cov}(T_{i}^{b}(k_1),T_{j}^{b}(k_2)))|\leq \epsilon_n
	\end{equation*}
	where $\epsilon_n=C_2D^2\max \Big(\sqrt{\dfrac{\log(nd)}{G}},\dfrac{\log^2(nd)\log^2(Gd)}{G}\Big)$ for some big enough constant $C_2>0$ and $D$ is defined in {Assumption A.2.}
\end{lemma}

The proof of Lemma \ref{lemma: large deviation} is given in Section \ref{section: proof of large deviation}. Under Lemma \ref{lemma: large deviation}, the following results shows that the two Gaussian processes $\bT^{\bG}(k)$ and $\bT^b(k)|\mathcal{X}$ with $G\leq k\leq n-G$ can be uniformly close to each other with the $\ell_\infty$-norm. In fact, using Theorem 4.1 and Remark 4.1 in \cite{chernozhukov2017central} and Lemma \ref{lemma: large deviation}, with probability at least $1-(nd)^{-C_1}$, we have
\begin{equation}\label{equation: gaussian approximation 2}
	\sup_{z\in(0,\infty)} \big|\P(\max_{G\leq k \leq n-G}\|\bT^{\bG}(k)\|_{\infty}> z\big)-\P(\max_{G\leq k\leq n-G}\|\bT^b(k)\|_{\infty}> z|\mathcal{X}\big)\big|=O(\epsilon_n^{1/3}\log(d(n-2G+1))^{2/3}).
\end{equation}
where $\epsilon_n$ is defined in Lemma \ref{lemma: large deviation}.
Hence, by  (\ref{equation: gaussian approximation 2}), in Step~3, we approximate the process $\bT^{\bG}(k)$ by $\bT^b(k)|\mathcal{X}$.\\
$\mathbf{Step~4}.$ In this step, we aim to combine the previous three steps to prove (\ref{equ: gaussian approximation result}) in Theorem \ref{theorem: gaussian approximation}.
To this end, we need both the upper bound and the lower bound for $v_0$, where
\begin{equation*}
	v_0=\P\big(\max_{G\leq k \leq n-G}\|\bT(k)\|_{\infty}\geq z\big)-\P\big(\max_{G\leq k \leq n-G}\|\bT^b(k)\|_{\infty}\geq z|\mathcal{X}\big).
\end{equation*}
We first obtain the upper bound. By plugging  $\bT^{(1)}(k)$ in $\max_{G\leq k \leq n-G}\|\bT(k)\|_{\infty}$, and using the triangle inequality for the $\ell_\infty$-norm, we have
\begin{equation}\label{step4: inequality1}
	\P\big(\max\limits_{G\leq k\leq n-G} \big\|\bT(k)\big\|_{\infty}\geq z \big)\leq \P\big(\max\limits_{G\leq k\leq n-G} \big\|\bT^{(1)}(k)\big\|_{\infty}\geq z-\epsilon \big)+v_1,
\end{equation}
where $v_1=\P(\max\limits_{G\leq k\leq n-G} \|\bT(k)-\bT^{(1)}(k)\|_{\infty}\geq \epsilon)$.
By Lemma \ref{lemma: bootstraped negligible}, we have $v_1=o(1)$. For $\P(\max_{G\leq k\leq n-G} \|\bT^{(1)}(k)\|_{\infty}\geq z-\epsilon)$, by the triangle inequality, we have
\begin{equation}\label{step4: inequality2}
	\P\big(\max\limits_{G\leq k\leq n-G} \big\|\bT^{(1)}(k)\big\|_{\infty}\geq z-\epsilon \big)\leq \P\big(\max\limits_{G\leq k\leq n-G} \big\|\bT^{\bG}(k)\big\|_{\infty}\geq z-\epsilon \big)+v_2,
\end{equation}
where
\begin{equation*}
	v_2=\max_{x>0}\big|\P(\max_{G\leq k\leq n-G}\|\bT^{\bG}(k)\|_{\infty}> x\big)-\P(\max_{G\leq k\leq n-G}\|\bT^{(1)}(k)\|_{\infty}> x\big)\big|.
\end{equation*}
By Lemma \ref{lemma: gaussian approxiation 1}, we have $v_2\leq Cn^{-\zeta_0}$. Therefore, by (\ref{step4: inequality1}) and (\ref{step4: inequality2}), we have proved that
\begin{equation}\label{step4: inequality3}
	\P\big(\max\limits_{G\leq k\leq n-G} \big\|\bT(k)\big\|_{\infty}\geq z \big)\leq\underbrace{\P\big(\max\limits_{G\leq k\leq n-G} \big\|\bT^{\bG}(k)\big\|_{\infty}\geq z-\epsilon \big)}_{v_3}+o(1). 
\end{equation}
We next consider $v_3$. We decompose $v_3$ as $v_3=v_4+v_5$, where $v_4$ and $v_5 $ are defined as
\begin{equation*}
	v_4=\P\big(z-\epsilon\leq \max\limits_{G\leq k\leq n-G} \big\|\bT^{\bG}(k)\big\|_{\infty}\leq z\big),~~v_5=\P(\max\limits_{G\leq k\leq n-G} \|\bT^{\bG}(k)\|_{\infty}\geq z).
\end{equation*}
By Lemmas \ref{lemma:anti consentration inequality}, we can show that $v_4=o(1)$. For $v_5$, we have
\begin{equation}\label{step4: inequality4}
	\P\big(\max\limits_{G\leq k\leq n-G} \big\|\bT^{\bG}(k)\big\|_{\infty}\geq z \big)\leq \P\big(\max\limits_{G\leq k\leq n-G} \big\|\bT^b(k)\big\|_{\infty}\geq z|\mathcal{X}) \big)+v_6,
\end{equation}
where $v_6=\sup_{z\in(0,\infty)} |\P(\max\limits_{G\leq k\leq n-G}\|\bT^{\bG}(k)\|_{\infty}> z)-\P(\max\limits_{G\leq k\leq n-G}\|\bT^b(k)\|_{\infty}> z|\mathcal{X})|$. By (\ref{equation: gaussian approximation 2}), we have $v_6=o(1)$  with probability at least $1-C(nd)^{-1}$. Therefore, by (\ref{step4: inequality1}) -- (\ref{step4: inequality4}), we have proved
\begin{equation*}
	\P\big(\max_{G\leq k \leq n-G}\|\bT(k)\|_{\infty}\geq z\big)-\P\big(\max_{G\leq k \leq n-G}\|\bT^b(k)\|_{\infty}\geq z|\mathcal{X}\big)=o_p(1),
\end{equation*}
uniformly for $z>0$. Similarly, we can obtain the lower bound. Finally, we have proved that
\begin{equation*}
	\sup_{z\in (0,\infty)}\big|\P\big(\max_{G\leq k \leq n-G}\|\bT(k)\|_{\infty}\geq z\big)-\P\big(\max_{G\leq k \leq n-G}\|\bT^b(k)\|_{\infty}\geq z|\mathcal{X}\big)\big|=o_p(1),
\end{equation*}
which finishes the proof of \ref{equ: gaussian approximation result} in Theorem \ref{theorem: gaussian approximation}.

\subsection{Proof of Theorem \ref{theorem: power results}}
In this section, we give the power results. 
The proof of Theorem \ref{theorem: power results} proceeds in two steps. In Step 1, we obtain the upper bound of $c_{W^b}(1-\alpha)$, where $c_{W^b}(1-\alpha)$ is  the $1-\alpha$ quantile of $W^{b}$, which is defined as $	c_{W^b}(1-\alpha):=\inf\big\{t:\P(W^b\leq t|\cX)\geq 1-\alpha  \big\}$. In Step 2, using the upper bound, we get the lower bound of $\P\big(W\geq c_{W^b}(1-\alpha)\big)$ to prove
\begin{equation}
	\P\big(W\geq c_{W^b}(1-\alpha)\big)\rightarrow1, ~\text{as}~n,d\rightarrow\infty.
\end{equation}
Note that $\{\Phi_{\alpha}=1\}\Leftrightarrow\{W\geq \hat{c}_{W^b}(1-\alpha) \}$, where 
\begin{equation}
	\hat{c}_{W^b}(1-\alpha):=\inf\Big\{t:
	\dfrac{1}{B+1}\sum_{b=1}^{B}\mathbf{1}\{W^{b}\leq t|\cX\}\geq 1-\alpha\Big\}.
\end{equation}
Finally, using the fact that $\hat{c}_{W^b}(1-\alpha)$ is the estimation for $c_{W^b}(1-\alpha)$ based on the bootstrap samples, we complete the proof. Now, we consider the two steps in detail.  \\
\textbf{Step~1}: In this step, we aim to obtain the upper bound for $c_{W^b}(1-\alpha)$. Recall $T_j^b(k)$ defined in $(\ref{equ: bootstrap based testing statistics for each coordinate})$. For each fixed $j$ and $k=G,\ldots,n-G$, conditional on $\cX$, we have $$T_j(k)\sim N(0,\sigma_{j}(k)^2).$$
Let $q'=d(n-2G+1)$. Next, we will prove that there exists some constant $A_0$ such that with probability tending to one, $$\max_{1\leq j\leq d}\max_{G\leq k\leq n-G}\sigma_{j}(k)^2\leq A_0^2,$$ which is shown in the following lemma.
\begin{lemma}\label{lemma: upper bound of variance of bootstrap samples}
	{Suppose Assumptions B.1-B.3 hold.} With probability at least $1-3(nd)^{-C_1}$, it holds that 
	\begin{equation*}
		\max_{j}\max_{k}\sigma_j(k)^2\leq C_2(D+\max_{1\leq m\leq M_0}\max_{1\leq j\leq d}|\theta^m_j|)^2:=A_0^2.
	\end{equation*}
\end{lemma}
The proof of Lemma \ref{lemma: upper bound of variance of bootstrap samples} is given in Section \ref{sec: proof of upper bound of bootstrap samples}.  Combining Lemmas \ref{lemma: upper bound of variance of bootstrap samples} and  \ref{lemma:maximum inequality}, for any $t>0$, we have 
\begin{equation}\label{inequality: maximum inequality}
	\E \Big(\max_{G\leq k\leq n-G}\max_{1\leq j\leq d} |T_{j}^b(k)|\Big)\leq \dfrac{\log(2q')}{t}+\dfrac{tA_0^2}{2}.
\end{equation}
Furthermore, taking $t=A_0^{-1}\sqrt{2\log(q')}$ in (\ref{inequality: maximum inequality}), we have
\begin{equation}\label{inequality: maximum inequality2}
	\E \Big(\max_{G\leq k\leq n-G}\max_{1\leq j\leq d}  |T_{j}^b(k)|\Big)\leq A_0\sqrt{2\log(q')}\big(1+\dfrac{1}{2\log q'}\big).
\end{equation}
By Theorem 5.8 in \cite{boucheron2013concentration}, we have
\begin{equation}\label{inequality: upper bound expectation}
	\P\Big(\max_{G\leq k\leq n-G}\max_{1\leq j\leq d}  |T_{j}^b(k)|\geq \E\Big[ \max_{G\leq k\leq n-G}\max_{1\leq j\leq d}  |T_{j}^b(k)|\Big]+z\Big|\mathcal{X}\Big)\leq \exp\Big(-\dfrac{z^2}{2A_0^2}\Big).
\end{equation}
Based on (\ref{inequality: maximum inequality2}), and taking $z=A_0\sqrt{2\log(\alpha^{-1})}$ in (\ref{inequality: upper bound expectation}), we have 
\begin{equation}\label{inequality: upper bound for 1-alpha quantile1}
	c_{W^b}(1-\alpha)\leq A_0\sqrt{2\log(q')}\big(1+\dfrac{1}{2\log q'}\big)+A_0\sqrt{2\log(\alpha^{-1})}.
\end{equation}
\textbf{Step~2}:
In this step, we aim to prove that $\P\big(W\geq c_{W^b}(1-\alpha)\big)\rightarrow1$ as $n,d\rightarrow\infty$. 
Let 
\begin{equation}\label{equation: upper bound for 1-alpha quantile}
	\begin{array}{ll}
		c^{\rm u}_{W^b}(1-\alpha)&=A_0\sqrt{2\log(q')}\big(1+\dfrac{1}{2\log q'}\big)+A_0\sqrt{2\log(\alpha^{-1})}
	\end{array}
\end{equation}
Considering the upper bound obtained in (\ref{equation: upper bound for 1-alpha quantile}), it is sufficient to prove $H_1\rightarrow1$, where
\begin{equation}\label{inequality: H1}
	H_1=\P\big(W\geq c^{\rm u}_{W^b}(1-\alpha)\big).
\end{equation}
By the definition of $W$, we have 
\begin{equation}\label{equ: lower bound of maximum statistics 1}
	W=\max_{G\leq k\leq n-G }\max_{1\leq j\leq d}|T_j(k)|=\max_{G\leq k\leq n-G}\|\bT(k)\|_{\infty}\geq \max_{1\leq m\leq M_0}\|\bT(\gamma_m)\|_{\infty}.
\end{equation}
Hence, it is sufficient to prove $H_2\rightarrow 1$, where 
\begin{equation}\label{inequality: H2}
	H_2=\P\big( \max_{1\leq m\leq M_0}\|\bT(\gamma_m)\|_{\infty}\geq c^{\rm u}_{W^b}(1-\alpha)\big).
\end{equation}
For each fixed $m=1\ldots,M_0$ and $j=1,\ldots,d$, we have
\begin{equation}\label{equ: Hoeffding at gamma}
	\begin{array}{ll}
		T_{j}(\gamma_m)=\sqrt{G}\theta_{j}^m+\dfrac{1}{\sqrt{G}}\sum\limits_{t_1=\gamma_m-G}^{\gamma_{m}}h_{1,j}(X_{t_1,j})+\dfrac{G}{G^{3/2}}\sum\limits_{t_2=\gamma_m+1}^{\gamma_m+G}h_{2,j}(X_{t_2,j})\\
		+\dfrac{1}{G^{3/2}}\sum\limits_{t_1=\gamma_m-G}^{\gamma_m}\sum\limits_{t_2=\gamma_m+1}^{\gamma_m+G}g_j(X_{t_1,j},X_{t_2,j}),
	\end{array}
\end{equation}
where $h_1(X_{t_1,j}):=\E h(X_{t_1,j},X_{\gamma_{m}+1,j}|X_{t_1,j})-\theta_j^m$, $h_2(X_{t_2,j}):=\E h(X_{\gamma_m,j},X_{t_2,j}|X_{t_2,j})-\theta_j^m$ and $g_j(X_{t_1,j},X_{t_2,j}):=h(X_{t_1,j},X_{t_2,j})-\theta_{j}^m-h_1(X_{t_1,j})-h_2(X_{t_2,j})$.
Hence, by Lemmas \ref{lemma: concentration for maxsimum sub-exponential1}, \ref{lemma: concentration for maxsimum sub-exponential} and \ref{lemma: Hoeffding's residual}, with probability at least $1-(nd)^{-C_1}$, it holds uniformly over  $m=1,\ldots,M_0$ 
\begin{equation}\label{equ: lower bound Hoeffding at gamma}
	\begin{array}{ll}
		\max\limits_{1\leq j\leq d}	|T_{j}(\gamma_m)|\geq \max\limits_{1\leq j\leq d}|\sqrt{G}\theta_{j}^m|-\max\limits_{1\leq j\leq d}|\dfrac{1}{\sqrt{G}}\sum\limits_{t_1=\gamma_m-G}^{\gamma_{m}}h_{1,j}(X_{t_1,j})|-\max\limits_{1\leq j\leq d}|\dfrac{G}{G^{3/2}}\sum\limits_{t_2=\gamma_m+1}^{\gamma_m+G}h_{2,j}(X_{t_2,j})|\\
		-\max\limits_{1\leq j\leq d}|\dfrac{1}{G^{3/2}}\sum\limits_{t_1=\gamma_m-G}^{\gamma_m}\sum\limits_{t_2=\gamma_m+1}^{\gamma_m+G}g_j(X_{t_1,j},X_{t_2,j})|\\
		\geq_{(1)} \sqrt{G}\|\btheta^m\|_{\infty}-C_2\sqrt{\log(nd)}-C_3\dfrac{G}{G^{3/2}}\log(nd)\max(\sigma^2+2\sqrt{\dfrac{\log^2(nd)}{G}})\\
		\geq_{(2)} \sqrt{G}\|\btheta^m\|_{\infty}-C_2\sqrt{\log(nd)}-C_3\log(nd)/\sqrt{G}\\
		\geq_{(3)} \sqrt{G}\|\btheta^m\|_{\infty}-C_4\sqrt{\log(nd)}\\
	\end{array}
\end{equation}
where $(1)$ comes from Lemmas \ref{lemma: concentration for maxsimum sub-exponential1}, \ref{lemma: concentration for maxsimum sub-exponential} and \ref{lemma: Hoeffding's residual}, (2) and (3) come from the assumption that $G\gg\log^2(nd)$. Hence, the above results in (\ref{equ: lower bound Hoeffding at gamma}) yield that with probability at least $1-(nd)^{-C_1}$
\begin{equation}\label{equ: lower bound of maximum statistics}
	\max_{1\leq m\leq M_0}\max\limits_{1\leq j\leq d}	|T_{j}(\gamma_m)|\geq \sqrt{G}	\max_{1\leq m\leq M_0}\|\btheta^m\|_{\infty}-C_3\sqrt{\log(nd)}.
\end{equation}
Considering (\ref{equation: upper bound for 1-alpha quantile}), (\ref{equ: lower bound of maximum statistics 1}),  (\ref{equ: lower bound of maximum statistics}) and the {assumption} in (\ref{inequality: theoretical signal strengh}), we have
\begin{equation*}
	\P\big( W\geq c^{\rm u}_{W^b}(1-\alpha)\big)\rightarrow 1,
\end{equation*}
which completes the proof.

\subsection{Proof of Theorem \ref{theory: initial estimators}}
{\bf{Proof of (1) of Theorem \ref{theory: initial estimators}}}.  Let $D_+$ be the bootstrap based critical value. We first prove that the estimated change point number $\hat{M}_0$ in (\ref{equation: intial number}) is consistent for the true number $M_0$. By the definition of $\hat{M}_0$, define the two events:
\begin{equation*}
	\begin{array}{ll}
		\cA_1=\{\max\limits_{k: \min_{1\leq m\leq M_0: |k-\gamma_m|\geq G}} \|\bT(k)\|_{\infty}\leq D_+\},\\
		\cA_2=\{\min\limits_{k:\min_{1\leq m\leq M_0: |k-\gamma_m|\leq (1-\eta)G}}\|\bT(k)\|_{\infty}> D_+\}.
	\end{array}
\end{equation*}
Note that $\cA_1\cap\cA_2\subset \{\hat{M}_0=M_0\}$ and $\cA_1\cap \cA_2\subset\{\max_{1\leq m\leq M_0}|\hat{\gamma}_m-\gamma_m|\leq G\}$. If we can prove that $\P(\cA_1\cap \cA_2)\rightarrow 1$, the proof of (1) is completed. Now, we consider $\P(\cA_1)$ and $\P(\cA_2)$, respectively.

{\bf{Proof of $\cA_1$}}.  For $\cA_1$, by the definition of $\bT(k)=(T_1(k),\ldots,T_{d}(k))^\top$, where 
\begin{equation*}
	T_j(k):=\dfrac{1}{G^{3/2}}\sum_{t_1=k-G}^k\sum_{t_2=k+1}^{k+G}h(X_{t_1,j},X_{t_2,j}), ~~j=1,\ldots,d.
\end{equation*}
Note that for any $k\in\{G,\ldots,n-G\}$ satisfying $ \min_{1\leq m\leq M_0}: |k-\gamma_m|\geq G$, it means that there is no change point in the interval $[k-G,k+G]$. In other words, $X_{t_1,j}$ and $X_{t_2,j}$ are identically distributed, say $F_j(x)$. In this case, by the Hoeffding's decomposition, we have
\begin{equation}\label{equation: no change point interval}
	h(X_{t_1,j},X_{t_2,j})=0+h_{1,j}(X_{t_1,j})+h_{2,j}(X_{t_2,j})+g_j(X_{t_1,j},X_{t_2,j}),
\end{equation} 
where $h_{1,j}(X_{t_1,j}):=\E h(X_{t_1,j},X_{t_2,j}|X_{t_1,j})$ and $h_{2,j}(X_{t_2,j}):=\E h(X_{t_1,j},X_{t_2,j}|X_{t_2,j})=-h_1(X_{t_1,j})$ and  $g_j(X_{t_1,j},X_{t_2,j})=h(X_{t_1,j},X_{t_2,j})-h_{1,j}(X_{t_1,j})-h_{2,j}(X_{t_2,j})$ is the residual term of the Hoeffding's decomposition. Hence, by (\ref{equation: no change point interval}), we have:
\begin{equation}\label{equation: Tjk}
	T_j(k)=\underbrace{G^{-1/2}\sum_{t_1=k-G}^{k}h_{1,j}(X_{t_1,j})}_{T_j^{(1)}(k)}+\underbrace{G^{-1/2}\sum_{t_2=k+1}^{k+G}h_{2,j}(X_{t_1,j})}_{T_j^{(2)}(k)}+{\underbrace{\dfrac{1}{G^{3/2}}\sum_{t_1=k-G}^k\sum_{t_2=k+1}^{k+G}g_j(X_{t_1,j},X_{t_2,j})}_{R_j(k)}}.
\end{equation}
Note that by definition of $k: \min_{1\leq m\leq M_0}: |k-\gamma_m|\geq G$, we can split those k into $m+1$ groups: 
\begin{equation}\label{equation: cT1-cTm}
	\cT_1=\{k: k\in[1,\gamma_{1}-G]\},~	\cT_2=\{ k: k\in[\gamma_{1}+G,\gamma_{2}-G]\},\ldots, \cT_{m+1}=\{k: k\in[\gamma_{m}+G,n]\}.
\end{equation}
Hence, it is sufficient to prove 
\begin{equation*}
	\P(\max_{1\leq j\leq d}\max_{1 \leq m\leq M_0}\max_{k\in \cT_m}|T_j(k)|\leq D_{+}\}\rightarrow 1.
\end{equation*}
For each fixed $j$ and $m$, we consider $\max_{k\in \cT_m}|T_j^{(1)}(k)|$, $\max_{k\in \cT_m}|T_j^{(2)}(k)|$, and $\max_{k\in \cT_m}|R_j^{(1)}(k)|$. 

{\bf{Control of $\max_{k\in \cT_m}|T_j^{(1)}(k)|$}}. By {Assumption B.2}, $h_{1,j}(X_{t_1,j})$ is sub-exponential distributed. Hence, by Lemma \ref{lemma: concentration for maxsimum sub-exponential}, for each  fixed $j$ and $m$, taking $\tau_1(k)=k-G$, $\tau_2(k)=k$, with probability at least $1-(nd)^{-C_2}$ for some $C_2>3$, we have
\begin{equation*}
	\max_{k\in \cT_m}G^{1/2}|T_{j}^{(1)}(k)|\leq C_1G\sqrt{\dfrac{\log(nd)}{G}}=C_1\sqrt{G\log(nd)}.
\end{equation*}
Taking the union bounds over $j$ and $m$, we have with probability at least $1-(nd)^{-C_3}$ for some $C_3>0$, it holds that 
\begin{equation*}
	\max_{1\leq j \leq d}\max_{1\leq m\leq M_0}	\max_{k\in \cT_m}|T_{j}^{(1)}(k)|\leq C_1G\sqrt{\dfrac{\log(nd)}{G}}=C_1\sqrt{\log(nd)}\leq D^+.
\end{equation*}
Similarly, we  can prove 
\begin{equation*}
	\max_{1\leq j \leq d}\max_{1\leq m\leq M_0}	\max_{k\in \cT_m}|T_{j}^{(1)}(k)=O_p(\log(nd))
\end{equation*}
For $R_j(k)$ with fixed $j\in\{1,\ldots,d\}$ and $m\in\{1,\ldots,M_0\}$, by Lemma \ref{lemma: Hoeffding's residual}, taking $\underline{\tau}_{1}(k)=k-G$, $\overline{\tau}_{1}(k)=k$, $\underline{\tau}_{2}(k)=k+1$, $\overline{\tau}_{2}(k)=k+G$, with probability at least $1-(nd)^{-C_2}$ for some $C_2>3$, we have 
\begin{equation*}
	\begin{array}{ll}
		G^{3/2}	\max_{k\in \cT_{m}}|R_j(k)|&\leq C_1 \log(nd)G\Big(\sqrt{\sigma^2}+\sqrt{\dfrac{\log^2(nd)}{G}}+\sqrt{\dfrac{\log^2(nd)}{G}}\Big)\\
		&\leq C_1G\log(nd)
	\end{array}
\end{equation*}
where we use the assumption that $G\gg \log^2(nd)$. Taking the  union bounds over $j$ and $m$, we then have 
\begin{equation*}
	\max_{1\leq j\leq d}\max_{1\leq m\leq M_0}\max_{k\in \cT_{m}}|R_j(k)|=O_p(\dfrac{\log(nd)}{\sqrt{G}})\ll D^+.
\end{equation*}
Hence, based on the above results,  with probability at least $1-(dn)^{-C_1}$, we have 
\begin{equation*}
	\max_{1\leq j\leq d}\max_{1 \leq m\leq M_0}\max_{k\in \cT_m}|T_j(k)|\leq C\sqrt{\log(nd)}\leq D_+.
\end{equation*}

{\bf{Proof of $\cA_2$}}. For $\cA_2$, Note that by definition of $\{k: \min_{1\leq m\leq M_0}: |k-\gamma_m|\leq (1-\eta)G\}$ we can split those k into $m$ groups: 
\begin{equation}\label{equ: cT1'-cTm'}
	\begin{array}{ll}
		\cT'_{1,1}=\{k: k\in[\gamma_{1}-(1-\eta)G,\gamma_{1}]\},~\cT'_{1,2}=\{k: k\in[\gamma_{1}+1,\gamma_{1}+(1-\eta)G],\}\\
		\cT'_{2,1}=\{k: k\in[\gamma_{2}-(1-\eta)G,\gamma_{2}]\},~\cT'_{2,2}=\{k: k\in[\gamma_{2}+1,\gamma_{2}+(1-\eta)G],\}\\
		\qquad	\qquad \qquad\qquad\qquad\qquad\qquad\vdots\\
		\cT'_{M_0,1}=\{k: k\in[\gamma_{M_0}-(1-\eta)G,\gamma_{M_0}]\},~\cT'_{M_0,2}=\{k: k\in[\gamma_{M_0}+1,\gamma_{M_0}+(1-\eta)G]\}
	\end{array}
\end{equation}
Hence, it is sufficient to prove 
\begin{equation*}
	\P(\min_{1 \leq m\leq M_0}\min_{1\leq s\leq 2}\min_{k\in \cT'_{m,s}}\max_{1\leq j\leq d}|T_j(k)|\geq D_{+})\rightarrow 1.
\end{equation*}
For any $m$ and $s$, there is only one change point in either 
$[k-G,k]$ or $[k+1,k+G]$.  Without loss of generality, we consider the case of $s=2$. In this case, for any  $k\in \cT'_{m,2}$, there is only one change point $\gamma_m$ in $[k-G,k]$. 
In this case, we can rewrite 
$T_j(k)$ as two parts:
\begin{equation}\label{equ: decomposition T_{j,k}}
	T_{j}(k)=\underbrace{\dfrac{1}{G^{3/2}}\sum_{t_1=k-G}^{\gamma_m}\sum_{t_2=k+1}^{k+G}h(X_{t_1,j},X_{t_2,j})}_{T_{j,1}(k)}+\underbrace{\dfrac{1}{G^{3/2}}\sum^{k}_{t_1=\gamma_m}\sum_{t_2=k+1}^{k+G}h(X_{t_1,j},X_{t_2,j})}_{T_{j,2}(k)},
\end{equation}
Let $\theta_{j}^{m}:=\E h(X_{\gamma_m,j},X_{\gamma_{m+1},j})$ be the signal jump at the change point $\gamma_{m}$ for the $j$-th coordinate. Let $\btheta^m =(\theta_1^m,\ldots,\theta_d^m)$. Based on the Hoeffding's decomposition, we can decompose $T_{j,1}(k)$ into three parts:
\begin{equation}\label{equ: Hoeffding T_{j,k}}
	\begin{array}{ll}
		T_{j,1}(k)=\dfrac{\gamma_m-k+G}{G^{1/2}}\theta_{j}^m+\underbrace{\dfrac{1}{\sqrt{G}}\sum\limits_{t_1=k-G}^{\gamma_{m}}h_{1,j}(X_{t_1,j})}_{T_{j,1,1}(k)}+\underbrace{\dfrac{\gamma_{m}-k+G}{G^{3/2}}\sum\limits_{t_2=k+1}^{k+G}h_{2,j}(X_{t_2,j})}_{T_{j,1,2}(k)}\\
		+\underbrace{\dfrac{1}{G^{3/2}}\sum\limits_{t_1=k-G}^{\gamma_m}\sum\limits_{t_2=k+1}^{k+G}g_j(X_{t_1,j},X_{t_2,j})}_{T_{j,1,3}(k)},
	\end{array}
\end{equation}
where $h_1(X_{t_1,j}):=\E h(X_{t_1,j},X_{\gamma_{m}+1,j}|X_{t_1,j})-\theta_j^m$, $h_2(X_{t_2,j}):=\E h(X_{\gamma_m,j},X_{t_2,j}|X_{t_2,j})-\theta_j^m$ and $g_j(X_{t_1,j},X_{t_2,j}):=h(X_{t_1,j},X_{t_2,j})-\theta_{j}^m-h_1(X_{t_1,j})-h_2(X_{t_2,j})$. For $T_{j,1,1}(k)$ and $T_{j,1,2}(k)$, using Lemma \ref{lemma: concentration for maxsimum sub-exponential} and the fact that for $k\in \cT'_{m,2}$, $\eta G\leq \gamma_m-k+G\leq G$,  taking the union bound over $m=1,\ldots,M_0$, $j=1,\ldots,d$ and $k\in\cT'_{m,s}$, we have with probability at least $(1-(nd)^{-C_2}$ for some $C_2>0$, it holds that
\begin{equation*}
	\max_{1\leq m\leq M_0} \max_{k\in \cT_{m,2}}\max_{1\leq j\leq d}|T_{j,1,1}(k)|\leq C_1\sqrt{\log(nd)},~~\max_{1\leq m\leq M_0} \max_{k\in \cT_{m,2}}\max_{1\leq j\leq d}|T_{j,1,2}(k)|\leq C_1\sqrt{\log(nd)}.
\end{equation*}
For $T_{j,1,3}(k)$, wtih fixed $j$ and $m$, taking  $\underline{\tau}_{1}(k)=k-G$, $\overline{\tau}_{1}(k)=\gamma_m$, $\underline{\tau}_{2}(k)=k+1$, $\overline{\tau}_{2}(k)=k+G$, with probability at least $1-(nd)^{-C_2}$ for some $C_2>3$, we have 
\begin{equation*}
	\begin{array}{ll}
		G^{3/2}	\max_{k\in \cT_{m}}|T_{j,1,3}(k)|&\leq C_1 \log(nd)G\Big(\sqrt{\sigma^2}+\sqrt{\dfrac{\log^2(nd)}{\eta G}}+\sqrt{\dfrac{\log^2(nd)}{G}}\Big)\\
		&\leq C_1G\log(nd)
	\end{array}
\end{equation*}
Taking the union bound over $m=1,\ldots,M_0$, $j=1,\ldots,d$, we have 
\begin{equation*}
	\max_{1\leq m\leq M_0} \max_{k\in \cT_{m,2}}\max_{1\leq j\leq d}|T_{j,1,3}(k)|=O_p(\dfrac{\log(nd)}{\sqrt{G}})\leq O_p(\sqrt{\log(nd)})
\end{equation*}
Lastly, using a similar proof, for $T_{j,2}(k)$ in (\ref{equ: decomposition T_{j,k}}), we can prove that 
\begin{equation*}
	\max_{1\leq m\leq M_0} \max_{k\in \cT_{m,2}}\max_{1\leq j\leq d}|T_{j,2}(k)|=O_p(\sqrt{\log(nd)}
\end{equation*}
Combining the above results, with probability at least $(1-(nd)^{-C_2}$ for some $C_2>0$, we have
\begin{equation}\label{equ: bound of non signal}
	\begin{array}{ll}
		\max\limits_{1\leq m\leq M_0} \max\limits_{k\in \cT_{m,2}}\max\limits_{1\leq j\leq d}\Big(|T_{j,1,1}(k)|+|T_{j,1,2}(k)|+|T_{j,2}(k)| \Big)\leq C_1\sqrt{\log(nd)}.
	\end{array}
\end{equation}
Hence, by the decomposition of $T_j(k)$ in (\ref{equ: decomposition T_{j,k}}) and the Hoeffding's decomposition in (\ref{equ: Hoeffding T_{j,k}}), and the bound in (\ref{equ: bound of non signal}), for $s=2$, with probability at least $1-(nd)^{-C_2}$ , we have
\begin{equation*}
	\begin{array}{ll}
		\min\limits_{1 \leq m\leq M_0}\min\limits_{k\in \cT'_{m,2}}\max\limits_{1\leq j\leq d}|T_j(k)|\\
		\geq_{(1)} \min\limits_{1 \leq m\leq M_0}\min\limits_{k\in \cT'_{m,2}}\max\limits_{1\leq j\leq d}\Big|\dfrac{\gamma_m-k+G}{G^{1/2}}\theta_{j}^m\Big|-\max\limits_{1\leq m\leq M_0} \max\limits_{k\in \cT_{m,2}}\max\limits_{1\leq j\leq d}\Big(|T_{j,1,1}(k)|+|T_{j,1,2}(k)|+|T_{j,2}(k)| \Big)\\
		\geq_{(2)} \dfrac{\eta G}{G^{1/2}}\min_{1\leq m\leq M_0}\|\btheta^m\|_{\infty}-C_1\sqrt{\log(nd)}\\
		\geq_{(3)} 0
	\end{array}
\end{equation*}
where $(2)$ comes from the fact that $\gamma_m\leq k \leq \gamma_m+(1-\eta)G$ and (\ref{equ: bound of non signal}),
and $(3)$ is due to the assumption that $\min_{1\leq m\leq M_0}\|\btheta^m\|_{\infty}\gg \sqrt{\log(nd)/G}$. Hence, we have completed the proof for $s=2$. Note that the proof for $s=1$ is similar, which is ommitted.

{\bf{Proof of (2) of Theorem \ref{theory: initial estimators}}}. In this part, we give the theoretical results of the estimated change point $\hat{\gamma}_m$ for $m=1,\ldots,M_0$.  In other words, we aim to prove that 
\begin{equation*}
	|\hat{\gamma}_m-\gamma_m|=O_p\Big(\dfrac{\log(nd)}{\|\btheta^m\|^2_{\infty}}\Big).
\end{equation*}
Note that we assume $\|\btheta^m\|_{2}\leq C_0$ for some $C_0>0$. Hence, we can assume $	|\hat{\gamma}_m-\gamma_m|\geq \log(nd)$, otherwise the result is trivial. 

Define the event
\[
M_n = \left\{\hat{M}_0 = M_0, \ \max_{1 \leq m \leq M_0} |\hat{\gamma}_m - \gamma_m| < G, \ \hat{\tau}_n > 0 \right\} 
\]
Let $\tilde{v}_m=\max(\gamma_m-G+1,v_m)$ and $\tilde{w}_m=\min(\gamma_m+G,w_m)$. Under the event $M_n$, we have
\begin{equation}\label{equ: definition of gamma-hat}
	\begin{array}{ll}
		\hat{\gamma}_m&=\argmax_{\tilde{v}_m\leq k\leq \tilde{w}_{m}} |\bT(k)|_{\infty}.
	\end{array}
\end{equation}
We first consider the case of $\hat{\gamma}_m\geq \gamma_m$ since the case of $\hat{\gamma}_m\leq \gamma_m$ is similar. Recal $\theta_{j}^{m}:=\E h(X_{\gamma_m,j},X_{\gamma_{m+1},j})$ be the signal jump at the change point $\gamma_{m}$ for the $j$-th coordinate and $\btheta^m =(\theta_1^m,\ldots,\theta_d^m)$. Without loss of generality, we assume $\theta_{j}^m\geq 0$.  Let $j^*\in\{1,\ldots,d\}$ such that $T_{j^*}(\hat{\gamma}_m)=\max\limits_{1\leq j\leq d}T_j(\hat{\gamma}_m)$. The following Lemma \ref{lemma: maximum signal at the identified cpt} shows that $\liminf_{n\rightarrow \infty}\theta^{m}_{j^*}/\|\btheta^m\|_{\infty}\geq 1$. The proof of Lemma \ref{lemma: maximum signal at the identified cpt} is given in Section \ref{sec: proof of lemma at the biggest signal}
\begin{lemma}\label{lemma: maximum signal at the identified cpt}
	{	Suppose Assumptions $\mathbf{(B.1)}$ -- $\mathbf{(B.3)}$ and $\mathbf{(C.1)}$ hold}. Then, with probability tending to one,  we have
	$\liminf_{n\rightarrow \infty}\theta^m_{j^*}/\|\btheta^m\|_{\infty}\geq 1$.
\end{lemma}
Furthermore, define the event 
\begin{equation}
	\begin{array}{ll}
		\cH_1=\Big\{\max_{1\leq j\leq d}T_j(\hat{\gamma}_m)=\max_{1\leq j\leq d}|T_j(\hat{\gamma}_m)|:=\|\bT(\hat{\gamma}_m)\|_{\infty}\Big\},\\
		\cH_2=\Big\{T_{j^*}({\gamma}_m)=|T_{j^*}({\gamma}_m)|\Big\}.
	\end{array}	
\end{equation}
The following Lemma \ref{lemma: sign consistent} shows that $\cH_1\cap\cH_2$ occurs with high probability. The proof of Lemma \ref{lemma: sign consistent} is provided in Section \ref{sec: proof of lemma of the sign consistency}

\begin{lemma}\label{lemma: sign consistent}
	Suppose Assumptions $\mathbf{(B.1)}$ -- $\mathbf{(B.3)}$ and $\mathbf{(C.1)}$ hold. Then we have
	\begin{equation}
		\P(\cH_1\cap\cH_2)\geq 1-C_1(nd)^{-C_2},
	\end{equation}
	where $C_1$ and $C_2$ are universal positive constants not depending on $n$ or $p$.
\end{lemma}
Moreover, our proof relies on the following lemma. 
\begin{lemma}\label{lemma: big error bound of initial estimators}
	{	Suppose Assumptions $\mathbf{(B.1)}$ -- $\mathbf{(B.3)}$ and $\mathbf{(C.1)}$ hold}. Then, with probability tending to one,  we have $|\hat{\gamma}_m-\gamma_m|\leq G/2$. 
\end{lemma}
Using the above lemmas, we are ready to give the proof. Specifically, by the above two lemmas, we have:
\begin{equation*}
	\begin{array}{ll}
		&	\|\bT(\gamma_m)\|_{\infty}-	\|\bT(\hat{\gamma}_m)\|_{\infty}\\
		&=\|\bT(\gamma_m)\|_{\infty}-T_{j^*}(\hat{\gamma}_m )\\
		&\geq T_{j^*}(\gamma_m)-T_{j^*}(\hat{\gamma}_m).	
	\end{array}
\end{equation*}

Moreover, for $k\in[\tilde{v}_m,\tilde{w}_m]$, by the definition of $T_j(k)$ in (\ref{equ: decomposition T_{j,k}}), we have
\begin{equation}\label{equ: Tj true minus Tj} 
	T_{j^*}(\gamma_m)-T_{j^*}(\hat{\gamma}_m)=A_1+A_2-A_3-A_{4},
\end{equation}
where $A_1,\ldots,A_4$ are defined as:
\begin{equation*}
	\begin{array}{ll}
		A_1=\dfrac{1}{G^{3/2}}\sum\limits_{t_1=\gamma_m-G}^{\gamma_m}\sum\limits_{t_2=\gamma_{m}}^{\hat{\gamma}_m}h(X_{t_1,j^*},X_{t_2,j^*}),\\
		A_2=\dfrac{1}{G^{3/2}}\sum\limits_{t_1=\gamma_m-G}^{\hat{\gamma}_m-G}\sum\limits^{\gamma_{m}+G}\limits_{t_2=\hat{\gamma}_m}h(X_{t_1,j^*},X_{t_2,j^*}),\\
		A_3=\dfrac{1}{G^{3/2}}\sum\limits_{t_1=\hat{\gamma}_m-G}^{\gamma_m}\sum\limits_{t_2=\gamma_{m}+G}^{\hat{\gamma}_m+G}h(X_{t_1,j^*},X_{t_2,j^*}),\\
		A_4=\dfrac{1}{G^{3/2}}\sum\limits_{t_1=\gamma_m}^{\hat{\gamma}_m}\sum\limits_{t_2=\hat{\gamma}_{m}}^{\hat{\gamma}_m+G}h(X_{t_1,j^*},X_{t_2,j^*})\\
	\end{array}	
\end{equation*}
Note that $A_1,A_2,A_3$ involve data before and after the change point $\gamma_m$ and $A_4$ involves homogeneous data after the change point. Hence, we can make Hoeffding's decomposition for the above four parts as:
\begin{equation}\label{equ: Hoeffding's decomposition for four parts}
	\begin{array}{ll}
		A_1=\dfrac{1}{G^{3/2}}\Big\{G(\hat{\gamma}_m-\gamma_m)\theta_{j^*}^m+\underbrace{(\hat{\gamma}_m-\gamma_m)\sum\limits_{t_1=\gamma_m-G}^{\gamma_m} h_{1,j^*}(X_{t_1,j^*})}_{A_{1,1}}+\underbrace{G\sum\limits_{t_2=\gamma_m}^{\hat{\gamma}_m}h_{2,j^*}(X_{t_2,j^*})}_{A_{1,2}}\\
		+\nb{\underbrace{\sum\limits_{t_1=\gamma_m-G}^{\gamma_m}\sum\limits_{t_2=\gamma_{m}}^{\hat{\gamma}_m}g_{j^*}(X_{t_1,j^*},X_{t_2,j^*})}_{A_{1,3}}\Big\}},\\
		A_2=\dfrac{1}{G^{3/2}}\Big\{(\hat{\gamma}_m-\gamma_m)(\gamma_m+G-\hat{\gamma}_m)\theta_{j^*}^m+\underbrace{({\gamma}_m+G-\hat{\gamma}_m)\sum\limits_{t_1=\gamma_m-G}^{\hat{\gamma}_m-G} h_{1,j^*}(X_{t_1,j^*})}_{A_{2,1}}\\+\underbrace{(\hat{\gamma}_m-\gamma_m)\sum\limits_{t_2=\hat{\gamma}_m}^{\gamma_m+G}h_{2,j^*}(X_{t_2,j^*})}_{A_{2,2}}
		+\nb{\underbrace{\sum\limits_{t_1=\gamma_m-G}^{\hat{\gamma}_m-G}\sum\limits^{\gamma_{m}+G}\limits_{t_2=\hat{\gamma}_m}g_{j^*}(X_{t_1,j^*},X_{t_2,j^*})}_{A_{2,3}}\Big\}},\\
		A_3=\dfrac{1}{G^{3/2}}\Big\{(\hat{\gamma}_m-\gamma_m)(\gamma_m+G-\hat{\gamma}_m)\theta_{j^*}^m+\underbrace{(\hat{\gamma}_m-{\gamma}_m)\sum\limits_{t_1=\hat{\gamma}_m-G}^{{\gamma}_m} h_{1,j^*}(X_{t_1,j^*})}_{A_{3,1}}\\+\underbrace{({\gamma}_m+G-\hat{\gamma}_m)\sum\limits_{t_2={\gamma}_m+G}^{\hat{\gamma}_m+G}h_{2,j^*}(X_{t_2,j^*})}_{A_{3,2}}
		+\nb{\underbrace{\sum\limits_{t_1=\hat{\gamma}_m-G}^{\gamma_m}\sum\limits_{t_2=\gamma_{m}+G}^{\hat{\gamma}_m+G}g_{j^*}(X_{t_1,j^*},X_{t_2,j^*})}_{A_{3,3}}\Big\}},\\
		A_4=\dfrac{1}{G^{3/2}}\Big\{\underbrace{G\sum\limits_{t_1={\gamma}_m}^{\hat{\gamma}_m} h^0_{1,j^*}(X_{t_1,j^*})}_{A_{4,1}}+\underbrace{(\hat{\gamma}_m-{\gamma}_m)\sum\limits_{t_2=\hat{\gamma}_m}^{\hat{\gamma}_m+G}h^0_{2,j^*}(X_{t_2,j^*})}_{A_{4,2}}
		+\nb{\underbrace{\sum\limits_{t_1=\gamma_m}^{\hat{\gamma}_m}\sum\limits_{t_2=\hat{\gamma}_{m}}^{\hat{\gamma}_m+G}g^0_{j^*}(X_{t_1,j^*},X_{t_2,j^*})}_{A_{4,3}}\Big\}}
	\end{array}
\end{equation}
where we use the superscript ``0" in $A_4$ to denote that the Hoeffding's decomposition is based on the homogeneous samples $X_{t,j^*}$ with $[\gamma_m,\gamma_m+2G]$. Hence, based on the (\ref{equ: Tj true minus Tj}) and the Hoeffding's decomposition in (\ref{equ: Hoeffding's decomposition for four parts}), we have:
\begin{equation}\label{equ: lower bound of Tj-est minus Tj true}
	\begin{array}{ll}
		T_{j^*}(\gamma_m)-T_{j^*}(\hat{\gamma}_m)\geq_{(1)} G^{-1/2}(\hat{\gamma}_m-\gamma_m)\theta_{j^*}^m-G^{-3/2}\Big(|A_{1,1}|+|A_{1,2}|+|A_{1,3}|\\
		+|A_{2,1}|+|A_{2,2}|+|A_{2,3}|	+|A_{3,1}|+|A_{3,2}|+|A_{3,3}|+	+|A_{4,1}|+|A_{4,2}|+|A_{4,3}|\Big)\\
		\geq_{(2)}G^{-1/2}(\hat{\gamma}_m-\gamma_m)\|\btheta^m\|_\infty-G^{-3/2}\Big(|A_{1,1}|+|A_{1,2}|+|A_{1,3}|\\
		+|A_{2,1}|+|A_{2,2}|+|A_{2,3}|+|A_{3,1}|+|A_{3,2}|+|A_{3,3}|	+|A_{4,1}|+|A_{4,2}|+|A_{4,3}|\Big)\\
	\end{array}
\end{equation}
where $(2)$ come from Lemma \ref{lemma: maximum signal at the identified cpt}. Hence, the next goal is to consider the 12 parts $|A_{1,1}|,\ldots,|A_{4,3}|$, respectively. 

For $A_{1,1}$, taking $\overline{k}=\underline{k}=1$, $\tau_1(k)=\gamma_m-G$ and $\tau_{2}(k)=\gamma_m$, using Lemma \ref{lemma: concentration for maxsimum sub-exponential}, with probability at least $1-C_1(nd)^{-C_2}$, we have
\begin{equation*}
	|A_{1,1}|\leq C_3(\hat{\gamma}_m-\gamma_m)\sqrt{G\log(nd)}.
\end{equation*}
Using a similar analysis, we can prove that with probability at least $1-C_1(nd)^{-C_2}$, it holds that 
\begin{equation*}
	\begin{array}{ll}
		|A_{1,2}|&\leq C_3G(\hat{\gamma}_m-\gamma_m)\max\Big(  \sqrt{\dfrac{\log(nd)}{\hat{\gamma}_m-\gamma_m}},  \dfrac{\log(nd)}{\hat{\gamma}_m-\gamma_m}  \Big)\\
		&=C_3G\max\Big(\sqrt{(\hat{\gamma}_m-\gamma_m)\log(nd)},\log(nd)\Big),\\
		|A_{2,1}|&\leq C_3(\gamma_m+G-\hat{\gamma}_m)(\hat{\gamma}_m-\gamma_m)\max\Big(  \sqrt{\dfrac{\log(nd)}{\hat{\gamma}_m-\gamma_m}},  \dfrac{\log(nd)}{\hat{\gamma}_m-\gamma_m}  \Big)\\
		&=C_3(\gamma_m+G-\hat{\gamma}_m)\max\Big(\sqrt{(\hat{\gamma}_m-\gamma_m)\log(nd)},\log(nd)\Big),\\
		|A_{2,2}|&\leq C_3(\hat{\gamma}_m-\gamma_m)(\gamma_m+G-\hat{\gamma}_m)\max\Big(  \sqrt{\dfrac{\log(nd)}{\gamma_m+G-\hat{\gamma}_m}},  \dfrac{\log(nd)}{\gamma_m+G-\hat{\gamma}_m}  \Big)\\
		&=C_3(\hat{\gamma}_m-\gamma_m)\max\Big(\sqrt{(\gamma_m+G-\hat{\gamma}_m)\log(nd)},\log(nd)\Big),\\
		|A_{3,1}|&\leq C_3(\hat{\gamma}_m-\gamma_m)(\gamma_m+G-\hat{\gamma}_m)\max\Big(  \sqrt{\dfrac{\log(nd)}{\gamma_m+G-\hat{\gamma}_m}},  \dfrac{\log(nd)}{\gamma_m+G-\hat{\gamma}_m}  \Big)\\
		&=C_3(\hat{\gamma}_m-\gamma_m)\max\Big(\sqrt{(\gamma_m+G-\hat{\gamma}_m)\log(nd)},\log(nd)\Big),\\
		|A_{3,2}|&\leq C_3(\gamma_m+G-\hat{\gamma}_m)(\hat{\gamma}_m-\gamma_m)\max\Big(  \sqrt{\dfrac{\log(nd)}{\hat{\gamma}_m-\gamma_m}},  \dfrac{\log(nd)}{\hat{\gamma}_m-\gamma_m}  \Big)\\
		&=C_3(\gamma_m+G-\hat{\gamma}_m)\max\Big(\sqrt{(\hat{\gamma}_m-\gamma_m)\log(nd)},\log(nd)\Big),\\
		|A_{4,1}|&\leq C_3G(\hat{\gamma}_m-\gamma_m)\max\Big(  \sqrt{\dfrac{\log(nd)}{\hat{\gamma}_m-\gamma_m}},  \dfrac{\log(nd)}{\hat{\gamma}_m-\gamma_m}  \Big)\\
		&=C_3G\max\Big(\sqrt{(\hat{\gamma}_m-\gamma_m)\log(nd)},\log(nd)\Big),\\
		|A_{4,2}|&\leq C_3(\hat{\gamma}_m-\gamma_m)\sqrt{G\log(nd)}.
	\end{array}
\end{equation*}
Note that by the assumption that $  \hat{\gamma}_m-\gamma_m\geq \log(nd)$, we have 
\begin{equation*}
	\max\Big(\sqrt{(\hat{\gamma}_m-\gamma_m)\log(nd)},\log(nd)\Big)=\sqrt{(\hat{\gamma}_m-\gamma_m)\log(nd)},~~\gamma_m+G-\hat{\gamma}_m\leq G
\end{equation*}
It implies that
\begin{equation}\label{equ: upper bound of the leading hoeffding}
	\begin{array}{ll}
		|A_{1,1}|+\ldots+|A_{4,2}|&\leq C_4\Big((\hat{\gamma}_m-\gamma_m)\sqrt{G\log(nd)}+G\sqrt{(\hat{\gamma}_m-\gamma_m)\log(nd)}\Big)\\
		&\leq C_5G\sqrt{(\hat{\gamma}_m-\gamma_m)\log(nd)},\\
	\end{array}
\end{equation}
where the second inequality comes from the fact that $\hat{\gamma}_m-\gamma_m\leq G$. 

Next, we analyze $A_{1,3}$, $A_{2,3}$, $A_{3,3}$ and $A_{4,3}$. For $A_{1.3}$, taking  $\underline{\tau}_{1}(k)=\gamma_m-G$, $\overline{\tau}_{1}(k)=\gamma_m$, $\underline{\tau}_{2}(k)=\gamma_m$, $\overline{\tau}_{2}(k)=\hat{\gamma}_m$ in Lemma \ref{lemma: Hoeffding's residual}, with probability at least $1-(nd)^{-C_2}$ for some $C_2>3$, we have 
\begin{equation}\label{equ: upper bound of the residual hoeffding 1}
	\begin{array}{ll}
		|A_{1,3}|&\leq_{(1)} C_1 \log(nd)\sqrt{G}\sqrt{\hat{\gamma_m}-\gamma_m}\Big(\sqrt{\sigma^2}+\sqrt{\dfrac{\log^2(nd)}{ G}}+\sqrt{\dfrac{\log^2(nd)}{\hat{\gamma}_m-\gamma_m}}\Big)\\
		&\leq_{(2)} C_1\sqrt{G}\sqrt{\hat{\gamma_m}-\gamma_m}\log(nd)\sqrt{\log(nd)}\\
		&\ll_{(3)} C_1G\sqrt{\hat{\gamma_m}-\gamma_m}\sqrt{\log(nd)}
	\end{array}
\end{equation}
where the $(2)$ uses the fact that $\log(nd)\leq \hat{\gamma}_m-\gamma_m\leq G$, $(3)$ uses the assumption that $G\gg \log^2(nd)$. Similarly, for $A_{2,3}$,  using Lemma \ref{lemma: Hoeffding's residual}, we can prove
\begin{equation*}
	\begin{array}{ll}
		|A_{2,3}|&\leq_{(1)} C_1 \log(nd)\sqrt{\hat{\gamma_m}-\gamma_m}\sqrt{\gamma_m-\hat{\gamma_m}+G}\Big(\sqrt{\sigma^2}+\sqrt{\dfrac{\log^2(nd)}{ \hat{\gamma_m}-\gamma_m}}+\sqrt{\dfrac{\log^2(nd)}{\gamma_m-\hat{\gamma_m}+G}}\Big)\\
		&=_{(2)} C_1 \log(nd)\sqrt{\hat{\gamma_m}-\gamma_m}\Big(\sqrt{\sigma^2(\gamma_m-\hat{\gamma_m}+G)}+\sqrt{(\gamma_m-\hat{\gamma_m}+G)\dfrac{\log^2(nd)}{ \hat{\gamma_m}-\gamma_m}}+\log(nd)\Big)\\
		& =_{(3)} C_1 \sqrt{\log(nd)}\sqrt{\hat{\gamma_m}-\gamma_m}\Big(\sqrt{\log(nd)\sigma^2(\gamma_m-\hat{\gamma_m}+G)}\\&+\sqrt{(\gamma_m-\hat{\gamma_m}+G)\dfrac{\log^3(nd)}{ \hat{\gamma_m}-\gamma_m}}+\log^{3/2}(nd)\Big)\\
	\end{array}
\end{equation*}
Note that under the event $M_n$ and the assumption that $\hat{\gamma}_m-\gamma_m\geq \log(nd)$, we have $\hat{\gamma}_m\in[\gamma_m-\log(nd),\gamma_m+G]$. Hence, using the assumption that $G\gg \log^2(nd)$, we have
\begin{equation*}
	\begin{array}{ll}
		\log(nd)\sigma^2(\gamma_m-\hat{\gamma_m}+G)\leq C\log(nd)(G-\log(nd))\ll G^2\\
		(\gamma_m-\hat{\gamma_m}+G)\dfrac{\log^3(nd)}{ \hat{\gamma_m}-\gamma_m}\leq \log^3(nd)\dfrac{G-\log(nd)}{\log(nd)}\leq \dfrac{G}{\log^2(nd)}\ll G^2\\
		\log^{3/2}(nd)\ll G
	\end{array}
\end{equation*}
where the second inequality uses the fact that $(\gamma_m-\hat{\gamma_m}+G)\dfrac{\log^3(nd)}{ \hat{\gamma_m}-\gamma_m}$ obtains its maximum at $\hat{\gamma}_m=\gamma_m+\log(nd)$. This implies that 
\begin{equation}\label{equ: upper bound of the residual hoeffding 2}
	|A_{2,3}|=o_p(G\sqrt{\hat{\gamma_m}-\gamma_m}\sqrt{\log(nd)}).
\end{equation}
Lastly, for $A_{3,3}$ and $A_{4,3}$, using a similar proof, we can prove that 
\begin{equation}\label{equ: upper bound of the residual hoeffding 3}
	|A_{3,3}|, 	|A_{4,3}|,=o_p(G\sqrt{\hat{\gamma_m}-\gamma_m}\sqrt{\log(nd)}).
\end{equation}
Combining (\ref{equ: lower bound of Tj-est minus Tj true}) - (\ref{equ: upper bound of the residual hoeffding 3}), we have with probability at least $1-C_1(nd)^{-C_2}$, it holds that 
\begin{equation*}
	T_{j^*}(\gamma_m)-T_{j^*}(\hat{\gamma}_m)\geq G^{-1/2}(\hat{\gamma}_m-\gamma_m)\|\btheta^m\|_\infty-C_*G^{-3/2}G\sqrt{(\hat{\gamma}_m-\gamma_m)\log(nd)}
\end{equation*}
for some large enough $C_*>0$. Since $\hat{\gamma}_m$ is the maximizer of $T_{j^*}(k)$, it implies that with probability at least $1-C_1(nd)^{-C_2}$, we have
\begin{equation*}
	G^{-1/2}(\hat{\gamma}_m-\gamma_m)\|\btheta^m\|_\infty-C_*G^{-3/2}G\sqrt{(\hat{\gamma}_m-\gamma_m)\log(nd)}\leq 0,
\end{equation*}
which further implies
\begin{equation*}
	\hat{\gamma}_m-\gamma_m \leq C_*\dfrac{\log(nd)}{\|\btheta^m\|^2_{\infty}},
\end{equation*}
which completes the proof for the case of $\hat{\gamma}_m\geq \gamma_m$. Note that the case of $\hat{\gamma}_m<\gamma_m$ is similar. The proof is omitted.

\newpage

\subsection{Proof of Theorem \ref{theorem: support recovery}}
Recall $\Pi_{m}:=\{j: 1\leq j\leq d: |\theta_{j}^{(m)}>0\}|$ be the coordinates having a change point at $\gamma_m$ and $\hat{\Pi}_{m}:=\{j: 1\leq j\leq d: |\hat{\theta}_{j}^{(m)}>w^+\}|$ be the estimator. To prove $\hat{\Pi}_m=\Pi_m$, it is sufficient to prove that 
\begin{equation}\label{equ: large deviation bound of parameter estimator}
	\max_{1\leq m \leq M_0}\max_{1\leq j\leq d}|\hat{\theta}_{j}^{(m)}-{\theta}_j^{m}|=O_p\Big(D\Big(\sqrt{\dfrac{\log(nd)}{G}}+\dfrac{\log(Gd)\log(nd)}{G}\Big)\Big).
\end{equation}
By (\ref{equ: large deviation bound of parameter estimator}), for any $j\notin \Pi_{m}$ with $\theta_j^{(m)}=0$, we have
\begin{equation*}
	|\hat{\theta}_j^{(m)}|=	|\hat{\theta}_j^{(m)}-\theta_j^{(m)}|\leq C_2D\Big(\sqrt{\dfrac{\log(nd)}{G}}+\dfrac{\log(Gd)\log(nd)}{G}\Big)\leq w^+.
\end{equation*}
This implies that $j\notin \hat{\Pi}_m$. For any $j\in \Pi_{m}$, we have
\begin{equation*}
	\begin{array}{ll}
		|\hat{\theta}_j^{(m)}|&=	|\hat{\theta}_j^{(m)}-\theta_j^{(m)}+\theta_j^{(m)}|\\
		&\geq |\theta_j^{(m)}|-|\hat{\theta}_j^{(m)}-\theta_j^{(m)}|\\
		&\geq \min_{j,m}|\theta_j^{(m)}|-\max_{j,m}|\hat{\theta}_j^{(m)}-\theta_j^{(m)}|\\
		&\geq \min_{j,m}|\theta_j^{(m)}|-C_1\sqrt{\dfrac{\log(dn)}{G}}\\
		&\geq w^+,
	\end{array}
\end{equation*}
This implies that $j\in \hat{\Pi}_m$. Hence, it remains to prove (\ref{equ: large deviation bound of parameter estimator}). By {Assumptions $\mathbf{(B.1)}$ -- $\mathbf{(B.3)}$ and $\mathbf{(C.1)}$}, Lemma A.3 in \cite{yu2022robust}, we can prove that with probability at least $1-(nd)^{-C_1}$ for some $C_1>2$, it holds that
\begin{equation*}
	\max_{1\leq j\leq d}|\hat{\theta}_{j}^{(m)}-{\theta}_j^{m}|\leq C_2D\Big(\sqrt{\dfrac{\log(nd)}{G}}+\dfrac{\log(Gd)\log(nd)}{G}\Big).
\end{equation*}
for some big enough constant $C_2>0$. Lastly, taking the union bound over $m=1,\ldots,M_0$, we finish the proof of (\ref{equ: large deviation bound of parameter estimator}).
\subsection{Proof of Theorem \ref{theorem: refined estimation}}

{\bf{Proof of (1) of Theorem \ref{theorem: refined estimation}}}.  Define the set $$\MM=
\left\{
\begin{aligned}
	&\hat{M}_0 = M_0, \hat{\Pi}_m={\Pi}_{m} ~~\text{for} ~~m=1,\ldots,M_0,\\
	&\max_{1 \leq m \leq \hat{M}_0} |\hat{\gamma}_m - \gamma_m| = O_P\left( \frac{\log(nd)}{\|\btheta^{(m)}\|_{\infty}^2} \right), \\
	&\max_{1 \leq m \leq \hat{M}_0} \|\hat{\btheta}^{(m)} - \btheta^{(m)}\|_{\infty} = O_P\left( \sqrt{\frac{\log(nd)}{G}} \right).
\end{aligned}
\right\}
$$ Note that by Theorem \ref{theory: initial estimators}, we have $\PP(\MM)\rightarrow 1$. Next, we give our proof given $\MM$. Note that by {Assumption C.1}, we have 
\begin{equation*}
	\frac{\log(nd)}{\|\btheta^{(m)}\|_{\infty}^2} \ll G.
\end{equation*}
For each change point $m=1,\ldots,M_0$, recall 
\begin{equation}
	\tilde{T}_j(k):=\dfrac{1}{G}\sum_{t_1=k-G}^{k}\sum_{t_2=k+1}^{k+G}\hat{\theta}_{j}^{(m)}h(X_{t_1,j},X_{t_2,j}),~~\text{for}~~j=1,\ldots,d.
\end{equation}
and 
\begin{equation}\label{equ: refined estimator}
	\tilde{\gamma}_m=\argmax_{|k-\hat{\gamma}_{m}|\leq G/4}\sum_{j\in\hat{\Pi}_m}\tilde{T}_j(k):=\argmax_{|k-\hat{\gamma}_{m}|\leq G/4}\tilde{T}(k),
\end{equation}
where $\tilde{T}(k)=\sum_{j\in\hat{\Pi}_m}\tilde{T}_j(k)$. Note that by Theorem \ref{theory: initial estimators}, we have proved that $|\hat{\gamma}_m-\gamma_m|=o_p(G)$. Hence, the search domain belongs to $[\gamma_m-G,\gamma_m+G]$. 
Next, we will prove that $|\tilde{\gamma}_m-{\gamma}_m|=O_P(\dfrac{1}{\|\btheta^{(m)}\|^2})$. 
Since $\tilde{\gamma}_m$ is the ``argmax'' of $\tilde{T}(k)$, it is also the ``argmax" of $	\tilde{T}(k)-\tilde{T}(\gamma_m)$. Equivalently, let $\tilde{\gamma_m}=\gamma_m+r$, we aim to prove that $|r|=O_P(\dfrac{1}{\|\btheta^{(m)}\|^2})$. Note that by the above search domain, we have $r\in[\gamma_m-G,\gamma_m+G]$. 
Let 
\begin{equation*}
	J_1(r)=G\times \big(\tilde{T}(\gamma_m+r)-\tilde{T}(\gamma_m)\big).
\end{equation*}
Since $r$ is the argmax of $G(r)$, we have $\tilde{T}(\gamma_m+r)-\tilde{T}(\gamma_m)\ge 0$. For simplicity, we assume $r\geq 0$. Moreover, by the definition of  $\tilde{T}(k)$, $G(r)$ can be decomposed into the following:

\begin{equation*}
	J_1(r)=I+II-III-IV,
\end{equation*}
where
\begin{equation*}
	\begin{array}{ll}
		I=\sum\limits_{j\in\hat{\Pi}_m}\sum\limits_{t_1=\gamma_m+r-G}^{\gamma_m}\sum\limits_{t_2=\gamma_m+G}^{\gamma_m+r+G}\hat{\theta}_{j}^{(m)}h(X_{t_1,j},X_{t_2,j}),\\
		II=\sum\limits_{j\in\hat{\Pi}_m}\sum\limits_{t_1=\gamma_m+1}^{\gamma_m+r}\sum\limits_{t_2=\gamma_m+r}^{\gamma_m+G+r}\hat{\theta}_{j}^{(m)}h(X_{t_1,j},X_{t_2,j}),\\
		III=\sum\limits_{j\in\hat{\Pi}_m}\sum\limits_{t_1=\gamma_m-G}^{\gamma_m+r-G}\sum\limits_{t_2=\gamma_m+1}^{\gamma_m+G}\hat{\theta}_{j}^{(m)}h(X_{t_1,j},X_{t_2,j}),\\
		IV=\sum\limits_{j\in\hat{\Pi}_m}\sum\limits_{t_1=\gamma_m+r-G}^{\gamma_m+1}\sum\limits_{t_2=\gamma_m+1}^{\gamma_m+r}\hat{\theta}_{j}^{(m)}h(X_{t_1,j},X_{t_2,j}).
	\end{array}
\end{equation*}
Moreover,  replacing $\hat{\theta}_{j}^{(m)}$ by ${\theta}_{j}^{(m)}$ in $I-IV$, we define 
\begin{equation*}
	\begin{array}{ll}
		\tilde{I}=\sum\limits_{j\in\hat{\Pi}_m}\sum\limits_{t_1=\gamma_m+r-G}^{\gamma_m}\sum\limits_{t_2=\gamma_m+G}^{\gamma_m+r+G}{\theta}_{j}^{(m)}h(X_{t_1,j},X_{t_2,j}),\\
		\tilde{II}=\sum\limits_{j\in\hat{\Pi}_m}\sum\limits_{t_1=\gamma_m+1}^{\gamma_m+r}\sum\limits_{t_2=\gamma_m+r}^{\gamma_m+G+r}{\theta}_{j}^{(m)}h(X_{t_1,j},X_{t_2,j}),\\
		\tilde{III}=\sum\limits_{j\in\hat{\Pi}_m}\sum\limits_{t_1=\gamma_m-G}^{\gamma_m+r-G}\sum\limits_{t_2=\gamma_m+1}^{\gamma_m+G}{\theta}_{j}^{(m)}h(X_{t_1,j},X_{t_2,j}),\\
		\tilde{IV}=\sum\limits_{j\in\hat{\Pi}_m}\sum\limits_{t_1=\gamma_m+r-G}^{\gamma_m}\sum\limits_{t_2=\gamma_m+1}^{\gamma_m+r}{\theta}_{j}^{(m)}h(X_{t_1,j},X_{t_2,j}).
	\end{array}
\end{equation*}
Conditional on $\MM$, it is sufficient to consider 
\begin{equation*}
	\tilde{J_1}(r)=\tilde{I}+\tilde{II}-\tilde{III}-\tilde{IV}.
\end{equation*}
Next, we consider the four parts respectively.  According to our considered model, for $t\in[\gamma_m-2G,\gamma_m+2G]$, we have
\begin{equation*}
	\bX_t=\bmu^{ (m)}\mathbf{1}\{\gamma_m-2G\leq t\leq \gamma_m\}+\bmu^{ (m+1)}\mathbf{1}\{\gamma_m< t\leq \gamma_m+2G\}+\bepsilon_t.
\end{equation*}
We next consider the Hoeffding's decomposition under the change point model . Specifically, based on the Hoeffding's decomposition, we have
\begin{equation} \label{equ: hoef1 H1}
	h(x,y)=\theta_{j}^{(m)} +h'_{1,j}(x)+h'_{2,j}(y)+g'_{j}(x,y),
\end{equation}
with
\begin{equation}
	\begin{array}{ll}
		\theta_{j}^{(m)}  =\E h(X_{j}',Y_{j}'), &h'_{1,j}(x)=\E h(x,Y_{j}')-\theta_j,\\
		h'_{2,j}(y)=\E h(X_{j}',y)-\theta_j,& g'_{j}(x,y) = h(x,y)-h'_{1,j}(x)-h'_{2,j}(y)-\theta_{j},
	\end{array}
\end{equation}
where $X_{j}'$ and $Y_{j}'$ are two independent random variables with the same marginal distributions as $X_{\gamma_m,j}$ and $X_{\gamma_m+1,j}$. Hence, by (\ref{equ: hoef1 H1}), we have
\begin{equation*}
	\begin{array}{ll}
		\tilde{I}=(G-r)r\sum\limits_{j\in{\Pi}_m}(\theta_{j}^{(m)})^2+r\sum\limits_{j\in{\Pi}_m}\theta_{j}^{(m)}\sum\limits_{t_1=\gamma_m+r-G}^{\gamma_m}h'_{1,j}(X_{t_1,j})\\
		+(G-r)\sum\limits_{j\in{\Pi}_m}\theta_{j}^{(m)}\sum\limits_{t_2=\gamma_m+G}^{\gamma_m+G+r}h'_{2,j}(X_{t_2,j})+\sum\limits_{j\in\hat{\Pi}_m}\sum\limits_{t_1=\gamma_m+r-G}^{\gamma_m}\sum\limits_{t_2=\gamma_m+G}^{\gamma_m+r+G}{\theta}_{j}^{(m)}g_j'(X_{t_1,j},X_{t_2,j}),\\
		\tilde{IV}=(G-r)r\sum\limits_{j\in{\Pi}_m}(\theta_{j}^{(m)})^2+r\sum\limits_{j\in{\Pi}_m}\theta_{j}^{(m)}\sum\limits_{t_1=\gamma_m+r-G}^{\gamma_m}h'_{1,j}(X_{t_1,j})\\
		+(G-r)\sum\limits_{j\in{\Pi}_m}\theta_{j}^{(m)}\sum\limits_{t_2=\gamma_m}^{\gamma_m+r}h'_{2,j}(X_{t_2,j})+\sum\limits_{j\in\hat{\Pi}_m}\sum\limits_{t_1=\gamma_m+r-G}^{\gamma_m}\sum\limits_{t_2=\gamma_m}^{\gamma_m+r}{\theta}_{j}^{(m)}g_j'(X_{t_1,j},X_{t_2,j}),\\
		\tilde{III}=Gr\sum\limits_{j\in{\Pi}_m}(\theta_{j}^{(m)})^2+G\sum\limits_{j\in{\Pi}_m}\theta_{j}^{(m)}\sum\limits_{t_1=\gamma_m-G}^{\gamma_m-G+r}h'_{1,j}(X_{t_1,j})\\
		+r\sum\limits_{j\in{\Pi}_m}\theta_{j}^{(m)}\sum\limits_{t_2=\gamma_m}^{\gamma_m+G}h'_{2,j}(X_{t_2,j})+\sum\limits_{j\in\hat{\Pi}_m}\sum\limits_{t_1=\gamma_m-G}^{\gamma_m-G+r}\sum\limits_{t_2=\gamma_m}^{\gamma_m+G}{\theta}_{j}^{(m)}g_j'(X_{t_1,j},X_{t_2,j}),\\
		\tilde{II}=\sum\limits_{j\in\hat{\Pi}_m}\sum\limits_{t_1=\gamma_m+1}^{\gamma_m+r}\sum\limits_{t_2=\gamma_m+r}^{\gamma_m+G+r}{\theta}_{j}^{(m)}h(\epsilon_{t_2,j}-\epsilon_{t_1,j}).\\
	\end{array}
\end{equation*}
Based on the above decomposition, we can further have:
\begin{equation}\label{equ: tilde{G}}
	\begin{array}{ll}
		\tilde{G}(r)=	-Gr\sum\limits_{j\in{\Pi}_m}(\theta_{j}^{(m)})^2+\underbrace{\tilde{II}}_{(1)}+	\underbrace{(G-r)\sum\limits_{j\in{\Pi}_m}\theta_{j}^{(m)}\sum\limits_{t_2=\gamma_m+G}^{\gamma_m+G+r}h'_{2,j}(X_{t_2,j})}_{(2)}\\-\underbrace{(G-r)\sum\limits_{j\in{\Pi}_m}\theta_{j}^{(m)}\sum\limits_{t_2=\gamma_m}^{\gamma_m+r}h'_{2,j}(X_{t_2,j})}_{(3)}
		-\underbrace{G\sum\limits_{j\in{\Pi}_m}\theta_{j}^{(m)}\sum\limits_{t_1=\gamma_m-G}^{\gamma_m-G+r}h'_{1,j}(X_{t_1,j})}_{(4)}-\underbrace{r\sum\limits_{j\in{\Pi}_m}\theta_{j}^{(m)}\sum\limits_{t_2=\gamma_m}^{\gamma_m+G}h'_{2,j}(X_{t_2,j})}_{(5)}\\+\underbrace{\sum\limits_{j\in\hat{\Pi}_m}\sum\limits_{t_1=\gamma_m+r-G}^{\gamma_m}\sum\limits_{t_2=\gamma_m+G}^{\gamma_m+r+G}{\theta}_{j}^{(m)}g_j'(X_{t_1,j},X_{t_2,j})}_{(6)}-\underbrace{\sum\limits_{j\in\hat{\Pi}_m}\sum\limits_{t_1=\gamma_m+r-G}^{\gamma_m}\sum\limits_{t_2=\gamma_m}^{\gamma_m+r}{\theta}_{j}^{(m)}g_j'(X_{t_1,j},X_{t_2,j})}_{(7)}\\
		-\underbrace{\sum\limits_{j\in\hat{\Pi}_m}\sum\limits_{t_1=\gamma_m-G}^{\gamma_m-G+r}\sum\limits_{t_2=\gamma_m}^{\gamma_m+G}{\theta}_{j}^{(m)}g_j'(X_{t_1,j},X_{t_2,j})}_{(8)}.
	\end{array}
\end{equation}
Recall $s_m=|\Pi_{m}|$ as the total number of coordinates having a change point at $\gamma_m$. For notational convenience, we assume $\Pi_{m}=\{1,\ldots,s_m\}$. Let $\btheta^{(m)}=(\theta_1^{(m)},\ldots,\theta_{s_m}^{(m)})\in\RR^{s_m}$ be the signal jump in $\Pi_{m}$. 
For (1), define $\bh(\bx,\by)=(h(x_1,y_1),\cdots,h(x_{s_m},y_{s_m}))\in\RR^{s_m}$. Note that (1) considers the noise term with no change points. We can further have:
\begin{equation*}
	(1)=(\btheta^{(m)})^\top\sum\limits_{t_1=\gamma_m+1}^{\gamma_m+r}\sum\limits_{t_2=\gamma_m+r}^{\gamma_m+G+r}\bh(\bepsilon_{t_1},\bepsilon_{t_2}).
\end{equation*}
By the central limiting theory for two-sample U-statistics, we can prove that $$(1)=\|\btheta^{(m)}\|O_p(G\sqrt{r}+r\sqrt{G})$$. For $(3)$, $(4)$ and (5), using the central limiting theory, we can prove 
\begin{equation*}
	(2),(3)=O_p(\sqrt{r}(G-r)\|\btheta^{(m)}\|), ~~(4)=O_p(\sqrt{r}G\|\btheta^{(m)}\|), ~~(5)=O_p(r\sqrt{G}\|\btheta^{(m)}\|).
\end{equation*}
For (6), using Lemma \ref{lemma: Hoeffding's residual}, taking $\underline{w}_1=\overline{w}_1=G-r$ and $\underline{w}_2=\overline{w}_2=r$, we have 
\begin{equation*}
	(6)=O_p(\|\btheta^{(m)}\|\log^2(G)\sqrt{(G-r)r})=o_p(\sqrt{r}G\|\btheta^{(m)}\|).
\end{equation*}
Similarly, we have $(7),(8)=o_p(\sqrt{r}G\|\btheta^{(m)}\|)$.

Note that we require $\tilde{G}(r)\ge 0$. Based on the above bounds, it is sufficient to prove that
\begin{equation}\label{ine: refined estimator}
	\begin{array}{ll}
		Gr\|\btheta^{(m)}\|^2\leq \|\btheta^{(m)}\|O_p(G\sqrt{r}+r\sqrt{G})+O_p(\sqrt{r}(G-r)\|\btheta^{(m)}\|)\\
		+O_p(\sqrt{r}G\|\btheta^{(m)}\|)+O_p(r\sqrt{G}\|\btheta^{(m)}\|)+o_p(\sqrt{r}G\|\btheta^{(m)}\|)\leq O_p(\sqrt{r}G\|\btheta^{(m)}\|),
	\end{array}
\end{equation}
where the last inequality come from the fact that $r\leq G$. Lastly, based on (\ref{ine: refined estimator}), we have $r\leq O_p(\dfrac{1}{\|\btheta^{(m)}\|^2})$, which finishes the proof of (1) of Theorem \ref{theorem: refined estimation}.

{\bf{Proof of (2) of Theorem \ref{theorem: refined estimation}}}. Consider the process
\begin{equation*}
	J_2(r)=\tilde{T}(\gamma_m+r)-\tilde{T}(\gamma_m).
\end{equation*}
We first consider the case of $r\geq 0$. By (\ref{equ: tilde{G}}), we have
\begin{equation*}
	J_2(r):=\dfrac{1}{G}J_1(r).
\end{equation*}
Moreover, replacing the estimated parameters by the true values, we define $\tilde{J}_2(r)$ as follows:
\begin{equation}\label{equ: tilde{G}}
	\begin{array}{ll}
		\tilde{J}_2(r)=	-r\sum\limits_{j\in{\Pi}_m}(\theta_{j}^{(m)})^2+\underbrace{\dfrac{1}{G}\sum\limits_{j\in\hat{\Pi}_m}\sum\limits_{t_1=\gamma_m+1}^{\gamma_m+r}\sum\limits_{t_2=\gamma_m+r}^{\gamma_m+G+r}{\theta}_{j}^{(m)}h(X_{t_1,j},X_{t_2,j})}_{(1)}\\+	\underbrace{\dfrac{(G-r)}{G}\sum\limits_{j\in{\Pi}_m}\theta_{j}^{(m)}\sum\limits_{t_2=\gamma_m+G}^{\gamma_m+G+r}h'_{2,j}(X_{t_2,j})}_{(2)}-\underbrace{\dfrac{(G-r)}{G}\sum\limits_{j\in{\Pi}_m}\theta_{j}^{(m)}\sum\limits_{t_2=\gamma_m}^{\gamma_m+r}h'_{2,j}(X_{t_2,j})}_{(3)}\\
		-\nb{\underbrace{\sum\limits_{j\in{\Pi}_m}\theta_{j}^{(m)}\sum\limits_{t_1=\gamma_m-G}^{\gamma_m-G+r}h'_{1,j}(X_{t_1,j})}_{(4)}}-\underbrace{\dfrac{r}{G}\sum\limits_{j\in{\Pi}_m}\theta_{j}^{(m)}\sum\limits_{t_2=\gamma_m}^{\gamma_m+G}h'_{2,j}(X_{t_2,j})}_{(5)}\\+\underbrace{\sum\limits_{j\in\hat{\Pi}_m}\dfrac{1}{G}\sum\limits_{t_1=\gamma_m+r-G}^{\gamma_m}\sum\limits_{t_2=\gamma_m+G}^{\gamma_m+r+G}{\theta}_{j}^{(m)}g_j'(X_{t_1,j},X_{t_2,j})}_{(6)}-\underbrace{\sum\limits_{j\in\hat{\Pi}_m}\dfrac{1}{G}\sum\limits_{t_1=\gamma_m+r-G}^{\gamma_m}\sum\limits_{t_2=\gamma_m}^{\gamma_m+r}{\theta}_{j}^{(m)}g_j'(X_{t_1,j},X_{t_2,j})}_{(7)}\\
		-\underbrace{\sum\limits_{j\in\hat{\Pi}_m}\dfrac{1}{G}\sum\limits_{t_1=\gamma_m-G}^{\gamma_m-G+r}\sum\limits_{t_2=\gamma_m}^{\gamma_m+G}{\theta}_{j}^{(m)}g_j'(X_{t_1,j},X_{t_2,j})}_{(8)}.
	\end{array}
\end{equation}
Note that 
\begin{equation*}
	(3)+(5)=\nb{-\sum\limits_{j\in{\Pi}_m}\theta_{j}^{(m)}\sum\limits_{t_2=\gamma_m}^{\gamma_m+r}h'_{2,j}(X_{t_2,j})}-\underbrace{\dfrac{r}{G}\sum\limits_{j\in{\Pi}_m}\theta_{j}^{(m)}\sum\limits_{t_2=\gamma_m+r}^{\gamma_m+G}h'_{2,j}(X_{t_2,j})}_{(9)}
\end{equation*}
Furthermore,  combining (2) and (9), we have
\begin{equation*}
	(2)+(9)=\nb{\sum\limits_{j\in{\Pi}_m}\theta_{j}^{(m)}\sum\limits_{t_2=\gamma_m+G}^{\gamma_m+G+r}h'_{2,j}(X_{t_2,j})}-\dfrac{r}{G}\sum\limits_{j\in{\Pi}_m}\theta_{j}^{(m)}\sum\limits_{t_2=\gamma_m+r}^{\gamma_m+r+G}h'_{2,j}(X_{t_2,j}).
\end{equation*}
Lastly, we analyze (1). Note that (1) involves homogeneous data after the change point. To this end, we introduce the Hoeffding's decomposition. 
According to our considered model, for $t\in[\gamma_m+1,\gamma_m+2G]$, we have
\begin{equation*}
	\bX_t=\bmu^{ (m+1)}\mathbf{1}\{\gamma_m< t\leq \gamma_m+2G\}+\bepsilon_t.
\end{equation*}
Based on the Hoeffding's decomposition, we have
\begin{equation} \label{equ: hoef1 H0}
	h(x,y)=0 +\overline{h}_{1,j}(x)+\overline{h}_{2,j}(y)+\overline{g}_{j}(x,y),
\end{equation}
with 
\begin{equation}\label{equation: Hoeffding's after cpt}
	\begin{array}{ll}
		0 =\E h(X_{j}',Y_{j}'), &\overline{{h}}_{1,j}(x)=\E h(x,Y_{j}'),\\
		\overline{h}_{2,j}(y)=\E h(X_{j}',y)=-\overline{h}_{1,j}(y),& \overline{g}_{j}(x,y) = h(x,y)-\overline{h}_{1,j}(x)-\overline{h}_{2,j}(y),
	\end{array}
\end{equation}
where $X_{j}'$ and $Y_{j}'$ are two independent random variables with the same marginal distributions as  $X_{\gamma_m+1,j}$. Note that we use the notation $\overline{{h}}_{1,j}(x)$ to denote that we consider the Hoeffding's decomposition for the data after the change point. Hence, by (\ref{equ: hoef1 H0}), for (1), we have
\begin{equation}\label{equ: decomposition of (1)}
	\begin{array}{ll}
		(1)=\nb{\sum\limits_{j\in{\Pi}_m}\theta_{j}^{(m)}\sum\limits_{t_2=\gamma_m+1}^{\gamma_m+r}\overline{h}_{1,j}(X_{t_2,j})}-\dfrac{r}{G}\sum\limits_{j\in{\Pi}_m}\theta_{j}^{(m)}\sum\limits_{t_2=\gamma_m+r}^{\gamma_m+G+r}\overline{h}_{1,j}(X_{t_2,j})\\+\underbrace{\dfrac{1}{G}\sum\limits_{j\in\hat{\Pi}_m}\sum\limits_{t_1=\gamma_m+1}^{\gamma_m+r}\sum\limits_{t_2=\gamma_m+r}^{\gamma_m+G+r}{\theta}_{j}^{(m)}\overline{g}_j(X_{t_2,j},X_{t_1,j})}_{(10)}
	\end{array}
\end{equation}
Combining the above results, we have
\begin{equation}
	\begin{array}{ll}
		\tilde{J}_2(r)=	-r\sum\limits_{j\in{\Pi}_m}(\theta_{j}^{(m)})^2	-\nb{\sum\limits_{j\in{\Pi}_m}\theta_{j}^{(m)}\sum\limits_{t_1=\gamma_m-G}^{\gamma_m-G+r}h'_{1,j}(X_{t_1,j})}+\nb{\sum\limits_{j\in{\Pi}_m}\theta_{j}^{(m)}\sum\limits_{t_2=\gamma_m+1}^{\gamma_m+r}\big(\overline{h}_{1,j}(X_{t_2,j})-h'_{2,j}(X_{t_2,j})\big)}\\
		+\nb{\sum\limits_{j\in{\Pi}_m}\theta_{j}^{(m)}\sum\limits_{t_2=\gamma_m+G}^{\gamma_m+G+r}h'_{2,j}(X_{t_2,j})}-\dfrac{r}{G}\sum\limits_{j\in{\Pi}_m}\theta_{j}^{(m)}\sum\limits_{t_2=\gamma_m+r}^{\gamma_m+r+G}h'_{2,j}(X_{t_2,j})-\dfrac{r}{G}\sum\limits_{j\in{\Pi}_m}\theta_{j}^{(m)}\sum\limits_{t_2=\gamma_m+r}^{\gamma_m+G+r}\overline{h}_{1,j}(X_{t_2,j})\\+(6)+(7)+(8)+(10).
	\end{array}
\end{equation}
Recall $s_m=|\Pi_{m}|$ as the total number of coordinates and $\btheta^{(m)}=(\theta_1^{(m)},\ldots,\theta_{s_m}^{(m)})\in\RR^{s_m}$ be the signal jump in $\Pi_{m}$. 
For (1), define  $\overline{\bh}_1(\bx)=(\underline{h}_{1,1}(x_1),\cdots,\underline{h}_{1,s_{m}}(x_{s_m}))\in\RR^{s_m}$, $\bh'_1(\bx)=(h'_{1,1}(x_1),\cdots,h'_{1,s_m}(x_{s_m}))\in\RR^{s_m}$, and $\bh'_2(\bx)=(h'_{2,1}(x_1),\cdots,h'_{2,s_m}(x_{s_m}))\in\RR^{s_m}$.  Similarly, define 
$\overline{\bg}'(\bx,\bx')=(\overline{g}'_{1}(x_1,x'_1),\cdots,\overline{g}'_{s_m}(x_{s_m},x'_{s_m}))\in\RR^{s_m}$.

Hence, we have
\begin{equation}
	\begin{array}{ll}
		\tilde{J}_2(r)=	-r\|\btheta^{(m)}\|^2	-\underbrace{(\btheta^{(m)})^\top\sum\limits_{t_1=\gamma_m-G}^{\gamma_m-G+r}\bh_1'(\bX_{t_1})}_{	\tilde{J}_{2,1}(r)}+\underbrace{(\btheta^{(m)})^\top\sum\limits_{t_2=\gamma_m+1}^{\gamma_m+r}\big(\overline{\bh}_{1}(\bX_{t_2})-\bh'_{2}(\bX_{t_2})}_{\tilde{J}_{2,2}(r)}\\
		+\underbrace{(\btheta^{(m)})^\top\sum\limits_{t_2=\gamma_m+G}^{\gamma_m+G+r}\bh'_{2}(\bX_{t_2})}_{\tilde{J}_{2,3}(r)}-\underbrace{\dfrac{r}{G}(\btheta^{(m)})^\top\sum\limits_{t_2=\gamma_m+r}^{\gamma_m+r+G}\bh'_{2}(\bX_{t_2})}_{\tilde{J}_{2,4}(r)}-\underbrace{\dfrac{r}{G}(\btheta^{(m)})^\top\sum\limits_{t_2=\gamma_m+r}^{\gamma_m+G+r}\overline{\bh}_{1}(\bX_{t_2})}_{\tilde{J}_{2,5}(r)}\\
		+\underbrace{(\btheta^{(m)})^\top\dfrac{1}{G}\sum\limits_{t_1=\gamma_m+r-G}^{\gamma_m}\sum\limits_{t_2=\gamma_m+G}^{\gamma_m+r+G}{\bg'(\bX_{t_1},\bX_{t_2})}}_{\tilde{J}_{2,6}(r)}-\underbrace{(\btheta^{(m)})^\top\dfrac{1}{G}\sum\limits_{t_1=\gamma_m+r-G}^{\gamma_m}\sum\limits_{t_2=\gamma_m}^{\gamma_m+r}{\bg'(\bX_{t_1},\bX_{t_2})}}_{\tilde{J}_{2,7}(r)}\\
		-\underbrace{(\btheta^{(m)})^\top\dfrac{1}{G}\sum\limits_{t_1=\gamma_m-G}^{\gamma_m-G+r}\sum\limits_{t_2=\gamma_m}^{\gamma_m+G}{\bg'(\bX_{t_1},\bX_{t_2})}}_{\tilde{J}_{2,8}(r)}+\underbrace{(\btheta^{(m)})^\top\dfrac{1}{G}\sum\limits_{t_1=\gamma_m+1}^{\gamma_m+r}\sum\limits_{t_2=\gamma_m+r}^{\gamma_m+G+r}{\overline{\bg}(\bX_{t_2},
				\bX_{t_1})}}_{\tilde{J}_{2,9}(r)}
	\end{array}
\end{equation}
Let $r=\dfrac{s}{\|\btheta^{(m)}\|^2}$ with $\btheta^{(m)}=(\theta_1^{(m)},\ldots,\theta_{s_m}^{(m)})\in\RR^{s_m}$ be the signal jump in $\Pi_{m}$ with $s>0$.  Define
\begin{equation*}
	\begin{array}{ll}
		\bSigma_{\bh'_1}:=\text{Cov}(\bh'_1(\bX_{\gamma_{m}}))=\text{Cov}(\bh'_1(\bmu^{(m)}+\bepsilon)),\\
		\bSigma_{\bh'_2}:=\text{Cov}(\bh'_2(\bX_{\gamma_{m}+1}))=\text{Cov}(\bh'_2(\bmu^{(m+1)}+\bepsilon)),\\
		\bSigma_{\overline{\bh}_1-\bh'_2}:=\text{Cov}(\overline{\bh}_1(\bX_{\gamma_{m}+1})-\bh'_2(\bX_{\gamma_{m}+1}))=\text{Cov}(\overline{\bh}_1(\bmu^{(m+1)}+\bepsilon)-\bh'_2(\bmu^{(m+1)}+\bepsilon)).
	\end{array}
\end{equation*}
By {Assumptions C.2}, we require there exists constants $\kappa_1$ and $\kappa_2$ such that
\begin{equation*}
	\begin{array}{ll}
		\kappa_1\leq \min(\lambda_{min}(\bSigma_{\bh'_1}),\lambda_{min}(\bSigma_{\bh'_2}),\lambda_{min}(\bSigma_{\overline{\bh}_1-\bh'_2})), \max(\lambda_{max}(\bSigma_{\bh'_1}),\lambda_{max}(\bSigma_{\bh'_2}),\lambda_{max}(\bSigma_{\overline{\bh}_1-\bh'_2}))\leq \kappa_2.
	\end{array}
\end{equation*}

Define $\tilde{J}_3(s):=\tilde{J}_2(\dfrac{s}{\|\btheta^{(m)}\|^2})$ with
\begin{equation*}
	\begin{array}{ll}
		\tilde{J}_3(s)&=-s+\tilde{J}_{2,1}(\dfrac{s}{\|\btheta^{(m)}\|^2})+\cdots+\tilde{J}_{2,9}(\dfrac{s}{\|\btheta^{(m)}\|^2}),\\
		&:=-s+\tilde{J}_{3,1}(s)+\cdots+\tilde{J}_{3,9}(s).
	\end{array}
\end{equation*}
Next, we aim to prove
\begin{equation*}
	\tilde{J}_3(s)\Rightarrow Z(s), \text{for ~any~} s\leq C_u.
\end{equation*}
To this end, we need two steps.

{\bf{Step 1}}: In this step, we aim to prove that 
\begin{equation*}
	\sup_{0<s\leq C_u}|\tilde{J}_{2,4}(\dfrac{s}{\|\btheta^{(m)}\|^2})|=o_p(1),\ldots,\sup_{0<s\leq C_u}|\tilde{J}_{2,9}(\dfrac{s}{\|\btheta^{(m)}\|^2})|=o_p(1).
\end{equation*}

\textbf{Control of $\tilde{J}_{2,4}(\dfrac{s}{\|\btheta^{(m)}\|^2})$ and $\tilde{J}_{2,5}(\dfrac{s}{\|\btheta^{(m)}\|^2})$}. For $\tilde{J}_{2,4}(\dfrac{s}{\|\btheta^{(m)}\|^2})$, by the central limitting's theorem, we have 
\begin{equation*}
	\sup_{0<s\leq C_u}\tilde{J}_{2,4}(\dfrac{s}{\|\btheta^{(m)}\|^2})=O_p(\dfrac{C_u}{\|\btheta^{(m)}\|^2}\dfrac{((\btheta^{(m)})^\top	\bSigma(\bh'_2)\btheta^{(m)})^{1/2}}{\sqrt{G}})=O_p(\dfrac{1}{\sqrt{G}\|\btheta^{(m)}\|}).
\end{equation*}
By {Assumption C.1}, we have $G\|\btheta^{(m)}\|^2\gg \log(nd)$. This implies that 
\begin{equation*}
	\sup_{0<s\leq C_u}\tilde{J}_{2,4}(\dfrac{s}{\|\btheta^{(m)}\|^2})=o_p(1). 
\end{equation*}
Similarly, we can prove 
$\sup_{0<s\leq C_u}\tilde{J}_{2,5}(\dfrac{s}{\|\btheta^{(m)}\|^2})=o_p(1)$. 

\textbf{Control of $\tilde{J}_{2,6}(\dfrac{s}{\|\btheta^{(m)}\|^2})$}. Let $\underline{w}_1=G-\dfrac{C_u}{\|\btheta^{(m)}\|^2}$, $\overline{w}_1=G$,
$\underline{w}_2=1$, $\overline{w}_2=\dfrac{C_u}{\|\btheta^{(m)}\|^2}$. 
By its definition and the independence of the samples, we have
\begin{equation*}
	\begin{array}{ll}
		&\sup\limits_{0\leq s \leq C_u}|\tilde{J}_{2,6}(\dfrac{s}{\|\btheta^{(m)}\|^2})|\\&\leq _{(1)}	\dfrac{1}{G}\sup\limits_{\tau_1\in[\underline{w}_1,\overline{w}_1]}\sup\limits_{\tau_2\in[\underline{w}_2,\overline{w}_2]}\Big|\sum\limits_{t_1=1}^{\tau_1}\sum\limits_{t_2=1}^{\tau_2}{(\btheta^{(m)})^\top\bg'(\bX_{t_1},\bY_{t_2})}\Big|\\
		&\leq_{(2)} O_p(\dfrac{1}{G}\log^2(G)\sqrt{\overline{w}_1\overline{w}_2})\\
		&=_{(3)} O_p(\dfrac{1}{G}\log^2(G)\sqrt{G}\dfrac{1}{\|\btheta^{(m)}\|})\\
		&=_{(4)}O_p(\log^2(G)\dfrac{1}{\sqrt{G}\|\btheta^{(m)}\|})=_{(5)}o_p(1)
	\end{array}
\end{equation*}
where $(2)$ comes from Lemma \ref{lemma: Hoeffding's residual} and $(5)$ comes from the assumption that $G\|\btheta^{(m)}\|^2\gg \log^2(nd)$.

\textbf{Control of $\tilde{J}_{2,7}(\dfrac{s}{\|\btheta^{(m)}\|^2})$, $\tilde{J}_{2,8}(\dfrac{s}{\|\btheta^{(m)}\|^2})$, $\tilde{J}_{2,9}(\dfrac{s}{\|\btheta^{(m)}\|^2})$}. Similar to the analysis of $\tilde{J}_{2,6}(\dfrac{s}{\|\btheta^{(m)}\|^2})$. By Lemma \ref{lemma: Hoeffding's residual}, we can prove that
\begin{equation*}
	\small
	\begin{array}{ll}
		\sup\limits_{0\leq s \leq C_u}|\tilde{J}_{2,7}(\dfrac{s}{\|\btheta^{(m)}\|^2})|, \sup\limits_{0\leq s \leq C_u}|\tilde{J}_{2,8}(\dfrac{s}{\|\btheta^{(m)}\|^2})|,\sup\limits_{0\leq s \leq C_u}|\tilde{J}_{2,9}(\dfrac{s}{\|\btheta^{(m)}\|^2})|
		&=O_p(\log^2(G)\dfrac{1}{\sqrt{G}\|\btheta^{(m)}\|})=o_p(1).
	\end{array}
\end{equation*}
{\bf{Step 2}}: In this step, we aim to prove that
\begin{equation*}
	\tilde{J}_{2,1}(\dfrac{s}{\|\btheta^{(m)}\|^2})+\tilde{J}_{2,2}(\dfrac{s}{\|\btheta^{(m)}\|^2})+\tilde{J}_{2,3}(\dfrac{s}{\|\btheta^{(m)}\|^2}))\Rightarrow Z(s).
\end{equation*}
For $\tilde{J}_{3,1}(s)$, we have
\begin{equation*}
	\begin{array}{ll}
		\tilde{J}_{3,1}(s)&=-(\btheta^{(m)})^\top\sum\limits_{t_1=\gamma_m-G}^{\gamma_m-G+{s}/{\|\btheta^{(m)}\|^2}}\bh_1'(\bX_{t_1})\\
		&=\dfrac{((\btheta^{(m)})^\top	\bSigma_{\bh'_1}\btheta^{(m)})^{1/2}}{\|\btheta^{(m)}\|}\|\btheta^{(m)}\|\sum\limits_{t_1=\gamma_m-G}^{\gamma_m-G+{s}/{\|\btheta^{(m)}\|^2}}\dfrac{-(\btheta^{(m)})^\top h_1'(\bX_{t_1})}{((\btheta^{(m)})^\top	\bSigma_{\bh'_1}\btheta^{(m)})^{1/2}}.
	\end{array}
\end{equation*}
{By the functional central limitng theorem, see, e.g., Theorem 4.3.2 of \cite{whitt2002stochastic})}, we have
\begin{equation*}
	\tilde{J}_{3,1}(s)\Rightarrow \dfrac{((\btheta^{(m)})^\top	\bSigma_{\bh'_1}\btheta^{(m)})^{1/2}}{\|\btheta^{(m)}\|}W_{1}(s), ~\text{for}~s>0,
\end{equation*}
where $\{W_{1}(s),s>0\}$ is the standard Browian motion. Using  a similar analysis, we can also prove that 
\begin{equation*}
	\tilde{J}_{3,2}(s)\Rightarrow \dfrac{((\btheta^{(m)})^\top	\bSigma_{\overline{\bh}_1-\bh_2'}\btheta^{(m)})^{1/2}}{\|\btheta^{(m)}\|}W_{2}(s),~~\tilde{J}_{3,3}(s)\Rightarrow \dfrac{((\btheta^{(m)})^\top	\bSigma_{\bh_2'}\btheta^{(m)})^{1/2}}{\|\btheta^{(m)}\|}W_{3}(s)
\end{equation*}
where $\{W_{2}(s),s>0\}$ and $\{W_{3}(s),s>0\}$  are the standard Browian motion on $[0,\infty)$. Note that $\tilde{J}_{3,1}(s)$, $\tilde{J}_{3,2}(s)$, $\tilde{J}_{3,3}(s)$ are independent. Hence, $\{W_{1}(s),s>0\}$, $\{W_{2}(s),s>0\}$  and $\{W_{3}(s),s>0\}$ are independent.

Lastly, combining Steps 1 and 2, for $s>0$, we have proved
\begin{equation*}
	\begin{array}{ll}
		\tilde{J}_3(s):=\tilde{J}_2(\dfrac{s}{\|\btheta^{(m)}\|^2})\\
		\Rightarrow -s+\dfrac{((\btheta^{(m)})^\top	\bSigma_{\bh'_1}\btheta^{(m)})^{1/2}}{\|\btheta^{(m)}\|}W_{1}(s)+ \dfrac{((\btheta^{(m)})^\top	\bSigma_{\overline{\bh}_1-\bh_2'}\btheta^{(m)})^{1/2}}{\|\btheta^{(m)}\|}W_{2}(s)+\dfrac{((\btheta^{(m)})^\top	\bSigma_{\bh_2'}\btheta^{(m)})^{1/2}}{\|\btheta^{(m)}\|}W_{3}(s).
	\end{array}
\end{equation*}

Next, we consider the case for $s<0$. Note that in this case, the analysis involves the data before the change point $\gamma_m$. According to our considered model, for $t\in[\gamma_m-2G,\gamma_m]$, we have
\begin{equation*}
	\bX_t=\bmu^{ (m)}\mathbf{1}\{\gamma_m-2G< t\leq \gamma_m\}+\bepsilon_t.
\end{equation*}
Based on the Hoeffding's decomposition, we have
\begin{equation} \label{equ: hoef1 H0 before}
	h(x,y)=0 +\underline{h}_{1,j}(x)+\underline{h}_{2,j}(y)+\underline{g}_{j}(x,y),
\end{equation}
with 
\begin{equation}
	\begin{array}{ll}
		0 =\E h(X_{j}',Y_{j}'), &\underline{{h}}_{1,j}(x)=\E h(x,Y_{j}'),\\
		\underline{h}_{2,j}(y)=\E h(X_{j}',y)=-\underline{h}_{1,j}(y),& \underline{g}_{j}(x,y) = h(x,y)-\overline{h}_{1,j}(x)-\overline{h}_{2,j}(y),
	\end{array}
\end{equation}
where $X_{j}'$ and $Y_{j}'$ are two independent random variables with the same marginal distributions as  $X_{\gamma_m,j}$. Note that we use the notation $\underline{{h}}_{1,j}(x)$ to denote that we consider the Hoeffding's decomposition for the data before the change point. Define  $\underline{\bh}_1(\bx)=(\underline{h}_{1,1}(x_1),\cdots,\underline{h}_{1,s_{m}}(x_{s_m}))\in\RR^{s_m}$, $\bh'_1(\bx)=(h'_{1,1}(x_1),\cdots,h'_{1,s_m}(x_{s_m}))\in\RR^{s_m}$ and let 
\begin{equation*}
	\bSigma_{\underline{\bh}_1+\bh'_1}:=\text{Cov}(\underline{\bh}_1(\bX_{\gamma_{m}})+\bh'_1(\bX_{\gamma_{m}}))=\text{Cov}(\underline{\bh}_1(\bmu^{(m)}+\bepsilon)+\bh'_1(\bmu^{(m)}+\bepsilon)).
\end{equation*}
Using the above notation,  for $s<0$,  we can prove that 
\begin{equation*}
	\begin{array}{ll}
		\tilde{J}_3(s):=\tilde{J}_2(\dfrac{s}{\|\btheta^{(m)}\|^2})\\
		\Rightarrow s+\dfrac{((\btheta^{(m)})^\top	\bSigma_{\bh'_1}\btheta^{(m)})^{1/2}}{\|\btheta^{(m)}\|}W'_{1}(s)+ \dfrac{((\btheta^{(m)})^\top	\bSigma_{\underline{\bh}_1+\bh_1'}\btheta^{(m)})^{1/2}}{\|\btheta^{(m)}\|}W'_{2}(s)+\dfrac{((\btheta^{(m)})^\top	\bSigma_{\bh_2'}\btheta^{(m)})^{1/2}}{\|\btheta^{(m)}\|}W'_{3}(s).
	\end{array}
\end{equation*}
where  $\{W'_{1}(s),s<0\}$, $\{W'_{2}(s),s<0\}$, $\{W'_{3}(s),s<0\}$  are the standard Browian motion on $(-\infty,0]$ which are independent with each other. 

Lastly, replacing $\btheta^{(m)}$ by $\hat{\btheta}^{(m)}$ and by Theorem \ref{theory: initial estimators}, we finish the proof of Theorem \ref{theorem: refined estimation}.

\newpage
\section{Proof of useful lemmas}\label{section: proof of useful lemmas}

\subsection{Proof of Lemma \ref{lemma: concentration for maxsimum sub-exponential1}}
Since $X_t$ are sub-exponential distributed, by the Bernstein’s inequality, for each fixed $1\leq k\leq n$, there exists some universal constant $C_1>0$ such that 
\begin{equation*}
	\P\Big(\Big|\dfrac{1}{\tau_2(k)-\tau_1(k)}\sum_{t=\tau_1(k)}^{\tau_2(k)}X_t\Big|>x\Big)\leq 2\exp\big(-C_1(\tau_2(k)-\tau_1(k))\min (x^2,x)\big),
\end{equation*}
this implies that 
\begin{equation*}
	\P\Big(\Big|\sum_{t=\tau_1(k)}^{\tau_2(k)}X_t\Big|>x\Big)\leq 2\exp\Big(-C_1\min\Big( \dfrac{x^2}{\tau_2(k)-\tau_1(k)},x\Big)\Big).
\end{equation*}
Let $C_2>0$ be some big enough constant and let 
\begin{equation*}
	x_1=C_2	(\tau_2(k)-\tau_1(k))\sqrt{\log(nd)},~~x_2=C_2\log(nd), ~~x_3=\max(x_1,x_2).
\end{equation*}
Then, by choosing a large enough constant $C_2>0$, for fixed $k$, we have
\begin{equation*}
	\P\Big(\Big|\sum_{t=\tau_1(k)}^{\tau_2(k)}X_t\Big|>x_3\Big)\leq 2\exp\Big(-C_1C_2\log(nd)\Big).
\end{equation*}
Moreover, we can further choose $C_2$ big enough such that $C_3=C_1C_2>4$. Hence, for each fixed $k$, we derive that
\begin{equation*}
	\P\Big(\Big|\sum_{t=\tau_1(k)}^{\tau_2(k)}X_t\Big|>x_3\Big)\leq 2\exp\Big(-C_3\log(nd)\Big).
\end{equation*}
Lastly, we take the union bound over $k=1,\ldots,n$, which finishes the proof.

\subsection{Proof of Lemma \ref{lemma: concentration for maxsimum sub-exponential}}
Similar to the proof of Lemma \ref{lemma: concentration for maxsimum sub-exponential1}, for each fixed $k$, we have
\begin{equation*}
	\begin{array}{ll}
		\P\Big(\Big|\dfrac{1}{\max_{k}(\tau_2(k)-\tau_1(k))}\sum\limits_{t=\tau_1(k)}^{\tau_2(k)}X_t\Big|>x\Big)
		&\leq \P\Big(\Big|\dfrac{1}{\tau_2(k)-\tau_1(k)}\sum\limits_{t=\tau_1(k)}^{\tau_2(k)}X_t\Big|>x\Big)\\
		&\leq 2\exp\big(-C_1(\tau_2(k)-\tau_1(k))\min (x^2,x)\big),\\
		&\leq 2\exp\big(-C_1\min_{k}(\tau_2(k)-\tau_1(k))\min (x^2,x)\big).
	\end{array}
\end{equation*}
Let $C_2>0$ be some big enough constant and let 
\begin{equation*}
	x_1=\sqrt{C_2	\dfrac{\log(nd)}{\min_{k}(\tau_2(k)-\tau_1(k))}},~~x_2=C_2\dfrac{\log(nd)}{\min_{k}(\tau_2(k)-\tau_1(k))}, ~~x_3=\max(x_1,x_2).
\end{equation*}
Then, by choosing a large enough constant $C_2>0$, for fixed $k$, we have
\begin{equation*}
	\P\Big(\Big|\dfrac{1}{\max_{k}(\tau_2(k)-\tau_1(k))}\sum\limits_{t=\tau_1(k)}^{\tau_2(k)}X_t\Big|>x_3\Big)\leq 2\exp\Big(-C_1C_2\log(nd)\Big).
\end{equation*}
Moreover, we can further choose $C_2$ big enough such that $C_3=C_1C_2>4$. Hence, for each fixed $k$, we derive that
\begin{equation*}
	\P\Big(\Big|\sum\limits_{t=\tau_1(k)}^{\tau_2(k)}X_t\Big|>{\max_{k}(\tau_2(k)-\tau_1(k))}x_3\Big)\leq 2\exp\Big(-C_3\log(nd)\Big).
\end{equation*}
Take the union bound over $k=1,\ldots,n$, we further have
\begin{equation*}
	\P\Big(\Big|\sum\limits_{t=\tau_1(k)}^{\tau_2(k)}X_t\Big|>{\max_{k}(\tau_2(k)-\tau_1(k))}x_3,k=1,\ldots,n\Big)\leq 2\exp\Big(-C_3'\log(nd)\Big).
\end{equation*}
for some $C_3'>3$, which finishes the proof.
\subsection{Proof of Lemma \ref{lemma: Hoeffding's residual}}\label{sec: proof of lemma of the hoeffding's decomposition}
\begin{proof}
	In this section, we prove Lemma \ref{lemma: Hoeffding's residual}. For each fixed $\underline{k}\leq k\leq \overline{k}$, 
	let 
	\begin{equation*}
		\begin{array}{ll}
			w_1(k)=\overline{\tau}_1(k)-\underline{\tau}_1(k),~~ \underline{w}_1=\min_{\underline{k}\leq k\leq \overline{k}}w_1(k),~~\overline{w}_1=\max_{\underline{k}\leq k\leq \overline{k}}w_1(k)\\
			w_2(k)=\overline{\tau}_2(k)-\underline{\tau}_2(k),~~ \underline{w}_2=\min_{\underline{k}\leq k\leq \overline{k}}w_2(k),~~\overline{w}_2=\max_{\underline{k}\leq k\leq \overline{k}}w_2(k).\\
		\end{array}
	\end{equation*}
	Moreover, for each fixed $k$, let 
	\begin{equation}\label{equ: V1(k)+V2(k)}
		V_1{(k)}=\dfrac{1}{w_1(k)}\sum_{t_1=\underline{\tau}_1(k)}^{\overline{\tau}_1(k)} h_1(X_{t_1})^2, 	V_2{(k)}=\dfrac{1}{w_2(k)}\sum_{t_2=\underline{\tau}_2(k)}^{\overline{\tau}_2(k)} h_2(Y_{t_2})^2,
	\end{equation}
	By Proposition B.2 in \cite{chang2016}, there are some constants such that with probability at least $1-C_1\exp(-C_2y)$, 
	\[
	\left| \sum_{t_1=\underline{\tau}_1(k)}^{\overline{\tau}_1(k)} \sum_{t_2=\underline{\tau}_2(k)}^{\overline{\tau}_2(k)} g(X_{t_1}, Y_{t_2}) \right|
	\lesssim C_3y  [w_1(k)]^{1/2} [w_2(k)]^{1/2}
	\sqrt{\sigma^2 + V_1(k) + V_2(k)}.
	\]
	Hence, for any fixed $\underline{k}\leq k \leq \overline{k}$, with probability smaller than $C_1\exp(-C_2y)$,  we have
	\begin{equation}\label{equ: upper bound of the residual1}
		\left| \sum_{t_1=\underline{\tau}_1(k)}^{\overline{\tau}_1(k)} \sum_{t_2=\underline{\tau}_2(k)}^{\overline{\tau}_2(k)} g(X_{t_1}, Y_{t_2}) \right|
		\lesssim C_3y  [w_1(k)]^{1/2} [w_2(k)]^{1/2}
		\sqrt{\sigma^2 + \max_{k}V_1(k) + \max_{k}V_2(k)}.
	\end{equation}
	Taking $y=C_4\log (nd)$ in (\ref{equ: upper bound of the residual1}) with some big enough constant $C_4>0$, taking the union bound over $k$, we have with probability smaller than $C_1(nd)^{-C_2}$ for some $C_2>2$,  we have
	\begin{equation}\label{equ: upper bound of the residual2}
		\begin{array}{ll}
			\P\Big(\max\limits_{\underline{k}\leq k \leq \overline{k}}\left| \sum\limits_{t_1=\underline{\tau}_1(k)}^{\overline{\tau}_1(k)} \sum\limits_{t_2=\underline{\tau}_2(k)}^{\overline{\tau}_2(k)} g(X_{t_1}, Y_{t_2}) \right|\\
			\lesssim C_5
			\log(nd)  [w_1(k)]^{1/2} [w_2(k)]^{1/2}
			\sqrt{\sigma^2 + \max_{k}V_1(k) + \max_{k}V_2(k)}\Big)\\
			\leq C_1(nd)^{-C_2}
		\end{array}
	\end{equation}
	Note that by {Assumption A.3 and the Jensen's inequality}, we have
	\begin{equation*}
		\P(h_1(X_t)^2\geq y)\leq K_1\exp(-K_2y^{1/2}),~~\P(h_2(Y_t)^2\geq y)\leq K_1\exp(-K_2y^{1/2}).
	\end{equation*}
	By Theorem 6 in \cite{delaigle2011robustness}, we have
	\begin{equation*}
		\P\Big(\dfrac{1}{w_1(k)}\sum_{t_1=\underline{\tau}_1(k)}^{\overline{\tau}_1(k)} (h_1(X_{t_1})^2-\E[h_1(X_{t_1})^2])\geq y\Big)\leq K_1\exp(-K_2(w_1(k)y)^{1/2})\leq K_1\exp(-K_2(\underline{w}_1y)^{1/2}),
	\end{equation*}
	where the last inequality uses the fact that $\underline{w}_1=\min_{\underline{k}\leq k\leq \overline{k}}\overline{\tau}_1(k)-\underline{\tau}_1(k)$. 
	Moreover, using the assumption that $[\underline{k},\overline{k}]\subset\{1,\ldots,n\}$, and by the union of the events over $k$, we have
	\begin{equation*}
		\P(\max_{\underline{k}\leq k\leq \overline{k}}\Big|\dfrac{1}{w_1(k)}\sum_{t_1=\underline{\tau}_1(k)}^{\overline{\tau}_1(k)} (h_1(X_{t_1})^2-\E[h_1(X_{t_1})^2])\Big|\geq y)\leq K_1n\exp(-K_2(\underline{w}_1y)^{1/2}).
	\end{equation*}
	Hence, taking $y=K_3\log^2(nd)/\underline{w}_1$ for some big enough constant $K_3>0$ we have 
	\begin{equation*}
		\P\Big(\max_{\underline{k}\leq k\leq \overline{k}}\Big|\dfrac{1}{w_1(k)}\sum_{t_1=\underline{\tau}_1(k)}^{\overline{\tau}_1(k)} (h_1(X_{t_1})^2-\E[h_1(X_{t_1})^2])\Big|\geq K_3\dfrac{\log^2(nd)}{\underline{w}_1}\Big)\leq K_4(nd)^{-K_5},
	\end{equation*}
	for some $K_5>2$. Similarly, we can prove that 
	\begin{equation*}
		\P\Big(\max_{\underline{k}\leq k\leq \overline{k}}\Big|\dfrac{1}{w_2(k)}\sum_{t_2=\underline{\tau}_2(k)}^{\overline{\tau}_2(k)} (h_2(Y_{t_2})^2-\E[h_2(Y_{t_2})^2])\Big|\geq K_3\dfrac{\log^2(nd)}{\underline{w}_2}\Big)\leq K_4(nd)^{-K_5}.
	\end{equation*}
	Define the two events:
	\begin{equation*}
		\begin{array}{ll}
			\cE_1=\Big\{\max_{\underline{k}\leq k\leq \overline{k}}  V_{1}(k)\leq \underbrace{K_3\dfrac{\log^2(nd)}{\underline{w}_1}+  \max_{\underline{k}\leq k\leq \overline{k}}\dfrac{1}{w_1(k)}\sum\limits_{t_1=\underline{\tau}_1(k)}^{\overline{\tau}_1(k)} \E[h_1(X_{t_1})^2]}_{M_1}\Big\},\\
			\cE_2=\Big\{\max_{\underline{k}\leq k\leq \overline{k}}  V_{2}(k)\leq \underbrace{K_3\dfrac{\log^2(nd)}{\underline{w}_2}+  \max_{\underline{k}\leq k\leq \overline{k}}\dfrac{1}{w_2(k)}\sum\limits_{t_2=\underline{\tau}_2(k)}^{\overline{\tau}_2(k)} \E[h_2(Y_{t_2})^2]}_{M_2} \Big\}\\
		\end{array}
	\end{equation*}
	By the above results, we have 
	\begin{equation}\label{equ: upper bound of the residual3}
		\P(\cE_1,\cE_2)\geq 1-C_6(nd)^{-C_7}, ~~\text{for}~~ C_7>2. 
	\end{equation}
	Combining (\ref{equ: upper bound of the residual2}) and (\ref{equ: upper bound of the residual3}), we have
	\begin{equation}
		\small
		\begin{array}{ll}
			\P\Big(\max\limits_{\underline{k}\leq k \leq \overline{k}}\left| \sum\limits_{t_1=\underline{\tau}_1(k)}^{\overline{\tau}_1(k)} \sum\limits_{t_2=\underline{\tau}_2(k)}^{\overline{\tau}_2(k)} g(X_{t_1}, Y_{t_2}) \right|
			\geq C_5
			\log(nd)  [\overline{w}_1]^{1/2} [\overline{w}_2]^{1/2}
			\sqrt{\sigma^2 + M_1+M_2}\Big)\\
			\leq	\P\Big(\max\limits_{\underline{k}\leq k \leq \overline{k}}\left| \sum\limits_{t_1=\underline{\tau}_1(k)}^{\overline{\tau}_1(k)} \sum\limits_{t_2=\underline{\tau}_2(k)}^{\overline{\tau}_2(k)} g(X_{t_1}, Y_{t_2}) \right|
			\geq C_5
			\log(nd)  [w_1(k)]^{1/2} [w_2(k)]^{1/2}
			\sqrt{\sigma^2 + M_1+M_2},\cE_1,\cE_2\Big)\\+\P((\cE_1)^c)+\P((\cE_2)^c)\\
			\leq	\P\Big(\max\limits_{\underline{k}\leq k \leq \overline{k}}\left| \sum\limits_{t_1=\underline{\tau}_1(k)}^{\overline{\tau}_1(k)} \sum\limits_{t_2=\underline{\tau}_2(k)}^{\overline{\tau}_2(k)} g(X_{t_1}, Y_{t_2}) \right|
			\geq C_5
			\log(nd)  [w_1(k)]^{1/2} [w_2(k)]^{1/2}
			\sqrt{\sigma^2 + \max_{k}V_1(k) + \max_{k}V_2(k)}\Big)\\
			+\P((\cE_1)^c)+\P((\cE_2)^c)\\
			\leq C_7 (nd)^{-C_8}~~\text{for~some}~ C_8>2.
		\end{array}
	\end{equation}
	Lastly, using the fact that $\sqrt{a+b+c}\leq \sqrt{a}+\sqrt{b}+\sqrt{c}$, and $\E[h_1(X_t)^2]\leq \E h(X_{t_1},Y_{t_2})^2:=\sigma^2$, for some large enough constant $C_9>0$, we have
	\begin{equation*}
		\begin{array}{ll}
			&\sqrt{\sigma^2+M_1+M_2}\\&\leq \sqrt{\sigma^2}+\sqrt{K_3\dfrac{\log^2(nd)}{\underline{w}_1}+\sigma^2}\sqrt{K_3\dfrac{\log^2(nd)}{\underline{w}_2}+\sigma^2}\\
			&\leq C_9 \Big(\sqrt{\sigma^2}+\sqrt{\dfrac{\log^2(nd)}{\underline{w}_1}}+\sqrt{\dfrac{\log^2(nd)}{\underline{w}_2}}\Big)
		\end{array}
	\end{equation*}
	which finishes the proof.
\end{proof}

\section{Proof of lemmas used in Section \ref{section: proof of main results}}\label{sec: proof of lemmas used in the main proof}

\subsection{Proof of Lemma \ref{lemma: bootstraped negligible}}\label{section: proof of bootstraped negligible}
In this section, we prove that the residual term of the Hoeffding's decomposition can be uniformly negligible in the sense that 
\begin{equation*}
	\begin{array}{ll}
		\max_{G\leq k\leq n-G} \big\|\bT^{(2)}(k)\big\|_{\infty}&=\max\limits_{G\leq k\leq n-G}\|\dfrac{1}{G^{3/2}}\sum\limits_{t_1=k-G+1}^k\sum\limits_{t_2=k+1}^{k+G}\bg(\bX_{t_1},\bX_{t_2})\|_{\infty}\\
		&=\max\limits_{1\leq j\leq d}\max\limits_{G\leq k\leq n-G}\Big|\dfrac{1}{G^{3/2}}\sum\limits_{t_1=k-G+1}^k\sum\limits_{t_2=k+1}^{k+G}g_j(X_{t_1,j},X_{t_2,j})\Big|
	\end{array}
\end{equation*}
Let $\underline{k}=G$,  $\overline{k}=n-G$, $\underline{\tau}_1(k)=k-G+1$, $\overline{\tau}_1(k)=k$, $\underline{\tau}_2(k)=k+1$, $\overline{\tau}_2(k)=k+G$, and $\underline{w}_1=\underline{w}_2=G$.  Hence, by Lemma \ref{lemma: Hoeffding's residual}, for each fixed $j=1,\ldots,d$, there exist some big enough constant $C_1>0$ such that with probability at least $1-(nd)^{-C_2}$ with $C_2>2$, the following holds:
\begin{equation*}
	\begin{array}{ll}
		\max\limits_{G\leq k\leq n-G}\Big|\dfrac{1}{G^{3/2}}\sum\limits_{t_1=k-G+1}^k\sum\limits_{t_2=k+1}^{k+G}g_j(X_{t_1,j},X_{t_2,j})\Big|\\\leq C_1\dfrac{\log(nd)}{\sqrt{G}}  \Big(\sqrt{\sigma_{j,j}^2}+\sqrt{\dfrac{\log^2(nd)}{G}}+\sqrt{\dfrac{\log^2(nd)}{G}}\Big)
		\leq C_1 \dfrac{\log(nd)}{\sqrt{G}}  
	\end{array}
\end{equation*}
where $\sigma_{j,j}^2=\text{Var}(h(X_{1,j},X_{2,j}))$ and the last inequality uses the assumption that $G\gg \log^2(nd)$. Lastly, taking the union bound over $j=1,\ldots,d$ and taking $\epsilon=C_1 \dfrac{\log(nd)}{\sqrt{G}} $ yields the results. 
\subsection{Proof of Lemma \ref{lemma: gaussian approxiation 1}}\label{section: proof of gaussian approximation 1}
In this section, we prove Lemma \ref{lemma: gaussian approxiation 1}. Recall
\begin{equation*}
	a_{t}(k)=\dfrac{\sqrt{n}}{\sqrt{G}}(\mathbf{1}\{k-G+1\leq t\leq k\}-\mathbf{1}\{k+1\leq t\leq k+G\}).
\end{equation*}
Let $\bZ_{t}(k)=a_t(k)\bh_1(\bX_{t})\in \RR^{d}$ for $G\leq k\leq n-G$. Moreover, for each $t$, define the $d(n-2G+1)$-dimensional random vector $\bZ_t=(Z_{t}(k,j))_{k,j}$ with $Z_{t}(k,j)=a_t(k)h_{1,j}(X_{t,j})$ for $G\leq k\leq n-G$ and $1\leq j\leq d$. Hence, by the above definition, 
\begin{equation*}
	\max_{G\leq k\leq n-G}\|\bT^{(1)}(k)\|_{\infty}=\max_{G\leq k\leq n-G,\atop 1\leq j\leq d}\Big|\dfrac{1}{\sqrt{n}}\sum\limits_{t=1}^nZ_{t}(k,j)\Big|=\Big\|\dfrac{1}{\sqrt{n}}\sum\limits_{t=1}^n\bZ_{t}\Big\|_{\infty}.
\end{equation*}
Using similar notations, for each $t$, define the $d(n-2G+1)$-dimensional Gaussian distributed random vector $\bH_t=(H_{t}(k,j))$ with $H_{t}(k,j)=a_t(k)G_{t,j}$ for $G\leq k\leq n-G$ and $1\leq j\leq d$. Hence, by the above definition, we have
\begin{equation*}
	\max_{G\leq k\leq n-G}\|\bT^{\bG}(k)\|_{\infty}=\max_{G\leq k\leq n-G,\atop 1\leq j\leq d}\Big|\dfrac{1}{\sqrt{n}}\sum\limits_{t=1}^nH_{t}(k,j)\Big|=\Big\|\dfrac{1}{\sqrt{n}}\sum\limits_{t=1}^n\bH_{t}\Big\|_{\infty}.
\end{equation*}
It is easy to see that $\bH_t$ has the same covariance of $\bZ_t$. To prove Lemma \ref{lemma: gaussian approxiation 1}, it remains to verify that the Conditions M.1, M.2 and E.1 of Proposition 2.1 hold in \cite{chernozhukov2017central}. Note that by {Assumptions A.2 and A.3},  By Jensen’s inequality, for $t=1,\ldots,n$ and $j=1,\ldots,d$, we have
\begin{equation*}
	\E |h_{1,j}(X_{t,j})|^{2+\ell}\leq D^\ell, \ell=1,2;~~\|h_{1,j}(X_{t,j})\|_{\psi_1}\leq D
\end{equation*}
{By Assumption A.1}, we have
\begin{equation*}
	\begin{array}{ll}
		\dfrac{1}{n}\sum\limits_{t=1}^{n} \E[Z_{t}(k,j)]^2&=_{(1)}	\dfrac{1}{n}\sum\limits_{t=1}^{n}a_{t}(k)^2\E|h_{1,j}(X_{t,j})|^2\\
		&=_{(2)}\dfrac{1}{G}\sum\limits_{t=k-G+1}^k\E|h_{1,j}(X_{t,j})|^2+\dfrac{1}{G}\sum\limits_{t=k+1}^{k+G}\E|h_{1,j}(X_{t,j})|^2\\
		&\geq_{(3)} b^2
	\end{array}
\end{equation*}
where $(2)$ is due to the definition of $a_{t}(k)$ and $(3)$ comes from {Assumption A.1}. Hence, Condition M.1 hold. Moreover, using similar analysis, it holds that
\begin{equation*}
	\begin{array}{ll}
		\dfrac{1}{n}\sum\limits_{t=1}^{n} \E[Z_{t}(k,j)]^{2+\ell}&=_{(1)}	\dfrac{1}{n}\sum\limits_{t=1}^{n}a_{t}(k)^{2+\ell}\E|h_{1,j}(X_{t,j})|^{2+\ell}\\
		&\leq_{(2)}\dfrac{1}{G}\sum\limits_{t=k-G+1}^k\E|h_{1,j}(X_{t,j})|^{2+\ell}+\dfrac{1}{G}\sum\limits_{t=k+1}^{k+G}\E|h_{1,j}(X_{t,j})|^{2+\ell}\\
		&\leq_{(3)} D^\ell
	\end{array}
\end{equation*}
Hence, Condition M.2 hold.  Lastly, by Assumption A.2 and the definition of $a_{t}(k)$, we have
\begin{equation*}
	\|Z_{t}(k,j)\|_{\psi_1}=\|a_t(k)h_{1,j}(X_{t,j})\|_{\psi_1}\leq \dfrac{\sqrt{n}}{\sqrt{G}}|h_{1,j}(X_{t,j})\|_{\psi_1}\leq \dfrac{\sqrt{n}}{\sqrt{G}}D.
\end{equation*}
Taking $B_n=\dfrac{\sqrt{n}}{\sqrt{G}}D$ in Proposition 2.1 of \cite{chernozhukov2017central}, we finish the proof. 

\subsection{Proof of Lemma \ref{lemma: large deviation}}\label{section: proof of  large deviation}
In this section, we prove Lemma \ref{lemma: large deviation}. Firstly, for each $G\leq k\leq n-G$ and $1\leq j\leq d$, we have
\begin{equation}\label{equ: TG of coordinate}
	\begin{array}{ll}	T_j^{\bG}(k)&=\dfrac{1}{\sqrt{G}}\sum\limits_{t_1=k-G+1}^{k}G_{t_1,j}-\dfrac{1}{\sqrt{G}}\sum\limits_{t_2=k+1}^{k+G}G_{t_2,j}.
	\end{array}
\end{equation}
and 
\begin{equation*}
	T_{j}^b(k)=	\dfrac{1}{G^{3/2}}\sum_{t_1=k-G+1}^{k}e_{t_1}\Big[\sum_{t_2=k+1}^{k+G}h(X_{t_1,j},X_{t_2,j})\Big]+\dfrac{1}{G^{3/2}}\sum_{t_2=k+1}^{k+G}e_{t_2}\Big[\sum_{t_1=k-G+1}^{k}h(X_{t_1,j},X_{t_2,j})\Big].
\end{equation*} 
Without loss of generality, we consider $k_1\leq k_2$. According to the locations of $k_1,k_2$, we consider four cases. 

\textbf{\text{Case 1:} $k_2-k_1>2G-1$.}  In this case, it is easy to see that
\begin{equation*}
	\text{Cov}(T_{i}^{\bG}(k_1),T_{j}^{\bG}(k_2))=\text{Cov}(T_{i}^{b}(k_1),T_{j}^{b}(k_2)))=0.
\end{equation*}

\textbf{\text{Case 2:}} $G-1<k_2-k_1\leq 2G-1$.  In this case, we can prove that 
\begin{equation*}
	\begin{array}{ll}
		&\text{Cov}(T_{j_1}^{\bG}(k_1),T_{j_2}^{\bG}(k_2))\\
		&=	\text{Cov}\Big(-\dfrac{1}{\sqrt{G}}\sum\limits_{t_2=k_1+1}^{k+G}G_{t_2,j_1},\dfrac{1}{\sqrt{G}}\sum\limits_{t_1=k_2-G+1}^{k_2}G_{t_1,j_2}\Big)\\
		&=\text{Cov}\Big(-\dfrac{1}{\sqrt{G}}\sum\limits_{t_2=k_2-G+1}^{k_1+G}G_{t_2,j_1},\dfrac{1}{\sqrt{G}}\sum\limits_{t_1=k_2-G+1}^{k_1+G}G_{t_1,j_2}\Big)\\
		&=-\dfrac{2G-(k_2-k_1)}{G}\text{Cov}(G_{t,j_1},G_{t,j_2})\\
		&=-\dfrac{2G-(k_2-k_1)}{G}\gamma_{j_1,j_2}.
	\end{array}
\end{equation*}
Similarly, for the bootstrap based testing statistic, conditional on $\cX$, we have
\begin{equation*}
	\begin{array}{ll}
		&	\text{Cov}(T_{j_1}^{b}(k_1),T_{j_2}^{b}(k_2))\\
		&=\text{Cov}\Big(\dfrac{1}{G^{3/2}}\sum\limits_{t_2=k_2-G+1}^{k_1+G}e_{t_2}\sum\limits_{t_1=k_1-G+1}^{k_1}h(X_{t_1,j_1},X_{t_2,j_1}),\dfrac{1}{G^{3/2}}\sum\limits_{t_1=k_2-G+1}^{k_1+G}e_{t_1}\sum\limits_{t_2=k_2+1}^{k_2+G}h(X_{t_1,j},X_{t_2,j})\Big)\\
		&=\dfrac{1}{G^{3}}\sum\limits_{t_1=k_2-G+1}^{k_1+G}\sum\limits_{t_2=k_1-G+1}^{k_1}\sum\limits_{t_3=k_2+1}^{k_2+G}h(X_{t_1,j_1},X_{t_2,j_1})h(X_{t_1,j_2},X_{t_3,j_2}).
	\end{array}
\end{equation*}
Note that for each fixed $k_1$ and $k_2$, there is no overlap of $t_1,t_2,t_3$.  Hence, it holds that
\begin{equation*}
	\E\Big[	\text{Cov}(T_{j_1}^{b}(k_1),T_{j_2}^{b}(k_2))\Big]=\text{Cov}(T_{j_1}^{\bG}(k_1),T_{j_2}^{\bG}(k_2))
\end{equation*}
Moreover, let $T^1=\{t: k_2-G+1\leq t\leq k_1+G\}$, $T^2=\{t: k_1-G+1\leq t\leq k_1\}$ and $T^3=\{t: k_2+1\leq t\leq k_2+G\}$. Let $n(k_1,k_2):=|T^1|+|T^2|+|T^3|=4G-|k_2-k_1|$. 
Recall 
\begin{equation*}
	\begin{array}{ll}
		\bh(\bX_{t_1},\bX_{t_2}):=(h(X_{t_1,1},X_{t_2,1}),\ldots,h(X_{t_1,d},X_{t_2,d}))^\top,\\
		\bh_1(\bX_{t}):=(h_{1,1}(X_{t,1}),\ldots,h_{1,d}(X_{t,d}))^\top,\\
		\bg(\bX_{t_1},\bX_{t_2}):=(g_1(X_{t_1,1},X_{t_2,1}),\ldots,g_{d}(X_{t_1,d},X_{t_2,d}))^\top.\\
	\end{array}
\end{equation*}
and $\bGamma=(\gamma_{i,j})\in \RR^{d\times d}=\text{Cov}(\bh(\bX_1))=\E \bh(\bX_1,\bX_2)\bh(\bX_1,\bX_3)^\top$. Let $H(\bx_1,\bx_2,\bx_3)=\bh(\bx_1,\bx_2)\bh(\bx_1,\bx_3)^\top\in \RR^{d\times d}$ and define the new kernel $H'(\bX_{t_{1}}, \bX_{t_2}, \bX_{t_3}) = \sum\limits_{\pi_3} \tilde{H}(\bX_{\pi_3(1)}, \bX_{\pi_3(2)}, \bX_{\pi_3(3)})$, where
\begin{equation*}
	\tilde{H}(\bX_{t_1},\bX_{t_2},\bX_{t_3})=\left\{
	\begin{array}{cc}
		H(\bX_{t_1},\bX_{t_2},\bX_{t_3})),& \text{if}~ t_1\in T^1, t_2\in T^2, t_3\in T^3\\
		0,& \text{else}.
	\end{array}
	\right.
\end{equation*}
and  $\pi_3$ \text{ is a permutation of } $t_1,t_2,t_3$. Using the above notations, we have
\begin{equation*}
	\begin{array}{ll}
		\text{Cov}(\bT^{b}(k_1),\bT^{b}(k_2))&=\dfrac{1}{G^{3}}\sum\limits_{t_1=k_2-G+1}^{k_1+G}\sum\limits_{t_2=k_1-G+1}^{k_1}\sum\limits_{t_3=k_2+1}^{k_2+G}\bh(\bX_{t_1},\bX_{t_2})\bh(\bX_{t_1},\bX_{t_3})\\
		&=\dfrac{1}{G^{3}}\sum\limits_{t_1\neq t_2 \neq t_3\in T^1\cup T^2\cup T^3}\tilde{H}(\bX_{t_1},\bX_{t_2},\bX_{t_3})\\
		&=\dfrac{1}{6G^{3}}\sum\limits_{t_1\neq t_2\neq t_3\in T^1\cup T^2\cup T^3}{H'}(\bX_{t_1},\bX_{t_2},\bX_{t_3}).\\
	\end{array}
\end{equation*}
Note that by the above notations, we see that $\text{Cov}(\bT^{b}(k_1),\bT^{b}(k_2))$ is a U-statistics of order 3. Moreover, for each fixed $G-1<k_2-k_1\leq 2G-1$, let 
\begin{equation}\label{equ: relation of W and cov}
	\begin{array}{ll}
		\bW(k_1,k_2):=	\dfrac{(n(k_1,k_2)-3)!}{n(k_1,k_2)!}\sum\limits_{t_1\neq t_2\neq t_3\in T^1\cup T^2\cup T^3}{H'}(\bX_{t_1},\bX_{t_2},\bX_{t_3})\\=\dfrac{6G^3(n(k_1,k_2)-3)!}{n(k_1,k_2)!}\text{Cov}(\bT^{b}(k_1),\bT^{b}(k_2)).
	\end{array}
\end{equation}
By Lemma E.1 in \cite{chen2018gaussian} and similar to the proof of Lemma A.1 in \cite{yu2022robust}, we can prove with probability at least $1-(nd)^{-C_0}$ for some $C_0>2$, it holds that 
\begin{equation*}
	\|\bW(k_1,k_2)-\E \bW(k_1,k_2)\|_{\infty}\leq C_1D^2\max \Big(\dfrac{\log(nd)}{n(k_1,k_2)},\dfrac{\log^2(nd)\log^2(n(k_1,k_2)d)}{n(k_1,k_2)}\Big).
\end{equation*}
By (\ref{equ: relation of W and cov}), it holds  with probability at least $1-(nd)^{-C_0}$ for some $C_0>2$ that 
\begin{equation*}
	\begin{array}{ll}
		\Big\|\text{Cov}(\bT^{b}(k_1),\bT^{b}(k_2))-\E	\text{Cov}(\bT^{b}(k_1),\bT^{b}(k_2)))\Big\|_{\infty}\\\leq C_1	\dfrac{n(k_1,k_2)!}{6G^3(n(k_1,k_2)-3)!}D^2\max \Big(\dfrac{\log(nd)}{n(k_1,k_2)},\dfrac{\log^2(nd)\log^2(n(k_1,k_2)d)}{n(k_1,k_2)}\Big)\\
	\end{array}
\end{equation*}
Note that $n(k_1,k_2):=|T^1|+|T^2|+|T^3|=4G-|k_2-k_1|$ and $G-1<k_2-k_1\leq 2G-1$, which implies that $2G+1\leq n(k_1,k_2)\leq 3G$. Hence, with probability at least $1-(nd)^{-C_0}$ for some $C_0>2$,  it holds  uniformly for all $(k_1,k_2)$ with $G-1<k_2-k_1\leq 2G-1$ that 
\begin{equation*}
	\begin{array}{ll}
		\Big\|\text{Cov}(\bT^{b}(k_1),\bT^{b}(k_2))-\E	\text{Cov}(\bT^{b}(k_1),\bT^{b}(k_2))\Big\|_{\infty}\\
		=\Big\|\text{Cov}(\bT^{b}(k_1),\bT^{b}(k_2))-\text{Cov}(\bT^{\bG}(k_1),\bT^{\bG}(k_2)))\Big\|_{\infty}\\
		\leq C_2D^2\max \Big(\sqrt{\dfrac{\log(nd)}{G}},\dfrac{\log^2(nd)\log^2(Gd)}{G}\Big).\\
	\end{array}
\end{equation*}
for some large enough constant $C_2>0$. 

\textbf{\text{Case 3:}} $0<k_2-k_1\leq G-1$.  In this case, for any $1\leq j_1,j_2\leq d$, we can prove that 
\begin{equation*}
	\begin{array}{ll}
		\text{Cov}(T_{j_1}^{\bG}(k_1),T_{j_2}^{\bG}(k_2))\\
		=	\text{Cov}\Big(\dfrac{1}{\sqrt{G}}\sum\limits_{t=k_2-G+1}^{k_1}G_{t,j_1},\dfrac{1}{\sqrt{G}}\sum\limits_{t=k_2-G+1}^{k_1}G_{t,j_2}\Big)+\text{Cov}\Big(-\dfrac{1}{\sqrt{G}}\sum\limits_{t=k_1+1}^{k_2}G_{t,j_1},\dfrac{1}{\sqrt{G}}\sum\limits_{t=k_1+1}^{k_2}G_{t,j_2}\Big)\\
		\quad+\text{Cov}\Big(\dfrac{1}{\sqrt{G}}\sum\limits_{t=k_2+1}^{k_1+G}G_{t,j_1},\dfrac{1}{\sqrt{G}}\sum\limits_{t=k_2+1}^{k_1+G}G_{t,j_2}\Big)\\
		=\dfrac{G-(k_2-k_1)}{G}\gamma_{j_1,j_2}-\dfrac{(k_2-k_1)}{G}\gamma_{j_1,j_2}+\dfrac{G-(k_2-k_1)}{G}\gamma_{j_1,j_2}
	\end{array}
\end{equation*}
This implies
\begin{equation}\label{equ: case 3 large deviation 1}
	\begin{array}{ll}
		\text{Cov}(\bT^{\bG}(k_1),\bT^{\bG}(k_2))
		=\underbrace{\dfrac{G-(k_2-k_1)}{G}\bGamma}_{\text{Cov}_1(\bT^{\bG}(k_1),\bT^{\bG}(k_2))}\underbrace{-\dfrac{(k_2-k_1)}{G}\bGamma}_{\text{Cov}_2(\bT^{\bG}(k_1),\bT^{\bG}(k_2))}+\underbrace{\dfrac{G-(k_2-k_1)}{G}\bGamma}_{\text{Cov}_3(\bT^{\bG}(k_1),\bT^{\bG}(k_2))}
	\end{array}
\end{equation}
Similarly, for the bootstrap based testing statistics, we can prove that 
\begin{equation}\label{equ: case 3 large deviation 2}
	\begin{array}{ll}
		\text{Cov}(\bT^{b}(k_1),\bT^{b}(k_2))=\text{Cov}_1(\bT^{b}(k_1),\bT^{b}(k_2))+\text{Cov}_2(\bT^{b}(k_1),\bT^{b}(k_2))+\text{Cov}_3(\bT^{b}(k_1),\bT^{b}(k_2))\\
	\end{array}
\end{equation}
where $\text{Cov}_1(\bT^{b}(k_1),\bT^{b}(k_2))$, $\text{Cov}_2(\bT^{b}(k_1),\bT^{b}(k_2))$, and $\text{Cov}_3(\bT^{b}(k_1),\bT^{b}(k_2))$ are  as
\begin{equation*}
	\begin{array}{ll}
		&\text{Cov}_1(\bT^{b}(k_1),\bT^{b}(k_2))\\
		&=\text{Cov}\Big(\dfrac{1}{G^{3/2}}\sum\limits_{t_1=k_2-G+1}^{k_1}e_{t_1}\sum\limits_{t_2=k_1+1}^{k_1+G}\bh(\bX_{t_1},\bX_{t_2}),\dfrac{1}{G^{3/2}}\sum\limits_{t'_1=k_2-G+1}^{k_1}e_{t_1'}\sum\limits_{t_2'=k_2+1}^{k_2+G}\bh(\bX_{t'_1},\bX_{t'_2})\Big)\\
		&=\dfrac{1}{G^{3}}\sum\limits_{t_1=k_2-G+1}^{k_1}\sum\limits_{t_2=k_1+1}^{k_1+G}\sum\limits_{t_3=k_2+1}^{k_2+G}\bh(\bX_{t_1},\bX_{t_2})\bh(\bX_{t_1},\bX_{t_3}).
	\end{array}
\end{equation*}
\begin{equation*}
	\begin{array}{ll}
		\text{Cov}_2(\bT^{b}(k_1),\bT^{b}(k_2))\\=\text{Cov}\Big(\dfrac{1}{G^{3/2}}\sum\limits_{t_2=k_1+1}^{k_1+G}e_{t_2}\sum\limits_{t_1=k_1-G+1}^{k_1}\bh(\bX_{t_1},\bX_{t_2}),\dfrac{1}{G^{3/2}}\sum\limits_{t'_1=k_2-G+1}^{k_2}e_{t_1'}\sum\limits_{t_2'=k_2+1}^{k_2+G}\bh(\bX_{t'_1},\bX_{t'_2})\Big)\\
		=-\dfrac{1}{G^{3}}\sum\limits_{t_1=k_1+1}^{k_2}\sum\limits_{t_2=k_1-G+1}^{k_1}\sum\limits_{t_3=k_2+1}^{k_2+G}\bh(\bX_{t_1},\bX_{t_2})\bh(\bX_{t_1},\bX_{t_3}).
	\end{array}
\end{equation*}
and 
\begin{equation*}
	\begin{array}{ll}
		\text{Cov}_3(\bT^{b}(k_1),\bT^{b}(k_2))\\=\text{Cov}\Big(\dfrac{1}{G^{3/2}}\sum\limits_{t_2=k_2+1}^{k_2+G}e_{t_2}\sum\limits_{t_1=k_1-G+1}^{k_1}\bh(\bX_{t_1},\bX_{t_2}),\dfrac{1}{G^{3/2}}\sum\limits_{t'_2=k_2+1}^{k_2+G}e_{t_2'}\sum\limits_{t_1'=k_2-G+1}^{k_2}\bh(\bX_{t'_1},\bX_{t'_2})\Big)\\
		=\dfrac{1}{G^{3}}\sum\limits_{t_1=k_2+1}^{k_1+G}\sum\limits_{t_2=k_1-G+1}^{k_1}\sum\limits_{t_3=k_2-G+1}^{k_2}\bh(\bX_{t_1},\bX_{t_2})\bh(\bX_{t_1},\bX_{t_3}).
	\end{array}
\end{equation*}
In the next, we aim to prove that
\begin{equation*}
	\begin{array}{ll}
		\|	\text{Cov}_1(\bT^{b}(k_1),\bT^{b}(k_2))-\text{Cov}_1(\bT^{\bG}(k_1),\bT^{\bG}(k_2))\|_{\infty}=o_p(1),\\ \|	\text{Cov}_2(\bT^{b}(k_1),\bT^{b}(k_2))-\text{Cov}_2(\bT^{\bG}(k_1),\bT^{\bG}(k_2))\|_{\infty}=o_p(1),\\ \|	\text{Cov}_3(\bT^{b}(k_1),\bT^{b}(k_2))-\text{Cov}_3(\bT^{\bG}(k_1),\bT^{\bG}(k_2))\|_{\infty}=o_p(1).
	\end{array}
\end{equation*}
We first consider the first term. Note that by the definition of $\text{Cov}_1(\bT^{b}(k_1),\bT^{b}(k_2))$, there is overlap between $t_2$ and $t_3$. Specifically, let $T^0=\{t: k_2-G+1\leq t\leq k_1\}$, $T^1=\{t: k_1+1\leq t\leq k_2\}$, $T^2=\{t: k_2+1\leq t\leq k_1+G\}$ and $T^3=\{t: k_1+G+1\leq t\leq k_2+G\}$. Moreover, define $H(\bx_1,\bx_2,\bx_3)=\bh(\bx_1,\bx_2)\bh(\bx_1,\bx_3)^\top\in \RR^{d\times d}$, 
and define the new kernel $H'(\bX_{t_{1}}, \bX_{t_2}, \bX_{t_3}) = \sum_{\pi_3} \tilde{H}(\bX_{\pi_3(1)}, \bX_{\pi_3(2)}, \bX_{\pi_3(3)})$, where
\begin{equation*}
	\tilde{H}(\bX_{t_1},\bX_{t_2},\bX_{t_3})=\left\{
	\begin{array}{cc}
		H(\bX_{t_1},\bX_{t_2},\bX_{t_3})),& \text{if}~ t_1\in T^0, t_2\in T^1, t_3\in T^2\cup T^3\\
		H(\bX_{t_1},\bX_{t_2},\bX_{t_3})),& \text{if}~ t_1\in T^0, t_2\in T^2, t_3\in  T^3\\
		H(\bX_{t_1},\bX_{t_2},\bX_{t_3})),& \text{if}~ t_1\in T^0, t_2\neq t_3, t_2\in T^2, t_3\in T^2 \\
		0,& \text{else}.
	\end{array}
	\right.
\end{equation*}
and  $\pi_3$ \text{ is a permutation of } $t_1,t_2,t_3$. Let $H_2(\bx_1,\bx_2)=\bh(\bx_1,\bx_2)\bh(\bx_1,\bx_2)^\top\in \RR^{d\times d}$. Then, we can rewrite $\text{Cov}_1(\bT^{b}(k_1),\bT^{b}(k_2))$ as:
\begin{equation}\label{equ: case 3 large deviation 3}
	\begin{array}{ll}
		\text{Cov}_1(\bT^{b}(k_1),\bT^{b}(k_2))\\
		=\dfrac{1}{G^{3}}\sum\limits_{t_1=k_2-G+1}^{k_1}\sum\limits_{t_2=k_1+1}^{k_1+G}\sum\limits_{t_3=k_2+1}^{k_2+G}\bh(\bX_{t_1},\bX_{t_2})\bh(\bX_{t_1},\bX_{t_3})\\
		=\dfrac{1}{G^{3}}\sum\limits_{t_1\neq t_2 \neq t_3\in T^0\cup T^1\cup T^2\cup T^3}\tilde{H}(\bX_{t_1},\bX_{t_2},\bX_{t_3})+\dfrac{1}{G^{3}}\sum\limits_{t_1\in T^0,t_2\in T^2}H_2(\bX_{t_1},\bX_{t_2})\\
		=\underbrace{\dfrac{1}{6G^{3}}\sum\limits_{t_1\neq t_2\neq t_3\in T1\cup T2\cup T3}{H'}(\bX_{t_1},\bX_{t_2},\bX_{t_3})}_{\text{Cov}_{1,1}(\bT^{b}(k_1),\bT^{b}(k_2))}+\underbrace{\dfrac{1}{G^{3}}\sum\limits_{t_1\in T^0,t_2\in T^2}H_2(\bX_{t_1},\bX_{t_2})}_{\text{Cov}_{1,2}(\bT^{b}(k_1),\bT^{b}(k_2))}\\
	\end{array}
\end{equation}
Note that by the above definition, we have
\begin{equation}\label{equ: case 3 large deviation 4}
	\begin{array}{ll}
		\E[\text{Cov}_1(\bT^{b}(k_1),\bT^{b}(k_2))]&=\E[\text{Cov}_{1,1}(\bT^{b}(k_1),\bT^{b}(k_2))]+\E[\text{Cov}_{1,2}(\bT^{b}(k_1),\bT^{b}(k_2))]\\
		&=\dfrac{(G-(k_2-k_1))(G^2-(G-(k_2-k_1))}{G^3}\bGamma+\dfrac{(G-(k_2-k_1))}{G^3}\bSigma
	\end{array}
\end{equation}
where $\bSigma:=\text{Cov}(\bh(\bX_1,\bX_2)):=\E \bh(\bX_1,\bX_2)\bh(\bX_1,\bX_2)^\top$. 
Similar to the analysis of Case 3, we can prove that, with probability at least $1-(nd)^{-C_0}$ for some $C_0>2$,  it holds  uniformly for all $(k_1,k_2)$ with $0<k_2-k_1\leq G-1$ that 
\begin{equation}\label{equ: case 3 large deviation 5}
	\begin{array}{ll}
		\Big\|\text{Cov}_{1,1}(\bT^{b}(k_1),\bT^{b}(k_2))-\E	\text{Cov}_{1,1}(\bT^{b}(k_1),\bT^{b}(k_2))\Big\|_{\infty}\\
		\leq C_2D^2\max \Big(\sqrt{\dfrac{\log(nd)}{G}},\dfrac{\log^2(nd)\log^2(Gd)}{G}\Big).\\
	\end{array}
\end{equation}
for some large enough constant $C_2>0$. Moreover, by Lemma A.4 in \cite{yu2022robust}, we can prove that 
with probability at least $1-(nd)^{-C_0}$ for some $C_0>2$,  it holds  uniformly for all $(k_1,k_2)$ with $0<k_2-k_1\leq G-1$ that 
\begin{equation}\label{equ: case 3 large deviation 6}
	\begin{array}{ll}
		\Big\|\text{Cov}_{1,2}(\bT^{b}(k_1),\bT^{b}(k_2))-\E	\text{Cov}_{1,2}(\bT^{b}(k_1),\bT^{b}(k_2))\Big\|_{\infty}\\
		\leq C_2D^2\max \Big(\sqrt{\dfrac{\log(nd)}{G}},\dfrac{\log^2(nd)\log^2(Gd)}{G}\Big).\\
	\end{array}
\end{equation}
for some large enough constant $C_2>0$. Combining (\ref{equ: case 3 large deviation 1})-(\ref{equ: case 3 large deviation 6}), we have 
\begin{equation*}
	\begin{array}{ll}
		\|	\text{Cov}_1(\bT^{b}(k_1),\bT^{b}(k_2))-\text{Cov}_1(\bT^{\bG}(k_1),\bT^{\bG}(k_2))\|_{\infty}\\
		\leq 	\|	\text{Cov}_{1,1}(\bT^{b}(k_1),\bT^{b}(k_2))-\E	\text{Cov}_{1,1}(\bT^{b}(k_1),\bT^{b}(k_2))\|_{\infty}\\+	\|	\text{Cov}_{1,2}(\bT^{b}(k_1),\bT^{b}(k_2))-\E	\text{Cov}_{1,2}(\bT^{b}(k_1),\bT^{b}(k_2))\|_{\infty}\\+\|\E	\text{Cov}_{1,1}(\bT^{b}(k_1),\bT^{b}(k_2))+\E	\text{Cov}_{1,2}(\bT^{b}(k_1),\bT^{b}(k_2))-\dfrac{G-(k_2-k_1)}{G}\bGamma\|_{\infty}\\
		\leq C_2D^2\max \Big(\sqrt{\dfrac{\log(nd)}{G}},\dfrac{\log^2(nd)\log^2(Gd)}{G}\Big)\\+\|\dfrac{(G-(k_2-k_1))(G^2-(G-(k_2-k_1))}{G^3}\bGamma-\dfrac{G-(k_2-k_1)}{G}\bGamma\|_{\infty}+\|\dfrac{(G-(k_2-k_1))}{G^3}\bSigma\|_{\infty}\\
		\leq C_2D^2\max \Big(\sqrt{\dfrac{\log(nd)}{G}},\dfrac{\log^2(nd)\log^2(Gd)}{G}\Big)+o(\dfrac{1}{G^2})=o_p(1)
	\end{array}
\end{equation*}
where the last inequality uses the fact that $0<k_2-k_1\leq G-1$. 
Using a similar proof, we can prove that 
\begin{equation*}
	\begin{array}{ll}
		\|\text{Cov}_2(\bT^{b}(k_1),\bT^{b}(k_2))-\text{Cov}_2(\bT^{\bG}(k_1),\bT^{\bG}(k_2))\|_{\infty}=O_p(D^2\max \Big(\sqrt{\dfrac{\log(nd)}{G}},\dfrac{\log^2(nd)\log^2(Gd)}{G}\Big)),\\ \|	\text{Cov}_3(\bT^{b}(k_1),\bT^{b}(k_2))-\text{Cov}_3(\bT^{\bG}(k_1),\bT^{\bG}(k_2))\|_{\infty}=O_p(D^2\max \Big(\sqrt{\dfrac{\log(nd)}{G}},\dfrac{\log^2(nd)\log^2(Gd)}{G}\Big)).
	\end{array}
\end{equation*}
which implies 
\begin{equation*}
	\|\text{Cov}(\bT^{b}(k_1),\bT^{b}(k_2))-\text{Cov}(\bT^{\bG}(k_1),\bT^{\bG}(k_2))\|_{\infty}=O_p(D^2\max \Big(\sqrt{\dfrac{\log(nd)}{G}},\dfrac{\log^2(nd)\log^2(Gd)}{G}\Big)).
\end{equation*}
This completes the proof of Case 3. 

\textbf{\text{Case 4:}} $k_2-k_1=0$.  Using a similar analysis of Cases 2 and 3, we can prove that  with probability at least $1-(nd)^{-C_0}$ for some $C_0>2$,  it holds  uniformly for all $(k_1,k_2)$ with $k_1=k_2$ that 
\begin{equation*}
	\|\text{Cov}(\bT^{b}(k_1),\bT^{b}(k_2))-\text{Cov}(\bT^{\bG}(k_1),\bT^{\bG}(k_2))\|_{\infty}\leq C_2D^2\max \Big(\sqrt{\dfrac{\log(nd)}{G}},\dfrac{\log^2(nd)\log^2(Gd)}{G}\Big).
\end{equation*}
Lastly, using the results of Cases 1-4 and taking the union over $G\leq k_1,k_2\leq n-G$, we prove that with probability at least $1-(nd)^{-C_0}$ for some $C_0>0$, it holds that
\begin{equation*}
	\max_{k_1,k_2}\|\text{Cov}(\bT^{b}(k_1),\bT^{b}(k_2))-\text{Cov}(\bT^{\bG}(k_1),\bT^{\bG}(k_2))\|_{\infty}\leq C_2D^2\max \Big(\sqrt{\dfrac{\log(nd)}{G}},\dfrac{\log^2(nd)\log^2(Gd)}{G}\Big),
\end{equation*}
which completes the proof of Lemma \ref{lemma: large deviation}.

\subsection{Proof of Lemma \ref{lemma: upper bound of variance of bootstrap samples}}\label{sec: proof of upper bound of bootstrap samples}
\begin{proof}
	In this section, we prove Lemma \ref{lemma: upper bound of variance of bootstrap samples}. Recall $T^b_{j}(k)$ defined in $(\ref{equ: bootstrap based testing statistics for each coordinate})$. It has the following form:
	\begin{equation*}
		T_{j}^b(k)=	\dfrac{1}{G^{3/2}}\sum_{t_1=k-G}^{k}e_{t_1}\Big[\sum_{t_2=k+1}^{k+G}h(X_{t_1,j},X_{t_2,j})\Big]+\dfrac{1}{G^{3/2}}\sum_{t_2=k+1}^{k+G}e_{t_2}\Big[\sum_{t_1=k-G}^{k}h(X_{t_1,j},X_{t_2,j})\Big].
	\end{equation*} 
	Note that 	$\{e_t\}_{t=1}^n$ are independent $N(0,1)$ random variables. Hence,  for each fixed $j$ and $k=G,\ldots,n-G$, conditional on $\cX=\{\bX_1,\ldots,\bX_n\}$, we have
	\begin{equation}\label{equ: variance of bootstrap samples}
		\begin{array}{ll}
			&	\text{Var}(T_j^b(k))\\
			&=_{(1)}\dfrac{1}{G^3}\sum\limits_{t_1=k-G}^k\Big(\sum\limits_{t_2=k+1}^{k+G}h(X_{t_1,j},X_{t_2,j})\Big)^2+\sum\limits_{t_2=k+1}^{k+G}\Big(\sum\limits_{t_1=k-G}^{k}h(X_{t_1,j},X_{t_2,j})\Big)^2\\
			&	\leq_{(2)} \dfrac{1}{G^2}\sum\limits_{t_1=k-G}^k\sum\limits_{t_2=k+1}^{k+G}h^2(X_{t_1,j},X_{t_2,j})+\dfrac{1}{G^2}\sum\limits_{t_2=k+1}^{k+G}\sum\limits_{t_1=k-G}^kh^2(X_{t_1,j},X_{t_2,j})\\
			&=\dfrac{2}{G^2}\sum\limits_{t_1=k-G}^k\sum\limits_{t_2=k+1}^{k+G}h^2(X_{t_1,j},X_{t_2,j}),
		\end{array}
	\end{equation}
	where $(1)$ comes from the Cauchy-Schwarz inequality. Note that we consider the alternative that there may exist multiple change points. Hence, we can split the search location $k\in[G,n-G]$ into three types of groups:		
	\begin{equation}\label{equation: cT1-cTm-2}
		\begin{array}{ll}
			\cT_1=\{k: k\in[1,\gamma_{1}-G]\},			
			\cT_2=\{ k: k\in[\gamma_{1}+G,\gamma_{2}-G]\},
			\ldots, 
			\cT_{M_0+1}=\{k: k\in[\gamma_{M_0}+G,n]\}.
		\end{array}
	\end{equation}
	and 		
	\begin{equation}\label{equ: cT1'-cTm'-2}
		\begin{array}{ll}
			\cT'_{1,1}=\{k: k\in[\gamma_{1}-G,\gamma_{1}]\},~\cT'_{1,2}=\{k: k\in[\gamma_{1}+1,\gamma_{1}+G]\}\\
			\cT'_{2,1}=\{k: k\in[\gamma_{2}-G,\gamma_{2}]\},~\cT'_{2,2}=\{k: k\in[\gamma_{2}+1,\gamma_{2}+G]\}\\
			\qquad	\qquad \qquad\qquad\qquad\qquad\qquad\vdots\\
			\cT'_{M_0,1}=\{k: k\in[\gamma_{M_0}-G,\gamma_{M_0}]\},~\cT'_{M_0,2}=\{k: k\in[\gamma_{M_0}+1,\gamma_{M_0}+G]\}.
		\end{array}
	\end{equation}
	Based on the above groups, we consider three cases: \\
	\textbf{Case 1: }  
	For $k\in \cT_{m,2}'$ there is one change point in the interval $[k-G,k+G]$. Hence, we can decompose it into two terms:
	\begin{equation*}
		\sum\limits_{t_1=k-G}^k\sum\limits_{t_2=k+1}^{k+G}h^2(X_{t_1,j},X_{t_2,j})=\underbrace{\sum_{t_1=k-G}^{\gamma_m}\sum_{t_2=k+1}^{k+G}h^2(X_{t_1,j},X_{t_2,j})}_{V_{j,1}(k)}+\underbrace{\sum^{k}_{t_1=\gamma_m}\sum_{t_2=k+1}^{k+G}h^2(X_{t_1,j},X_{t_2,j})}_{V_{j,2}(k)}.
	\end{equation*}
	Next, we prove $\max_{1\leq j\leq d}V_{j,1}(k)$ and $\max_{1\leq j\leq d}V_{j,2}(k)$ are unifomly bounded. We first analyze $V_{j,1}(k)$.  Let $n_1(k)=\gamma_m+G-k$ and $n_2(k)=G$. Moreover, let $\tilde{t}_1=t_1-(k-G)+1$, $\tilde{t}_2=t_2-(k+1)+1$, 
	$\tilde{X}_{\tilde{t}_1,j}=X_{\tilde{t}_1,j}$ and $\tilde{Y}_{\tilde{t}_2,j}=X_{\tilde{t}_2,j}$. Hence, we have
	\begin{equation*}
		\begin{array}{ll}
			V_{j,1}(k)&=\sum\limits_{t_1=k-G}^{\gamma_m}\sum\limits_{t_2=k+1}^{k+G}h^2(X_{t_1,j},X_{t_2,j})\\
			&=\sum\limits_{\tilde{t}_1=1}^{n_1(k)}\sum\limits_{\tilde{t}_2=1}^{n_2(k)}h^2(\tilde{X}_{\tilde{t}_1,j},\tilde{Y}_{\tilde{t}_2,j})-\E[h^2(\tilde{X}_{\tilde{t}_1,j},\tilde{Y}_{\tilde{t}_2,j})]+\sum\limits_{\tilde{t}_1=1}^{n_1(k)}\sum\limits_{\tilde{t}_2=1}^{n_2(k)}\E[h^2(\tilde{X}_{\tilde{t}_1,j},\tilde{Y}_{\tilde{t}_2,j})].
		\end{array}
	\end{equation*}
	Note that for $k\in \cT'_{m,2}$, we have $n_1(k)\leq n_2(k)$. Let $$t^+=C_1B^2n_2(k)\max(\sqrt{n_1(k)\log(nd)},\log^2(nd)\log^2(n_1(k)d)),$$ for some big enough universal constant $C_1>0$, where $B:=D+\max_m\max_j|\theta_j^m|$ with $D$ defined in {Assumption B.2}. {By Assumptions B.1, B.2, and Lemma A.4 in \cite{yu2022robust} }, we can prove that 
	\begin{equation}\label{equ: upper bound of V(j,k)-1}
		\P\Big(\max_{1\leq j\leq d}\sum\limits_{\tilde{t}_1=1}^{n_1(k)}\sum\limits_{\tilde{t}_2=1}^{n_2(k)}h^2(\tilde{X}_{\tilde{t}_1,j},\tilde{Y}_{\tilde{t}_2,j})-\E[h^2(\tilde{X}_{\tilde{t}_1,j},\tilde{Y}_{\tilde{t}_2,j})]\geq 2t^+,k\in  \cT'_{m,2} \Big)\leq (nd)^{-C_2}
	\end{equation}
	for some $C_2>3$. Moreover, { by Assumption B.2}, we have
	\begin{equation}\label{equ: upper bound of V(j,k)-2}
		\begin{array}{ll}
			&\sum\limits_{\tilde{t}_1=1}^{n_1(k)}\sum\limits_{\tilde{t}_2=1}^{n_2(k)}\E[h^2(\tilde{X}_{\tilde{t}_1,j},\tilde{Y}_{\tilde{t}_2,j})]\\
			&=_{(1)}n_1(k)n_2(k)\E[h^2(\tilde{X}_{\tilde{t}_1,j},\tilde{Y}_{\tilde{t}_2,j})]\\
			&=_{(2)}n_1(k)n_2(k)\E[h^2({X}_{\gamma_m-1,j},{X}_{\gamma_m+1,j})]\\
			&=_{(3)}n_1(k)n_2(k)\E(h({X}_{\gamma_m-1,j},{X}_{\gamma_m+1,j})-\theta_j^m+\theta_j^m)^2\\
			&\leq_{(4)} n_1(k)n_2(k)(2D+2|\theta_j^m|^2)\leq 2n_1(k)n_2(k)B^2,
		\end{array}
	\end{equation}
	where $(4)$ comes from {Assumption B.2}.  Considering (\ref{equ: upper bound of V(j,k)-1}) and (\ref{equ: upper bound of V(j,k)-2}),  we have
	\begin{equation}\label{equ: upper bound of V(j,k)-3}
		\P\Big(\max_{1\leq j\leq d}V_{j,1}(k)\geq 2(t^++n_1(k)n_2(k)B^2),k\in  \cT'_{m,2} \Big)\leq (nd)^{-C_2}.
	\end{equation}
	For $V_{2,j}(k)$, using a similar analysis, by letting $n'_1(k)=k-\gamma_m$ and $n'_2(k)=G$, we can prove that 
	\begin{equation}\label{equ: upper bound of V(j,k)-3}
		\P\Big(\max_{1\leq j\leq d}V_{j,2}(k)\geq 2(t^{++}+n'_1(k)n'_2(k)B^2),k\in  \cT'_{m,2} \Big)\leq (nd)^{-C_2}.
	\end{equation}
	where 
	$$t^{++}=C_1B^2n'_2(k)\max(\sqrt{n'_1(k)\log(nd)},\log^2(nd)\log^2(n'_1(k)d)),$$ for some big enough universal constant $C_1>0$. Note that $n_1(k)\leq G$ and $n_1'(k)\leq G$ for $k\in  \cT'_{m,2} $. Let 
	$$t^{'}=4(C_1B^2G\max(\sqrt{G\log(nd)},\log^2(nd)\log^2(Gd))+G^2B^2).$$
	By the above bound, taking the union bound over $m=1,\ldots,M_0$, we have proved that 
	\begin{equation*}
		\P\Big(V_{j,1}(k)+V_{j,2}(k)\geq t',k\in  \cT'_{m,2},m=1,\ldots,M_0 \Big)\leq (nd)^{-C_2}.
	\end{equation*}
	for some $C_2>2$.\\
	\textbf{Case 2: }  Similar to the analysis of Case 1, we can prove that  
	\begin{equation*}
		\begin{array}{ll}
			\P\Big(\max_{1\leq j\leq d}	\sum\limits_{t_1=k-G}^k\sum\limits_{t_2=k+1}^{k+G}h^2(X_{t_1,j},X_{t_2,j})\\
			\leq 4(C_1B^2G\max(\sqrt{G\log(nd)},\log^2(nd)\log^2(Gd))+G^2B^2),k\in  \cT'_{m,1},m=1,\ldots,M_0\Big)\\\leq (nd)^{-C_2}
		\end{array}
	\end{equation*}
	for some $C_2>2$.\\
	\textbf{Case 3: }   Similar to the analysis of Case 1, we can prove that  
	\begin{equation*}
		\begin{array}{ll}
			\P\Big(\max_{1\leq j\leq d}	\sum\limits_{t_1=k-G}^k\sum\limits_{t_2=k+1}^{k+G}h^2(X_{t_1,j},X_{t_2,j})\\
			\leq 4(C_1B^2G\max(\sqrt{G\log(nd)},\log^2(nd)\log^2(Gd))+G^2B^2),k\in  \cT_{m},m=1,\ldots,M_0+1\Big)\\\leq (nd)^{-C_2}
		\end{array}
	\end{equation*}
	for some $C_2>2$.
	Lastly, considering  (\ref{equ: variance of bootstrap samples}) and the obtained upper bound of the three cases, we have with probability at least $1-3(nd)^{-C_2}$, it holds uniformly over $k=G,\ldots,n-G$ that 
	\begin{equation*}
		\begin{array}{ll}
			&\max_{1\leq j\leq d}\text{Var}(T_j^b(k))\\&\leq \dfrac{2}{G^2}\sum\limits_{t_1=k-G}^k\sum\limits_{t_2=k+1}^{k+G}h^2(X_{t_1,j},X_{t_2,j})\\
			&\leq_{(1)} 8\dfrac{(C_1B^2G\max(\sqrt{G\log(nd)},\log^2(nd)\log^2(Gd))+G^2B^2)}{G^2}\\
			&\leq_{(2)} 8(B^2+o(1))\leq 9B^2,
		\end{array}
	\end{equation*}
	where $(2)$ comes from the assumption that $G\gg \log^2(nd)$ and $\dfrac{\log^2(nd)\log^2(Gd)}{G}=o(1)$, which finishes the proof. 
\end{proof}
\subsection{Proof of Lemma \ref{lemma: maximum signal at the identified cpt}}\label{sec: proof of lemma at the biggest signal}
\begin{proof}
	In this section, we prove $\liminf_{n\rightarrow \infty}\theta^m_{j^*}/\|\btheta^m\|_{\infty}\geq 1$. 	We give the proof by contradiction. Suppose there is a constant $c<1$ such that 
	\begin{equation*}
		\theta^m_{j^*}\leq c\|\btheta^m\|_{\infty}.
	\end{equation*}
	On one hand, by the decomposition of $\bT(k) $ in (\ref{equ: decomposition T_{j,k}}) and (\ref{equ: Hoeffding T_{j,k}}), at time point $\hat{\gamma}_m$, we have:
	\begin{equation}\label{equ: upper bound of statistics at k}
		\begin{array}{ll}
			\|\bT(\hat{\gamma}_m)\|_\infty=T_{j^*}(\hat{\gamma}_m)\\
			\leq_{(1)} \dfrac{\gamma_m-\hat{\gamma}_m+G}{G^{1/2}}\theta_{j^*}^m+\Big|\dfrac{1}{\sqrt{G}}\sum\limits_{t_1=\hat{\gamma}_m-G}^{\gamma_{m}}h_{1,j^*}(X_{t_1,j^*})\Big|+\Big|\dfrac{\gamma_{m}-\hat{\gamma}_m+G}{G^{3/2}}\sum\limits_{t_2=\hat{\gamma}_m+1}^{\hat{\gamma}_m+G}h_{2,j^*}(X_{t_2,j^*})\Big|\\
			+\Big|\dfrac{1}{G^{3/2}}\sum\limits_{t_1=\hat{\gamma}_m-G}^{\gamma_m}\sum\limits_{t_2=\hat{\gamma}_m+1}^{\hat{\gamma}_m+G}g_{j^*}(X_{t_1,j^*},X_{t_2,j^*})\Big|+\Big|\dfrac{1}{G^{3/2}}\sum^{\hat{\gamma}_m}\limits_{t_1=\gamma_m}\sum\limits_{t_2=\hat{\gamma}_m+1}^{\hat{\gamma}_m+G}h(X_{t_1,j^*},X_{t_2,j^*})\Big|\\
			\leq_{(2)} \dfrac{\gamma_m-\hat{\gamma}_m+G}{G^{1/2}}c\|\btheta^m\|_{\infty}+G^{-1/2}\max\Big(O_p(\sqrt{(\gamma_m+G-\hat{\gamma}_m)\log(nd)}),O_p(\log(nd))\Big)\\+\dfrac{\gamma_{m}-\hat{\gamma}_m+G}{G^{3/2}}O_p(\sqrt{G\log(nd)})+o_p(\sqrt{\log(nd)})\\
			\leq_{(3)} c\sqrt{G}\|\btheta^m\|_{\infty}+O_p(\sqrt{\log(nd)})\\
			\leq_{(4)} c\sqrt{G}(1+o_p(1))\|\btheta^m\|_{\infty}
		\end{array}
	\end{equation}
	where $(2)$ comes from {Lemmas \ref{lemma: concentration for maxsimum sub-exponential} and \ref{lemma: Hoeffding's residual}}, $(3)$ comes from the fact that ${\gamma}-\hat{\gamma}_m+G\leq G$, $(4)$ comes from the assumption that $G\min_{1\leq m \leq M_0}\|\btheta^m\|^2_{\infty}\gg \log(nd)$. 
	
	On the other hand, 	at time point $\gamma_m$, we have:
	\begin{equation}\label{equ: lower bound of the statistics at the true cpt}
		\begin{array}{ll}
			\|\bT({\gamma}_m)\|_\infty= \max_{j}|T_{j}({\gamma}_m)|\\
			\geq_{(1)} \dfrac{G}{G^{1/2}}\|\btheta^m\|_{\infty}-\max_{j}\Big|\dfrac{1}{\sqrt{G}}\sum\limits_{t_1=\gamma_m-G}^{\gamma_{m}}h_{1,j}(X_{t_1,j})\Big|-\max_{j}\Big|\dfrac{G}{G^{3/2}}\sum\limits_{t_2=\gamma_m+1}^{\gamma_m+G}h_{2,j}(X_{t_2,j})\Big|\\
			+\max_{j}\Big|\dfrac{1}{G^{3/2}}\sum\limits_{t_1=\gamma_m-G}^{\gamma_m}\sum\limits_{t_2=\gamma_m+1}^{\gamma_m+G}g_j(X_{t_1,j},X_{t_2,j})\Big|\\
			\geq_{(2)} \sqrt{G}\|\btheta^m\|_{\infty}-O_p(\sqrt{\log(nd)})-o_p(\sqrt{\log(nd)})\\
			\geq_{(4)} \sqrt{G}(1-o_p(1))\|\btheta^m\|_{\infty}.
		\end{array}
	\end{equation}	
	Considering the above results, we have: $\P(\|\bT(\gamma_m)\|_{\infty}>\|\bT(\hat{\gamma}_m)\|_{\infty})\rightarrow 1$, which is contradicted to the fact that $\hat{\gamma}_m$ is the maximizer of $\|\bT(k)\|_{\infty}$. This completes the proof. 
\end{proof}

\subsection{Proof of Lemma \ref{lemma: sign consistent}}\label{sec: proof of lemma of the sign consistency}
\begin{proof}
	
	{\bf{Proof of $\cH_1$}}. Without loss of generality, we assume $\hat{\gamma}_m\geq \gamma_m$. The proof proceeds in two steps. In Step 1, we prove that 
	\begin{equation}
		\big|\max_{j\in\{1,\ldots,d\}}T_j(\hat{\gamma}_m)\big|\geq \big|\min_{j\in\{1,\ldots,d\}}T_j(\hat{\gamma}_m)\big|
	\end{equation}
	By noting that 	$$\max_{j\in\{1,\ldots,d\}}|T_j(\hat{\gamma}_m)|=
	\big|\max_{j\in\{1,\ldots,d\}}T_j(\hat{\gamma}_m)\big|\vee \big|\min_{j\in\{1,\ldots,d\}}T_j(\hat{\gamma}_m)\big|$$
	we have $\max_{j\in\{1,\ldots,d\}}|T_j(\hat{\gamma}_m)|=
	\big|\max_{j\in\{1,\ldots,d\}}T_j(\hat{\gamma}_m)\big|$. 
	In Step 2, we prove $\max_{j\in\{1,\ldots,d\}}T_j(\hat{\gamma}_m)\geq 0$. Note that by definition, $T_{j^*}(\hat{\gamma}_m)=|\max\limits_{j\in\{1,\ldots,d\}}T_j(\hat{\gamma}_m)|$. Combining Steps 1 and 2, we complete the proof. 
	
	By the decomposition of $T_{j}(k)$ at $k=\hat{\gamma}_m$ with $\hat{\gamma}_m\geq \gamma_m$  in (\ref{equ: decomposition T_{j,k}}) and (\ref{equ: Hoeffding T_{j,k}}), at time point $\hat{\gamma}_m$, we have:
	\begin{equation*}
		\begin{array}{ll}
			T_{j}(\hat{\gamma}_m)=_{(1)} \dfrac{\gamma_m-\hat{\gamma}_m+G}{G^{1/2}}\theta_{j}^m+\dfrac{1}{\sqrt{G}}\sum\limits_{t_1=\hat{\gamma}_m-G}^{\gamma_{m}}h_{1,j}(X_{t_1,j})+\dfrac{\gamma_{m}-\hat{\gamma}_m+G}{G^{3/2}}\sum\limits_{t_2=\hat{\gamma}_m+1}^{\hat{\gamma}_m+G}h_{2,j}(X_{t_2,j})\\
			+\dfrac{1}{G^{3/2}}\sum\limits_{t_1=\hat{\gamma}_m-G}^{\gamma_m}\sum\limits_{t_2=\hat{\gamma}_m+1}^{\hat{\gamma}_m+G}g_{j}(X_{t_1,j},X_{t_2,j})+\dfrac{1}{G^{3/2}}\sum^{\hat{\gamma}_m}\limits_{t_1=\gamma_m}\sum\limits_{t_2=\hat{\gamma}_m+1}^{\hat{\gamma}_m+G}h(X_{t_1,j},X_{t_2,j})\\
		\end{array}
	\end{equation*}
	Note that by {Lemma \ref{lemma: big error bound of initial estimators}}, we have $|\hat{\gamma}_m-\gamma_m|\leq G/2$. Moreover, by Lemma \ref{lemma: concentration for maxsimum sub-exponential}, we have  with probability at least $1-C_1(nd)^{-C_2}$, we have
	\begin{equation*}
		\begin{array}{ll}
			&	\max\limits_{1\leq j \leq d}|\dfrac{1}{\sqrt{G}}\sum\limits_{t_1=\hat{\gamma}_m-G}^{\gamma_{m}}h_{1,j}(X_{t_1,j})|\\
			&\leq_{(1)} C_2\dfrac{1}{G^{1/2}}\max(\sqrt{(\gamma_m-\hat{\gamma}_m+G)\log nd},\log(nd))\\
			&\leq_{(2)} C_2\dfrac{1}{G^{1/2}}\max(\sqrt{G\log nd},\log(nd))\\
			&\leq_{(3)}  C_2\dfrac{1}{G^{1/2}}\sqrt{G\log nd}=C_3\sqrt{\log(nd)}.
		\end{array}
	\end{equation*}
	where $(1)$ comes from Lemma \ref{lemma: concentration for maxsimum sub-exponential}, $(2)$ is due to the fact that $|\hat{\gamma}_m-\gamma_m|\leq G/2$, and $(3)$ uses the assumption that $G\gg \log^2(nd)$. Simiarly, we have 
	\begin{equation*}
		\begin{array}{ll}
			&	\max\limits_{1\leq j\leq d}	\Big|\dfrac{\gamma_{m}-\hat{\gamma}_m+G}{G^{3/2}}\sum\limits_{t_2=\hat{\gamma}_m+1}^{\hat{\gamma}_m+G}h_{2,j}(X_{t_2,j})\Big|\\
			&\leq C_2\dfrac{\gamma_{m}-\hat{\gamma}_m+G}{G^{3/2}}\max(\sqrt{G\log(nd)},\log(nd))\\
			&\leq C_2\dfrac{G}{G^{3/2}}\sqrt{G\log(nd)}\leq C_3\sqrt{\log(nd)}
		\end{array}
	\end{equation*}
	Moreover, using Lemma \ref{lemma: Hoeffding's residual}, with probability at least $1-C_1(nd)^{-C_2}$, we have
	\begin{equation*}
		\begin{array}{ll}
			\max\limits_{1\leq j\leq d}\Big|\dfrac{1}{G^{3/2}}\sum\limits_{t_1=\hat{\gamma}_m-G}^{\gamma_m}\sum\limits_{t_2=\hat{\gamma}_m+1}^{\hat{\gamma}_m+G}g_{j}(X_{t_1,j},X_{t_2,j})\Big|\\
			\leq_{(1)} C_2\dfrac{1}{G^{3/2}}\log(nd)\sqrt{\gamma_{m}-\hat{\gamma}_m+G}\sqrt{G}\Big(\sigma^2+\sqrt{\dfrac{\log^2(nd)}{\gamma_{m}-\hat{\gamma}_m+G}}+\sqrt{\dfrac{\log^2(nd)}{G}}\Big)\\
			\leq_{(2)} C_2\log(nd)\dfrac{G}{G^{3/2}}\Big(\sigma^2+\sqrt{\dfrac{\log^2(nd)}{G/2}}+\sqrt{\dfrac{\log^2(nd)}{G}}\Big)\\
			\leq_{(3)} C_3\dfrac{\log(nd)}{G^{1/2}}\ll_{(4)} \sqrt{\log(nd)}
		\end{array}
	\end{equation*}
	where $(1)$ comes from Lemma \ref{lemma: Hoeffding's residual}, $(2)$ comes from the fact that $|\hat{\gamma}_m-\gamma_m|\leq G/2$, and $(3)$ and $(4)$ uses the assumption that $G\gg \log^2(nd)$. Lastly, using a similar proof, we can prove that 
	\begin{equation*}
		\begin{array}{ll}
			\max\limits_{1\leq j \leq d}\Big|	\dfrac{1}{G^{3/2}}\sum^{\hat{\gamma}_m}\limits_{t_1=\gamma_m}\sum\limits_{t_2=\hat{\gamma}_m+1}^{\hat{\gamma}_m+G}h(X_{t_1,j},X_{t_2,j})\Big|=O_p(\sqrt(\log(nd))).
		\end{array}
	\end{equation*}
	The above results impliy that with probability at least $1-C_1(nd)^{-C_2}$, it holds that 
	\begin{equation*}
		\max_{1\leq j\leq d}\Big|T_{j}(\hat{\gamma}_m)-\dfrac{\gamma_m-\hat{\gamma}_m+G}{G^{1/2}}\theta_{j}^m\Big|\leq C^*\sqrt{\log(nd)}:=K^*
	\end{equation*}
	for some large enough $C^*>0$.  Note that for any sequences $\{a_i\}$ and $\{b_i\}$, we have
	$|\max (a_i)-\max (b_i)|\leq \max_{i}|a_i-b_i|$. Hence, it holds that 
	\begin{equation}\label{equ: lower bound of max}
		\begin{array}{ll}
			&\max_{1\leq j\leq d}T_{j}(\hat{\gamma}_m)\\
			&\geq \max_{1\leq j\leq d}\dfrac{\gamma_m-\hat{\gamma}_m+G}{G^{1/2}}\theta_{j}^m-K^*\\
			&\geq_{(1)} \dfrac{G/2}{G^{1/2}}\|\btheta^{(m)}\|_{\infty}-K^*\\
			&\geq_{(2)} K^*
		\end{array}
	\end{equation}
	where $(1)$ comes from the fact that $|\hat{\gamma}_m-\gamma_m|\leq G/2$ and $(2)$ uses the assumption that $G\|\btheta^{(m)}\|^2_{\infty}\gg \log(nd)$.  Moreover, note that for any sequences $\{a_i\}$ and $\{b_i\}$, we have 
	\begin{equation*}
		\min(a_i)-\min (b_i)\geq \min (a_i-b_i)\geq \min(-|a_i-b_i|)=-\max (|a_i-b_i|).
	\end{equation*}
	Let $T_{j}(\hat{\gamma}_m)=a_j$ and $\dfrac{\gamma_m-\hat{\gamma}_m+G}{G^{1/2}}\theta_{j}^m=b_j$, we have
	\begin{equation}\label{equ: lower bound of min}
		\min_{1\leq j\leq d}T_{j}(\hat{\gamma}_m)\geq -K^*.
	\end{equation}
	Combining (\ref{equ: lower bound of max}) and (\ref{equ: lower bound of min}), we have $\big|\max_{j\in\{1,\ldots,d\}}T_j(\hat{\gamma}_m)\big|\geq \big|\min_{j\in\{1,\ldots,d\}}T_j(\hat{\gamma}_m)\big|$. Moreover, (\ref{equ: lower bound of max}) further implies that $\max_{1\leq j\leq d}T_{j}(\hat{\gamma}_m)\geq K^*>0$, which completes the proof.

	Proof of $\cH_2$. Note that the proof of $\cH_2$ is similar and easier, to save space, we omit the details.

\end{proof}

\subsection{Proof of Lemma \ref{lemma: big error bound of initial estimators}}\label{sec: proof of lemma of intial estimator}
\begin{proof}
	Recall $\hat{\gamma}_m$ defined in (\ref{equ: definition of gamma-hat}). We give the proof by contradiction. Specifically, we will prove that  with probability at least $1-C_1(nd)^{-C_2}$, it holds that 
	\begin{equation*}
		\|\bT(\gamma_m)\|_{\infty}> \sup_{k: |k-\gamma_m|\geq G/2}\|\bT(k)\|_{\infty}. 
	\end{equation*}
	Note that $\hat{\gamma}_m$ is the maximizer of $\|\cT(k)\|_{\infty}$, we must have $|\hat{\gamma}_m-\gamma_m|\leq G/2$, which finishes the proof.  Without loss of generality, we assume $\hat{\gamma}_m\geq \gamma_m$. On one hand, by (\ref{equ: lower bound of the statistics at the true cpt}), at the true change point $\gamma_m$, we have
	\begin{equation}\label{equ: lower bound of statistics at gamma-m}
		\|\bT({\gamma}_m)\|_\infty\geq \sqrt{G}(1-o_p(1))\|\btheta^m\|_{\infty}
	\end{equation}
	On the other hand, similar to the analysis in (\ref{equ: upper bound of statistics at k}),  uniformly in   $k\in[\gamma_m+G/2,\gamma_m+G]$, with probability at least $1-C_1(nd)^{-C_2}$, we have
	\begin{equation}\label{equ: upper bound of statistics at kk}
		\begin{array}{ll}
			\max_{1\leq j\leq d}|T_{j}(k)|\\
			\leq_{(1)} \max\limits_{1\leq j\leq d}|\dfrac{\gamma_m-k+G}{G^{1/2}}|\theta_{j}^m|+	 \max\limits_{1\leq j\leq d}\Big|\dfrac{1}{\sqrt{G}}\sum\limits_{t_1=k-G}^{\gamma_{m}}h_{1,j}(X_{t_1,j})\Big|\\+	 \max\limits_{1\leq j\leq d}\Big|\dfrac{\gamma_{m}-k+G}{G^{3/2}}\sum\limits_{t_2=k+1}^{k+G}h_{2,j}(X_{t_2,j})\Big|\\
			+	 \max\limits_{1\leq j\leq d}\Big|\dfrac{1}{G^{3/2}}\sum\limits_{t_1=k-G}^{\gamma_m}\sum\limits_{t_2=k+1}^{k+G}g_{j}(X_{t_1,j},X_{t_2,j})\Big|\\+	 \max\limits_{1\leq j\leq d}\Big|\dfrac{1}{G^{3/2}}\sum^{k}\limits_{t_1=\gamma_m}\sum\limits_{t_2=k+1}^{k+G}h(X_{t_1,j},X_{t_2,j})\Big|\\
			\leq_{(2)} \dfrac{\sqrt{G}}{2}\|\btheta^m\|_{\infty}+C_3G^{-1/2}\max\Big(\sqrt{(\gamma_m+G-k)\log(nd)},\log(nd)\Big)\\+\dfrac{\gamma_{m}-k+G}{G^{3/2}}\max\Big(\sqrt{G\log(nd)},\log(nd)\Big)+o_p(\sqrt{\log(nd)})\\
			\leq_{(3)} \dfrac{\sqrt{G}}{2}\|\btheta^m\|_{\infty}+C_3G^{-1/2}\max\Big(\sqrt{\dfrac{G}{2}\log(nd)},\log(nd)\Big)+\dfrac{G/2}{G^{3/2}}\sqrt{\dfrac{G}{2}\log(nd)}+o(\log(nd))\\
			\leq_{(4)}\dfrac{\sqrt{G}}{2}\|\btheta^m\|_{\infty}+C_4\sqrt{\log(nd)}
		\end{array}
	\end{equation}
	where $(2)$ comes from Lemmas \ref{lemma: concentration for maxsimum sub-exponential1} and \ref{lemma: Hoeffding's residual}, and the fact that $k\in[\gamma_m+G/2,\gamma_m+G]$. This implies that 
	\begin{equation*}
		\max_{k\in[\gamma_m+G/2,\gamma_m+G]}\max_{1\leq j\leq d}|T_{j}(k)|\leq \dfrac{\sqrt{G}}{2}\|\btheta^m\|_{\infty}+O_p(\sqrt{\log(nd)}).
	\end{equation*}
	Considering (\ref{equ: lower bound of statistics at gamma-m}), we have 
	\begin{equation*}
		\begin{array}{ll}
			&	\|\bT(\gamma_m)\|_{\infty}-\sup_{k: |k-\gamma_m|\geq G/2}\|\bT(k)\|_{\infty}\\
			&\geq \sqrt{G}(1-o_p(1))\|\btheta^m\|_{\infty}-\dfrac{\sqrt{G}}{2}\|\btheta^m\|_{\infty}-O_p(\sqrt{\log(nd)})\\
			&\geq \sqrt{G}(1/2-o_p(1))\|\btheta^m\|_{\infty}-O_p(\sqrt{\log(nd)})\\
			&>0
		\end{array}
	\end{equation*}
	which finishes the proof.
\end{proof}

\end{document}